\documentclass[11pt]{article}

\usepackage{indentfirst}
\usepackage{times}
\usepackage{amsmath}
\usepackage{amsthm}
\usepackage{color}
\usepackage{eufrak}
\usepackage{verbatim}
\usepackage[makeroom]{cancel}
\usepackage{rotating}
\usepackage{float}
\usepackage{fullpage}
\usepackage{amssymb}
\usepackage{amsthm}
\usepackage{mathrsfs}
\usepackage{epsfig}
\usepackage{graphicx}
\usepackage{indentfirst}
\usepackage{framed}
\usepackage[numbers]{natbib}
\usepackage{color}
\usepackage{hyperref}
\usepackage{graphicx}
\usepackage{float}
\usepackage{bm}
\usepackage[capitalize,noabbrev,nameinlink]{cleveref}
\usepackage{enumitem}
\usepackage{mathtools}
\usepackage{mleftright}
\usepackage{thmtools} 
\usepackage{relsize}\usepackage{scalerel}
\usepackage{tikz}
\usepackage{multicol}
\usetikzlibrary{positioning,shapes,shadows,arrows}
\usepackage[labelfont=bf]{caption}
\usepackage{algorithm}
\usepackage[noend]{algpseudocode}
\usepackage{tablefootnote}

\newcommand{\numberthis}{\addtocounter{equation}{1}\tag{\theequation}}

\definecolor{corlinks}{RGB}{0,0,200}
\definecolor{cormenu}{RGB}{0,0,200}
\definecolor{corurl}{RGB}{0,0,200}

\hypersetup{
	colorlinks=true,
	urlcolor=corlinks,
	linkcolor=corlinks,
	menucolor=cormenu,
	citecolor=corlinks,
	pdfborder= 0 0 0
}

\definecolor{blue-violet}{rgb}{0.54, 0.17, 0.89}
\definecolor{deep-green}{rgb}{0,0.5,0}

\newtheorem{theorem}{Theorem}
\newtheorem{lemma}[theorem]{Lemma}

\newtheorem{claim}[theorem]{Claim}

\newtheorem{corollary}[theorem]{Corollary}
\newtheorem{proposition}[theorem]{Proposition}

\theoremstyle{definition}

\newtheorem{definition}[theorem]{Definition}

\newtheorem{problem}[theorem]{Problem}

\RenewCommandCopy{\theHtheorem}{\thetheorem} %

\makeatletter
\newenvironment{breakablealgorithm}
  {%
   \begin{center}
     \refstepcounter{algorithm}%
     \hrule height.8pt depth0pt \kern2pt%
     \renewcommand{\caption}[2][\relax]{%
       {\raggedright\textbf{\ALG@name~\thealgorithm} ##2\par}%
       \ifx\relax##1\relax %
         \addcontentsline{loa}{algorithm}{\protect\numberline{\thealgorithm}##2}%
       \else %
         \addcontentsline{loa}{algorithm}{\protect\numberline{\thealgorithm}##1}%
       \fi
       \kern2pt\hrule\kern2pt
     }
  }{%
     \kern2pt\hrule\relax%
   \end{center}
  }
\makeatother

\newcommand{\D}{\mathcal{D}}

\let\phi\varphi

\newcommand{\binset}{\{0,1\}}
\newcommand{\bin}{\binset}
\newcommand{\poly}{\mathsf{poly}}

\newcommand{\SIZE}{\mathsf{SIZE}}

\DeclareMathOperator*\Exp{{\bf E}}
\DeclareMathOperator*\Prob{{\bf Pr}}
\newcommand{\SAMP}{\mathsf{SAMP}}
\renewcommand{\tt}{\mathsf{tt}}

\newcommand{\bool}{\mleft\{0,1\mright\}}
\newcommand{\Kt}{\mathsf{Kt}}

\newcommand{\N}{\mathbb{N}}

\newcommand{\nat}{\N}
\newcommand{\rKt}{\mathsf{rKt}}
\newcommand{\pKt}{\mathsf{pKt}}

\newcommand{\rK}{\mathsf{rK}}
\newcommand{\pK}{\mathsf{pK}}
\newcommand{\K}{\mathsf{K}}
\renewcommand{\P}{\mathsf{P}}
\newcommand{\NP}{\mathsf{NP}}
\newcommand{\coNP}{\mathsf{coNP}}
\newcommand{\PH}{\mathsf{PH}}
\newcommand{\BPP}{\mathsf{BPP}}
\newcommand{\PSPACE}{\mathsf{PSPACE}}

\newcommand{\NEXP}{\mathsf{NEXP}}
\newcommand{\ZPP}{\mathsf{ZPP}}
\newcommand{\BPE}{\mathsf{BPE}}

\newcommand{\BPEXP}{\mathsf{BPEXP}}
\newcommand{\EE}{\mathsf{EE}}

\newcommand{\polylog}{\mathsf{polylog}}
\newcommand{\co}{\mathsf{co}}

\newcommand{\nKt}{\mathsf{nKt}}
\newcommand{\nK}{\mathsf{nK}}
\newcommand{\rnKt}{\mathsf{rnKt}}
\newcommand{\rnK}{\mathsf{rnK}}
\newcommand{\pnKt}{\mathsf{pnKt}}
\newcommand{\pnK}{\mathsf{pnK}}

\newcommand{\eqdef}{\triangleq}

\newcommand{\MnKtP}{\mathsf{MnKtP}}

\newcommand{\NTIME}{\mathsf{NTIME}}
\newcommand{\coNTIME}{\mathsf{coNTIME}}
\newcommand{\coAM}{\mathsf{coAM}}
\newcommand{\AM}{\mathsf{AM}}
\newcommand{\coAMTIME}{\mathsf{coAMTIME}}
\newcommand{\AMTIME}{\mathsf{AMTIME}}
\newcommand{\AMEXP}{\mathsf{AMEXP}}
\newcommand{\coAMEXP}{\mathsf{coAMEXP}}

\newcommand{\MKtP}{\mathsf{MKtP}}

\newcommand{\gapmrnktp}{\mathsf{Gap}\textnormal{-}\mathsf{MrnKtP}}
\newcommand{\gapmnktp}{\mathsf{Gap}\textnormal{-}\mathsf{MnKtP}}

\newcommand{\EXP}{\mathsf{EXP}}
\renewcommand{\P}{\mathsf{P}}

\newcommand{\pr}{\mathsf{pr}\text{-}}

\newcommand{\SAT}{\mathsf{SAT}}

\newcommand{\BPTIME}{\mathsf{BPTIME}}
\newcommand{\U}{\mathcal{U}}
\newcommand{\E}{\mathsf{E}}
\newcommand{\Few}{\mathsf{Few}}

\begin{document}

\title{Asymmetry and Complexity of Nondeterministic Computations\vspace{0.4cm}}

\author{
Halley Goldberg\thanks{University of Warwick. \texttt{E-mail:~halley.goldberg@warwick.ac.uk}}
\and
Jinqiao Hu\thanks{University of Warwick. \texttt{E-mail:~jinqiao.hu@warwick.ac.uk}}
\and
Zhenjian Lu\thanks{University of Victoria. \texttt{E-mail:~zhenjianlu@uvic.ca}}
\and
Jingyi Lyu\thanks{University of Michigan. \texttt{E-mail:~jylyu@umich.edu}}
\and
Igor C. Oliveira\vspace{0.6cm}\thanks{University of Warwick. \texttt{E-mail:~igor.oliveira@warwick.ac.uk}}
}

\date{October 13, 2026}
\maketitle
	
\vspace{-0.7cm}

\begin{abstract}
Symmetry of Information (SoI) \citep{zvonkin1970complexity} is a fundamental theorem of classical Kolmogorov complexity  whose status for time-bounded computation is closely related to central open problems in complexity theory and cryptography. In a recent development, \citep{HMO26} showed that SoI fails for $\rKt$ complexity, namely in the setting of randomized computation. By contrast, many basic properties of Kolmogorov complexity are known to remain valid in the presence of nondeterminism \citep{HLO25}, suggesting that SoI may provide a fruitful structural principle for exploring the relative power of randomness and nondeterminism in computation.

Our results reveal that the validity of SoI for nondeterministic computation, namely for $\nKt$ complexity, is a subtle question connected to several frontiers. We show that symmetry of information holds on average for randomized extensions of $\nKt$, while its worst case validity would imply breakthrough complexity separations. Although the standard cryptographic argument for asymmetry does not extend to the nondeterministic setting, we provide evidence that information remains asymmetric in nondeterministic computations. Among other results, we show that SoI would imply a collapse of exponential time classes to low levels of the polynomial hierarchy.

Then, building on the work of \citep{DBLP:conf/coco/Hirahara22, GK22, DBLP:journals/eccc/KK25, HMO26}, we develop a general framework connecting symmetry of information, meta-complexity, and explicit constructions in the nondeterministic setting. More precisely, we prove that suitable formulations of these three notions are quantitatively equivalent. We apply different aspects of this correspondence to make progress on complexity lower bounds and on the failure of SoI for nondeterministic computation. These consequences are unconditional and provide some of the strongest known lower bounds in meta-complexity.
\end{abstract}

\newpage

\normalsize

\newpage

\small

\tableofcontents

\normalsize

\newpage

\section{Introduction}

\subsection{Overview}

\subsubsection{Context and Motivation}\label{ss:context}

An overarching goal in theoretical computer science is to understand the relative power of deterministic, randomized, and nondeterministic computation. In addition to comparing the  problems that these resources can solve, one can seek structural principles that hold in one computational setting but fail in another. Such principles offer a different lens on the separation and collapse questions at the heart of complexity theory.

The structural principle at the center of this work is \emph{Symmetry of Information} (SoI), a fundamental theorem of algorithmic information theory \citep{zvonkin1970complexity}. This result states that for all $n$-bit strings $x$ and $y$,
\vspace{0.1cm}
$$
    \K(x,y)
    \;\;=\;\; \K(x) + \K(y \mid x) \pm O(\log n)
    \;\;=\;\; \K(y) + \K(x \mid y) \pm O(\log n),
$$

\vspace{0.2cm}
\noindent where $\K(z)$ denotes the Kolmogorov complexity of $z$. Equivalently, an optimal description of $(x,y)$ can be decomposed into an optimal description of one string and an optimal description of the other conditioned on the first. In this sense, SoI is a powerful structural principle: describing two objects separately, while allowing one to serve as side information for the other, is essentially as efficient as describing them jointly. Another fundamental consequence of these equations is that the information that $x$ contains about $y$, namely $\K(y) - \K(y \mid x)$, is roughly the same as the information that $y$ contains about $x$, namely $\K(x) - \K(x \mid y)$.\footnote{This is analogous to the corresponding phenomenon in Shannon's information theory, where the mutual information between two random variables $X$ and $Y$ satisfies $I(X;Y) = H(X) - H(X \mid Y) = H(Y) - H(Y \mid X)$.} In this sense, information is symmetric when considering \emph{time-unbounded} computations.

The classical proof of SoI relies on exhaustive search and does not extend to time-bounded computations. Indeed, the status of SoI for time-bounded Kolmogorov complexity is intertwined with fundamental questions about the power of efficient computation. Its validity for polynomial-time measures is connected to the $\P$ versus $\NP$ problem \citep{DBLP:journals/iandc/LongpreW95, DBLP:conf/coco/Hirahara22}, the existence of one-way functions \citep{DBLP:conf/stoc/HiraharaILNO23, DBLP:conf/focs/HiraharaLN24}, and the feasibility of learning \citep{DBLP:conf/focs/HiraharaN23, DBLP:journals/eccc/HN25}. Collectively, these connections provide strong evidence that SoI fails for polynomial-time measures such as $\K^t$, $\rK^t$, and $\pK^t$, although proving this unconditionally would resolve longstanding open problems.

The situation is better understood for Levin-style measures, which charge only logarithmically for decoding time and therefore permit exponentially long computations.\footnote{Levin's $\Kt$ complexity of a string $x$ is defined as $\Kt(x) = \min_{\Pi,t}{|\Pi|+\lceil\log t\rceil \colon \Pi \text{ outputs } x \text{ within } t \text{ steps}}$,with respect to a fixed universal machine. Thus, $\Kt$ accounts for both description length and decoding time.} Ronneburger \citep{RonneburgerThesis} proved unconditionally that SoI fails for Levin's $\Kt$ complexity, and more recently \citep{HMO26} established the same conclusion for its randomized analogue $\rKt$. Thus, in both the deterministic and randomized settings, efficient descriptions violate a central structural principle of classical Kolmogorov complexity. We refer the reader to \citep{HMO26} for applications of this phenomenon to complexity lower bounds and cryptography.

In this work, we investigate what happens when descriptions are allowed to use \emph{nondeterminism}. Our main focus in on the nondeterministic measure $\nKt$ \citep{DBLP:conf/coco/AllenderKR03} and its probabilistic  extensions $\rnKt$ and $\pnKt$.\footnote{Under plausible derandomization assumptions, it is not hard to show that $\pnKt \approx \rnKt \approx \nKt$, just as $\pKt \approx \rKt \approx \Kt$. We refer to the Appendix of \citep{GKLO22} for similar results.} Informally, $\nKt$ is the nondeterministic analogue of Levin's $\Kt$. A short nondeterministic program for $x$ must uniquely specify $x$: every computation path outputs either $x$ or a failure symbol, and at least one path outputs $x$. Nondeterminism may therefore yield more succinct descriptions, even though deterministically recovering the encoded string can require searching over the program's choices.

We say that SoI holds for $\nKt$ with error $e(n)$ if, for all $x,y \in \{0,1\}^n$,
$$
    \nKt(x,y) \geq \nKt(x) + \nKt(y \mid x) - e(n).
$$
\noindent If this inequality holds with $e(n)=O(\log n)$, we simply say that SoI holds for $\nKt$.\footnote{The reverse inequality holds for all the complexity measures considered here under natural assumptions on the computational model.} We use the analogous terminology for $\rnKt$ and $\pnKt$.

There is a basic reason why this question is particularly subtle for nondeterministic computation. For deterministic and randomized time-bounded measures, the standard evidence against SoI comes from the asymmetry created by cryptographic hardness. If $f$ is a one-way function, then $y = f(x)$ is easy to compute from $x$, whereas recovering $x$ from $f(x)$ is hard on average. A joint description can exploit the easy forward computation, while a decomposition in the opposite order encounters the hard inverse problem. This intuition does not transfer directly to nondeterministic descriptions: even if $\nKt(x \mid f(x))=O(\log n)$, such a short description gives no evident way to find a successful computation path and hence no efficient procedure for inverting $f$.\footnote{We mention that it is also unclear whether the super-bit conjecture of \citep{DBLP:conf/random/Rudich97} implies the failure of SoI for $\nKt$.}

At the same time, nondeterministic Kolmogorov complexity retains several features of the classical, time-unbounded theory that are conjectured to fail for deterministic and randomized time-bounded measures, including forms of language compression and conditional coding \citep{HLO25}. This suggests that nondeterministic Kolmogorov complexity may be structurally closer to its time-unbounded counterpart. If SoI were also preserved, it would provide an intrinsic distinction between randomized and nondeterministic computation and could point toward a new route for separating these resources. If it fails, then its failure must be explained by a source of asymmetry beyond the standard cryptographic argument.

These considerations lead to the following basic question:
\vspace{0.3cm}
\begin{center}
\emph{Does Symmetry of Information hold for nondeterministic computation?}
\end{center}
\vspace{0.3cm}

Despite substantial progress on nondeterministic Kolmogorov complexity and its applications \citep{DBLP:conf/coco/AllenderKR03, DBLP:journals/cc/BuhrmanLM05, DBLP:conf/cie/Allender10, DBLP:journals/jcss/AllenderKRR11, HLO25}, this question has remained elusive. More broadly, studying SoI in the nondeterministic setting offers a possible route to new complexity lower bounds against nondeterministic computation, continuing a line of work that uses information symmetry as a lens on complexity lower bounds \citep{DBLP:conf/coco/Hirahara22, HMO26}.

\subsubsection{Summary of Contributions}

We develop a general theory for symmetry of information for nondeterministic computation, connecting this problem to questions in complexity lower bounds and explicit constructions. We then explore different aspects of this theory to obtain new unconditional results concerning meta-complexity and asymmetry of information.
 Our main contributions can be summarized as follows:\footnote{A preliminary version of a subset of these results appeared as part of the unpublished ECCC preprint~\citep{ECCC2025TR215}.}

\begin{enumerate}
    \item First, we highlight a significant distinction  between nondeterministic and randomized computations: while the existence of cryptographic one-way functions implies that SoI fails \emph{on average} for $\rKt$, we observe that SoI holds \emph{on average} for $\rnKt$. In particular, it follows that the standard cryptographic argument does not yield failure of SoI for nondeterministic computation.
    
    \item Formalizing the intuition that SoI might be useful for separating complexity classes, we prove that if SoI holds for $\rnKt$ in the \emph{worst case} then $\coNP \nsubseteq \AMTIME[n^{\poly(\log n)}]$ (and in particular $\NP \nsubseteq \BPP$). We complement this result by establishing that any proof of worst-case SoI for nondeterministic computation must use a non-relativizing technique.

    \item Bypassing the cryptographic approach, we provide evidence that nondeterministic computation is \emph{asymmetric}. We show that the validity of SoI for either $\nKt$ or $\rnKt$ would imply the unlikely inclusion of nondeterministic exponential-time classes in $\mathsf{BPP}^{\NP}$ and $\NP/\poly \cap \coNP/\poly$. In particular, if $\PH$ does not collapse then SoI fails for both $\nKt$ and $\rnKt$.

    \item Building on recent work in the simpler setting of randomized computation, we develop a  framework connecting \emph{symmetry of information}, \emph{meta-complexity}, and \emph{explicit constructions} in the nondeterministic setting. We show that suitable formulations of these  notions are quantitatively equivalent. Thus progress on explicit constructions is formally necessary to  show the failure of SoI. 
    \item We apply different facets of this correspondence to derive new \emph{unconditional} results. 
 We consider the task of estimating $\nKt$ and $\rnKt$, and derive lower bounds against quasi-polynomial size Boolean circuits and exponential-time Arthur-Merlin classes. Finally, we combine one of these complexity lower bounds with additional ideas to establish that SoI fails for nondeterministic computation in the presence of side information.
\end{enumerate}

 Our techniques  build on the early literature on nondeterministic Kolmogorov complexity \citep{DBLP:conf/coco/AllenderKR03, DBLP:journals/cc/BuhrmanLM05, DBLP:conf/cie/Allender10, DBLP:journals/jcss/AllenderKRR11}  and on recent advances in the study of meta-complexity \citep{DBLP:conf/coco/Hirahara22, GK22, HLO25}, explicit constructions \citep{DBLP:conf/stoc/0001LL25}, and time-bounded Kolmogorov complexity \citep{DBLP:conf/tcc/HiraharaLO24, DBLP:journals/eccc/KK25, HMO26}, among other results from the theory of computational pseudorandomness.

\paragraph{Broader Perspective.}  Over the last decade, several fundamental problems from algorithms and complexity theory have been investigated through the lens of  Kolmogorov complexity and meta-complexity. Influential advances have occurred in areas such as learning theory \citep{DBLP:conf/coco/CarmosinoIKK16}, cryptography \citep{DBLP:conf/focs/LiuP20}, and average-case complexity \citep{DBLP:conf/stoc/Hirahara21}. In some cases, such as \citep{DBLP:conf/stoc/Hirahara21}, although the result itself is not about Kolmogorov complexity, the only known proof technique uses Kolmogorov complexity in an essential way. We refer to the surveys \citep{DBLP:journals/eatcs/LuO22, DBLP:journals/eatcs/Hirahara22, 10.1145/3744715.3744722} for more context and background.

These and other developments have led to a renewed interest in time-bounded Kolmogorov complexity, meta-complexity, and their applications in theoretical computer science (see, e.g.,~\citep{SimonsMetaComplexity2023}). An important perspective explored in many results, specially in the last few years, has been the use of extensions of the classical theory of time-bounded Kolmogorov complexity to more powerful computational settings, such as \emph{probabilistic Kolmogorov complexity} \citep{GKLO22, DBLP:journals/eatcs/LuO22}. Indeed, many if not most results in the area from the last few years employ probabilistic variants of Kolmogorov complexity (see, e.g.,~\citep{DBLP:conf/focs/Ilango23, DBLP:conf/coco/Santhanam23,  DBLP:conf/stoc/HiraharaILNO23, DBLP:conf/focs/HiraharaN23, DBLP:conf/crypto/LiuP23,  DBLP:conf/stoc/Hirahara23, DBLP:conf/coco/GoldbergK23, DBLP:conf/tcc/HiraharaLO24, DBLP:conf/focs/HiraharaLN24, DBLP:conf/icalp/LuS24, DBLP:conf/coco/HiraharaKLO24, DBLP:conf/approx/GoldbergK24, DBLP:conf/focs/LuORS24, DBLP:journals/eccc/HN25, DBLP:conf/crypto/LiP25, DBLP:journals/eccc/KK25}).

In this work, we advance this program by further developing the theory of \emph{nondeterministic Kolmogorov complexity}. We connect this theory to some questions at the frontier of algorithms and complexity theory and use these connections to establish new lower bounds in meta-complexity.\\

We discuss our results in detail in \Cref{ss:results} below. In \Cref{s:techniques}, we present an overview of some of the underlying proofs and techniques.

\subsection{Results}\label{ss:results}

First, we need to define the main time-bounded Kolmogorov complexity measures appearing in our results. As alluded to above, we consider a natural extension of Levin's $\Kt$ complexity \citep{DBLP:journals/iandc/Levin84} to the setting of nondeterministic computation \cite{DBLP:journals/jcss/AllenderKRR11}. As usual, in order to define time-bounded notions of Kolmogorov complexity, we fix a time-efficient universal machine $U$, and consider programs $\Pi$ encoded with respect to $U$ (see \Cref{sec:preliminaries}). 

 \paragraph{Notions from Time-Bounded Kolmogorov Complexity.} For a string $x\in\bool^*$, the \emph{nondeterministic time-bounded Kolmogorov complexity of $x$} is defined as
	\[
	\mathsf{nKt}(x) \eqdef \min_{\Pi \in \{0,1\}^*,\,t \in \mathbb{N}} \mleft\{|\Pi| + \lceil \log t \rceil \,\,\middle\vert\,\,   
	\begin{aligned}
		&\forall w\in\bool^t, \Pi(w) \textrm{ outputs $x$ or $\bot$ in $t$ steps}; \\
        &\exists w \in \{0, 1\}^t, \Pi(w) \textrm{ outputs $x$ in $t$ steps}
	\end{aligned} 
	\mright\}.
	\]
In other words, by viewing $w$ as a nondeterministic string, in at least one computational path, the program $\Pi$ outputs $x$, while in every computational path, $\Pi$ outputs either $x$ or the failure symbol $\bot$. For this reason, the program $\Pi$ uniquely specifies $x$. In addition,  $x$ can be recovered from $\Pi$ in time at most $t$ using nondeterminism.	

We will also consider probabilistic extensions of $\nKt$ complexity in the sense of \citep{DBLP:conf/icalp/Oliveira19} and \citep{GKLO22}. For a string $x\in\bool^*$, the \emph{randomized nondeterministic time-bounded Kolmogorov complexity of $x$} is defined as
    \begin{equation*}
        \rnKt(x) \eqdef  \min_{\Pi \in \{0,1\}^*,\,t \in \mathbb{N}} \left\{ |\Pi| + \lceil \log t \rceil 
        \ \bigg|\ 
        \Prob_{r \sim \{0, 1\}^t}\left[
        \begin{aligned}
        & \forall w \in \{0, 1\}^t, \Pi(w; r) \textrm{ outputs $x$ or $\bot$ in $t$ steps}; \\
        & \exists w \in \{0, 1\}^t, \Pi(w; r) \textrm{ outputs $x$ in $t$ steps}
        \end{aligned}
        \right] \ge \frac{2}{3}
        \right\}.
    \end{equation*}
Equivalently, there exists a nondeterministic program $\Pi(r)$ such that, with probability at least $2/3$ over $r$, $\Pi(r)$ uniquely specifies $x$ and enables its recovery in nondeterministic time at most $t$.

Recall that $\pKt$ is simply $\Kt$ in the presence of public randomness, i.e., $x$ has ``small'' $\pKt$ complexity if with probability at least $2/3$ over the randomness $r$, $x$ has ``small'' $\Kt$ complexity given access to $r$. Formally, for $x \in \{0,1\}^*$, its \emph{probabilistic time-bounded Kolmogorov complexity} \citep{DBLP:conf/tcc/HiraharaLO24} is given by
\[
		\mathsf{pKt}(x) \eqdef \min\mleft\{k\in\mathbb{N}  \;\bigg|\;  \Prob_{r\sim\bool^{2^k}}\bigg [\Kt(x\mid r)\leq k\bigg ]\geq \frac{2}{3}\mright\},
\]
 where $\Kt(x\mid r)$ denotes the conditional $\Kt$ complexity of $x$ given $r$, i.e., the minimum value $|\Pi|$ + $\lceil \log t \rceil$ over all pairs $(\Pi, t)$ such that $\Pi$ is a program that outputs $x$ when given $r$ and runs in at most $t$ steps.

Analogously, for a string $x \in \{0,1\}^*$, its \emph{probabilistic nondeterministic time-bounded Kolmogorov complexity} (see also \citep{HLO25})  is defined as 
\[
		\mathsf{pnKt}(x) \eqdef \min\mleft\{k\in\mathbb{N}  \;\bigg|\;  \Prob_{r\sim\bool^{2^k}}\bigg [\nKt(x\mid r)\leq k\bigg ]\geq \frac{2}{3}\mright\}.
\]
In other words, bounded $\pnKt$ complexity means that in the presence of a typical  random string $r$, $x$ has bounded $\nKt$ complexity.

\subsubsection{The Intricate Setting of Nondeterministic Computation}

As explained in \Cref{ss:context}, it is unclear how to rule out SoI for nondeterminsitic measures assuming the existence of one-way functions. Note that, in the randomized setting, this cryptographic assumption provides an even stronger guarantee: \citep{DBLP:conf/tcc/HiraharaLO24} showed that one-way functions exist if and only if SoI fails on average for $\pKt$, i.e., there is a polynomial-time samplable distribution family $\{\mathcal{D}_n\}$, where each $\D_n$ is supported over $\{0,1\}^n\times \{0,1\}^n$, and a polynomial $q$ such that for every constant $c>0$ and for every large enough $n$, 
  			\[
  			\Prob_{(x,y)\sim \mathcal{D}_n} \mleft[\;\mathsf{pKt}(x,y)< \mathsf{pKt}(x) + \mathsf{pKt}(y\mid x) - c\cdot \log n \;\mright] \geq \frac{1}{q(n)}.
  			\]
This means that, under the widely-believed assumption that one-way functions exist, SoI \emph{fails on average} for $\pKt$ (and for related notions such as $\Kt$ and $\rKt$).

In contrast, we show (unconditionally) that SoI \emph{holds on average} for the probabilistic extensions of $\nKt$.

    \begin{theorem}[Average-Case Symmetry of Information for Nondeterministic Kolmogorov Complexity]\label{t:soi-pKt-avg}
        For every polynomial-time samplable distribution family $\{\mathcal{D}_n\}_{n \in\mathbb{N}}$, where each $\mathcal{D}_n$ is  supported over $\{0,1\}^n\times \{0,1\}^n$, there exists a constant $c>0$ such that for all large enough $n, k\in\mathbb{N}$,
        \[ \Prob_{(x,y)\sim \mathcal{D}_n} \left[\; \mathsf{pnKt}(x,y)\geq \mathsf{pnKt}(x) + \mathsf{pnKt}(y\mid x) - c\cdot \log n - \log k \; \right] \geq 1-\frac{1}{k}. \]
        Similarly, 
        \[ \Prob_{(x,y)\sim \mathcal{D}_n} \left[\; \mathsf{rnKt}(x,y)\geq \mathsf{rnKt}(x) + \mathsf{rnKt}(y\mid x) - c\cdot \log^3 n - \log k \; \right] \geq 1-\frac{1}{k}. \]
    \end{theorem}

It is unclear how to extend the proof of \Cref{t:soi-pKt-avg} to the \emph{worst case}, i.e., to prove that the relevant inequality holds for all strings $x, y \in \{0,1\}^n$. Interestingly, we are able to show that this would separate $\mathsf{coNP}$ from $\mathsf{NP}$ in a strong sense, as stated below. Let $\mathsf{TAUT} \subseteq \{0,1\}^*$ denote the set of tautologies, i.e., Boolean formulas $F$ that always evaluate to $1$.

    \begin{theorem}[Complexity Lower Bounds from  Worst-Case SoI]\label{t:soi-NP-coAM}
Item 1 implies Item 2 in the following.
\begin{enumerate}
    \item\label{i:soi-pnkt-1} \emph{(Worst-Case Symmetry of Information for Nondeterministic Kolmogorov Complexity).} There exists a constant $c>0$ such that for all large enough $n\in\mathbb{N}$ and $x,y\in\bool^{n}$,
		\[
		    \pnKt(x,y) \geq \pnKt(x) + \pnKt(y\mid x) - e, 
            \]
            where $e = \log^c n$.
    \item \label{i:soi-pnkt-2} \emph{(Hardness of Proving Tautologies).} $\mathsf{TAUT} \notin \mathsf{AMTIME}[t(n)]$, where $t(n)= 2^{\log^d n},~\text{for any}~d \geq 1$. 
\end{enumerate}
The above implication remains valid if $\pnKt$ is replaced with $\rnKt$ in the first item.
\end{theorem}

Note that the consequence $\mathsf{TAUT} \notin \AM$ implies  that $\NP \nsubseteq \BPP$. In particular, \Cref{t:soi-NP-coAM} formalizes the intuition discussed in \Cref{ss:context} that the validity of SoI in a nondeterminsitic setting could point toward a new route for separating the computational power of randomness and nondeterminism.

 While several results establish that lower bounds follow from the \emph{failure} of SoI (see, e.g.,~\citep{DBLP:journals/iandc/LongpreW95, DBLP:conf/coco/Hirahara22, DBLP:conf/stoc/HiraharaILNO23, DBLP:conf/tcc/HiraharaLO24, DBLP:conf/focs/HiraharaLN24}), an interesting aspect of \Cref{t:soi-NP-coAM} is that here a lower bound follows from the \emph{validity} of SoI in a time-bounded setting. In a sense, \Cref{t:soi-NP-coAM} reduces the \emph{lower bound} problem of separating $\mathsf{coNP}$ from $\mathsf{AM}$ to an \emph{upper bound} task: showing that for every large enough $n$ and $x, y \in \{0,1\}^n$, we have 
\begin{equation}\label{eq:SoI_Ineq}
\pnKt(y \mid x) \leq \pnKt(x,y) - \pnKt(x) + \poly(\log n). \tag{$\ast$}
\end{equation}
It is worth mentioning that SoI has been successfully employed as an approach to obtain unconditional complexity lower bounds \citep{DBLP:conf/coco/Hirahara22}.

 Next, we show a general result implying that any proof of \Cref{eq:SoI_Ineq} must employ non-relativizing techniques.

 \begin{theorem}[Strong Failure of SoI in Some Oracle World]\label{thm:neg-o-soi}

    There exist an oracle $\mathcal{O}$ and a constant $\varepsilon > 0$ such that, for every complexity measure $\mu \in \{\rKt, \pKt, \nKt, \rnKt, \pnKt\}$, the following holds. For all sufficiently large $n \in \N$, there exist strings $x, y \in \bool^n$ such that
    \[\mu^{\mathcal{O}}(x, y) < \mu^{\mathcal{O}}(x) + \mu^{\mathcal{O}}(y \mid x) - \varepsilon \cdot n.\]
    
\end{theorem}

The relativized failure of SoI was known for deterministic and non-deterministic polynomial-time (distinguishing) complexity (see \citep{BuhrmanFortnow1995, DBLP:journals/tcs/LeeR05, TroyLeeThesis}). In contrast, \Cref{thm:neg-o-soi} considers Levin-style measures and their randomized extensions,  providing an  oracle world where SoI fails in a strong sense for all these measures.

 We observe that, in order to show the weaker separation  $\NP \nsubseteq \mathsf{BPP}$, it is enough to establish \Cref{eq:SoI_Ineq} under the assumption that $\NP \subseteq \BPP$ (since as observed above  $\mathsf{TAUT} \notin \mathsf{AM}$ implies in particular that $\NP \nsubseteq \BPP$). By a similar argument, to show $\NEXP \nsubseteq \BPP$, it suffices to prove that $\NEXP \subseteq \BPP \Longrightarrow$~\Cref{eq:SoI_Ineq}. Consequently, as an approach to lower bounds, symmetry of information might be useful even if \Cref{eq:SoI_Ineq} turns out to be false.\footnote{Interestingly, it is known that $\mathsf{NP} \subseteq \BPP$ implies that symmetry of information holds for other time-bounded Kolmogorov complexity measures \citep{DBLP:conf/coco/Hirahara22}.}

\subsubsection{Nondeterministic Computation is Likely Asymmetric}\label{ss:likely_asymmetric}

As discussed in \Cref{ss:context} and explained in the preceding section, any failure of SoI for $\nKt$ and its extensions must arise from a source of asymmetry not captured by the standard cryptographic argument. We pursue a different approach and obtain strong evidence that SoI indeed fails in the nondeterministic setting.

Recall that $\mathsf{FewEXP}$ is the set of languages in $\NEXP$ where each positive instance admits at most $2^{\poly(n)}$ witnesses. Note that $\EXP \subseteq \mathsf{FewEXP} \subseteq \NEXP$. We use $||\NP$ to denote parallel access to an $\NP$-oracle. 

\begin{theorem}[Strong Complexity Collapses from SoI] \label{thm:collapses_from_SoI}
    The following results hold:
    \begin{itemize}
        \item[\emph{(}i\emph{)}] \emph{[Uniform Collapses]}~\\
        -- If SoI holds for $\nKt$ or $\rnKt$ with error $e(n) = O(\log n)$, then  $\mathsf{FewEXP} \subseteq \mathsf{BPP}^{\NP}$.\\ 
        -- If SoI holds for $\nKt$ or $\rnKt$ with error $e(n) = \poly(\log n)$, then $\mathsf{FewEXP} \subseteq \mathsf{BPTIME}^{\NP}[n^{\poly(\log n)}]$.
         \item[\emph{(}ii\emph{)}]  \emph{[Nonuniform Collapses]}~\\ 
         -- If SoI holds for $\rnKt$ with error $e(n) = O(\log n)$, then $\mathsf{BPEXP}^{||\NP}\subseteq\NP/\poly \cap \coNP/\poly$.\\ 
         -- If SoI holds for $\nKt$ with error $e(n) = O(\log n)$, then $\mathsf{EXP}^{||\NP}\subseteq\NP/\poly \cap \coNP/\poly$.
    \end{itemize}
\end{theorem}

In particular, it follows from \Cref{thm:collapses_from_SoI} Part (\emph{ii}) and  a standard result \citep{DBLP:journals/tcs/Yap83} that if SoI holds for either $\rnKt$ or $\nKt$, the polynomial-time hierarchy collapses to the third level. Thus, under a  widely believed computational assumption, SoI does not hold for nondeterministic computation.\footnote{In a concurrent and independent work \citep{Khrushchev-June-2026}, Viacheslav Khrushchev  also established that the validity of SoI for $\nKt$ implies the collapse of the polynomial hierarchy. We thank Bruno Bauwens for bringing this to our attention.}

Before stating our next result, we need the following definitions. For a complexity measure $\kappa$ (such as $\Kt$, $\nKt$, or $\rnKt$), we consider the standard formulation of the corresponding meta-computational problem, where given a string $x \in \{0,1\}^n$ and a threshold parameter $s$, the goal is to decide if $\kappa(x) \leq s$. We will also consider the gap version of this problem, where for fixed functions $s_1(n)$ and $s_2(n)$, we are given the promise that either $\kappa(x) \leq s_1(n)$, or $\kappa(x) \geq s_2(n)$, and the goal is to distinguish the two cases. As usual, these computational problems will be denoted by $\mathsf{M}\kappa\mathsf{P}$ and $\mathsf{Gap}\text{-}\mathsf{M}\kappa\mathsf{P}[s_1,s_2]$, respectively. When $s_1(n) = n^\varepsilon$ for a fixed but arbitrarily small $\varepsilon > 0$  and $s_2(n) = n-1$, we simply write $\mathsf{Gap}\text{-}\mathsf{M}\kappa\mathsf{P}$. For a promise problem $\Pi = (\mathsf{YES}, \mathsf{NO})$ and a complexity class $\mathcal{C}$, we say that $\Pi \in \mathcal{C}$ if there is a language $L \in \mathcal{C}$ such that $\mathsf{YES} \subseteq L$ and $\mathsf{NO} \subseteq \{0,1\}^* \setminus L$. 

Recall that showing that $\MKtP \notin \P$ is a longstanding open problem in meta-complexity that dates back at least to \citep{DBLP:journals/siamcomp/AllenderBKMR06}. In the next result, we show that the validity of SoI in the nondeterministic setting implies both \emph{upper bounds} and \emph{lower bounds} for meta-computational problems.

\begin{theorem}[Refined Meta-Complexity Landscape Under  SoI]\label{thm:consequences_SoI}
The following results hold:
\begin{itemize}
\item[\emph{(}i\emph{)}]
   \emph{[Nonuniform Upper Bound]} If SoI with error $e(n) = O(\log n)$ holds for $\rnKt$, then 
   $$\gapmrnktp[n - O(\log n), n-1] \in \NP/\poly \cap \coNP/\poly.$$
\item[\emph{(}ii\emph{)}]
   \emph{[Uniform Lower Bound]} If SoI with error $e(n) = \poly(\log n)$ holds for either $\nKt$ or $\rnKt$, then $$\MKtP \notin \NTIME[n^{\poly(\log n)}] \cap \coNTIME[n^{\poly(\log n)}].$$
    \end{itemize}
\end{theorem}

 It seems unlikely that an approximation to $\Kt(x)$ can be certified in $\NP/\poly \cap \coNP/\poly$, let alone for the more expressive measure $\rnKt$, as in the conclusion of \Cref{thm:consequences_SoI} Item (\emph{i}). Therefore, we interpret \Cref{thm:consequences_SoI} as further evidence that information is likely asymmetric in nondeterministic computation.

\subsubsection{A Correspondence Between Asymmetry, Hardness, and Explicit Constructions}

Next, we establish a connection among symmetry of information for $\rnKt$, the meta complexity of estimating $\rnKt$, and $\pr\BPTIME[\cdot]^{|| \NP}$-explicit constructions. This correspondence extends recent work \citep{DBLP:conf/coco/Hirahara22, GK22, DBLP:journals/eccc/KK25, HMO26} relating SoI to meta-complexity in the simpler setting of randomized computation. We prove a more general result for nondeterministic computation, where the analysis is more subtle, as explained below and in \Cref{s:techniques}.

\begin{theorem}[Correspondence Between Asymmetry of Information, Hardness, and Explicit Constructions] \label{thm:equivalence_SoI} The following statements are equivalent:
    \begin{enumerate}[label=(\roman*)]
        \item \emph{(SoI holds for $\rnKt$)} \;Symmetry of information holds for $\rnKt$ with error $e(n)=\poly(\log n)$.\label{i:SoI-rnKt}
        \item \emph{(Upper bound for $\gapmrnktp$)} \;$\gapmrnktp[n-\poly(\log n),n-1]\in\pr\BPTIME[n^{\poly(\log n)}]^{||\NP}$.\label{i:easiness-Gap-MrnKtP}
        \item \emph{(Natural property against $\rnKt$)} \;There is $P=(\mathcal{YES},\mathcal{NO})\in\pr\BPTIME[n^{\poly(\log n)}]^{||\NP}$ such that
        \begin{itemize}
            \item $|\mathcal{YES}\cap\{0,1\}^n|\ge2^{n-1}$, and
            \item For every $x\in\{0,1\}^n$ where $\rnKt(x)<n-\poly(\log n)$, $x\in\mathcal{NO}$.
        \end{itemize}\label{i:natural-property-rnKt}
        \item \emph{(Hardness of explicit constructions for dense $\pr \BPTIME[\cdot]^{||\NP}$ properties)} \;The following is false:
        \begin{quote}
            For any promise problem $P=(\mathcal{YES},\mathcal{NO})\in\pr\BPTIME[n^{\poly(\log n)}]^{||\NP}$ such that $|\mathcal{YES}\cap\{0,1\}^n|\ge2^{n-1}$, there exists $t(\cdot),\ell(\cdot):\N\to\N$ where $\ell(n)+\log t(n)<n-\poly(\log(n))$ for each $n$, and an algorithm $A(1^n, \cdot)$ such that, for infinitely many $n$, there exists advice $\alpha_n\in\{0,1\}^{\ell(n)}$ and a string $x_n\in\{0,1\}^n\setminus\mathcal{NO}$, such that
            \[
            \Prob_{r\sim\bool^{t(n)}}\mleft[
	\begin{aligned}
		&\forall w\in\bool^{t(n)}, A(1^n,\alpha_n,w,r) \textrm{ outputs $x$ or $\bot$ in $t(n)$ steps}; \\
        &\exists w \in \{0, 1\}^{t(n)}, A(1^n,\alpha_n,w,r) \textrm{ outputs $x$ in $t(n)$ steps}
	\end{aligned} \mright] \geq\frac23.
            \]
        \end{quote}\label{i:hardness-explicit-construction}
    \end{enumerate}
\end{theorem}

General explicit construction results are known for dense properties in $\pr\BPP$ \citep{DBLP:conf/stoc/OliveiraS17} and in $\coAM$ \citep{DBLP:conf/stoc/0001LL25}. In contrast, \Cref{thm:equivalence_SoI} shows that SoI fails for $\rnKt$ if a suitable positive result can be established for dense properties in  $\pr\BPP^{|| \NP}$. Moreover, the design of such an algorithm is formally necessary to establish the asymmetry of nondeterministic computation. 

Interestingly, the \emph{nonadaptive access} to the $\NP$ oracle and the precise formulation of the nondeterministic complexity of the explicit construction algorithm $A$ are essential for the equivalence result to hold. In fact, the construction from \citep{DBLP:conf/stoc/OliveiraS17} relativizes to an $\NP$-oracle, but the access to the $\NP$-oracle by the corresponding algorithm $A$ is \emph{adaptive}, which is insufficient to rule out SoI. This is another subtle aspect of the nondeterministic setting, and we discuss it in more detail in \Cref{sec:nonadaptive}.

\paragraph{More on Adaptive vs Nonadaptive Queries.} 
For every string $x$, $\nKt(x) = \Theta(\Kt^{||\SAT}(x))$ (see \Cref{lem:ktA-leq-nkt} and \Cref{lem:nkt-leq-ktsat}). In other words, in settings where a constant multiplicative overhead is inessential, we can think of $\nKt$ as $\Kt$ in the presence of \emph{nonadaptive} queries to a $\SAT$ oracle. (Note that this is not necessarily the case for SoI, where a large overhead when relating different measures can impact SoI.) Interestingly, since Ronneburger's proof  \citep{RonneburgerThesis} of the failure of SoI for $\Kt$ relativizes (see \Cref{s:nKt_SoI} for a related discussion), we know that SoI fails for both $\Kt$ and $\Kt^{\SAT}$. However, similarly to the situation for $\nKt$ and related measures, we currently do not know if SoI fails for $\Kt^{|| \SAT}$. Thus  the validity of SoI in the nondeterministic setting is closely related to the role of adaptive queries in computation. For another example of the relevance of (non)adaptive queries for SoI, see \Cref{sec:nonadaptive}.\\

In the next section, we apply \Cref{thm:equivalence_SoI} and its underlying correspondences to derive new unconditional lower bounds in meta-complexity and to establish the failure of SoI in the presence of side information.

\subsubsection{Unconditional Consequences: Complexity Lower Bounds and Asymmetry of Information}

As highlighted in \Cref{ss:likely_asymmetric}, showing that $\Kt$ cannot be computed in polynomial time, i.e., that $\MKtP \notin \P$, is a longstanding open problem in meta-complexity \citep{DBLP:journals/siamcomp/AllenderBKMR06}. By considering whether SoI holds for $\rnKt$ with error $\poly(\log n)$,  \Cref{thm:consequences_SoI} (Item 2) and \Cref{thm:equivalence_SoI} imply the following \emph{unconditional} lower bound result.

\begin{theorem}[A Quick Lower Bound Consequence] \label{thm:application_correspondence} The following result holds:
    $$\MKtP \notin \NTIME[n^{\poly(\log n)}] \cap \coNTIME[n^{\poly(\log n)}] ~\;\text{or}\;~ \gapmrnktp[s_1,s_2] \notin \pr \mathsf{BPTIME}[n^{\poly(\log n)}]^{||\NP},$$
    where $s_1(n) = n - \poly(\log n)$ and $s_2(n) = n-1$.
\end{theorem}

This result provides a quick example of how SoI can offer a fruitful viewpoint for establishing complexity lower bounds. In our next results, we apply the correspondences behind \Cref{thm:equivalence_SoI} in a more sophisticated, non-black-box way to establish new  results on the asymmetry and complexity of nondeterministic computations.

\paragraph*{Lower Bounds for $\nKt$ via Symmetry of Information.} Let $\mathsf{MnKtP}$ be the following problem: Given $(x,1^s)$, where $x\in \bool^*$ and $s \in \mathbb{N}$, decide whether $\nKt(x) \leq s$ (``yes'' case) or $\nKt(x) > s$ (``no'' case). Moreover, for  $s\colon\mathbb{N}\to\mathbb{N}$, we let $\mathcal{R}_{\nKt[s]} = \{x \in \{0,1\}^* \mid \nKt(x) \geq s(|x|)\}$. We say that a language $L$ is \emph{dense} if there is a positive constant $b$ such that, for every large enough $n$, we have $|L \cap \{0,1\}^n| \geq 2^n/n^b$.

	\begin{theorem}[Complexity Lower Bounds for Estimating \texorpdfstring{$\nKt$}{}]\label{thm:intro_nKt_lb}
		There is a constant $\varepsilon>0$ such that
		\[
		\mathsf{MnKtP}\not\in\mathsf{NTIME} \mleft[2^{\varepsilon n}\mright] \cap \mathsf{coNTIME}\mleft[2^{\varepsilon n}\mright].
		\]
		Moreover, given a constant $1/2 < \gamma < 1$, there exists $\varepsilon = \varepsilon(\gamma) > 0$ such that, for every dense language $L \subseteq \mathcal{R}_{\nKt[\gamma\cdot n]}$, we have $L \not\in\mathsf{NTIME} \mleft[2^{\varepsilon n}\mright] \cap \mathsf{coNTIME}\mleft[2^{\varepsilon n}\mright]$. 
	\end{theorem}

It was previously known that $\mathsf{MnKtP} \notin \NP \cap \mathsf{coNP}$ \citep{DBLP:conf/cie/Allender10, DBLP:journals/jcss/AllenderKRR11}.\footnote{We remark that the proof of \cite{DBLP:journals/jcss/AllenderKRR11} can be extended to a larger time bound, but the approach fundamentally faces a \emph{half-exponential barrier}. In particular, within the framework of \cite{DBLP:conf/cie/Allender10, DBLP:journals/jcss/AllenderKRR11}, the strongest lower bound one can hope to prove is $\MnKtP\not\in\NTIME[T(n)]\cap\co\NTIME[T(n)]$ for a time bound $T(n)$ satisfying $T(T(n))=o(2^n)$.} In contrast, \Cref{thm:intro_nKt_lb} provides an exponential lower bound on the complexity of $\mathsf{MnKtP}$. More generally, it shows that distinguishing strings of bounded complexity (``structure'') from random strings (``randomness'') is hard.

We observe that, in order to obtain a lower bound of the form \( 2^{\Omega(n)} \) for \( \mathcal{R}_{\nKt[\gamma \cdot n]} \), our techniques require assuming \( \gamma > 1/2 \); see \Cref{ss:limitations_gamma} for details. Establishing lower bounds of magnitude $2^{\Omega(n)}$ for \( \mathcal{R}_{\nKt[\gamma \cdot n]} \) when \( 0 < \gamma \leq 1/2  \) remains an open problem.

Since $\mathsf{MnKtP} \in \mathsf{P}^{\mathsf{NEXP}} \subseteq \mathsf{NEXP}/\poly$ and every language in $\mathsf{NEXP}$ non-uniformly reduces to  $\mathsf{MnKtP}$  \citep{DBLP:journals/jcss/AllenderKRR11}, it follows that $\mathsf{MnKtP}\not\in \mathsf{SIZE}[\poly]$ if and only if $\NEXP \nsubseteq \SIZE[\poly]$. In particular, a breakthrough complexity separation would follow if we could show that $\mathsf{MnKtP}\not\in \mathsf{SIZE}[\poly]$. While we currently don't know how to prove the required lower bound for $\mathsf{MnKtP}$, we are able to establish such a lower bound for the related problem of estimating $\rnKt$, the \emph{randomized extension} of $\nKt$.

\paragraph*{Lower Bounds for $\rnKt$ via Explicit Constructions.} Fix two functions $s_1,s_2\colon\mathbb{N}\to\mathbb{N}$. We let $\mathsf{Gap}$-$\mathsf{MrnKtP}[s_1,s_2]$ be the following promise problem: Given $x\in \bool^*$, decide whether $\rnKt(x) \leq s_1(|x|)$ (``yes'' case) or $\rnKt(x) > s_2(|x|)$ (``no'' case). Recall that for a promise problem $\Pi = (\mathsf{YES}, \mathsf{NO})$ and a complexity class $\mathcal{C}$, we say that $\Pi \notin \mathcal{C}$ if there is no language $L \in \mathcal{C}$ such that $\mathsf{YES} \subseteq L$ and $\mathsf{NO} \subseteq \{0,1\}^* \setminus L$. For a language $\mathcal{O}$ and a function $s \colon \mathbb{N} \to \mathbb{N}$, we let $\mathsf{SIZE}^\mathcal{O}[s]$ denote the set of languages computed by Boolean circuits with oracle access to $\mathcal{O}$ and of size at most $s(n)$. Similarly to $\mathcal{R}_{\nKt[s]}$, we can also consider the set $\mathcal{R}_{\rnKt[s]}$, which is defined in the natural way using $\rnKt$ instead of $\nKt$.

	\begin{theorem}[Complexity Lower Bounds for Estimating \texorpdfstring{$\rnKt$}{}]\label{thm:intro_rnKt_lbs}
		The following statements hold\emph{:}
		\begin{enumerate}
			\item \label{i:rnkt_lb_2} For every constant $1/2 < \gamma < 1$, there exists a constant $\varepsilon>0$ such that \[\gapmrnktp\mleft[\gamma\cdot n,n-1\mright]\not\in \mathsf{AMTIME}\mleft[2^{\varepsilon n}\mright] \cap \mathsf{coAMTIME}\mleft[2^{\varepsilon n}\mright].
			\]
            Moreover, given a constant $1/2 < \gamma < 1$, there exists $\varepsilon(\gamma)> 0$ such that, for every dense language $L \subseteq \mathcal{R}_{\rnKt[\gamma\cdot n]}$, we have $L \not\in\mathsf{AMTIME} \mleft[2^{\varepsilon n}\mright] \cap \mathsf{coAMTIME}\mleft[2^{\varepsilon n}\mright]$. 

            		\item \label{i:rnkt_lb_1} For every constant $\gamma>0$,
			\[\gapmrnktp\mleft[2^{\log^{\gamma} n},n-1\mright]\not\in \mathsf{AMTIME}\mleft[n^{\mathsf{polylog}(n)}\mright].
			\]  
            Moreover, for every dense language $L \subseteq \mathcal{R}_{\rnKt[2^{\log^{\gamma} n}]}$, we have $L \not\in\mathsf{coAMTIME}\mleft[n^{\mathsf{polylog}(n)}\mright]$.

			\item \label{i:rnkt_lb_3} For every constant $\gamma>0$ and $\mathcal{O} \in \AM \cap \coAM$,
			\[\gapmrnktp\mleft[2^{\log^{\gamma} n},n-1\mright]\not\in  \mathsf{SIZE}^{\mathcal{O}}\mleft[n^{\mathsf{polylog}(n)}\mright].
			\]
            Moreover, for every dense language $L \subseteq \mathcal{R}_{\rnKt[2^{\log^{\gamma} n}]}$, we have $L \not\in  \mathsf{SIZE}^\mathcal{O}\mleft[n^{\mathsf{polylog}(n)}\mright].$
		\end{enumerate}
	\end{theorem}

    While we were able to establish that $\mathsf{MrnKtP} \notin \mathsf{AM}$, it remains an interesting open problem to show that $\mathsf{MnKtP} \notin \NP$.

    To the best of our knowledge, no previous lower bounds against $\mathsf{AM}$, $\SIZE[\poly]$, and exponential time nondeterministic classes were known for natural meta-computational problems. These results thus provide the strongest unconditional lower bounds in meta-complexity to date.

\paragraph*{An Unconditional Failure of SoI via Computational Hardness.} While we are currently unable to determine, unconditionally, if symmetry of information fails in a nondeterministic setting, we obtain an unconditional result in the presence of a conditional string, which one can think of as an advice string or side information.      
    In more detail, we say that conditional SoI holds for $\nKt$  with error $e(n)$ if for any strings $x,y,w\in\{0,1\}^*$, we have that
\[
    \nKt(x,y\mid w)\ge\nKt(x\mid y,w)+\nKt(y\mid w)-e(|x|+|y|).
\]
We note that we do not assume random-access to the conditional string in the setting of conditional symmetry of information.  
Then, unconditionally, it is possible to show that  conditional SoI for $\nKt$ does not hold with error $e(n) = \poly(\log n)$.

\begin{theorem}[Failure of Conditional SoI for $\nKt$]\label{thm:failure-cond-SoI}
    For any $e(n)\le\polylog(n)$, conditional SoI for $\nKt$ does not hold with error $e(n)$.
\end{theorem}

Our setting is related to the failure of SoI (without a conditional string) in some oracle world (\Cref{thm:neg-o-soi}). There are however a couple of differences. While the oracle can be interpreted as a conditional string, in the oracle setting one assumes query access to the oracle, while here we do not need to assume query access to the conditional string. In addition, the oracle construction requires techniques from computational pseudorandomness, while the proof given here relies on a reduction to a computational lower bound, in the spirit of \Cref{thm:equivalence_SoI}.

\vspace{0.6cm}

\noindent \textbf{Use of AI Tools.} All ideas, arguments, and proofs presented in this paper were developed by the authors. The proofs were checked for correctness using GPT 5.6 and GPT 6.0.\\

\noindent \textbf{Acknowledgements.} 
We thank Jiatu Li, Hanlin Ren, Ninad Rajgopal, Rahul Santhanam, Rohan Acharya, and Yahel Manor for related discussions. This work was supported in part by the UKRI Frontier Research Guarantee Grant EP/Y007999/1 and the Centre for Discrete Mathematics and its Applications (DIMAP) at the University of Warwick. The authors would also like to thank the Isaac Newton Institute for Mathematical Sciences, Cambridge, for support and hospitality during the programme Logical Foundations of Complexity Theory (LFC), where part of the work on this paper was undertaken. The LFC programme was supported by EPSRC grant EP/Z000580/1.

\section{Techniques}\label{s:techniques}

We now describe the key ideas and techniques employed in the proof of our main results.\\

\Cref{t:soi-pKt-avg} states that, unconditionally, SoI holds for $\pnKt$ with high probability over random strings sampled efficiently. The theorem may proved by an argument closely resembling one from \cite{HLO25}, so we will not give a further sketch here.  We now move on to describe the proof of \Cref{t:soi-NP-coAM}.

\paragraph*{Complexity separations from symmetry of information (Theorem \ref{t:soi-NP-coAM}).}

The argument builds on techniques from \citep{RonneburgerThesis, DBLP:conf/stoc/0001LL25, DBLP:conf/tcc/HiraharaLO24}. We would like to argue in the contrapositive: namely, assuming $\NP \subseteq \coAMTIME[n^{\polylog(n)}]$ and demonstrating \emph{asymmetry} of information for $\pnKt$. 

Our starting point is the approach of \cite{RonneburgerThesis} for excluding SoI for $\Kt$. Unfortunately, the same argument does not yield asymmetry in the case of $\nKt$, let alone $\pnKt$. The problem in the case of $\nKt$ is that, during the exhaustive search over programs of length $n$, one must determine whether a given program has a unique nondeterministic output. This requires searching over witnesses of length up to $2^n$, meaning that exhaustive search would take time doubly-exponential in $n$. The problem in cases of randomized variants of $\Kt$ complexity is that known exponential-time procedures for determining complexity are probabilistic, and they cannot be relied upon to produce canonical high-complexity strings. In our case, of course, we must overcome both of these issues.

To that end, a first observation is that there exists a randomized $\SAT$-oracle algorithm (denoted $A^{\SAT}$ here) for approximating $\pnKt$ complexity. More specifically, given a pair of strings $u \in \bool^n$ and $v \in \bool^{\leq 2^n}$, $A^{\SAT}$ runs in time $2^{O(n)}$, accepts with high probability if $\pnKt(u \mid v) \geq 2n/3$, and rejects with high probability if $\pnKt(u \mid v) < n/2$.\footnote{An algorithm with these parameters can be obtained from \Cref{l:gapMpKtP-in-prBPE} combined with success amplification for $\pnKt$ following \cite{GKLO22}.} At a high level, our approach entails derandomizing $A^{\SAT}$ appropriately, yielding an efficient enough procedure to produce a canonical string $x$ with $\pnKt(x) \geq n/2$ (and, likewise, a canonical string $y$ with $\pnKt(y \mid x) \geq n/2$). 

We now describe how to derandomize $A^{\SAT}$. Here, we apply a recent \textit{explicit construction} result due to Chen et al. \cite{DBLP:conf/stoc/0001LL25} (discussed further below), which we adapt to obtain  a family of multisets $\{T_s\}_{s \in \N}$ pseudorandom infinitely often against ($\AM \cap \co\AM$)-oracle circuits of size $s$. For any desired constant $\delta > 0$, each $T_s$ has a representation of $\rnKt$ complexity at most $2^{(\log s)^{\delta}}$, and each $T_s$ contains at most $\poly(s)$ elements. Making use of our assumption that $\NP \subseteq \co\AMTIME[n^{\polylog(n)}]$, on input $(u,v)$, $A^{\SAT}$ can be simulated by an $(\AM \cap \co\AM)$-oracle circuit $C_{(u, v)}$ of size at most $2^{\poly(n)}$ that treats its input as internal randomness for $A^{\SAT}$. Thus, for infinitely many choices of $s = 2^{\poly(n)}$, the pseudorandom set $T_s$ fools such circuits $C$.

Combining the above ideas, for infinitely many input lengths $n \in \N$ and any desired constant $\delta' > 0$, we may obtain a pair of strings $x,y \in \bool^n$ with $\pnKt(x) \geq n/2$ and $\pnKt(y\mid x) \geq n/2$, but $\pnKt(x, y) \leq \rnKt(x, y) \leq 2^{n^{\delta'}}$. In particular, let $x$ be the lexicographically first string of length $n$ such that $C_{(x, \epsilon)}$ accepts with probability greater than $1/2$ over $T_s$, where $\epsilon$ denotes the empty string. On one hand, the lower bound $\pnKt(x) \geq n/2$ follows from the definition of $C_{(x, \epsilon)}$ and the pseudorandom property of $T_s$. (One may define $y$ such that $\pnKt(y \mid x) \geq n/2$ similarly.) On the other hand, given $T_s$, an exhaustive procedure to \emph{construct} $x$ and $y$ runs in time $2^{\poly(n)}$ provided an $(\AM \cap \co\AM)$-oracle. This, together with the upper bound on $\rnKt(T_s)$, implies $\rnKt^{\AM \cap \co\AM}(x,y) \leq 2^{n^{\delta'/2}}$. The desired bound $\rnKt(x,y) \leq 2^{n^{\delta'}}$ follows by ``collapsing'' the $(\AM \cap \co\AM)$-oracle into the nondeterminism in the definition of $\rnKt$.

Of course, the pair $x, y$ obtained above does not yet witness asymmetry, since $n/2 + n/2 \ll 2^{n^{\delta'}}$. To obtain a suitable pair of strings, we apply the above approach iteratively, as in \cite{DBLP:conf/tcc/HiraharaLO24}; in particular, we obtain a sequence of $\ell := 2^{n^{\delta'}}$ strings $x_1, ..., x_{\ell} \in \bool^n$ such that 
\[\sum_{i \in [\ell]}\pnKt(x_i \mid x_1, ..., x_{i - 1}) \geq \ell \cdot n/2\]
but 
\[\pnKt(x_1, ..., x_{\ell}) \leq \ell.\]
The above can then be used to show a violation of the standard formulation of SoI for $\pnKt$; see \Cref{sec:worst-soi-pnkt} for further details.\\

The proof of \Cref{thm:neg-o-soi} is basically a relativizing version of the argument used in \Cref{t:soi-NP-coAM}, adapted so that it works for all measures $\mu \in \{\rKt, \pKt, \nKt, \rnKt, \pnKt\}$ and
gives a linear gap $\varepsilon \cdot n$ instead of just polylogarithmic.\footnote{To the best of our knowledge, our approach to show failure of SoI in some relativized world is completely different than the one in \citep{DBLP:journals/tcs/LeeR05, TroyLeeThesis}.}

\paragraph*{Complexity collapses from symmetry of information (Theorem \ref{thm:collapses_from_SoI}).}

\Cref{thm:collapses_from_SoI} shows that surprising uniform and non-uniform complexity collapses would follow from SoI for $\nKt$ or $\rnKt$. We start by discussing the proof of the item stating that $\Few\EXP \subseteq \BPP^{\NP}$ if SoI holds for $\nKt$ with error $e(n) = O(\log n)$. The proof employs the PRG construction of \cite{DBLP:journals/jcss/ImpagliazzoW01, DBLP:journals/cc/TrevisanV07} as well as a version of the ``easy witness method'' due to \cite{DBLP:conf/coco/Kabanets00}.

Later in this section, we outline the proof of a correspondence theorem (\Cref{thm:equivalence_SoI}), which shows, in part, that SoI for $\rnKt$ implies the easiness of $\gapmrnktp$ for randomized algorithms with non-adaptive access to an $\NP$ oracle. We may rely on very similar techniques here to show that SoI for $\nKt$ with error $O(\log n)$ implies $\gapmnktp \in \pr\P^{\parallel \NP}$. See below for more detail on this piece of the argument. The rest of the proof of our stated uniform collapse consists in contradicting this easiness of $\gapmnktp$, assuming $\Few\EXP \nsubseteq \BPP^{\NP}$. 

We follow a win-win approach to obtain a PRG computable single-valued non-deterministically with small advice and secure against $\NP$-oracle circuits. On one hand, if $\EXP \nsubseteq \BPP^{\NP}$, then the desired PRG follows from \cite{DBLP:journals/jcss/ImpagliazzoW01, DBLP:journals/cc/TrevisanV07}; this step employs an intermediate win-win argument (see \Cref{t:io-PRG-EXP}). On the other hand, if $\EXP \subseteq \BPP^{\NP}$, then $\Few\EXP \neq \EXP$ by the assumption that $\Few\EXP \nsubseteq \BPP^{\NP}$. In this case, for any fixed $\Few\EXP \backslash \EXP$ language $L$, there are infinitely many strings $x \in L$ that do not have any witnesses of $\NP$-oracle circuit complexity $\poly(|x|)$. The idea is then to instantiate the PRG of \cite{DBLP:journals/cc/BabaiFNW93} with the lexicographically first witness for $x \in L$. The advice required is $(x, |W_x|) \in \bool^{\poly(|x|)}$, where $W_x$ is the witness set of $x$. Using this advice, one may non-deterministically guess all of the $2^{\poly(|x|)}$-many witnesses for $x$ and output the lexicographically first as a canonical high-complexity truth-table, taking time at most $2^{\poly(|x|)}$. 

The resulting PRG is secure against polynomial-size $\NP$-oracle circuits. Moreover, the outputs of this PRG can be shown to have bounded $\nKt$ complexity, so any algorithm for an appropriate gap version of $\MnKtP$ could be used to distinguish PRG outputs from random. Therefore, $\gapmnktp \notin \pr\P^{\NP}$, yielding the desired contradiction.

We omit proof sketches for the remaining uniform collapses, since all the proofs are quite similar to the one described above. Next, we briefly describe the proof of the result stating that, if SoI holds for $\nKt$ with error $e(n) = O(\log n)$, then $\EXP^{\parallel \NP} \subseteq \NP/\poly \cap \coNP/\poly$. Again, we will make use of the fact that SoI for $\nKt$ implies $\gapmnktp \in \pr\P^{\parallel \NP}$. Using techniques similar to those of \cite{DBLP:journals/siamcomp/AllenderBKMR06}, one may show that any fixed language $L \in \E^{\parallel \NP}$ has a non-uniform non-adaptive poly-time Turing reduction to $\gapmnktp$. Namely, using $L$ to instantiate the PRG of \cite{DBLP:journals/cc/BabaiFNW93}, outputs of the PRG have bounded $\nKt$ complexity (by hard-wiring the number of YES queries of the $\NP$ oracle), while uniformly random strings almost always have near-maximum $\nKt$ complexity. The reconstruction property of the PRG then yields a reduction from $L$ to $\gapmnktp$, and therefore we have $\E^{\parallel \NP} \subseteq \P^{\parallel \gapmnktp}/\poly \subseteq \P^{\parallel \NP}/\poly$. The stated inclusion $\EXP^{\parallel \NP} \subseteq \NP/\poly \cap \coNP/\poly$ follows from a downward collapse result of \cite{DBLP:journals/jacm/ShaltielU05}, given that $\E^{\parallel \NP}$ allows \emph{low-degree extension}, together with a padding argument.

The result stating that SoI for $\rnKt$ implies $\BPEXP^{\parallel \NP} \subseteq \NP/\poly \cap \coNP/\poly$ follows a similar approach together with the unconditional inclusion $\BPE^{\parallel \NP} \subseteq \AM\E/O(n)$.\\

\Cref{thm:consequences_SoI} gives a non-uniform upper-bound and a uniform lower-bound for meta-complexity that hold under SoI for $\nKt$ or $\rnKt$. The proofs rely on the complexity collapse results described above, the correspondence theorem described next, a result from \cite{DBLP:journals/siamcomp/AllenderBKMR06}, and a straightforward diagonalization argument. We refer to \Cref{sec:meta_SoI} for more detail.

\paragraph*{Correspondence Theorem (Theorem \ref{thm:equivalence_SoI}).} 

We now describe the techniques used to prove \Cref{thm:equivalence_SoI}, which exhibits a broad equivalence between asymmetry of information for $\rnKt$, hardness of approximating $\rnKt$, natural properties against $\rnKt$, and explicit constructions in the context of $\pr\BPTIME[\cdot]^{\parallel \NP}$. The equivalence between \Cref{i:natural-property-rnKt} and \Cref{i:hardness-explicit-construction} is straightforward, and likewise the implication \Cref{i:easiness-Gap-MrnKtP} $\implies$ \Cref{i:natural-property-rnKt}, so we will not discuss the proofs further. In order to prove the \Cref{i:natural-property-rnKt} $\implies$ \Cref{i:easiness-Gap-MrnKtP}, obtaining a worst-case  algorithm for $\gapmrnktp$ from a natural property against $\rnKt$, we rely on the worst-case to average-case reduction framework of \cite{DBLP:conf/focs/Hirahara18}. A couple of differences in our setting compared to \cite{DBLP:conf/focs/Hirahara18} are as follows. We make use of a PRG based on a construction of Raz et al. \cite{RazRV02_jcss_journals} as in \cite{DBLP:conf/focs/Hirahara20} to obtain $\rK^{\poly(n, t_D), \parallel D}$ upper bounds in the case that $D$ is a $t_D$-time distinguisher. Moreover, the distinguisher algorithm runs in time ${\sf quasi}\text{-}\poly(n)$ and makes non-adaptive $\NP$-oracle queries, resulting in an upper bound on $\rK^{{\sf quasi}\text{-}\poly(n), \parallel \NP}$ and thus $\rnK^{{\sf quasi}\text{-}\poly(n)}$ (see \Cref{l:rKt--parallel-NP}). 

The proof of \Cref{i:easiness-Gap-MrnKtP} $\implies$ \Cref{i:SoI-rnKt}, stating that easiness of $\gapmrnktp$ for $\pr\BPTIME[n^{\poly(\log n)}]^{||\NP}$ implies SoI for $\rnKt$ with error $e(n) = \poly\log n$, follows along the framework of \cite{DBLP:conf/coco/Hirahara22, GK22} for proving SoI via hybrid argument, assuming easiness of meta-complexity. This framework is applied in a setting similar to that described above: the PRG of \cite{RazRV02_jcss_journals, DBLP:conf/focs/Hirahara20}, distinguishers in $\BPTIME[n^{\poly(\log n)}]^{\parallel \NP}$, and bounds on $\rK^{{\sf quasi}\text{-}\poly(n), \parallel \NP}$ collapsing to bounds on $\rnK^{{\sf quasi}\text{-}\poly(n)}$. 

To prove that Symmetry of Information for $\rnKt$ implies a worst-case $\pr\BPTIME[n^{\poly(\log n)}]^{||\NP}$ algorithm for $\gapmrnktp[n-\poly(\log n),n-1]$ (\Cref{i:SoI-rnKt} $\implies$ \Cref{i:easiness-Gap-MrnKtP}), our starting point is a recent paper from \cite{HMO26}, which proved a similar statement for the case of $\rKt$. Though we make use of some key ideas from that work, the setting of $\rnKt$ introduces additional challenges. 

One concept from \cite{HMO26} that we are able to borrow quite directly yields that, assuming SoI for $\rnKt$ with error $e(n)$, for some $t \leq 2^{O(e(n))}$, for any $x \in \bool^n$, $\rnK^t(x) \leq \rnKt(x) + O(e(n))$.\footnote{Actually, \cite{HMO26} obtained the statement with $O(e(n) \cdot \log n)$ as the leftmost term; we use an alternative proof that slightly improves this parameter. See \Cref{lem:rnk-t-rnkt}.} Moreover, arguing again along the same lines as \cite{HMO26}, we can show that if $m$ is a witness for $\rnK^t(x) \leq \rnKt(x) + O(e(n))$ (namely, $|m| \leq \rnKt(x) + O(e(n))$ and $m$ is an $\rnK$-style description of $x$ with running time $t$), then $\rnKt(m \mid x) \leq O(e(n))$. At this point, with this statement for the setting of $\rKt$, the task of \cite{HMO26} is nearly complete: omitting some technical details concerning randomness in the complexity measure, one may essentially brute-force over programs of length $O(e(n))$ to obtain a list containing $m$, run each program on the list to see if it outputs $x$, and output the length of the shortest one that does.

In our setting, of course, the same idea does not work directly. Even with a candidate program $m$ in hand, we cannot simply run it to see if it outputs $x$. The same is true in the first part of the process, where we would brute-force over all $O(e(n))$-length strings and collect a list of their outputs. We are also not guaranteed that any such string represents a valid single-valued nondeterministic procedure. Luckily, these issues may be handled quite easily by the availability of an $\NP$ oracle. Roughly, with a string $v$ in hand, we can ask queries of the following form, where $r$ is chosen randomly and $U$ denotes a universal TM:
\begin{itemize}
    \item Does there exist a witness $w$ such that $U(v,w, r) = z$ with $z \neq \bot$ in time $t$?
    \item Do there exist two witnesses $w$ and $w'$ such that $U(v, w, r) = z$ and $U(v, w', r) = z'$ in time $t$, with $z \neq z'$?
    \item Does there exist a witness $w$ such that $U(v,w, r) = z \neq \bot$ in time $t$, and $z_i = 1$?
\end{itemize}
Clearly, if the first query is true and the second is false, then non-adaptive queries of the third type can be used to determine the unique nondeterministic output of $v$ on randomness $r$.

However, a more serious obstacle remains: if we use non-adaptive $\NP$ queries to produce a list of strings containing $m$, and then we ask another round of queries to determine the outputs of the strings on this list, then we have introduced a layer of adaptivity. At a high level, we overcome this with a careful restructuring of the procedure so that the second layer of queries does not depend on the strings collected in the first. For every string $v \in O(e(n))$, we really only need to know the length of the smallest $m$ such that $(v, x)$ outputs $m$ and $m$ outputs $x$. The key observation is that this kind of question can be collapsed into a single $\NP$ query, which does not require $m$ as input. Several auxiliary queries will be still be required: checking that $(v,x)$ and $m$ are single-valued nondeterministic, for example, and determining the value of their outputs on several independent choices of randomness to pinpoint the majority output in the sense of $\rnK$. Fortunately, this can all be handled non-adaptively. See \Cref{l:SoI-rnKt} for more detail.\\

As mentioned above, the proof of \Cref{thm:application_correspondence} follows from \Cref{thm:consequences_SoI} (Item 2) and \Cref{thm:equivalence_SoI} via a win-win argument. We refer to \Cref{sec:quick_lower_bound}.

\paragraph*{Lower bounds for computing $\mathsf{nKt}$ via symmetry of information (Theorem \ref{thm:intro_nKt_lb}).} 

We now discuss the ingredients in the proof that $\MnKtP \notin \NTIME[2^{\varepsilon n}] \cap \coNTIME[2^{\varepsilon n}]$: i.e., the first part of \Cref{thm:intro_nKt_lb}. We will build on an approach to lower bounds due to \cite{DBLP:conf/coco/Hirahara22}, which was originally used to prove that $\mathsf{MrKtP}\not\in\mathsf{BPTIME}[2^{\varepsilon n}]$ for some constant $\varepsilon>0$. The proof consists of two steps. In the first step, the author shows that $\mathsf{MrKtP}\in\mathsf{BPTIME}[2^{\varepsilon n}]$ implies SoI for $\rKt$, employing the ``meta-complexity $\implies$ SoI'' hybrid argument approach discussed above. The second step of the proof follows the technique of \cite{RonneburgerThesis} (along with the assumption that $\mathsf{MrKtP} \in \mathsf{BPTIME}[2^{\varepsilon n}]$) for showing that SoI does not hold. Combining these steps, \cite{DBLP:conf/coco/Hirahara22} reaches a contradiction and thus may conclude that $\mathsf{MrKtP}$ is hard.

A limitation of this argument, which blocks its direct translation to $\mathsf{MnKtP}$, is the randomness required in the reconstruction procedure of the PRG of \cite{RazRV02_jcss_journals, DBLP:conf/focs/Hirahara20}, which is employed in the first step. Randomized complexity measures such as $\rKt$ naturally incorporate the presence of randomness \emph{by definition}, but randomness is not allowed in the definition of $\nKt$. Thus, it is not straightforward to give an upper bound on $\nKt$ via reconstruction. (For the same reason, the argument does not generalize to prove $\MKtP$ lower bounds, and in fact it remains open whether $\MKtP \notin \P$ unconditionally.)

We get around this limitation in the case of $\nKt$ by making use of nondeterminism to derandomize reconstruction. As a first attempt, we notice that if $\MnKtP\in\NTIME[2^{\varepsilon n}]\cap\co\NTIME[2^{\varepsilon n}]$ for a small constant $\varepsilon > 0$, then for any $m \in \N$, any string $s$ such that $\nKt(s)\ge m$ will have $\MnKtP$-oracle circuit complexity at least $m/\polylog(m)$. In particular, arguing in the contrapositive, assume that $s$ is the truth-table of an $\MnKtP$-oracle circuit $C$ of size less than $m/(\log m)^2$. Then, we may encode the circuit $C$ with fewer than $m/\log m$ bits into an $\nKt$-description for $s$, and moreover, we may collapse the non-determinism required to simulate the assumed $\NTIME[2^{\varepsilon n}]\cap\co\NTIME[2^{\varepsilon n}]$ oracle gates into the non-determinism in the definition of $\nKt$. Overall, this implies
\begin{align*}
    \nKt(s) &< \frac{m}{\log m} + \log(m\cdot 2^{\varepsilon m})
    \\&< m. 
\end{align*}
To obtain such $s$ efficiently, again applying the assumption that $\MnKtP$ is easy, we can non-deterministically guess $s$ and refute when $\nKt(s)$ is not large enough. Then, we can plug $s$ into the (relativized) pseudo-random generator from \cite{IW97} to derandomize the reconstruction procedure of \cite{RazRV02_jcss_journals, DBLP:conf/focs/Hirahara20}, resulting in an $\nKt^{\MnKtP}$ upper bound rather than an $\rKt^{\MnKtP}$ upper bound. Next, we may again use the assumed easiness of $\MnKtP$ to collapse the oracle into the non-determinism in the definition of $\nKt$, resulting in an $\nKt$ upper bound without any oracle.

The technique in the previous paragraph can be used to obtain the lower bound $\MnKtP \notin \NTIME[T(n)]\cap\co\NTIME[T(n)]$ for $T(n)=2^{n^\varepsilon}$, but it does not suffice for $T(n)=2^{\varepsilon n}$. This is because, in order to derandomize the reconstruction procedure of \cite{RazRV02_jcss_journals, DBLP:conf/focs/Hirahara20} via \cite{IW97}, we would need a string of $\MnKtP$-oracle complexity some large polynomial in $n$, so we would need to query the $\MnKtP$-oracle at length $n^c$ for some constant $c \in \N$. Hence, the derandomized reconstruction would take time at least $T(n^c)$, and when $T(n)>2^{n^{1/c}}$, it holds that $T(n^c)>2^n$, yielding a trivial upper bound on the relevant $\nKt$ complexity. To obtain a lower bound with $T(n)=2^{\varepsilon n}$, a key observation is that the reconstruction procedure in question only requires $O(n)$ bits of randomness. Then, instead of \cite{IW97}, we apply the PRG of \cite{RazRV02_jcss_journals, DBLP:conf/focs/Hirahara20} \emph{for a second time}, instantiated with a different set of parameters, to derandomize the reconstruction procedure. The upshot is that a truth-table with $\nKt$ complexity $O(n)$ suffices to generate the $O(n)$ pseudo-random bits to fool reconstruction. We note that the analysis here is quite delicate. For example, one must keep careful track of the randomness length in the definition of $\rK$ as well as the input lengths of oracle queries. See Section \Cref{ss:derand-rK} for more details.

One problem still remains: the reconstruction procedure of the PRG only recovers $x$ with $1/\poly(n)$ probability, so we need some additional advice to identify $x$ after derandomization. But the advice depends on our non-deterministic choice of hard truth table $s$, so we cannot store it in the description.
Nevertheless, using ideas from pair-wise independent sampling (see e.g. \cite{DBLP:journals/eccc/ECCC-TR97-020}), we can amplify the success probability of reconstruction from $1/\poly(n)$ to $2/3$, but still only using $O(n)$ bits of randomness (see \Cref{ss:rK-rand-amp}). This eliminates the need for additional advice, as we can just output the string that appears the most often.

We have given a brief overview of the proof of the main part of \Cref{thm:intro_nKt_lb}. Next, we discuss the ``moreover'' part. We observe that the first step of the proof essentially extends to any dense language $L\subseteq\mathcal R_{\nKt[n]}$. On the other hand, for the second step of the proof, letting $y \in \bool^n$ be such that $\nKt(y) \geq n$ and $x \in \bool^n$ such that $\nKt(x \mid y) \geq n$, SoI guarantees that $\nKt(x,y) \gtrsim 2n$, but it does not guarantee that $(x,y)\in L$, so it is possible that we cannot find $x$ by brute-force with access to $L$. Fortunately, the first step of the proof actually gives us something stronger than SoI: for any $x, y \in \{0, 1\}^*$, letting $G_1$ and $G_2$ be PRGs with appropriate output lengths approximately $\nKt(x \mid y) \geq n$ and $\nKt(y)$, there exist $z,z'$ such that $G_1(x,z)\circ G_2(y,z')\in L$. (See \Cref{l:nKt-SoI} for more detail.) Then, given $y \in L$, we enumerate not only over $x$, but also over $z,z'$ to find $G_1(x, z) \circ G_2(y, z') \in L$.

This works for any dense language  $L\in\mathcal{R}_{\nKt[(1 - o(1)) \cdot n]}$, but not for $L\in\mathcal{R}_{\nKt[\gamma\cdot n]}$ for $\gamma\in(0,1)$ in general. In particular, assuming easiness of $L$, we can only find $y \in L$ such that $\nKt(y)\ge\gamma \cdot n$, and then we can find $x \in \{0, 1\}^n$ and $y' = G_1(x, z)\circ G_2(y_1, z') \in L$ such that
\[
    \nKt(y')\gtrsim\gamma\cdot(n + \nKt(y)) \geq  (\gamma+\gamma^2)\cdot n.
\]
We can see the problem here: when $\gamma\le(\sqrt5-1)/2 \approx 0.618$, we are not guaranteed that $\nKt(y')\ge(1+\Omega(1))\cdot n$, so we cannot derive a contradiction. To get around this obstacle, the key idea is to \emph{iterate the construction $k$ times} for a large enough constant $k$, applying our extended SoI statement again to $y'$, and so on. The detailed proof can be found in \Cref{ss:nKt-iter-cons}. On the other hand, \Cref{ss:limitations_gamma} discusses the limits of our techniques for $\gamma \leq 1/2$.

\paragraph*{Lower bounds for computing $\bm{\mathsf{rnKt}}$ via explicit constructions (Theorem \ref{thm:intro_rnKt_lbs}).}

We now move on to discuss the techniques behind \Cref{thm:intro_rnKt_lbs}. The proof of \Cref{i:rnkt_lb_2} naturally extends the lower bound proof for $\rKt$ in \cite{DBLP:conf/coco/Hirahara22}. Conceptually, this setting is simpler than the proof of \Cref{thm:intro_nKt_lb} since derandomization is not required. Still, the combination of randomness and nondeterminism makes some parts of the argument slightly more involved, and we include a complete proof in \Cref{ss:rnKt_lb_2} to keep the presentation self-contained. 

The proofs of \Cref{i:rnkt_lb_1} and \Cref{i:rnkt_lb_3} exploit new results from \cite{DBLP:conf/stoc/0001LL25}: namely, an \textit{explicit construction algorithm} for dense $\coAM$ properties in the form of Arthur-Merlin protocols with non-binary output. At the start of the protocol, prover $P$ (which is computationally unbounded) and verifier $V$ both receive an advice string $\alpha$. After $k = O(1)$ rounds of interaction, $V$ outputs a string $x$ based on $\alpha$, the public randomness $\{r_j\}_{j \in [k]}$, and the prover's responses $\{w_j\}_{j \in [k]}$. Such a protocol is said to be \emph{single-valued} if it satisfies certain \emph{conformity} and \emph{resiliency} properties; in particular, these properties imply that $(P, V)(\alpha)$ always outputs a canonical string $x$, and that, for any prover $P'$, $(P', V)(\alpha)$ belongs to $\{x, \bot\}$ with high probability. See \Cref{ss:rnKt_lb_1} for detail.

The main technical result of \cite{DBLP:conf/stoc/0001LL25} proves that for every dense enough $L \in \coAM$, there exists a sequence of strings $\{x_n\}$ and a non-binary public-coin protocol $(P, V)$ satisfying
\begin{itemize}
    \item $(P, V)$ is single-valued Arthur-Merlin protocol computing $\{x_n\}$;
    \item the advice complexity of $(P, V)$ is small, and $V$ can be computed efficiently;
    \item for infinitely many $n$, $x_n \in L \cap \{0, 1\}^n$. 
\end{itemize} 

A sketch of the proof of \Cref{i:rnkt_lb_1} is as follows. We first observe that, after applying round reduction to the protocol above, the sequence $\{x_n\}$ admits small $\rnK$ complexity: say, $\rnKt(x_n) \le \theta(n)$ for some $\theta(n) < n - 1$. Then, we consider any dense language $L \subseteq \mathcal{R}_{\rnKt[\theta(n) + 1]}$ and assume $L \in \coAM$. By the existence of an Arthur-Merlin protocol as above, we know that for infinitely many $n$, $L \cap \{0, 1\}^n$ contains a string with $\rnKt$ complexity smaller than $\theta(n) + 1$, contradicting the definition of $L$. In our formal proof of \Cref{i:rnkt_lb_1}, we extend this idea to quasi-polynomial time regime by a simple padding argument.

One important corollary of the explicit construction protocol in \cite{DBLP:conf/stoc/0001LL25} states that $(\AMEXP \cap \coAMEXP)/2^{n^\varepsilon}\nsubseteq \SIZE^{\mathcal{O}}[2^n/n]$ for any oracle $\mathcal{O}\in\AM\cap\coAM$ (see \cite[Theorem 1.2]{DBLP:conf/stoc/0001LL25}). This follows from the close connection between explicit constructions and circuit lower bounds: roughly speaking, if $\Pi_{\mathsf{hard}}$ is defined as the set of strings in $\{0,1\}^{2^n}$ that are not truth-tables of $\mathcal{O}$-oracle circuits of size at most $2^n/n$, then $\Pi_{\mathsf{hard}}$ is dense and $\Pi_{\mathsf{hard}}\in\coAM$. By the explicit construction protocol from above, one obtains a sequence of truth-tables defining a language $L_{\mathsf{hard}}$ such that  $L_{\mathsf{hard}}\in(\AMEXP \cap \coAMEXP)/2^{n^\varepsilon}$ but $L_{\mathsf{hard}} \notin\SIZE^\mathcal{O}[2^n/n]$.

Our idea for proving \Cref{thm:intro_rnKt_lbs}, \Cref{i:rnkt_lb_3} is to instantiate the PRG of \cite{IW97} with the truth-table of $L_{\mathsf{hard}}$, obtaining a family of multi-sets $\{T_s\}_{s \in \mathbb{N}}$ which are pseudorandom against $(\AM \cap \coAM)$\textnormal{-}oracle circuits of size $s$ on infinitely many input lengths $s$. 

Note that, because $L_{\mathsf{hard}}$ is computable in $(\AMEXP \cap \coAMEXP)/{2^{n^\varepsilon}}$ and the \cite{IW97} generator is polynomial-time computable, every element inside $T_s$ has $\rnKt$ complexity bounded by $2^{O(\log^\varepsilon s)}$. Therefore, any dense language $L$ consisting only of strings with large $\rnKt$ complexity distinguishes $T_s$ from uniformly random strings, for every sufficiently large $s$. The pseudorandom condition then gives $L \notin \SIZE^{\AM \cap \coAM}[\poly(n)]$. Similarly to \Cref{i:rnkt_lb_1}, the lower bound is later amplified to quasi-polynomial by padding argument.

\paragraph{Unconditional failure of conditional SoI (Theorem \ref{thm:failure-cond-SoI})}

Lastly, we describe the proof of \Cref{thm:failure-cond-SoI}, which states that conditional SoI does not hold: for any error term $e(n) = \poly\log n$, there exist strings $x, y$, and $w$ such that 
\[\nKt(x, y \mid w) < \nKt(x \mid y, w) + \nKt(y \mid w) - e(|x| + |y|).\]
The proof of this statement builds on the unconditional lower bound for computing $\gapmnktp$ as in \Cref{thm:intro_nKt_lb}. That is, assuming conditional SoI, we exhibit an algorithm for $\gapmnktp$ contradicting this lower bound.

As a first attempt, assume conditional SoI, and consider the following deterministic algorithm for $\gapmnktp[0.6n, 0.9n]$.
\begin{quote}

    On input $x \in \bool^n$, partition $x$ into $k = n/(100e(n))$ blocks $(x_1, ..., x_k)$, with each $x_i$ being of length roughly $100e(n)$. For each $i \in [k]$, compute $v_i := \nKt(x_i \mid x_1, ..., x_{i - 1})$ by brute force, and accept iff $\sum_{i \in [k]}v_i \leq 0.75 n$.
    
\end{quote}
On one hand, it holds trivially that
\[\nKt(x) \leq \sum_{i \in [k]}\nKt(x_i \mid x_1, ..., x_{i - 1}) + O(\log n) = \sum v_i.\]
On the other hand, by our assumption of conditional SoI applied $k$ times,\footnote{In fact, standard (non-conditional) SoI would suffice for this step.} we have that 
\[\nKt(x) \geq \sum\nKt(x_i \mid x_1, ..., x_{i - 1}) - k\cdot e(n) = \sum v_i - \frac{n}{100}.\]
Since $\sum v_i$ closely approximates $\nKt(x)$, the algorithm above clearly solves $\gapmnktp[0.6n, 0.9n]$.

The problem with this approach is that our algorithm is highly inefficient. Specifically, to determine $\nKt(x_i \mid x_1, ..., x_{i - 1})$ by brute force, one must consider non-deterministic programs running in time up to $2^{100e(n)} = 2^{\poly\log n}$ and therefore must exhaustively search over witnesses of that length, taking time $2^{2^{\poly \log n}}$. The lower bound of \Cref{thm:intro_nKt_lb} does not rule out an algorithm with this time complexity.

To improve the running time of the algorithm, we argue again as in \cite{HMO26} and the proof of \Cref{thm:equivalence_SoI}; namely, we argue that $\rnKt$ approximately upper-bounds $\rnK^t$. Crucially, by assuming conditional SoI, we get a stronger version of the statement that allows us to condition on an auxiliary string $w$. In particular, for any strings $x, w$, 
\[\nK^{O(2^{e(|x|)})}(x \mid w) \leq \nKt(x \mid w) + O(e(|x|)).\]

Returning to our algorithm for $\gapmnktp[0.6n, 0.9n]$, instead of computing $\nKt(x_i \mid x_1, ..., x_{i - 1})$ by brute force, we may instead compute $\nK^{t}(x_i \mid x_1, ..., x_{i - 1})$ by brute force, where $t = O(2^{e(|x_i|)}) = O(2^{e(100e(n))})$. This only incurs an additional additive error term of $n/(\poly \log n)$. Moreover, the algorithm now takes time only $2^{100e(n)}\cdot 2^{O(2^{e(100e(n))})} \leq 2^{2^{\poly(\log \log n)}} = 2^{o(n)}$, contradicting the lower bound of \Cref{thm:intro_nKt_lb}.

\section{Preliminaries}\label{sec:preliminaries}

\subsection{Notation}

We denote the empty string by $\epsilon$, and the length of a string $x$ by $|x|$.  
For strings $x \in \{0,1\}^n$ and $y \in \{0,1\}^m$, their concatenation is written as $x \circ y \in \{0,1\}^{n+m}$.  

If $x \sim \mathcal{D}$, then $x$ is sampled from the distribution $\mathcal{D}$.  
For any $a$ in the support of $\mathcal{D}$, $\mathcal{D}(a)$ denotes its probability.  
We write $\mathcal{U}_\ell$ for the uniform distribution over $\{0,1\}^\ell$.

\subsection{Computational Model}

We consider multi-tape Turing machines over the alphabet $\Sigma = \{0,1,\bot\}$. In addition to the usual input tape and work tapes, in some contexts we allow a machine to access three additional tapes: one holding an auxiliary string, one holding a non-deterministic guess, and one holding the randomness. We assume that the machine has \emph{sequential access} (i.e., no random access) to the randomness tape and the non-deterministic tape.\footnote{To avoid confusion, sequential access means that the tape head can remain at a given tape cell or move left or right at each step of the computation. In particular, we do not require read-only or one-way access to the tape.}

For a machine $M$, time bound $t$, and string $x$, we let $\mathsf{K}^t_M(x)$ be the minimum length of a string $p \in \{0,1\}^*$ such that $M$ outputs $x$ on input string $p$ when running for at most $t$ steps. We say that a (universal) machine $U$ is \emph{time-efficient} if for every machine $U'$ there is a constant $C$ such that, for every string $x$ and time bound $t$, we have $\mathsf{K}_U^{C \cdot t \log t}(x) \leq \mathsf{K}_{U'}^{t}(x) + C$. It is known that time-efficient machines exist (see, e.g.,~\citep{DBLP:series/txcs/LiV19}). We fix such a machine $U$ when specifying Kolmogorov complexity.

We also assume an encoding function $\langle \cdot \rangle$ for machines (which we omit for notational simplicity) such that $U$, when given (the encoding of) a machine $M$, an input $\alpha$, an auxiliary string $x$, a non-deterministic guess $w$, and a random string $r$, runs in time $\widetilde{O}(|M|, t)$ and outputs the string produced by $M(\alpha; x; w; r)$, where $|M|$ denotes the encoding length of $M$ and $t$ is the number of steps $M$ runs on $(\alpha; x; w; r)$.

A \emph{program} $\Pi$ is defined as a pair $\Pi \vcentcolon= (M, \alpha)$, where $M$ is a Turing machine and $\alpha \in \bool^*$ is an ``advice'' string. The output of $\Pi$ on $(x; w; r)$ (which corresponds to the auxiliary input, the non-deterministic guess, and the randomness), denoted $\Pi(x; w; r)$, is given by $M(\alpha; x; w; r)$.

The \emph{running time} of the program $\Pi=(M,\alpha)$ on input $(x, w, r)$ is defined to be the number of steps $U$ runs on input $\mleft(\langle M\rangle;\alpha; x; w; r\mright)$, where $\langle M \rangle$ is an encoding of the machine $M$.\footnote{As noted above, we assume a fixed standard encoding for Turing machines.}

We also define the \emph{size} of the program $\Pi=(M,\alpha)$. To encode a machine-advice pair $(M, \alpha)$, we use the following format:
\[
s_\alpha \,01\, \langle M \rangle\, \alpha,
\]
where $s_\alpha$ is the binary representation of the integer specifying the length of $\alpha$, with each bit duplicated. The size of $\Pi$, denoted $|\Pi|$, is the total length of this description.

We say that a program $\Pi$ \emph{non-deterministically} outputs $y$ on input $(x; r)$ in time $t$ if:
\begin{itemize}
    \item $\forall w \in \{0,1\}^t$, $\Pi(x; w; r)$ outputs either $y$ or $\bot$ in time $t$, and
    \item $\exists w \in \{0,1\}^t$, $\Pi(x; w; r)$ outputs $y$ in time $t$.
\end{itemize}

\subsection{Time-Bounded Kolmogorov Complexity}\label{ss:time_bounded_definitions}

We fix a time-efficient universal machine $U$, and define time-bounded Kolmogorov complexity with respect to $U$ (see \citep{DBLP:series/txcs/LiV19} for the required background in Kolmogorov complexity). In other words, whenever we say that a program $\Pi$ outputs $x$ in $t$ steps, formally, this means that $U(\Pi) = x$ when $U$ runs for at most $t$ steps on input $\Pi$. Our results are robust with respect to encoding choices.

Next, we provide formal definitions for various notions of time-bounded Kolmogorov complexity. We begin with deterministic versions of time-bounded Kolmogorov complexity.

	\begin{definition}[$\mathsf{K}^t$]
	For $x,y,r\in\bool^*$, $t\in\mathbb{N}$, and a oracle $\mathcal{O}$, the \emph{$t$-time-bounded Kolmogorov complexity of $x$ conditioning on $(y;r)$ and given oracle $\mathcal{O}$} is defined as 
		\[
		\mathsf{K}^{t,\mathcal{O}}(x\mid y; r) \eqdef \min_{\Pi \in \{0,1\}^*} \mleft\{|\Pi|  \,\,\middle\vert\,\,  \Pi^{\mathcal{O}}(y;\epsilon;r) \textnormal{ outputs $x$ within $t$ steps}\mright\}.
		\]
        We define $\K^{t,\mathcal{O}}(x\mid y)$ as $\K^{t,\mathcal{O}}(x\mid y;\epsilon)$, and $\K^{t,\mathcal{O}}(x)$ as $\K^{t,\mathcal{O}}(x\mid \epsilon)$.
	\end{definition}

        \begin{definition}[$\mathsf{Kt}$]
	For $x,y,r\in\bool^*$ and a oracle $\mathcal{O}$, the \emph{time-bounded Levin--Kolmogorov complexity of $x$ conditioning on $(y;r)$ and given oracle $\mathcal{O}$} is defined as 
		\[
		\mathsf{Kt}^{\mathcal{O}}(x\mid y;r) \eqdef \min_{\substack{\Pi \in \{0,1\}^* \\ t\in\mathbb{N}}} \mleft\{|\Pi| + \lceil \log t \rceil\,\,\middle\vert\,\,  \Pi^{\mathcal{O}}(y;\epsilon;r) \textnormal{ outputs $x$ within $t$ steps}\mright\}.
		\]
        We define $\Kt^{\mathcal{O}}(x\mid y)$ as $\Kt^{\mathcal{O}}(x\mid y;\epsilon)$, and $\Kt^{\mathcal{O}}(x)$ as $\K^{\mathcal{O}}(x\mid \epsilon)$.
	\end{definition}
    
    We can consider the case where the oracle calls are \emph{non-adaptive}. In this case, we will write $\mathsf{K}^{t,\parallel\mathcal{O}}(x)$ or $\mathsf{Kt}^{\parallel\mathcal{O}}(x)$.

Next, we will define other variants of the time-bounded Kolmogorov complexity measure.
For simplicity, we will define only the basic versions without oracles. These can easily be generalized to settings where an oracle is present.

\paragraph{Time-Bounded Kolmogorov Complexity with Randomness.}
We now define probabilistic variants of time-bounded Kolmogorov complexity. These definitions include a parameter $\lambda$ in the subscript to denote the ``success probability''. We will often omit the subscript when $\lambda = 2/3$.

    \begin{definition}[$\mathsf{rK}^t$]
    For $x,y\in\bool^*$, $\lambda\in [0,1]$, $t\in\mathbb{N}$, the \emph{randomized $t$-time-bounded Kolmogorov complexity of $x$ conditioning on $y$} is defined as
   \[
        \rK^{t,\mathcal{O}}_\lambda(x \mid y) = \min_{\Pi\in\bool^*}\mleft\{|\Pi|
  \,\,\middle\vert\,\,  \Prob_{r\sim\bool^{t}}\mleft[\Pi(y;\epsilon;r) \textnormal{ outputs $x$ within $t$ steps}\mright]\geq \lambda\mright\}.
    \]
	\end{definition}

        \begin{definition}[$\mathsf{rKt}$]
    For $x,y\in\bool^*$, $\lambda\in [0,1]$, the \emph{randomized time-bounded Levin--Kolmogorov complexity of $x$ conditioning on $y$} is defined as
   \[
        \rKt^{\mathcal{O}}_\lambda(x \mid y) = \min_{\substack{\Pi \in \{0,1\}^* \\ t\in\mathbb{N}}}\mleft\{|\Pi|  + \lceil \log t \rceil 
  \,\,\middle\vert\,\,  \Prob_{r\sim\bool^{t}}\mleft[\Pi(y;\epsilon;r) \textnormal{ outputs $x$ within $t$ steps}\mright]\geq \lambda\mright\}.
    \]
	\end{definition}

    	\begin{definition}[$\mathsf{pK}^t$]
    For $x,y\in\bool^*$, $\lambda\in [0,1]$, $t\in\mathbb{N}$, the \emph{probabilistic $t$-time-bounded Kolmogorov complexity of $x$ conditioning on $y$} is defined as
		\[
		\mathsf{pK}^t_{\lambda}(x\mid y) \eqdef \min \mleft\{s\in\mathbb{N}  \,\,\middle\vert\,\,  \Prob_{r\sim\bool^{t}}\mleft[\exists\Pi\in\bool^s\text{ s.t. }\Pi(y;\epsilon;r) \textnormal{ outputs $x$ within $t$ steps}\mright]\geq \lambda\mright\}.
        \]
    Equivalently,
    \[
        \pK^t_\lambda(x\mid y)\eqdef\min\mleft\{s\in\mathbb{N}\;\middle\vert\;\Prob_{r\sim\{0,1\}^t}\mleft[\K^t(x\mid y;r)\le s\mright]\ge\lambda\mright\}.
    \]
	\end{definition}

    \begin{definition}[$\mathsf{pKt}$]
    For $x,y\in\bool^*$, $\lambda\in [0,1]$, $t\in\mathbb{N}$, the \emph{probabilistic time-bounded Levin--Kolmogorov complexity of $x$ conditioning on $y$} is defined as
    \[
        \pKt_\lambda(x\mid y)\eqdef\min\mleft\{s\in\mathbb{N}\;\middle\vert\;\Prob_{r\sim\{0,1\}^{2^s}}\mleft[\Kt(x\mid y;r)\le s\mright]\ge\lambda\mright\}.
    \]
	\end{definition}

    \paragraph{Time-Bounded Kolmogorov Complexity with Nondeterminism.}
Next, we define nondeterministic variants of time-bounded Kolmogorov complexity.

	\begin{definition}[$\mathsf{nK}^t$]
		For $x,y,r\in\bool^*$ and $t\in\mathbb{N}$, the \emph{nondeterministic $t$-time-bounded Kolmogorov complexity of $x$ conditioning on $(y;r)$} is defined as
		\[
		\mathsf{nK}^t(x\mid y;r) \eqdef \min_{\Pi \in \{0,1\}^*} \mleft\{|\Pi|  \,\,\middle\vert\,\,  
	\begin{aligned}
		&\forall w\in\bool^t, \Pi(y;w;r) \textrm{ outputs $x$ or $\bot$ in $t$ steps}; \\
        &\exists w \in \{0, 1\}^t, \Pi(y;w;r) \textrm{ outputs $x$ in $t$ steps}
	\end{aligned} 
		\mright\}.
		\]
        We define $\nK^t(x\mid y)$ as $\nK^t(x\mid y;\epsilon)$, and $\nK^t(x)$ as $\nK^t(x\mid \epsilon)$.
	\end{definition}

    	\begin{definition}[$\mathsf{nKt}$]
		For $x,y,r\in\bool^*$, the \emph{nondeterministic time-bounded Levin--Kolmogorov complexity of $x$ conditioning on $(y;r)$} is defined as
		\[
		\mathsf{nKt}(x\mid y;r) \eqdef \min_{\substack{\Pi \in \{0,1\}^* \\ t\in\mathbb{N}}} \mleft\{|\Pi|   + \lceil \log t \rceil \,\,\middle\vert\,\,  
	\begin{aligned}
		&\forall w\in\bool^t, \Pi(y;w;r) \textrm{ outputs $x$ or $\bot$ in $t$ steps}; \\
        &\exists w \in \{0, 1\}^t, \Pi(y;w;r) \textrm{ outputs $x$ in $t$ steps}
	\end{aligned} 
		\mright\}.
		\]
        We define $\nKt(x\mid y)$ as $\nKt(x\mid y;\epsilon)$, and $\nKt(x)$ as $\nKt(x\mid \epsilon)$.
	\end{definition}

        \paragraph{Time-Bounded Kolmogorov Complexity with Randomness and Nondeterminism.}
Finally, we define time-bounded Kolmogorov complexity notions that incorporate both randomness and nondeterminism.
    \begin{definition}[$\mathsf{rnK}^t$]
	For $x,y\in\bool^*$, $\lambda\in[0,1]$ and $t\in\mathbb{N}$, the \emph{randomized nondeterministic $t$-time-bounded time-bounded Kolmogorov complexity of $x$ conditioning on $y$} is defined as
	\[
	\mathsf{rnK}_\lambda^t(x\mid y) \eqdef \min_{\Pi \in \{0,1\}^*} \mleft\{|\Pi|  \,\,\middle\vert\,\,   \Prob_{r\sim\bool^t}\mleft[
	\begin{aligned}
		&\forall w\in\bool^t, \Pi(y;w;r) \textrm{ outputs $x$ or $\bot$ in $t$ steps}; \\
        &\exists w \in \{0, 1\}^t, \Pi(y;w;r) \textrm{ outputs $x$ in $t$ steps}
	\end{aligned} \mright] \geq\lambda
	\mright\}.
	\]
	\end{definition}

        \begin{definition}[$\mathsf{rnKt}$]
	For $x,y\in\bool^*$, $\lambda\in[0,1]$, the \emph{randomized nondeterministic time-bounded Levin--Kolmogorov complexity of $x$ conditioning on $y$} is defined as
	\[
	\mathsf{rnKt}_\lambda(x\mid y) \eqdef \min_{\substack{\Pi \in \{0,1\}^* \\ t\in\mathbb{N}}} \mleft\{|\Pi|  + \lceil \log t \rceil\,\,\middle\vert\,\,   \Prob_{r\sim\bool^t}\mleft[
	\begin{aligned}
		&\forall w\in\bool^t, \Pi(y;w;r) \textrm{ outputs $x$ or $\bot$ in $t$ steps}; \\
        &\exists w \in \{0, 1\}^t, \Pi(y;w;r) \textrm{ outputs $x$ in $t$ steps}
	\end{aligned} \mright] \geq\lambda
	\mright\}.
	\]\end{definition}

\begin{definition}[$\mathsf{pnK}^t$]\label{def:pnK-t}
    For $x,y\in\bool^*$, $\lambda\in [0,1]$ and $t\in\mathbb{N}$, the \emph{probabilistic nondeterministic $t$-time-bounded time-bounded Kolmogorov complexity of $x$ conditioning on $y$} is defined as
		\[
		\mathsf{pnK}_{\lambda}^t(x\mid y) \eqdef \min\mleft\{s\in\mathbb{N} \,\,\middle\vert\,\,   \Prob_{r\sim\bool^{t}}\mleft[ 
        	\begin{aligned}
        &\exists \Pi\in\bool^{\leq s} \textrm{ such that the following conditions hold:}\\
		&\forall w\in\bool^t, \Pi(y;w;r) \textrm{ outputs $x$ or $\bot$ in $t$ steps}; \\
        &\exists w \in \{0, 1\}^t, \Pi(y;w;r) \textrm{ outputs $x$ in $t$ steps}
	\end{aligned}  \mright] \geq \lambda 
		\mright\}.
        \]
    Equivalently,
    \[
        \pnK^t_\lambda(x\mid y)\eqdef\min\mleft\{s\in\mathbb{N}\;\middle\vert\;\Prob_{r\sim\{0,1\}^t}\mleft[\nK^t(x\mid y;r)\le s\mright]\ge\lambda\mright\}.
    \]
    \end{definition}

        \begin{definition}[$\mathsf{pnKt}$]
    For $x,y\in\bool^*$, $\lambda\in [0,1]$, $t\in\mathbb{N}$, the \emph{probabilistic nondeterministic time-bounded Levin--Kolmogorov complexity of $x$ conditioning on $y$} is defined as
    \[
        \pnKt_\lambda(x\mid y)\eqdef\min\mleft\{s\in\mathbb{N}\;\middle\vert\;\Prob_{r\sim\{0,1\}^{2^s}}\mleft[\nKt(x\mid y;r)\le s\mright]\ge\lambda\mright\}.
    \]
	\end{definition}

\subsection{Meta-Computational Problems}

We let $\MKtP$ denote the set of input pairs $(x,1^s)$ such that $x \in \{0,1\}^*$ and $\Kt(x) \leq s$.

We also consider the promise problem $\mathsf{GapMnKtP}[\alpha, \beta]$ on inputs $x\in\{0,1\}^*$ defined by
\[
\textsc{YES}\eqdef\{x:\nKt(x)\le \alpha(|x|)\},\qquad \textsc{NO}\eqdef\{x:\nKt(x)\ge \beta(|x|)\}.
\]

Similarly, we can consider meta-computational problems associated with  other Kolmogorov complexity measures, e.g., $\mathsf{GapMKtP}[\alpha,\beta]$ captures the problem of estimating $\Kt$ complexity.

\subsection{Basic Tools in Meta-Complexity}

Here, we state some useful tools about Kolmogorov complexity.

\begin{lemma}[Standard Counting Argument]\label{l:incomp-strings}
    For any $\lambda \in (1/2, 1)$, $t\in\mathbb{N}$, and $\kappa \in \{\K, \K^t, \rK^t_\lambda, \nK^t, \rnK^t_\lambda\}$, and for all $n, \gamma \in \mathbb{N}$, we have
    \[
        \Prob_{x \sim \{0,1\}^n} \left[ \kappa(x) < n - \gamma \right] < \frac{1}{2^\gamma}.
    \]
    Also, for any $\lambda \in (0, 1)$ and $\kappa \in \{\pK_{\lambda}^t, \pnK_{\lambda}^t\}$, and for all $n, \gamma \in \mathbb{N}$, we have
    \[
        \Prob_{x \sim \{0,1\}^n} \left[ \kappa(x) < n - \gamma \right] < \frac{1}{\lambda \cdot 2^\gamma}.
    \]
    Moreover, the inequalities above continue to hold in the presence of an oracle $\mathcal{O}$ and when conditioning on a fixed string $y$.
\end{lemma}

\begin{proof}

    For simplicity, we consider the oracle-free case and $y = \epsilon$. (The general case is analogous.) The case of $\K$ directly follows from counting. Indeed, the number of programs with length at most $n - \gamma$ is less than $2^{n - \gamma}$, therefore
    \[
        \Prob_{x \sim \{0, 1\}^n}\left[\K(x) < n - \gamma\right] < \frac{2^{n - \gamma}}{2^n} = \frac{1}{2^{\gamma}}.
    \]

    The remaining cases (except $\kappa = \rK^t_\lambda$ and $\kappa = \rnK^t_\lambda$) follows similarly, since each randomized programs can correspond to at most one string under the setting that $\lambda > 1/2$.

    For $\kappa = \pK^t_\lambda, \pnK^t_\lambda$, consider a bipartite graph where the left vertices represent randomness-program pairs  $\{0,1\}^t \times \{0,1\}^{<n - \gamma}$, and the right vertices represent strings in $\{0,1\}^n$. There is an edge between a left vertex $(r, M)$ and a right vertex $y$ if $M(x)$ outputs $x$. For a right vertex $x$ to have $\pK_{\lambda}^t$-complexity at most $n - \gamma$, it must have degree at least $2^t \cdot \lambda$. 
    
    However, there are less than $2^t \cdot 2^{n - \gamma}$ total edges in the graph. Therefore, the number of right vertices with degree at least $2^t \cdot \lambda$ is less than $2^{n-\gamma}/\lambda$.
\end{proof}

	\begin{lemma}[See {\cite[Lemma 9]{DBLP:conf/stoc/HiraharaILNO23}}]\label{f:incompressible}
	There exists a universal constant $b>0$ such that for any distribution family $\{\D_n\}_{n\in\mathbb{N}}$, where each $\D_n$ is over $\{0,1\}^{n}$, and $\gamma\in\mathbb{N}$,
	\[
	\Prob_{x\sim \D_n}\mleft[\mathsf{K}(x)< \log\frac{1}{\D_n(x)}-\gamma\mright]< \frac{n^b}{2^{\gamma}}.
	\]
\end{lemma}

\begin{theorem}[Efficient Coding Theorem~\cite{DBLP:conf/icalp/LuOZ22}]\label{t:coding_pkt}
	For every polynomial-time samplable distribution family $\{\mathcal{D}_n\}_{n\in\mathbb{N}}$, where each $\mathcal{D}_n$ is supported over $\bool^n$, there exists a polynomial $p$ such that for every $x\in\mathsf{Support}(\mathcal{D}_n)$,
	\[
	\mathsf{pK}^{p(n)}(x) \,\leq\,  \log\frac{1}{\mathcal{D}_n(x)} +  \log p(n).
	\]
\end{theorem}

\begin{lemma}[Following {\cite[Lemma 18]{GKLO22}}]\label{l:K-pK}
	For every computable language $L$, there exists a constant $c>0$ such that for any $x\in\bool^*$ and time bound $t\in\mathbb{N}$, $\K(x\mid t) \leq \pK^{t,L}(x) + c\cdot \log |x|$.
\end{lemma}

\subsection{Pseudorandomness}

For a finite set $A$, we write $x \sim A$ to denote that $x$ is uniformly distributed over $A$.

Let $\mathcal{D}_n$ be a distribution supported over $\{0,1\}^n$. Let $\varepsilon \in [0,1]$. Finally, let $f \colon \{0,1\}^n \to \{0,1\}$. We say that $\mathcal{D}_n$ $\varepsilon$\emph{-fools} $f$ if
\[
\Big | \Prob_{x \sim \mathcal{D}_n}[f(x) = 1] - \Prob_{x \sim \{0,1\}^n}[f(x) = 1]  \Big | \leq \varepsilon.
\]

\begin{theorem}[\cite{DBLP:journals/cc/BabaiFNW93, DBLP:journals/siamcomp/KlivansM02}; see also Theorem 17 of \cite{DBLP:journals/siamcomp/AllenderBKMR06}]\label{t:BFNW93}
	For every $\varepsilon>0$ and $b\in\mathbb{N}$, there exist $\delta<\varepsilon$, a constant $a \in \nat$, and a function
	\[
	F\colon\{0,1\}^{2^{n^\delta}}\times\{0,1\}^{n^\varepsilon}\to\{0,1\}^{n^b},
	\]
    computable in time $\poly(2^{n^{\delta}})$ and satisfying the following for all sufficiently large $n \in \nat$.

	For any function $f: \bool^{n^{\delta}} \to \bool$, let $G_f(\text{-})$ denote $F(\tt(f), \text{-})$. Then for any such $f$ and $D: \bool^{n^b} \to \bool$ such that 
    \[\left | \Pr_{z \sim \mathcal{U}_{n^{\varepsilon}}}[D(G_f(z)) = 1] - \Pr_{w \sim \mathcal{U}_{n^b}}[D(w) = 1]\right | \geq \frac{1}{n^b},\]
    there exists a non-adaptive oracle circuit $R^{(\text{-})}$ of size $n^a$ such that $R^D(x) = f(x)$ for all $x \in \bool^{n^{\delta}}$.
    
    As a consequence, there exists a constant $c \in \nat$ (chosen so that $n^{c\delta}$ exceeds $n^a\cdot n^b$) such that for any oracle $O$, if $f: \bool^{n^{\delta}} \to \bool$ is a function with $O$-oracle Boolean circuit complexity at least $n^{c\delta}$, then $G_f: \bool^{n^{\varepsilon}} \to \bool^{n^b}$ is a pseudorandom generator that $(n^{-b})$-fools every $O$-oracle Boolean circuit $D: \{0, 1\}^{n^b} \to \{0, 1\}$ of size at most $n^b$.

\end{theorem}

\begin{theorem}[\cite{IW97}]{\label{t:IW97-PRG}}
     For every $\varepsilon > 0$, there exist a polynomial-time computable function $F$ and $c, d \in \mathbb{N}$ such that such that the following holds for any oracle $O$.
    \begin{equation*}
        F: \{0, 1\}^{n^c} \times \{0, 1\}^{d \log n} \to \{0, 1\}^n,
    \end{equation*}
    and if $f: \{0, 1\}^{c\log n} \to \{0, 1\}$ is a function with $O$-oracle Boolean circuit complexity at least $n^{\varepsilon c}$, then the function $G_f(\text{-}) := F(\tt(f), \text{-})$ is a pseudorandom generator mapping $\{0, 1\}^{d\log n}$ to $\{0, 1\}^n$ that $(1/n)$-fooles every $O$-oracle Boolean circuit $D: \{0, 1\}^n \to \{0, 1\}$ of size at most $n$.
\end{theorem}

\begin{theorem}[\cite{DBLP:journals/jcss/ImpagliazzoW01, DBLP:journals/cc/TrevisanV07}]\label{t:io-PRG}
	Let $O$ be any oracle. Assume $\mathsf{PSPACE}\not\subseteq\BPP^{O}$. Then for every $\varepsilon>0$ and $b \in\mathbb{N}$, there is a sequence $\{G_n\}_{n \in\mathbb{N}}$, where $G_{n}\colon \bool^{n^{\varepsilon}} \to \bool^{n^{b}}$ is computable in time $2^{O(n^{\varepsilon})}$, such that the following holds. For every distribution family $\{\mathcal{C}_{n}\}_{n\in\mathbb{N}}$ of Boolean $O$-oracle circuits samplable in time $n^b$, there are infinitely many $n\in\mathbb{N}$ such that with probability at least $1-n^{-b}$ over $C$ sampled from $\mathcal{C}_{n}$, $G_n$ $(n^{-b})$-fools $C$.
\end{theorem}

\begin{corollary}\label{t:io-PRG-EXP}
		Let $O\in\PSPACE$. Assume $\mathsf{EXP}\not\subseteq\BPP^{O}$. Then for every $\varepsilon>0$ and $b \in\mathbb{N}$, there is a sequence $\{G_n\}_{n \in\mathbb{N}}$, where $G_{n}\colon \bool^{n^{\varepsilon}} \to \bool^{n^{b}}$ is computable in time $2^{O(n^{\varepsilon})}$, such that the following holds. For every distribution family $\{\mathcal{C}_{n}\}_{n\in\mathbb{N}}$ of Boolean $O$-oracle circuits samplable in time $n^b$, there are infinitely many $n\in\mathbb{N}$ such that with probability at least $1-n^{-b}$ over $C$ sampled from $\mathcal{C}_{n}$, $G_n$ $(n^{-b})$-fools $C$.
\end{corollary}
\begin{proof}[Proof Sketch]
Assume $\mathsf{EXP} \ne \BPP^{O}$. We consider two cases.

If $\mathsf{PSPACE} \not\subseteq \BPP^{O}$, then the desired PRG follows directly from \Cref{t:io-PRG}.

Otherwise, suppose $\mathsf{PSPACE} \subseteq \BPP^{O}$. Combined with the assumption that $\mathsf{EXP} \not\subseteq \BPP^{O}$, this implies that $\mathsf{EXP} \ne \mathsf{PSPACE}$. We claim that $\mathsf{EXP} \nsubseteq \SIZE^{O}[\poly]$. Indeed, if $\mathsf{EXP} \subseteq \SIZE^{O}[\poly]$, then
\begin{align*}
	\mathsf{EXP} 
	&\subseteq \mathsf{MA}^{O} \tag{by \cite{DBLP:journals/cc/BabaiFNW93}} \\
	&\subseteq \mathsf{PSPACE},
\end{align*}
where the last inclusion follows from the fact that $O \in \mathsf{PSPACE}$, so both the $\mathsf{MA}$ computation and the oracle queries can be simulated within $\mathsf{PSPACE}$. This yields $\mathsf{EXP} \subseteq \mathsf{PSPACE}$, a contradiction.

Therefore, $\mathsf{EXP} \nsubseteq \SIZE^{O}[\poly]$. Applying \Cref{t:BFNW93}, we obtain an infinitely-often secure PRG with sub-polynomial seed length based on the hard truth table of a language in $\mathsf{EXP}$.
\end{proof}

We will also need the following black-box pseudorandom generator. We closely follow the construction in \cite[Section~4.2]{DBLP:conf/focs/Hirahara20}, which uses the weak design from \cite{RazRV02_jcss_journals} and the list-decodable error-correcting code from \cite{DBLP:journals/jc/Sudan97}. But our construction is different in two ways. First, in \cite{DBLP:conf/focs/Hirahara20}, ECC list-decodable for distance $1/2-1/2m^2$ is used. Here we have an additional parameter $\varepsilon$ (which is the advantage of the distinguisher), and we use the same family of ECC, but the list-decodable distance is set as $1/2-\varepsilon/2m$. Second, in \cite{DBLP:conf/focs/Hirahara20}, the reconstruction algorithm outputs a list of strings. Here, after we obtain the list of strings, we just take a uniformly random sample from the list.
\begin{lemma}\label{l:black-box-PRG}
    For all sufficiently large $n,m\in\mathbb{N}$ such that $m\le 2n$ and any $\varepsilon>0$, there exists a triple $(G,A,R^{(\text{-})})$ such that
    \begin{align*}
        G&:\{0,1\}^n\times\{0,1\}^d\to\{0,1\}^m,\\
        A&:\{0,1\}^n\times\{0,1\}^d\to\{0,1\}^m,\\
        R^{(\text{-})}&:\{0,1\}^m\times\{0,1\}^d\times\{0,1\}^r\to\{0,1\}^n,
    \end{align*}
    and for every $x\in\{0,1\}^n$ and any $D:\{0,1\}^m\to\{0,1\}$ such that
    \[
        \Prob_{z\sim\mathcal{U}_d}\left[D(G(x;z))=1\right]-\Prob_{w\sim\mathcal{U}_m}\left[D(w)=1\right]\ge\varepsilon,
    \]
    it holds that
    \[
        \Prob_{\substack{w\sim\mathcal{U}_r\\z\sim\mathcal{U}_d}}\left[ R^D(A(x,z),z,w)=x\right]\ge\frac1{\poly(m/\varepsilon)}.
    \]
    Here, we have $d=O(\log ^3(n/\varepsilon))$ and $r=O(m)$. Moreover, $G$ and $A$ can be computed in time $\poly(n/\varepsilon)$, and $R^D$ can be computed in time $\poly(n/\varepsilon)$ with \emph{non-adaptive} oracle access to $D$.
\end{lemma}
\begin{proof}
    Similar to \cite[Section~4.2]{DBLP:conf/focs/Hirahara20}.
\end{proof}

\begin{corollary}\label{cor:rnkt-recon}

    For any $n,m\in\N$ and $\varepsilon>0$, let $G: \bool^n \times \bool^d \to \bool^m$ be as in \Cref{l:black-box-PRG}. Consider any $x \in \bool^n$, $\ell \leq \poly(n/\varepsilon)$, and $v \in \bool^{\leq \ell}$, and let $D: \bool^m \times \bool^\ell \times \bool^{\le\ell} \to \bool$ be any function such that 
    \[
        \Prob_{\substack{z\sim\mathcal{U}_d\\u \sim \mathcal{U}_{\ell}}}\left[D(G(x;z), u, v)=1\right]-\Prob_{\substack{w\sim\mathcal{U}_m\\ u \sim \mathcal{U}_{\ell}}}\left[D(w, u, v)=1\right]\ge2\varepsilon.
    \]
    Then, for some $t = \poly(n/\varepsilon)$, it holds that 
    \[\rK^{t, \,\parallel D}(x \mid v) \leq m + O(\log^3(n/\varepsilon)).\]
    
\end{corollary}

\begin{proof}[Proof Sketch]
    By Markov's inequality, we have that
    \[
        \Prob_{u\sim\U_\ell}\left[\Prob_{z\sim\U_d}\left[D(G(x;z),u,v)=1\right]-\Prob_{w\sim\U_m}\left[D(w,u,v)=1\right]\ge\varepsilon\right]\ge\varepsilon.
    \]
    In other words, with probability at least $\varepsilon$ over $u\sim\U_\ell$, $D(\cdot,u,v)$ is an $\varepsilon$-distinguisher for $G(x;\cdot)$. By \Cref{l:black-box-PRG}, for such $u$, the reconstruction succeeds with probability $1/\poly(m/\varepsilon)$. Hence we have
    \[
        \Prob_{\substack{w\sim\U_r\\z\sim\U_d\\u\sim\U_\ell}}\left[R^{D(\cdot,u,v)}(A(x,z),z,w)=x\right]\ge\varepsilon\cdot\frac{1}{\poly(m/\varepsilon)}=\frac{1}{\poly(m/\varepsilon)}.
    \]
    By an averaging argument, there exists some $z$ such that
    \[
        \Prob_{\substack{w\sim\U_r\\u\sim\U_\ell}}\left[R^{D(\cdot,u,v)}(A(x,z),z,w)=x\right]\ge\frac{1}{\poly(m/\varepsilon)}.
    \]
    Therefore we can design a program that saves $z$, $A(x,z)$, and the description of $R$. When given access to $v$, it samples $w,u$ uniformly, then runs for $\poly(n/\varepsilon)$ steps, queries $D$ non-adaptively, and finally prints $x$ with probability at least $1/\poly(n/\varepsilon)$. By definition of $\rK$, this gives
    \[
        \rK^{\poly(n/\varepsilon),||D}_{1/\poly(m/\varepsilon)}(x\mid v)\le |z|+|A(x,z)|+|R|\le m+O(\log^3 (n/\varepsilon)).
    \]
    By success amplification of $\rK^t$ \cite{DBLP:conf/icalp/LuO21}, we have
    \[
        \rK^{\poly(n/\varepsilon),||D}(x\mid v)\le \rK^{\poly(n/\varepsilon),||D}_{1/\poly(m/\varepsilon)}(x\mid v)+O(\log(m/\varepsilon))\le m+O(\log^3 (n/\varepsilon)).
    \]
\end{proof}

\paragraph{Direct Product Generator .}
For $x,y\in\{0,1\}^n$, we denote their inner product by $x\cdot y=\bigoplus_{i=1}^nx_iy_i$.
\begin{definition}[Direct Product Generator (DPG) \cite{DBLP:conf/stoc/Hirahara21}]\label{def:DPG}
    For $n, k \in \mathbb{N}$, the $k$-wise direct product generator is the function
    \[
        \mathsf{DP}_k \colon \{0,1\}^n \times \{0,1\}^{nk} \to \{0,1\}^{nk + k},
    \]
    defined by
    \[
        \mathsf{DP}_k(x, z_1 \circ \cdots \circ z_k) = (z_1 \circ \cdots \circ z_k \circ x \cdot z_1 \circ \cdots \circ x \cdot z_k).
    \]
\end{definition}
\begin{lemma}[$\pK^t$ Reconstruction Lemma \cite{GKLO22}]\label{lem:recon}
		There is a polynomial $p_{_\mathsf{DP}}$ such that the following holds. For every $n,m,k\in\mathbb{N}^+$, $x \in \bool^n$, let $D$ be a function that $(1/m)$-distinguishes $\mathsf{DP}_k(x; \mathcal{U}_{nk})$ from $\mathcal{U}_{nk + k}$. Then  
		\[\pK^{p_{_\mathsf{DP}}(n+m+k), \parallel D}(x) \leq k +\log p_{_\mathsf{DP}}(n+m+k).\]
	\end{lemma}

\subsection{Technical Results}

\begin{lemma}[\citep{DBLP:journals/siamcomp/AllenderBKMR06}]\label{lem:MKtP_collapse}
   If $\MKtP \in \P$ then $\PSPACE \subseteq \ZPP$.
\end{lemma}

\begin{lemma}[\citep{DBLP:journals/jcss/AllenderKRR11}] \label{lem:MnKtP_lb}
For every $0 < \varepsilon < 1$, we have $\mathsf{GapMnKtP}[n^\varepsilon, n-1] \notin \NP \cap \coNP$.\footnote{In other words, there is no language $L \in \NP \cap \coNP$ that agrees with $\mathsf{GapMnKtP}[n^\varepsilon, n-1]$.}
\end{lemma}

We will use the downward collapse theorem of Shaltiel--Umans.

\begin{definition}[Low-degree extension {\cite[Def.~3.1]{SU06}}]
    We say a complexity class $\mathcal{C}$ allows low-degree extension if $\E^{\mathcal{C}_{\le O(n)}}\subseteq\mathcal{C}$, where $\mathcal{C}_{\le O(n)}$ means that the machine only makes linear-length queries to $\mathcal{C}$.
\end{definition}
\begin{lemma}[Shaltiel--Umans downward collapse {\citep[Thm.~3.2]{SU06}}]
    \label{lem:su-collapse-low-deg-ext}
    For any complexity class $\mathcal{C}$ that allows low-degree extension, we have
    \[
        \mathcal{C}\subseteq\P^{||\NP}/\poly\implies\mathcal{C}\subseteq\NP/\poly\cap\coNP/\poly.
    \]
\end{lemma}
It is not hard to verify that both $\E^{||\NP}$ and $\BPE^{||\NP}$ allows low-degree extension. Hence we get the following corollaries:
\begin{corollary}\label{cor:su-collapse-E}
    $\E^{||\NP}\subseteq\P^{||\NP}/\poly\implies\E^{||\NP}\subseteq\NP/\poly\cap\coNP/\poly$
\end{corollary}
\begin{corollary}\label{cor:su-collapse-BPE}
    $\BPE^{||\NP}\subseteq\P^{||\NP}/\poly\implies\BPE^{||\NP}\subseteq\NP/\poly\cap\coNP/\poly$
\end{corollary}

\section{Useful Results About Nondeterministic Kolmogorov Complexity}
In this section, we present several useful facts about various nondeterministic and relativized notions of Kolmogorov complexity. We note that all the results presented here also hold for the corresponding notions in the presence of any oracle. For simplicity, in below, we consider only the cases without oracles. It is straightforward to adapt all the proofs to handle the more general setting.

\subsection{\texorpdfstring{$\nK$}{}}\label{ss:nK_prelim}

\begin{lemma}\label{lem:kA-leq-nk-super-script}
    There exists a constant $c > 0$ such that for any strings $x, y \in \bool^*$ and $t \in \mathbb{N}$,
    \[
        \K^{t', \parallel \SAT}(x \mid y) \leq \nK^t(x \mid y) + c\cdot \log t,
    \]
    where $t' = (t + |y|)^c$.
\end{lemma}

\begin{proof}
    Let $x, y \in \bool^*$ and $t \in \mathbb{N}$. We assume without loss of generality that $t \geq |x|$. Define $k = \nK^t(x \mid y)$, and let $\Pi \in \bool^k$ be a program that nondeterministically outputs $x$ on input $y$ in time $t$. 

    Given $\Pi$, $|x|$, $t$, and $y$, we reconstruct $x$ as follows: For each $i \in [|x|]$, we use the $\SAT$ oracle to answer the following query: 
    \begin{quote}
        Does there exist a $w \in \bool^t$ such that $\Pi^{\mathcal{O}}(y; w)$ outputs a string $z \in \bool^{|x|}$ in time $t$ and the $i$-th bit of $z$ is $1$?
    \end{quote}
    We then set the $i$-th bit of $x$ to be $1$ if and only if the answer to this query is yes.    

    It is easy to verify that the above procedure runs in time $\poly(|x| + |y| + t)$. Also, by the fact that $\Pi$ nondeterministically outputs $x$ on input $y$ in time $t$, it can correctly recover $x$. Thus, for $t' = \poly(|y| + t)$, we obtain
    \begin{align*}
        \K^{t', \parallel \SAT}(x \mid y) &\leq |\Pi| + O(\log |t|) \\
        &\leq k + O(\log t),
    \end{align*}
    as desired.
\end{proof}

\begin{corollary}\label{lem:ktA-leq-nkt}
     There exists a constant $c > 0$ such that for any strings $x, y \in \bool^*$,
    \[
        \Kt^{\parallel \SAT}(x \mid y) \leq c\cdot \nKt(x \mid y) + O(\log |y|).
    \]
\end{corollary}
\begin{proof}
    Let $x, y \in \bool^*$ and $k = \nKt(x \mid y)$. Then there exists a nondeterministic program $\Pi$ that outputs $x$ given $y$ in time $t$, where $|\Pi| + \lceil \log t \rceil \leq k$. This implies $\nK^t(x \mid y) \leq |\Pi|$.

    By \Cref{lem:kA-leq-nk-super-script}, we have
    \[
        \K^{t', \parallel \SAT}(x \mid y) \leq |\Pi| + O(\log t),
    \]
    where $t' = \poly(|y| + t)$. It follows that there exists a nondeterministic program of size at most $|\Pi| + O(\log t)$ that outputs $x$ given $y$ in time $t'$, which implies
    \[
        \nKt(x \mid y) \leq |\Pi| + O(\log t) + \log t' \leq c\cdot\nKt(x \mid y) + O(\log |y|),
    \]
    as desired.
\end{proof}

\begin{lemma}\label{lem:nk-leq-ksat-superscript}
    There exists a constant $c' \in \nat$ as follows. For every language $A \in \NP$, there exists a constant $c \in \mathbb{N}$ such that for any strings $x, y \in \bool^*$ and $t \in \mathbb{N}$,
    \[
        \nK^{t^c}(x \mid y) \leq \K^{t,\parallel A}(x \mid y) + c \cdot \log t.
    \]
\end{lemma}
\begin{proof}
    Let $x, y \in \bool^*$ and $t \in \mathbb{N}$, and define $k = \K^{t,\parallel A}(x \mid y)$. Then there exists a program $\Pi \in \bool^k$ that on input $y$, makes only non-adaptive queries to $A$ and outputs $x$ within time $t$.

As in the proof of \cite[Lemma 31]{HLO25}, we will eliminate the oracle access to $A$ by guessing the answers to the (non-adaptive) queries to $A$ and verifying them with the help of an additional small amount of advice. We provide the details for completeness.

    Suppose $A$ is the following oracle:
\[
A(q) = 1 \iff \exists w \in \bool^{\poly(|q|)} \text{ such that } V(q, w) = 1,
\]
where $V$ is some polynomial-time verifier. Let $q_1, q_2, \dots, q_{\ell}$, where $\ell \leq t$, be the (non-adaptive) queries made by $d$ to the oracle $A$. Also, let $p$ be the number of positive answers to these queries, i.e., $p \vcentcolon= \sum_{i=1}^{\ell} A(q_i)$. We then non-deterministically guess $\ell$ witnesses $(w_1, w_2, \dots, w_\ell)=\vcentcolon w$ and count the number of indices $i$ for which $V(q_i, w_i) = 1$. Denote this number by $u_w$. If $u_w$ is not equal to $p$, we output $\bot$. Otherwise, for each $w_i$ that satisfies the verifier $V$, we set the answer to the query $q_i$ to be $1$, and we let the rest of the answers be $0$. We then run $d$ while simulating the oracle $A$ using the answers obtained as described above.

Note that, given the number $p$, which can be specified using $\lceil \log \ell \rceil \leq \lceil \log t \rceil$ bits, the above procedure can be implemented to run in time $\poly(t)$. Also, it is easy to observe that there exists some guess of $w=(w_1, w_2, \dots, w_\ell)$ that satisfies $u_w= p$. Moreover, for any guess of $w$ that satisfies $u_w= p$, the answers obtained are the correct answers to the queries $(q_1, q_2, \dots, q_{\ell})$ for $A$. It follows that there exists a program of size at most $|\Pi| + O(\log t)$ that, given $y$, runs in time $\poly(t)$ and outputs $x$ non-deterministically.
Thus, for some constant $c \in \mathbb{N}$,
    \begin{align*}
        \nK^{t^c}(x \mid y) &\leq |\Pi| + c \cdot \log t
        \\&\leq k + c \cdot \log t,
    \end{align*}
    as desired.
\end{proof}

As a corollary of \Cref{lem:nk-leq-ksat-superscript}, we get the following.
\begin{corollary}\label{lem:nkt-leq-ktsat}
    For every language $A \in \NP$, there exists a constant $c \in \mathbb{N}$ such that for any strings $x, y \in \bool^*$,
    \[
        \nKt(x \mid y) \leq c \cdot \Kt^{\parallel A}(x \mid y).
    \]
\end{corollary}

\begin{lemma}\label{l:K-NP-coNP-superscript} 
There is a constant $c>0$ such that the following holds.
    For any time-constructible, non-decreasing $T: \mathbb{N} \to \mathbb{N}$, $T(n) \ge n$ and every $L \in \NTIME[T(n)] \cap \coNTIME[T(n)]$, $x, y \in \bool^*$, and sufficiently large $t \in \mathbb{N}$,
    \[
        \nK^{t \cdot T(t)\cdot \log^c(T(t))}(x \mid y) \leq \nK^{t,L}(x \mid y) + c\cdot \log \log t.
    \]
    Specifically, restricting query length to at most $\ell \in \mathbb{N}$, for every $x, y \in \bool^*$ and sufficiently large $t \in \mathbb{N}$,
    \[
        \nK^{t\cdot T(\ell) \cdot \log^c(T(\ell) \cdot t)}(x \mid y) \leq \nK^{t,L_{\le \ell}}(x \mid y) + c \cdot \log \log t.
    \]
\end{lemma}
\begin{proof}
    Let $L \in \NTIME[T(n)] \cap \coNTIME[T(n)]$, $x, y \in \bool^*$, and $t \in \mathbb{N}$. Without loss of generality, assume $t$ is a power of $2$; if not, replace $t$ with $2^{\lceil \log t \rceil}$ in the argument below at the cost of at most a factor of $2$ in the resulting running time. Define $k = \nK^{t,L_{\leq \ell}}(x \mid y)$, and let $\Pi \in \bool^k$ be a program that non-deterministically outputs $x$ in time $t$ given $y$ and \emph{oracle access to $L_{\leq \ell}$}.

    We now simulate the program $\Pi$ \emph{without} oracle access to $L$. To do this, whenever $\Pi$ makes a query $q \in \{0, 1\}^{\leq \ell}$ to $L$, we non-deterministically guess an answer (``yes'' or ``no'') along with a witness $w_q \in \{0, 1\}^{\leq T(\ell)}$ certifying that answer. Since $L \in \NTIME[T(n)] \cap \coNTIME[T(n)]$, each such guess can be verified in time $\widetilde{O}(T(\ell))$. Whenever the verification fails, halt with $\bot$ being the output. Observe that the resulting program is single-valued, by definition of $\Pi$ and since all oracle queries are verified.
     
    This yields a nondeterministic program of size at most $|\Pi| + O(\log \log t)$ that outputs $x$ in time $t\cdot T(\ell) \cdot \poly\log(T(\ell) \cdot t)$ given $y$ (without oracle access to $L$). Therefore,
    \[
        \nK^{t\cdot T(\ell) \cdot \log^c(T(\ell) \cdot t)}(x \mid y) \leq \nK^{t,L_{\leq \ell}}(x \mid y) + c \cdot \log\log t
    \]
    for some constant $c > 0$.\footnote{Note that the $L$-verifier descriptions and their constant-factor simulation overhead depend on $L$ but not on
    $t$, and are therefore absorbed into the stated bounds once $t$ is sufficiently large; the exponent
    $c$ comes only from the universal machine's simulation overhead and is independent of $L$.}
    
    The first bound is the case
    $\ell = t$. This concludes the proof.
    
\end{proof}

    Again, as a corollary of \Cref{l:K-NP-coNP-superscript}, we get the following.

\begin{corollary}\label{l:Kt-NP-coNP}
    For every $L \in \NP \cap \mathsf{coNP}$, there exists a constant $c > 0$ such that for every $x, y \in \bool^*$,
    \[
        \nKt(x \mid y) \leq c \cdot \nKt^L(x \mid y).
    \]
\end{corollary}

\subsection{\texorpdfstring{$\rnK$}{}}

We prove that the success probability of $\rnK^t_\delta$ can be amplified, a task that is non-trivial if we start with a probability of success $\delta$ that is not noticeably larger than $1/2$. For randomized notions of Kolmogorov complexity such as $\rK^t$, success amplification in this regime can be achieved  using standard hashing-based techniques. However, in $\rnK^t$ we have both randomness and nondeterminism, and must ensure that a good choice of the randomness leads to a nondeterministic encoding that never outputs the wrong string. This creates an additional difficulty and requires a modified construction. 

To achieve success amplification for $\rnK^t$, we use an idea from \cite{DBLP:journals/cc/BuhrmanLM05}. Roughly speaking, on top of the hashing-based framework for success amplification, we also consider the average number of positive $\NP$ queries. Then, using a frequency analysis and concentration bounds, one can prove that when we take enough samples, with high probability the number of wrong $\NP$-oracle answers is small. The technique essentially allows one to reduce the  argument to that of randomized time-bounded Kolmogorov complexity, i.e., when nondeterminism is not present.

\begin{lemma}[Success Amplification for $\rnK^t$]\label{l:success-amp-rnK}
    For any oracle $\mathcal{O}$, any $n \in \nat$, any time bound $t\in\mathbb{N}$, any $\varepsilon,\delta\in(0,1)$ and any strings $x,y\in\{0,1\}^*$ such that $|x| = n$ and $|x|\le t$, we have
    \[
        \rnK^{t\cdot p(n\log(t/\varepsilon)/\delta), \mathcal{O}}_{1-\varepsilon}(x\mid y)\le\rnK^{t, \mathcal{O}}_\delta(x\mid y)+\log p(n\log(t/\varepsilon)/\delta).
    \]
\end{lemma}
\begin{proof}
    Let $M$ be the program for $\rnK^{t, \mathcal{O}}_\delta(x\mid y)$. For simplicity, assume $t$ is a power of $2$; if not, the argument can be adapted by replacing $t$ with $2^{\lceil \log t\rceil } \leq 2t$. By definition, we have
    \[
        \Prob_{r\sim\mathcal{U}_t}\left[\begin{aligned}&\exists w\in\{0,1\}^t,M^{\mathcal{O}}(y,r,w)\text{ outputs $x$ in $t$ steps,}\\
        &\forall w\in\{0,1\}^t,M^{\mathcal{O}}(y,r,w)\text{ outputs $x$ or $\bot$ in $t$ steps}.\end{aligned}\right]\ge\delta.
    \]
    For any string $r\in\{0,1\}^t$, define a string $s_r\in\{0,1\}^n$, such that the $i$'th bit of $s_r$ is $1$ iff there exists some $w\in\{0,1\}^t$ such that the $i$'th bit of $M^{\mathcal{O}}(y,r,w)$ is $1$. In other words, $s_r$ is the bit-wise OR of outputs on all computation paths when the random string is $r$. We also define $d_r\in[0,1]$ as
    \[
        d_r:=\frac{w(s_r)}{n},
    \]
    where $w(\cdot)$ is the Hamming weight of a string. We define $\overline{d}$ as the average $d_r$ over $r\sim\mathcal{U}_t$, i.e.
    \[
        \overline{d}=\Exp_{r\sim\mathcal{U}_t}[d_r].
    \]
    For any $\lambda\in(0,1)$, define $A_\lambda$ as
    \[
        A_\lambda:=\left\{s\in\{0,1\}^n\ \middle\vert\ \Prob_{r\sim\mathcal{U}_t}[s_r=s]\ge\lambda\right\}.
    \]
    By definition, $x\in A_{\delta}$. We take $k$ independent samples $r_1,\ldots,r_k\sim\mathcal{U}_t$. The following claim is crucial to our proof:
    \begin{claim}\label{c:rnK-amplify-concentration}
    For $k=\left\lceil\frac{1024n^2\log(6/\varepsilon)}{\delta^2}\right\rceil$, with probability at least $1-\varepsilon$ over $r_1,\ldots,r_k\sim\mathcal{U}_t$, the following three conditions hold\emph{:}
    \begin{enumerate}[label=(\roman*)]
        \item $|\{i\in[k]\mid s_{r_{i}}=x\}|\ge\frac{\delta k}{2}$.\label{i:rnK-amplify-1}
        \item For any string $x'\notin A_{\delta/8}$, $|\{i\in[k]\mid s_{r_i}=x'\}|\le\frac{\delta k}{4}$.\label{i:rnK-amplify-2}
        \item $|\frac{d_{r_1}+\cdots+d_{r_k}}{k}-\overline{d}|\le\frac{\delta}{32n}$.\label{i:rnK-amplify-3}
    \end{enumerate}
    \end{claim}
    \begin{proof}[Proof of \Cref{c:rnK-amplify-concentration}]\renewcommand{\qedsymbol}{$\diamond$}
        Since $\Prob_{r\sim\mathcal{U}_t}[s_r=x]\ge\delta$, by Chernoff bound, the probability of \Cref{i:rnK-amplify-1} being violated is at most
        \[
            \Prob_{r_1,\ldots,r_k\sim\mathcal{U}_t}\left[\left|\{i\in[k]\mid s_{r_i}=x\}\right|\le\frac{\delta}{2}\cdot k\right]\le\exp\left(-2k\cdot\frac{\delta^2}{4}\right)\le\frac{\varepsilon}{3}.
        \]
        Similarly, the probability of \Cref{i:rnK-amplify-3} being violated is at most
        \[
            \Prob_{r_1,\ldots,r_k\sim\mathcal{U}_t}\left[\left|\frac{d_{r_1}+\cdots+d_{r_k}}{k}-\overline{d}\right|\ge\frac{\delta}{32n}\right]\le2\exp\left(-2k\cdot\frac{\delta^2}{1024n^2}\right)\le\frac{\varepsilon}{3}.
        \]
        For \Cref{i:rnK-amplify-2}, let $x'$ be any string in $\{0,1\}^n\setminus A_{\delta/8}$. Then we have $\Prob_{r\sim\mathcal{U}_t}[s_r=x']\le\delta/8$. Then by Chernoff bound, we get
        \[
            \Prob_{r_1,\ldots,r_k\sim\mathcal{U}_t}\left[|\{i\in[k]\mid s_{r_i}=x'\}|\ge\frac{\delta}{4}\cdot k\right]\le\exp\left(-2k\cdot\frac{\delta^2}{64}\right)\le\frac{\varepsilon}{3}\cdot 2^{-n}.
        \]
        Applying union bound over all $x'\in\{0,1\}^n\setminus A_{\delta/8}$, the probability of \Cref{i:rnK-amplify-2} being violated is at most $\varepsilon/3$. Applying union bound over \Cref{i:rnK-amplify-1,i:rnK-amplify-2,i:rnK-amplify-3}, these three conditions hold together with probability at least $1-\varepsilon$.
    \end{proof}
    Notice that $|A_{\delta/8}|\le8/\delta$. Applying \Cref{l:string-isolation} to $A_{\delta/8}$, there exists a string $u\in\{0,1\}^{O(\log(n/\delta))}$, such that the algorithm $B(1^n,u)$ runs in $\mathsf{poly}(n)$ time and outputs a circuit that computes a function $H:\{0,1\}^n\to\{0,1\}^{O(\log(1/\delta))}$, satisfying $H(x)\ne H(x')$ for every distinct pair $x,x'\in A_{\delta/8}$. Then our algorithm for recovering $x$ works as follows:
    \begin{enumerate}[label=(\alph*)]
        \item Set $k=\left\lceil\frac{1024n^2\log(6/\varepsilon)}{\delta^2}\right\rceil$, and take $k$ independent samples $r_1,\ldots,r_k\sim\mathcal{U}_{t}$.\label{i:rnK-alg-1}
        \item For each $i\in[k]$, non-deterministically guess $n$ strings $w_{i,1},\ldots w_{i,n}\in\{0,1\}^{t}$. By simulating $M$ for at most $t$ steps $nk$ times (and treating non-halting runs as outputting $\bot$), we then compute $\hat{s}_{r_i}\in\{0,1\}^n$, where the $j$'th bit of $\hat{s}_{r_i}$ is $1$ iff the $j$'th bit of $M^{\mathcal{O}}(y,r_i,w_{i,j})$ is $1$. We also compute $\hat{d}_{r_i}:=w(\hat{s}_{r_i})/n$, where $w(\cdot)$ is the Hamming weight of a string.\label{i:rnK-alg-2}
        \item Set the threshold $a=(1/n) \cdot \lceil nk(\overline{d}-\frac{\delta}{32n})\rceil$. If $\hat d_{r_1}+\cdots+\hat d_{r_k}< a$, then outputs $\bot$ and halts.\label{i:rnK-alg-3}
        \item Let $\hat{A}$ be the set of strings appearing at least $\lfloor \frac{3\delta}{8}\cdot k\rfloor$ times in $\hat s_{r_1},\ldots,\hat s_{r_k}$. If there exists a unique string $\hat{x}\in\hat{A}$ such that $H(\hat x)$ matches the stored hash value $H(x)$, outputs $\hat{x}$; otherwise outputs $\bot$.\label{i:rnK-alg-4}
    \end{enumerate}
    Notice that we only need to store $(M,n, \log t, k, \delta, a,u,H(x))$, which takes $|M|+O(\log \frac{n\log(t/\varepsilon)}{\delta})$ bits. (Here, $a$ is stored as the integer $\lceil nk(\overline{d}-\frac{\delta}{32n})\rceil$.) The running time of the above algorithm is also bounded by $t\cdot\poly(\frac{n\log(t/\varepsilon)}{\delta})$. It remains to estimate the success probability.

    \begin{claim}\label{c:good-randomness}
        If the random strings $r_1,\ldots,r_k$ sampled by our algorithm satisfy \Cref{i:rnK-amplify-1,i:rnK-amplify-2,i:rnK-amplify-3} of \Cref{c:rnK-amplify-concentration}, then
        \begin{enumerate}
            \item For all nondeterministic guesses $\{w_{i,j}\}_{i\in[k],j\in[n]}$, our algorithm either outputs $x$ or $\bot$; and\label{i:rnK-claim-1}
            \item There exists some nondeterministic guess $\{w_{i,j}\}_{i\in[k],j\in[n]}$ such that our algorithm outputs $x$.\label{i:rnK-claim-2}
        \end{enumerate}
    \end{claim}
    \begin{proof}[Proof of \Cref{c:good-randomness}]\renewcommand{\qedsymbol}{$\diamond$}
        We first prove \Cref{i:rnK-claim-1}. If the algorithm reaches step \ref{i:rnK-alg-4} of the algorithm, then we have
        \[
            \hat d_{r_1}+\cdots+\hat d_{r_k}\ge k\left(\overline d-\frac{\delta}{32n}\right).
        \]
        But since $r_1,\ldots,r_k$ also satisfies \Cref{i:rnK-amplify-3} of \Cref{c:rnK-amplify-concentration}, we have
        \[
            d_{r_1}+\cdots+d_{r_k}\le k\left(\overline d+\frac{\delta}{32n}\right).
        \]
        Subtracting the two inequalities, we get
        \[
            (d_{r_1}-\hat d_{r_1})+\cdots+(d_{r_k}-\hat d_{r_k})\le\frac{\delta k}{16n}.
        \]
        By averaging argument, we get
        \[
            \left|\left\{i\in[k]\ \middle|\  d_{r_i}-\hat d_{r_i}\ge\frac 1n\right\}\right|\le \frac{\delta k}{16}.
        \]
        Since $d_{r_i}$ counts the number of 1 in the bitwise OR of all computation paths, $d_{r_i}\ge\hat d_{r_i}$; and if $d_{r_i}=\hat d_{r_i}$, then $s_{r_i}=\hat s_{r_i}$ must hold. Also because $d_{r_i}, \hat d_{r_i}$ are both multiples of $1/n$, $d_{r_i}\ne\hat d_{r_i}$ implies $d_{r_i}-\hat d_{r_i}\ge1/n$. Therefore $s_{r_i}\ne \hat s_{r_i}$ implies $d_{r_i}-\hat d_{r_i}\ge1/n$, and we get
        \begin{equation}
            |\{i\in[k]\mid s_{r_i}\ne\hat s_{r_i}\}|\le\frac{\delta k}{16}.\label{eq:rnK-amplify-1}
        \end{equation}
        By \Cref{i:rnK-amplify-1} of \Cref{c:rnK-amplify-concentration}, we get
        \begin{equation}
            |\{i\in[k]\mid s_{r_i}=x\}|\ge\frac{\delta k}{2}.\label{eq:rnK-amplify-2}
        \end{equation}
        Combining \Cref{eq:rnK-amplify-1,eq:rnK-amplify-2}, we get
        \[
            |\{i\in[k]\mid \hat s_{r_i}=x\}|\ge\frac{\delta k}{2}-\frac{\delta k}{16}=\frac{7\delta k}{16}>\frac{3\delta k}{8}.
        \]
        Therefore, by step \ref{i:rnK-alg-4} of the algorithm, $x\in\hat A$ holds. Also, by \Cref{i:rnK-amplify-2} of \Cref{c:rnK-amplify-concentration}, for any $x'\not\in A_{\delta/8}$, we have
        \begin{equation}
            |\{i\in[k]\mid s_{r_i}=x'\}|\le\frac{\delta k}{4}.\label{eq:rnK-amplify-3}
        \end{equation}
        Combining \Cref{eq:rnK-amplify-1,eq:rnK-amplify-3}, we get
        \[
            |\{i\in[k]\mid \hat s_{r_i}=x'\}|\le\frac{\delta k}{4}+\frac{\delta k}{16}=\frac{5\delta k}{16}<\frac{3\delta k}{8}.
        \]
        Therefore, by step \ref{i:rnK-alg-4} of the algorithm, $x'\notin\hat A$ holds. Thus we get $\hat A\subseteq A_{\delta/8}$. By \Cref{l:string-isolation}, we can always find the correct $x$ in step \ref{i:rnK-alg-4}.

        We have shown that whenever the algorithm reaches step \ref{i:rnK-alg-4}, it always outputs $x$. Therefore to prove \Cref{i:rnK-claim-2} of the claim, it suffices to show that there exists some $\{w_{i,j}\}_{i\in[k],j\in[n]}$ such that the algorithm reaches step \ref{i:rnK-alg-4}. In fact, since $s_{r_i}$ is the bitwise OR of all computation paths, there exists some $\{w_{i,j}\}_{i\in[k],j\in[n]}$, such that for all $i\in[k]$ we have $s_{r_i}=\hat s_{r_i}$, which implies $d_{r_i}=\hat d_{r_i}$. In this case, \Cref{i:rnK-amplify-3} gives $\hat d_{r_1}+\cdots+\hat d_{r_k}=d_{r_1}+\cdots+d_{r_k}\ge k\left(\overline d-\frac{\delta}{32n}\right)$. Since $n\cdot(\hat d_{r_1}+\cdots+\hat d_{r_k})$ is an integer, it follows that $\hat d_{r_1}+\cdots+\hat d_{r_k}\ge a$, so our algorithm reaches step \ref{i:rnK-alg-4}.
    \end{proof}
    Therefore the success probability is at least $1-\varepsilon$ by \Cref{c:rnK-amplify-concentration}, which finishes our proof.
\end{proof}

    \begin{lemma}\label{l:rKt--parallel-NP}
        For every $L\in\NP$, there is a a constant $c>0$ such that for every $t\in\mathbb{N}$, $\lambda\in (0,1)$, and $x,y\in\bool^*$,
        \[
            \rnK^{t^c}_{\lambda/t}(x\mid y)\leq \rK_{\lambda}^{t, \parallel L}(x\mid y) + c\cdot \log t.
        \]
    \end{lemma}
\begin{proof}
Let $x, y \in \bool^*$ and $t \in \mathbb{N}$, and define $k = \rK^{t,\parallel A}(x \mid y)$. Then there exists a program $\Pi \in \bool^k$ such that for a $\lambda$ fraction of $r \in \bool^t$, $\Pi$, when run on $(y; r)$ with non-adaptive oracle access to $L$, can non-deterministically output $x$ in time $t$.

Fix any good $r$ such that the above holds. Then, as in the proof of \Cref{lem:nk-leq-ksat-superscript}, one can eliminate the oracle access to $L$ by hard-coding the number $p$ of \emph{positive} answers to the oracle queries made by $\Pi$, and non-deterministically guessing witnesses (of length at most $\poly(t)$) for each query. Simulate $\Pi$, attempting to answer the queries using these witnesses: if a witness verifies, answer "yes"; otherwise, answer "no." At the end, verify that exactly $p$ witnesses were accepted.

Note that in the above, the additional advice $p$ depends on the randomness $r$, and therefore we cannot have a single advice string that works for all good choices of $r$. However, such advice can be encoded using at most $\log t$ bits. Then for any good $r$, if we uniformly pick an advice string from $\bool^{\log t}$, the program can successfully output $x$ in a nondeterministic manner as described above. It follows that there exists a program, which includes the description of $\Pi$ and has size at most $|\Pi| + O(\log t)$, such that with probability at least
\[
\lambda \cdot \frac{1}{2^{\log t}} = \lambda / t,
\]
over $r \sim \bool^t$ and $p \sim \bool^{\log t}$, it outputs $x$ non-deterministically in time $\poly(t)$. This concludes the proof.
\end{proof}

\begin{lemma}\label{l:rK-AM-coAM-superscript} 
    There exists a constant $c'$ such that the following holds. For any time-constructible, non-decreasing $T: \mathbb{N} \to \mathbb{N}$, $T(n) \ge n$, for every $L \in \AMTIME[T(n)] \cap \co\AMTIME[T(n)]$, all $x, y \in \bool^*$, and sufficiently large $t \in \mathbb{N}$ with $t > |x|$,
    \[
        \rnK^{T(t)^{c'}}(x \mid y) \leq \rnK^{t,L}(x \mid y) + O(\log t).
    \]
    Moreover, restricting query length to at most $\ell$ for an integer $\ell \geq 1$, 
    \[
        \rnK^{t\cdot (T(\ell) \cdot |x| \cdot \log t)^{c'}}(x \mid y) \leq \rnK^{t,L_{\le \ell}}(x \mid y) + O(\log(\ell \cdot |x| \cdot \log t)).
    \]
\end{lemma}

\begin{proof}[Proof Sketch]
The proof is similar to that of \Cref{l:K-NP-coNP-superscript}. It relies on the idea that, given a nondeterministic program with oracle access to a language $L \in \AMTIME[T(n)] \cap \coAMTIME[T(n)]$, one can eliminate the oracle by guessing an answer to a query along with a witness certifying that answer, and then verifying the answer efficiently in time $\poly(T(\ell))$. 

The main difference here is that $L$ lies in $\AMTIME[T(n)] \cap \coAMTIME[T(n)]$, so the verifiers are randomized. However, using standard error reduction techniques, we can ensure that the verifications have perfect completeness and soundness error at most $O(2^{-2\ell})$, within time overhead of $O(\ell)$. We use the same choice of randomness on every query on every nondeterministic branch. Then, by union bound over possible inputs to $L$ of length at most $\ell$, the overall success probability of the simulated program is at least $2/3 - O(2^{-\ell})$, which can be amplified back to $2/3$ using \Cref{l:success-amp-rnK} at the expense of $t \cdot \poly(T(\ell)\cdot |x| \cdot \log t)$ time overhead and $O(\log(\ell \cdot |x| \cdot \log t))$ description overhead.
    \end{proof}

     \begin{corollary}\label{l:rKt-AM-coAM}
                    For every $L\in\AM\cap\coAM$, there is a constant $c>0$ such that for every $x,y\in\bool^*$,
        \[
            \rnKt(x\mid y)\leq c\cdot \rnKt^{L}(x\mid y).
        \]
        \end{corollary}

\subsection{\texorpdfstring{$\pnK$}{}}

    \begin{lemma}\label{l:pKt-SAT-superscript}
        For every oracle $\mathcal{O}$ and every language $L$ that is hard for $\NP^{\mathcal{O}}$ under deterministic polynomial-time reductions, there exists a constant $c  > 0$ such that, for all $x,y\in\bool^*$ and $t\in\mathbb{N}$,
        \[
            \pK^{t^c,L}(x\mid y)\leq\pnK^{t, \mathcal{O}}(x\mid y) + c\cdot \log t.
        \]
    \end{lemma}
    \begin{proof}

        The proof follows the same idea as \Cref{lem:kA-leq-nk-super-script}. Let $k = \pnK^{t,\mathcal{O}}(x \mid y)$, and suppose that with probability at least $2/3$ over $r \sim \bool^t$, there exist $\Pi \in \bool^k$ such that, on input $(y;r)$ and with oracle access to $\mathcal{O}$, $\Pi$ non-deterministically outputs $x$ within time $t$.

        Let $L_0 \in \NP^{\mathcal{O}}$ be the following language:
        \[L_0 = \mleft\{(\Pi,y,r, i, 1^t, 1^n) \,\,\middle\vert\,\, 	
        	\begin{aligned}
		&\exists w\in\bool^t \textrm{ such that
		$\Pi^{\mathcal{O}}(y;w;r)$ outputs a string}\\
        &x \in \bool^n \textrm{ within time $t$ and the $i$-th bit of $x$ is $1$.}
	\end{aligned}
    \mright\}.\]

        Then, with probability at least $2/3$ over $r \sim \bool^t$, given $(\Pi, y, r, t, n)$ as above hard-coded, one can recover the string $x$ within time $\poly(t)$ given oracle access to $L_0$. In particular, one can simply query $(\Pi,y,r, i, 1^t, 1^n)$ to $L_0$ for every $i \in [n]$ and return the concatenation of the answers from the oracle. By definition of $\Pi$ and $t$, this process is guaranteed to output $x$. 
        
        Lastly, for any $\NP^{\mathcal{O}}$-complete language $L$, we may replace the $L_0$-oracle with an $L$-oracle by simulating the reduction from $L_0$ to $L$. This increases the time to produce $x$ by at most a factor of $\poly(t)$. 

        Overall, we have
        \begin{align*}
            \pK^{\poly(t), L}(x \mid y) &\leq |(\Pi, t, n)| + O(\log t)
            \\&\leq  |\Pi| + O(\log t)
            \\&= \pnK^{t,\mathcal{O}}(x \mid y)) + O(\log t),
        \end{align*}
        as desired.
    \end{proof}

    \begin{corollary}\label{l:pKt-SAT}
        For every oracle $\mathcal{O}$ and every language $L$ that is hard for $\NP^{\mathcal{O}}$ under deterministic c polynomial-time reductions, there exists a constant $c  > 0$ such that, for all $x,y\in\bool^*$,
        \[
            \pKt^{L}(x\mid y)\leq c \cdot \pnKt^{\mathcal{O}}(x\mid y).
        \]
    \end{corollary}

\section{Average-Case SoI for Nondeterministic Computation: Proof of \Cref{t:soi-pKt-avg}}

\subsection{Average-Case SoI for \texorpdfstring{$\pnKt$}{}}\label{sec:avg-soi-pnKt}

In this subsection, we show that Symmetry of Information holds for $\pnKt$ \emph{on average} over efficiently samplable distributions.  Toward proving this result, we observe that Symmetry of Information holds for $\pK^{\poly}$ if one replaces the term $\pK^{\poly}(y \mid x)$ with its nondeterministic counterpart $\pnK^{\poly}(y \mid x)$. 

\subsubsection{A Semi-Symmetry of Information Theorem for \texorpdfstring{$\pnK^t$}{}}\label{sec:semi-soi}

    We obtain a version of semi-SoI for $\pnK^t$ via the worst-case conditional coding theorem of \cite{HLO25}. The proof is similar to an argument from \cite{DBLP:conf/stoc/HiraharaILNO23} used to conditionally show average-case SoI for $\pK^t$. The statement slightly improves on the worst-case semi-SoI theorem of \cite{HLO25}, which would include a term for $\pnK^{\poly(t)}(y)$ rather than $\pK^{\poly(t)}(y)$.

    \begin{theorem}\label{thm:semi-soi-pk-poly}
    
        There exists a polynomial $p$ such that, for all sufficiently large $n, t \in \N$ with $t \geq 4n$, the following holds. For all $x, y \in \bool^n$,
        \[\pK^{t}(x, y) \geq \pK^{p(t)}(x) + \pnK^{p(t)}(y \mid x) - \log p(t).\]
        
    \end{theorem}
    \begin{proof}
        
        Let $\mathcal{D} \in \P\SAMP$ be a distribution family as in Lemma 36 of \cite{DBLP:conf/icalp/LuOZ22} that dominates the universal distribution for $\pK^{t}$; that is, for some constant $c \in \N$, for any $n, t \in \N$ and $z \in \bool^{2n}$, 
        \begin{align}\label{eq:domination}
            \mathcal{D}_{n,t}(z) \geq \frac{1}{2^{\pK^{t}(z)}\cdot n^c}.
        \end{align}
        For $n, t \in \N$, let $\mathcal{D}^{(2)}_{n,t}$ denote the marginal distribution over the second half of the output of $\mathcal{D}_{n,t}$, and note that $\mathcal{D}^{(2)} = \{\mathcal{D}^{(2)}_{n,t}\}_{n,t \in \N} \in \P\SAMP$. 
        
        Now fix $n, t \in \N$ such that $t \geq 4n$ along with strings $x, y \in \bin^n$. By the worst-case conditional coding theorem for $\pnK^t$ with respect to $\mathcal{D}$ (\cite[Theorem 6]{HLO25}), for some polynomial $q$,
        \begin{align*}
            \pnK^{q(t)}(y \mid x) &\leq \log\frac{1}{\mathcal{D}_{n,t}(y \mid x)} + \log q(t)
            \\&= \log\frac{\mathcal{D}_{n,t}^{(2)}(x)}{\mathcal{D}_{n,t}(x,y)} + \log q(t)
            \\&= \log \frac{1}{\mathcal{D}_{n,t}(x,y)} - \log\frac{1}{\mathcal{D}^{(2)}_{n,t}(x)} + \log q(t)
            \\&\leq  \pK^t(x, y) - \log\frac{1}{\mathcal{D}^{(2)}_{n,t}(x)} + \log q(t) + c\log n. \tag{by \Cref{eq:domination}}    
        \end{align*}
        Moreover, by the worst-case coding theorem for $\pK^t$ complexity \cite{DBLP:conf/icalp/LuOZ22}, for some polynomial $q'$, 
        \[\pK^{q'(t)}(x) \leq \log\frac{1}{\mathcal{D}^{(2)}_{n,t}(x)} + \log q'(t).\]
        Combining this with the above,
        \[\pnK^{q(t)}(y \mid x) + \pK^{q'(t)}(x) \leq \pK^t(x, y) + \log q(t) + c\log n + \log q'(t).\]
        Let the polynomial $p$ be such that $p(m) \geq 4 \cdot \max\{q(m), m^c, q'(m)\}$ for all $m \in \N$. It follows that 
        \[\pnK^{p(t)}(y \mid x) + \pK^{p(t)}(x) \leq \pK^t(x, y) + \log p(t).\]
        This completes the proof.
    \end{proof}

\subsubsection{From Semi-SoI for \texorpdfstring{$\pnK^t$}{} to Average-Case SoI for \texorpdfstring{$\pnKt$}{}}

    Making use of the statement from the previous subsection, we prove that average-case Symmetry of Information holds for $\pnKt$. The proof is a simple adaptation of an argument from \cite{HLO25}; we include the details for completeness.
    
    Below, we state Theorem~\ref{t:soi-pKt-avg} in a seemingly stronger form, with a term for $\pKt(x)$ on the right-hand side rather than $\pnKt(x)$. However, we remark that this statement is actually essentially equivalent to Theorem~\ref{t:soi-pKt-avg}, since $\pKt(x) \leq \pnKt(x) + O(\log n)$ with high probability over strings $x$ sampled efficiently.

    \begin{theorem}\label{thm:avg-soi-pnKt-strong}
        For every polynomial-time samplable distribution family $\{\mathcal{D}_n\}_{n \in\mathbb{N}}$ supported over $\{0,1\}^n\times \{0,1\}^n$, there exists a constant $c>0$ such that for all large enough $n, k\in\mathbb{N}$,
        \[ \Prob_{(x,y)\sim \mathcal{D}_n} \left[\; \mathsf{pnKt}(x,y)\geq \mathsf{pKt}(x) + \mathsf{pnKt}(y\mid x) - c\cdot \log n - \log k \; \right] \geq 1-\frac{1}{k}. \]
    \end{theorem}
    \begin{proof}
        The proof is following {\cite[Theorem 8]{HLO25}}. Let $\{\mathcal{D}_n\}_{n \in\mathbb{N}}$ be as in the statement of this theorem, and let the polynomial $p$ be as in \Cref{thm:semi-soi-pk-poly}. Fix any $n \in \N$, and for a sufficiently large polynomial $q$, let $t = q(t_D(n))$, where $t_D(n)$ upper-bounds the time required to sample $\mathcal{D}_n$. Observe that, for any $x, y \in \bool^n$,
        \[\pnKt(y \mid x) \leq \pnK^{p(t)}(y \mid x) + O(\log p(t))\]
        and 
        \[\pKt(x) \leq \pK^{p(t)}(x) + O(\log p(t)).\]
        By the coding theorem for $\pK^{t}$ (\Cref{t:coding_pkt}) and \Cref{f:incompressible}, the following holds: for any $k \in \N$, with probability at least $1 - k^{-1}$ over $(x, y) \sim D_n$,
    	\begin{align*}
    		\pK^{t}(x, y) &\leq \log \frac{1}{D_n(x,y)} + O(\log n)\\
    		&\leq \K(x, y) + O(\log n) + \log k
    		\\&\leq \pnKt(x, y) + O(\log n) + \log k.
    	\end{align*}
        Combining the above with \Cref{thm:semi-soi-pk-poly} and letting $c \in \N$ be a sufficiently large constant, we have that with probability at least $1 - k^{-1}$ over $(x, y) \sim D_n$,
        \[\pnKt(x, y) \geq \pKt(x) + \pnKt(y \mid x) - c\log n - \log k,\]
        as desired.
    \end{proof}

\subsection{Average-Case SoI for \texorpdfstring{$\rnKt$}{}}\label{ss:soi-rnkt}

We first show a version of semi-symmetry of information for $\rnK^t$, which is analogous to \Cref{thm:semi-soi-pk-poly}.
\begin{theorem}\label{thm:semi-soi-rk-poly}
	There exists a polynomial $p$ such that, for all sufficiently large $n, t \in \N$ with $t \geq 4n$, the following holds. For all $x, y \in \bool^n$,
	\[
	\pK^{t}(x, y) \geq \rnK^{p(t)}(x) + \rnK^{p(t)}(y \mid x) - \log p(t) - \log^3 p(n).
	\]
\end{theorem}
\begin{proof}[Proof Sketch]
	The proof follows a similar approach to that used in \cite{GK22,DBLP:conf/coco/Hirahara22} to establish time-bounded versions of SoI under the assumption that $\mathsf{MINKT}$ is easy (see also \cite[Appendix A]{DBLP:conf/coco/HiraharaKLO24}). This approach employs the pseudorandom generator construction stated in \Cref{l:black-box-PRG} and its reconstruction property, while using $\mathsf{MINKT}$ (which is in $\NP$) as a distinguisher. In fact, by adapting this approach and using the fact that the reconstruction procedure makes only non-adaptive queries, one can show that for any $x, y \in \bool^n$ and sufficiently large $t$ that is polynomial in $n$,
	\begin{equation}\label{eq:semi-sot-rK-inter}
		\pK^t(x, y) \geq \rK_{1/\poly(n)}^{\poly(t), \parallel\mathsf{MINKT}}(x) + \rK_{1/\poly(n)}^{\poly(t), \parallel\mathsf{MINKT}}(y \mid x) - O(\log^3 n) - O(\log t).
	\end{equation}
	We now upper-bound the right-hand side of the above. We have
	\begin{align*}
		\rK_{1/\poly(n)}^{\poly(t), \parallel\mathsf{MINKT}}(x)
		&\geq \rnK_{1/\poly(n \cdot t)}^{\poly(t)}(x) - O(\log t) \tag{by \Cref{l:rKt--parallel-NP}}\\
		&\geq \rnK^{\poly(t)}(x) - O(\log t) \tag{by \Cref{l:success-amp-rnK}}
	\end{align*}
	By a similar argument, we have
	\[
	\rK_{1/\poly(n)}^{\poly(t), \parallel\mathsf{MINKT}}(y \mid x) \geq \rnK^{\poly(t)}(y \mid x) - O(\log t).
	\]
	Substituting the above two inequalities into \Cref{eq:semi-sot-rK-inter} yields the desired semi-symmetry of information for $\rnK^t$.
\end{proof}

The following corresponds to \Cref{t:soi-pKt-avg} for the case of $\rnKt$.
        \begin{theorem}\label{thm:avg-soi-rnKt}
        For every polynomial-time samplable distribution family $\{\mathcal{D}_n\}_{n \in\mathbb{N}}$ supported over $\{0,1\}^n\times \{0,1\}^n$, there exists a constant $c>0$ such that for all large enough $n, k\in\mathbb{N}$,
        \[ \Prob_{(x,y)\sim \mathcal{D}_n} \left[\; \mathsf{rnKt}(x,y)\geq \mathsf{rnKt}(x) + \mathsf{rnKt}(y\mid x) - c\cdot \log^3 n - \log k \; \right] \geq 1-\frac{1}{k}. \]
    \end{theorem}
    \begin{proof}[Proof Sketch]
The proof can be easily adapted from that of \Cref{thm:avg-soi-pnKt-strong}, with \Cref{thm:semi-soi-pk-poly} replaced by \Cref{thm:semi-soi-rk-poly}. We omit the details.
    \end{proof}

\section{Complexity Separations from SoI: Proof of \Cref{t:soi-NP-coAM}}    

\subsection{Consequence of Worst-Case SoI for \texorpdfstring{$\pnKt$}{}}\label{sec:worst-soi-pnkt}

\Cref{t:soi-NP-coAM-pnKt}, stated and proved in \Cref{sec:worst-soi-pnkt}, shows that if worst-case symmetry of information holds for $\pnKt$, then $\mathsf{coNP}\not\subseteq\AM$. \Cref{t:soi-NP-coAM-rKt} establishes an analogous result for $\rnKt$ and is proved in \Cref{sec:worst-soi-rnkt}.

\begin{theorem}\label{t:soi-NP-coAM-pnKt}
Item 1 implies Item 2 in the following.
\begin{enumerate}
    \item \emph{\textbf{(Worst-Case Symmetry of Information for $\pnKt$).}} There exists a constant $c>0$ such that for all large enough $n\in\mathbb{N}$ and $x,y\in\bool^{n}$,
		\[
		\pnKt(x,y) \geq \pnKt(x) + \pnKt(y\mid x) - \log^c n.
		\]
    \item $\mathsf{UNSAT}\notin\mathsf{AMTIME}\mleft[2^{\log^d n}\mright]$ for any $d \geq 1$.
\end{enumerate}
\end{theorem}

We will need the following lemmas.

\begin{lemma}\label{l:gapMpKtP-in-prBPE}
There exists a probabilistic oracle algorithm $B^{(-)}$ such that, given inputs $x \in \{0,1\}^m$ and $y \in \{0,1\}^{\leq 2^m \cdot m}$, the algorithm $B^{\SAT}(x, y)$ runs in time $2^{O(m)}$ with access to a $\SAT$-oracle, and
\begin{itemize}
    \item if $\pnKt_{2/3}(x \mid y) \leq m/2$, then $B^{\SAT}(x, y)$ rejects with probability at least $2/3$, and
    \item if $\pnKt_{1/3}(x \mid y) > m/2$, then $B^{\SAT}(x, y)$ accepts with probability at least $2/3$.
\end{itemize}
\end{lemma}

\begin{proof}
Consider the following algorithm: on input $(x, y)$, where $x \in \{0,1\}^m$ and $y \in \{0,1\}^{\leq 2^m \cdot m}$,
\begin{enumerate}
    \item Sample $r \sim \bool^{2^{m/2}}$.
    \item For all programs $\Pi \in \bool^*$ and running times $t \in \mathbb{N}$ such that $|\Pi| + \lceil \log t \rceil \leq m/2$, use the $\SAT$ oracle to answer the following two queries:
    \begin{enumerate}
        \item Does there exist a $w \in \bool^t$ such that $\Pi(y; w; r)$ outputs $x$ in $t$ steps?
        \item Does there exist a $w \in \bool^t$ such that $\Pi(y; w; r)$ either fails to terminate within $t$ steps or outputs some $x' \not\in \{x, \bot\}$ within $t$ steps?
    \end{enumerate}
    \item Reject if and only if there exists a pair $(\Pi, t)$ such that the answer to the first query is yes and the answer to the second query is no.
\end{enumerate}

It is clear that the running time of the above algorithm is $2^{O(m)}$, and correctness follows directly from the definition of $\pnKt$. Indeed, if $\pnKt_{2/3}(x \mid y) \leq m/2$, then the algorithm rejects with probability at least $2/3$ (over the randomness $r$), and if $\pnKt_{1/3}(x \mid y) > m/2$, the algorithm rejects with probability at most $1/3$.
\end{proof}

    \begin{lemma}\label{l:SoI-core-lemma}
    Suppose $\mathsf{SAT} \in \mathsf{coAMTIME}[n^{\polylog(n)}]$. Then for every constant $d > 0$, there are infinitely many integers $m$ such that for every $\ell \leq 2^m$, there exist strings $x_1, \dots, x_{\ell} \in \{0,1\}^m$ satisfying
    \[
        \sum_{i \in [\ell]} \pnKt(x_i \mid x_1, \dots, x_{i - 1}) \geq \ell \cdot \frac{m}{2}
    \]
    and
    \[
        \rnKt(x_1, \dots, x_{\ell}) \leq 2^{m^{1/d}}.
    \]
\end{lemma}
 \begin{proof}
Let $B$ be the (probabilistic) algorithm from \Cref{l:gapMpKtP-in-prBPE}. We first observe the following.

\begin{claim}\label{c:pad-circuit}
There exists a constant $c > 0$ and a language $L \in \AM \cap \mathsf{coAM}$ such that the following holds. Given any $x \in \bool^m$ and $y \in \bool^{\leq 2^m \cdot m}$, one can construct, in deterministic time $\poly\mleft(2^{m^c}\mright)$, an $L$-oracle circuit $C_{(x,y)}$ of size $2^{m^c}$ that treats its input as internal randomness used by the algorithm $B$ and computes $B^{\SAT}(x, y)$; that is, $C_{(x,y)}(r) = B^{\SAT}(x, y; r)$ for every $r$.
\end{claim}
\begin{proof}[Proof of \Cref{c:pad-circuit}]\renewcommand\qedsymbol{$\diamond$}
By \Cref{l:gapMpKtP-in-prBPE}, for every $x \in \bool^m$ and $y \in \bool^{\leq 2^m \cdot m}$, we can construct a $\SAT$-oracle circuit $C_0$ of size $s = 2^{O(m)}$ that treats its input as internal randomness and computes $B^{\SAT}(x, y)$. Note that $C_0$ makes queries to $\SAT$ of length at most~$s$.

By assumption, $\SAT \in \mathsf{coAMTIME}[t(n)]$, where $t(n) = n^{\polylog(n)}$. Applying a standard padding argument, we can define a language $L \in \AM \cap \mathsf{coAM}$ that is an appropriately padded version of $\SAT$, such that $C_0$ can be simulated by an $L$-oracle circuit of size at most $\poly(t(s)) \leq 2^{m^c}$ for a constant $c$ chosen large enough.

\end{proof}

Let $c > 0$ be the constant, and let $L \in \AM \cap \mathsf{coAM}$ be the language from
\Cref{c:pad-circuit}. Consider \Cref{t:prg-circuit} with $\varepsilon \vcentcolon= 1/(4cd)$, and let
$\mathcal{S} \subseteq \nat$ be the infinite set of lengths $s$ at which the multiset $T_s$ is
pseudorandom against $L$-oracle circuits of size $s$. For $s \in \mathcal{S}$, let
\[
    m(s) \ \vcentcolon= \ \max\{\,m \in \nat \ :\ 2^{m^{c}} \le s\,\},
\]
so that $2^{m(s)^{c}} \le s < 2^{(m(s)+1)^{c}}$. We say that an integer $m$ is \emph{good} if
$m = m(s)$ for some $s \in \mathcal{S}$. Note that there are infinitely many good values of $m$.

Fix any (sufficiently large) good $m$. Let $s$ be the integer that witnesses the goodness of $m$, and let $T_s$ be the corresponding pseudorandom set.

Consider the following deterministic procedure $A$, with oracle access to $L$, which takes as input $m$, $T_s$, and a string $y \in \bool^{\leq 2^m \cdot m}$, and proceeds as follows:
\begin{quote}
    For each $x \in \bool^m$, compute
    \[
        \mu_x \vcentcolon= \Prob_{z \sim T_s}\left[C_{(x,y)}(z) = 1\right],
    \]
    where $C_{(x,y)}$ is the $L$-oracle circuit obtained using \Cref{c:pad-circuit}.  
    Output the first $x$ such that $\mu_x > \frac{1}{3} + \frac{1}{10}$.
\end{quote}

\begin{claim}\label{c:pKt-construct}
The procedure $A$, on input $(m, T_s, y)$ with $y \in \bool^{\leq 2^m \cdot m}$, runs in deterministic time at most $2^{\poly(m)}$ when given oracle access to $L$ and outputs an $m$-bit string $x$ satisfying $\pnKt(x \mid y) \geq m/2$.
\end{claim}
\begin{proof}[Proof of \Cref{c:pKt-construct}]\renewcommand\qedsymbol{$\diamond$}
We first analyze the running time. Since $|T_s| \leq \poly(s)$ and $s \leq 2^{(m+1)^c}$, it follows that the algorithm runs in time at most $2^{\poly(m)}$.

We now show correctness. Because $T_s$ $(1/10)$-fools $C_{(x,y)}$ for every $x \in \bool^m$ and $y \in \bool^{\leq 2^m \cdot m}$, the value $\mu_x$ approximates $\Prob_B[B^{\SAT}(x,y) = 1]$ within additive error $1/10$. Thus, any output $x$ cannot satisfy $\pnKt_{2/3}(x \mid y) \leq m/2$, as that would imply $\Prob_B[B^{\SAT}(x,y) = 1] < 1/3$ and hence $\mu_x < 1/3 + 1/10$, contradicting the condition for output. Furthermore, since we enumerate all $x \in \bool^m$, at least one of them satisfies $\pnKt_{1/3}(x \mid y)>m/2$. It follows that $B^{\SAT}(-, y)$ accepts at least one such $x$. Therefore, we are guaranteed to find one with $\mu_x \geq 2/3 - 1/10 \geq 1/3 + 1/10$.
\end{proof}

Define $x_1, x_2, \dots, x_\ell \in \bool^m$ as follows:
\begin{itemize}
    \item $x_1$ is the output of $A(m, T_s, \epsilon)$.
    \item For $i = 2, 3, \dots, \ell$, let $x_i$ be the output of $A(m, T_s, x_1, \dots, x_{i-1})$.
\end{itemize}

By correctness of $A$, for each $i\in[\ell]$, we have
\[
    \pnKt(x_i \mid x_1, \dots, x_{i - 1})\geq \frac{m}{2}
\]
It follows that
\[
    \sum_{i \in [\ell]} \pnKt(x_i \mid x_1, \dots, x_{i - 1}) \geq \ell \cdot m/2,
\]
as desired.

To upper bound the $\rnKt$-complexity, observe that we can first compute $T_s$ and then run procedure $A$ repeatedly to obtain the sequence $x_1, \dots, x_\ell$.

From the efficiency guarantee in \Cref{t:prg-circuit}, we have:
\begin{align*}
    \rnKt(T_s) &\leq 2^{\log^\varepsilon s} \leq 2^{m^{1/3d}},
\end{align*}

Given $m$ and $T_s$, the sequence $x_1, \dots, x_\ell$ can be computed via $\ell$ invocations of $A$, each running in time $2^{\poly(m)}$ \emph{given oracle access to $L$}. This yields:
\begin{align*}\label{l:x-rnkt-upper}
    \rnKt^{L}(x_1, \dots, x_\ell) &\leq O\mleft (2^{m^{1/3d}}\mright) + \log\left(\ell \cdot 2^{\poly(m)}\right) \\
    &\leq 2^{m^{1/(2d)}}.\numberthis
\end{align*}
In the above we use the fact that $m$ is sufficiently large.

Finally, since $L\in \AM\cap\coAM$, we get 
\begin{align*}
    \rnKt(x_1, \dots, x_\ell) &\leq  O\mleft(\rnKt^{L}(x_1, \dots, x_\ell)\mright) \tag{by \Cref{l:rKt-AM-coAM}}\\
    &\leq O\mleft(2^{m^{1/(2d)}}\mright) \tag{by \Cref{l:x-rnkt-upper}}\\
    &\leq 2^{m^{1/d}}.
\end{align*}
This completes the proof.
\end{proof}

We now proceed to prove \Cref{t:soi-NP-coAM-pnKt}.
\begin{proof}[Proof of \Cref{t:soi-NP-coAM-pnKt}]
We show the contrapositive. Assume that $\mathsf{UNSAT} \in \mathsf{AMTIME}[n^{\polylog(n)}]$. Note that this implies $\mathsf{SAT} \in \mathsf{coAMTIME}[n^{\polylog(n)}]$. Below, we show the failure of worst-case symmetry of information for $\pnKt$. That is, for every constant $c>0$, there exist infinitely many $n\in\mathbb{N}$ and $x,y\in\bool^{n}$ such that
		\[
		\pnKt(x,y) < \pnKt(x) + \pnKt(y\mid x) - \log^c n.
		\]

Let $c > 1$ be any constant. We first establish the following claim.

\begin{claim}\label{c:pre-padding}
    The following holds for infinitely many values of $m$. There exist $v \in \bool^{m}$ and $u \in \bool^{n}$, where $m \leq n \leq 2^{m^{1/(2c)}}$, such that
    \[
        \pnKt(u,v) < \pnKt(u) + \pnKt(v \mid u) - \frac{m}{4},
    \]
\end{claim}

\begin{proof}[Proof of \Cref{c:pre-padding}]\renewcommand{\qedsymbol}{$\diamond$}
    Suppose, for the sake of contradiction, that for all but finitely many $m$, the following holds: for all $v \in \bool^m$ and $u \in \bool^n$ satisfying $m \leq n \leq 2^{m^{1/(2c)}}$, we have
    \[
        \pnKt(u,v) \geq \pnKt(u) + \pnKt(v \mid u) - \frac{m}{4}.
    \]

    Consider \Cref{l:SoI-core-lemma} with $d \coloneqq 3c$. Let $m$ be a sufficiently large ``good'' integer such that the statement of \Cref{l:SoI-core-lemma} holds for $\ell \coloneqq 2^{m^{1/d}}$, and let $x_1, \dots, x_\ell$ be the associated sequence of strings.

    Observe that $\ell \cdot m \leq 2^{m^{1/(2c)}}$. Therefore, we can apply the assumption to $v = x_\ell$ and $u = (x_1, \dots, x_{\ell - 1})$, which yields: 
    \begin{align*}
        \pnKt(x_1, \dots, x_\ell) &\geq \pnKt(x_1, \dots, x_{\ell-1}) + \pnKt(x_\ell \mid x_1, \dots, x_{\ell-1}) - \frac{m}{4}
    \end{align*}
    Repeating this argument inductively gives:
    \begin{align*}
        \pnKt(x_1, \dots, x_\ell)
        &\geq \sum_{i \in [\ell]} \pnKt(x_i \mid x_1, \dots, x_{i-1}) - \ell \cdot \frac{m}{4} \\
        &\geq \ell \cdot \frac{m}{2} - \ell \cdot \frac{m}{4} \tag{by \Cref{l:SoI-core-lemma}} \\
        &= 2^{m^{1/d}} \cdot \frac{m}{4}.
    \end{align*}
    However, this contradicts the fact that
    \[
        \pnKt(x_1, \dots, x_\ell) \leq 2^{m^{1/d}}.
    \]
    This completes the proof of the claim.
\end{proof}

By \Cref{c:pre-padding}, for infinitely many values of $m$, there exist $v \in \bool^m$ and $u \in \bool^n$ such that
\begin{equation}\label{eq:n}
    m \leq n \leq 2^{m^{1/(2c)}},
\end{equation}
and
\begin{equation}\label{eq:uv}
    \pnKt(u, v) < \pnKt(u) + \pnKt(v \mid u) - \frac{m}{4}.
\end{equation}
Define
\begin{equation}\label{e:x-y}
    x \coloneqq u \quad \text{and} \quad y \coloneqq v0^{n - m}.
\end{equation}
Then we have:
\begin{align*}
    \pnKt(x, y) 
    &\leq \pnKt(u, v) + O(\log n) \tag{by \Cref{e:x-y}} \\
    &\leq \pnKt(u) + \pnKt(v \mid u) - \frac{m}{4} + O(\log n) \tag{by \Cref{eq:uv}} \\
    &\leq \pnKt(x) + \pnKt(y \mid x) + O(\log n) - \frac{m}{4} + O(\log n) \tag{by \Cref{e:x-y}} \\
    &\leq \pnKt(x) + \pnKt(y \mid x) + O(\log n) - \frac{\log^{2c} n}{4} \tag{by \Cref{eq:n}} \\
    &< \pnKt(x) + \pnKt(y \mid x) - \log^c n,
\end{align*}
as desired.
\end{proof}

\subsection{Consequence of Worst-Case SoI for \texorpdfstring{$\rnKt$}{}}\label{sec:worst-soi-rnkt}

The following corresponds to \Cref{t:soi-NP-coAM} for the case of $\rnKt$.

    \begin{theorem}\label{t:soi-NP-coAM-rKt}
Item 1 implies Item 2 in the following.
\begin{enumerate}
    \item \emph{\textbf{(Worst-Case Symmetry of Information for $\rnKt$).}} There exists a constant $c>0$ such that for all large enough $n\in\mathbb{N}$ and $x,y\in\bool^{n}$,
		\[
		\rnKt(x,y) \geq \rnKt(x) + \rnKt(y\mid x) - \log^c n.
		\]
    \item $\mathsf{UNSAT}\notin\mathsf{AMTIME}\mleft[2^{\log^d n}\mright]$ for any $d \geq 1$.
\end{enumerate}
\end{theorem}
\begin{proof}[Proof Sketch]
    The proof is a easy adaption of that of \Cref{t:soi-NP-coAM-pnKt}, which employs \Cref{l:SoI-core-lemma}. It suffices to note that in \Cref{l:SoI-core-lemma}, we can replace $\pnKt$ with $\rnKt$, since for every $x,y\in\bool^n$, it holds that $\pnKt(x\mid y)\leq \rnKt(x\mid y)$.
\end{proof}

\section{Relativization Barrier for SoI: Proof of \Cref{thm:neg-o-soi}}\label{ss:rel_barrier_soi}

We will prove \Cref{thm:neg-o-soi}, restated here for convenience.

\begin{theorem}

    There exist an oracle $\mathcal{O}$ and a constant $\varepsilon > 0$ such that, for every complexity measure $\mu \in \{\rKt, \pKt, \nKt, \rnKt, \pnKt\}$, the following holds. For all sufficiently large $n \in \N$, there exist strings $x, y \in \bool^n$ such that
    \[\mu^{\mathcal{O}}(x, y) < \mu^{\mathcal{O}}(x) + \mu^{\mathcal{O}}(y \mid x) - \varepsilon \cdot n.\]
    
\end{theorem}

The proof is basically a relativizing version of the argument used in \Cref{t:soi-NP-coAM-pnKt}, adapted so that it works for all measures $\mu \in \{\rKt, \pKt, \nKt, \rnKt, \pnKt\}$ and
   gives a linear gap $\varepsilon \cdot n$ instead of just polylogarithmic.

We will make use of the following lemmas.

\begin{lemma}[\cite{DBLP:journals/siamcomp/BakerGS75}]\label{lem:bgs}

    There exists an oracle $\mathcal{O}$ such that $\PH^{\mathcal{O}} \subseteq \P^{\mathcal{O}}$.
    
\end{lemma}

\begin{lemma}\label{lem:ph-high-np-o}

    There is a deterministic polynomial-time oracle machine $M_0^{(-)}$ as follows. Let $\mathcal{O}$ be any oracle. There exists an oracle $\mathcal{X} \in \Sigma^{\P, \mathcal{O}}_3$ such that for any $n \in \N$, $M_0^{\mathcal{X}}(1^n)$ outputs a string $x \in \bool^n$ with $\NP^\mathcal{O}$-oracle circuit complexity at least $n^{1/2}$. 
    
\end{lemma}
\begin{proof}

   For a given oracle $\mathcal{O}$, let $\mathcal{X} \in \Sigma^{\P, \mathcal{O}}_3$ be as follows.
    \begin{align*}
        \mathcal{X} = \{ (1^n, y) \mid \; &y \in \bool^{<n} \; ; \; \exists x \in \bool^n \text{ and } w \in \bool^{n - |y| - 1}\text{ such that } x = y0w \text{ and} 
        \\&\forall \text{ oracle circuit } C^{(-)} \text{ of size at most  }n^{1/2}, x \text{ is not the truth-table of } C^{\NP^{\mathcal{O}}}\},    
    \end{align*}
    The machine $M_0^{(-)}$ is defined as follows.
    \begin{quote}
        On input $1^n$, start with $y$ being the empty string. Repeat the following $n$ times. Query $(1^n, y)$ to the oracle $\mathcal{X}$; if $\mathcal{X}$ accepts, let $y := y0$, and if $\mathcal{X}$ rejects, let $y := y1$. After this is done, return the string $x := y\in \bool^n$.
    \end{quote}
    By a standard counting argument, there exists a string $x \in \bool^n$ with $\NP^{\mathcal{O}}$-oracle circuit complexity greater than $n^{1/2}$. (In fact, most strings have this property.) Therefore, it is clear that $M_0^{\mathcal{X}}$ will find such a string in polynomial time, as desired. 
\end{proof}

\begin{lemma}\label{lem:high-np-o}

    There exist an oracle $\mathcal{O}$ and a deterministic polynomial-time oracle machine $M_1^{(-)}$ as follows. For any $n \in \N$, $M_1^{\mathcal{O}}(1^n)$ outputs a string $x \in \bool^n$ with $\NP^{\mathcal{O}}$-oracle circuit complexity at least $n^{1/2}$. Moreover, $\PH^{\mathcal{O}} \subseteq \P^{\mathcal{O}}$. 
    
\end{lemma}
\begin{proof}

    Let $\mathcal{O}$ be the oracle of \Cref{lem:bgs}. Let $M_0^{(-)}$ be the machine and $\mathcal{X} \in \Sigma^{\P, \mathcal{O}}_3$ the oracle of \Cref{lem:ph-high-np-o} applied to $A$. Since $\PH^{\mathcal{O}} \subseteq \P^{\mathcal{O}}$, there exists a deterministic polynomial-time oracle machine $N^{(-)}$ such that $N^{\mathcal{O}}$ decides $\mathcal{X}$. Let $M_1^{\mathcal{O}}$ be the machine that simulates $M_0^{(-)}$ and answers its oracle queries with the machine $N^{\mathcal{O}}$. By \Cref{lem:ph-high-np-o}, for any $n \in \N$, $M_1^{\mathcal{O}}(1^n)$ produces a string $x \in \bool^n$ with $\NP^{\mathcal{O}}$-oracle circuit complexity at least $n^{1/2}$. The lemma follows.
\end{proof}

We will require the following lemma, which is similar to \Cref{l:gapMpKtP-in-prBPE}, but pertains to $\pKt^{\mathcal{O}}$ complexity for arbitrary oracles $\mathcal{O}$ rather than $\pnKt$ complexity.

\begin{lemma}\label{lem:B-pkt-O}

    There exists a probabilistic oracle algorithm $B^{(-)}$ such that, for any oracle $\mathcal{O}$, given inputs $x \in \{0,1\}^m$ and $y \in \{0,1\}^{\leq m^2}$, $B^{\mathcal{O}}(x, y)$ runs in time $2^{O(m)}$ with oracle access to $\mathcal{O}$, and 
    \begin{itemize}
        \item if $\pKt^{\mathcal{O}}_{2/3}(x \mid y) \leq m/2$, then $B^{\mathcal{O}}(x, y)$ rejects with probability at least $2/3$, and 
        \item if $\pKt^{\mathcal{O}}_{1/3}(x \mid y) > m/2$, then $B^{\mathcal{O}}(x, y)$ accepts with probability at least $2/3$.
    \end{itemize}

\end{lemma}
\begin{proof}

    The algorithm $B^{\mathcal{O}}$ is as follows. Given $(x, y)$,
    \begin{enumerate}
        \item Sample $r \sim \bool^{2^{m/2}}$.
        \item For all programs $\Pi$ and running times $t \in \N$ with $|\Pi| + \lceil \log t \rceil \leq m/2$, check if $\Pi^{\mathcal{O}}(y; r)$ outputs $x$ within $t$ steps. Reject if so.
        \item Accept if Step 2 did not reject. 
    \end{enumerate}

    It is easy to see that $B^{\mathcal{O}}$ runs in time $2^{O(m)}$ given oracle access to $\mathcal{O}$. Correctness follows easily from the definition of $\pKt^{\mathcal{O}}$. 
\end{proof}

\begin{lemma}\label{lem:multi-oracle-soi}

    There exist an oracle $\mathcal{O}$ and a constant $a \in \N$ such that, for every $\ell, n \in \N$ with $\ell \leq n$, there exist strings $x_1, ..., x_{\ell}\in\bool^{n}$ satisfying
    \[\sum_{i \in [\ell]}\pnKt^{\mathcal{O}}(x_i \mid x_1, ..., x_{i - 1}) > \ell \cdot \frac{n}{a}\]
    and        
    \[\Kt^{\mathcal{O}}(x_1, ..., x_{\ell}) < a \cdot n.\]
    
\end{lemma}
\begin{proof}

    We follow the proof technique of \Cref{l:SoI-core-lemma}. Let $B^{(-)}$ be the randomized oracle algorithm of \Cref{lem:B-pkt-O}. Recall that, on inputs $x \in \bool^n$ and $y \in \bool^{\leq n^2}$, $B^{(-)}$ runs in time at most $2^{O(n)}$. For a constant $b \in \N$, $x \in \bool^n$, and $y \in \bool^{\leq n^2}$, let $C^{(-)}_{(x, y)}: \bool^{2^{bn}} \to \bool$ be the oracle circuit of size at most $2^{bn}$ that treats its input as internal randomness and simulates $B^{(-)}(x, y)$.

    Let $\mathcal{O}$ be the oracle and $M_1$ the machine of \Cref{lem:high-np-o}. Let $c, d \in \N$ be as in \Cref{t:IW97-PRG} applied with $\varepsilon = 1/2$, and let $z \in \bool^{2^{bc\cdot n}}$ be the output of $M_1^{\mathcal{O}}(1^{2^{bc\cdot n}})$. That is, $z$ has $\NP^{\mathcal{O}}$-oracle circuit complexity greater than $2^{bc\cdot n/2}$. Let $G_z: \bool^{bd\cdot n} \to \bool^{2^{bn}}$ be the PRG of \Cref{t:IW97-PRG}. In particular, $G_z$ is computable in time at most $\poly(2^{bc\cdot n}, bd \cdot n) = 2^{O(n)}$ and $(1/2^{bn})$-fools $C^{\NP^{\mathcal{O}}}_{(x,y)}$.
    
    Consider the following procedure $A$.
    \begin{quote}

        On input $n \in \N$, $y \in \bool^{\leq n^2}$, and $z \in \bool^{2^{bc\cdot n}}$, for each $x \in \bool^n$, compute
        \[\rho_x = \Prob_{\sigma \sim \bool^{bd\cdot n}}[C^{\NP^{\mathcal{O}}}_{(x, y)}(G_z(\sigma)) = 1].\]
        Output the first $x$ such that $\rho_x > \frac{1}{3} + \frac{1}{10}.$
        
    \end{quote}
    It is easy to see that, with $z$ as above, $A(n, y, z)$ runs in time at most $2^{O(n)}$, using an $\NP^{\mathcal{O}}$ oracle to simulate $C^{\NP^{\mathcal{O}}}$. Moreover, by the guarantee of \Cref{lem:B-pkt-O}, if $\pKt^{\NP^{\mathcal{O}}}(x \mid y) \leq n/2$, then $B^{\NP^{\mathcal{O}}}$ accepts with probability at most $1/3$. Therefore, $A$ must output a string $x \in \bool^n$ such that $\pKt^{\NP^{\mathcal{O}}}(x \mid y) > n/2$.

    Now, define
    \begin{itemize}
        \item  $x_1 = A(n, \epsilon, z)$, where $\varepsilon$ is the empty string;
        \item $x_i = A(n, (x_1, ..., x_{i - 1}), z)$ for $i \in \{2, ..., \ell\}$.
    \end{itemize}
    We claim that the strings $x_1, ..., x_{\ell} \in \bool^n$ have the properties stated in the lemma. Firstly, by the correctness of $A$ and \Cref{l:pKt-SAT}, for some constant $k \in \N$, for each $i \in [\ell]$, we have 
    \[\pnKt^{\mathcal{O}}(x_i \mid x_1, ..., x_{i-1}) \geq \frac{1}{k}\cdot \pKt^{\NP^{\mathcal{O}}}(x_i \mid x_1, ..., x_{i-1}) \geq \frac{n}{2k}.\]
    Thus,
    \[\sum_{i \in [\ell]}\pnKt^{\mathcal{O}}(x_i \mid x_1, ..., x_{i-1}) \geq \ell \cdot \frac{n}{2k}.\]

    Secondly, we observe that $x_1, ..., x_{\ell}$ can be produced uniformly by running $M_1^{\mathcal{O}}$ to produce a hard truth-table $z$ followed by $\ell \leq n$ invocations of $A$. This takes deterministic time at most $2^{O(n)}$ given an $\NP^{\mathcal{O}}$ oracle. Thus,
    \[\Kt^{\NP^{\mathcal{O}}}(x_1, ..., x_{\ell}) \leq O(n).\]
    Since $\PH^{\mathcal{O}} \subseteq \P^{\mathcal{O}}$, for any language $L \in \NP^{\mathcal{O}}$, there exists a deterministic polynomial-time oracle machine $N^{\mathcal{O}}$ that decides $L$. Therefore, replacing oracle calls to $L \in \NP^{\mathcal{O}}$ by simulating $N^{\mathcal{O}}$, we obtain
    \begin{align*}
        \Kt^{\mathcal{O}}(x_1, ..., x_{\ell}) &< 2a'\cdot (\Kt^{\NP^{\mathcal{O}}}(x_1, ..., x_{\ell})) 
        \\&\leq O(n),
    \end{align*}
    where $a' \in \N$ is such that $t^{a'}$ upper-bounds the running time of $N^{(-)}$ on inputs of length $t$.

    This completes the proof.
\end{proof}

The following is immediate from the definitions of the complexity measures in question.

\begin{corollary}\label{cor:multi-soi-mu}

    There exist an oracle $\mathcal{O}$ and a constant $d \in \N$ as follows. For every $\ell, n \in \N$ with $\ell \leq n$, there exist strings $x_1, ..., x_{\ell}\in\bool^{n}$ such that, for every complexity measure $\mu \in \{\rKt, \pKt, \nKt, \rnKt, \pnKt\}$, 
    \[\sum_{i \in [\ell]}\mu^{\mathcal{O}}(x_i \mid x_1, ..., x_{i - 1}) > \ell \cdot \frac{n}{d}\]
    and        
    \[\mu^{\mathcal{O}}(x_1, ..., x_{\ell}) < d \cdot n.\]
    
\end{corollary}

Given the above, we are ready to prove \Cref{thm:neg-o-soi} as stated at the beginning of this section. The argument is quite similar to that of \Cref{t:soi-NP-coAM-pnKt}. We include a proof for completeness.

\begin{proof}[Proof of \Cref{thm:neg-o-soi}]

    Fix $\mu \in \{\rKt, \pKt, \nKt, \rnKt, \pnKt\}$. We start with the following claim.

    \begin{claim}\label{claim:soi-o-contradiction}

        Let $\mathcal{O}$ be the oracle of \Cref{cor:multi-soi-mu}. For some constant $c \in \N$, for all sufficiently large $m \in \N$, there exist $j \in \N$ with $m \leq j \leq c^2\cdot m$ and strings $u \in \bool^{j}$ and $v \in \bool^{m}$ such that
        \[\mu^{\mathcal{O}}(u, v) < \mu^{\mathcal{O}}(u) + \mu^{\mathcal{O}}(v \mid u) - \frac{m}{c}.\]
        
    \end{claim}
    \begin{proof}[Proof of \Cref{claim:soi-o-contradiction}]

        Toward a contradiction, suppose that for any constant $c \in \N$, there exist infinitely many $m \in \N$ such that for all $j \in \N$ with $m \leq j \leq c^2 \cdot m$ and all $u \in \bool^{j}$ and $v \in \bool^{m}$, it holds that 
        \begin{align}\label{eq:for-cont-soi}
            \mu^{\mathcal{O}}(u, v) \geq \mu^{\mathcal{O}}(u) + \mu^{\mathcal{O}}(v \mid u) - \frac{m}{c}.
        \end{align}
        Let $d$ be the constant from \Cref{cor:multi-soi-mu}, let $c = 2d$, and let $m$ be as in \Cref{eq:for-cont-soi}. Let $\ell := c^2$, and let $x_1, ..., x_{\ell} \in \bool^m$ be the strings guaranteed by \Cref{cor:multi-soi-mu}. That is,
        \begin{align*}
            \sum_{i \in [\ell]}\mu^{\mathcal{O}}(x_i \mid x_1, ..., x_{i - 1}) &> \ell \cdot \frac{m}{d}
            \\&= \frac{(2d)^2 \cdot m}{d} = 4d \cdot m
        \end{align*}
        and
        \begin{align}\label{eq:mu-o-all-small}
            \mu^{\mathcal{O}}(x_1, ..., x_{\ell}) < d \cdot m.
        \end{align}
        Consider $u = (x_1, ..., x_{\ell - 1}) \in \bool^{(\ell - 1)\cdot m}$ and $v = x_{\ell} \in \bool^m$. By \Cref{eq:for-cont-soi}, we have
        \[\mu^{\mathcal{O}}(x_1, ..., x_{\ell}) \geq \mu^{\mathcal{O}}(x_1, ..., x_{\ell - 1}) + \mu^{\mathcal{O}}(x_{\ell} \mid x_1, ..., x_{\ell - 1}) - \frac{m}{c}.\]
        Repeating inductively, we obtain
        \begin{align*}
            \mu^{\mathcal{O}}(x_1, ..., x_{\ell}) &\geq \sum_{i \in [\ell]}\mu^{\mathcal{O}}(x_i \mid x_1, ..., x_{i - 1}) - \ell \cdot \frac{m}{c} 
            \\&\geq 4d \cdot m - 2d \cdot m
            \\&= 2d \cdot m,
        \end{align*}
        contradicting \Cref{eq:mu-o-all-small}.
    \end{proof}

    Now, let the constant $c$ be as in \Cref{claim:soi-o-contradiction}. For any sufficiently large $n \in \N$, let $m = \lfloor n/c^4\rfloor$. By \Cref{claim:soi-o-contradiction}, there exists $j \in \N$ with $m \leq j \leq c^2\cdot m < n$ and strings $u \in \bool^{j}$ and $v \in \bool^{m}$ violating SoI. Then, define $n$-length strings
    \[x = u1^{n - j} \; \text{    and    }\; y = v1^{n - m}.\]
    Observe that 
    \[\mu^{\mathcal{O}}(x, y) \leq \mu^{\mathcal{O}}(u, v) + O(\log n)\] 
    and similarly 
    \[\mu^{\mathcal{O}}(u) \leq \mu^{\mathcal{O}}(x) + O(\log n)\] 
    and 
    \[\mu^{\mathcal{O}}(v \mid u) \leq \mu^{\mathcal{O}}(y \mid x) + O(\log n).\]

    It follows that 
    \begin{align*}
        \mu^{\mathcal{O}}(x, y) &\leq \mu^{\mathcal{O}}(u, v) + O(\log n)
        \\&< \mu^{\mathcal{O}}(u) + \mu^{\mathcal{O}}(v \mid u) - \frac{m}{c} + O(\log n)
        \\&\leq \mu^{\mathcal{O}}(x) + \mu^{\mathcal{O}}(y \mid x)  - \frac{m}{c} + O(\log n)
        \\&\leq \mu^{\mathcal{O}}(x) + \mu^{\mathcal{O}}(y \mid x) - \varepsilon \cdot n. \tag{choosing $\varepsilon := 1/c^6$}
    \end{align*}

    This completes the proof of the theorem.
\end{proof}

\section{Complexity Collapses from SoI: Proof of \texorpdfstring{\Cref{thm:collapses_from_SoI}}{}}

\subsection{Uniform Collapses}\label{sec:uniform_collapse}

\newcommand{\Bad}{\mathsf{Bad}}
\newcommand{\Good}{\mathsf{Good}}

In this subsection, we prove the first item of \Cref{thm:collapses_from_SoI}. We present the proof for the case where the SoI error is $O(\log n)$; it can be easily adapted to the case of $\poly(\log n)$ error.

\begin{theorem}[SoI for $\rnKt$ $\Rightarrow$ $\mathsf{FewEXP} \subseteq \mathsf{BPP}^{\mathsf{NP}}$]\label{thm:SoI_uniform_collapse}
	If either $\rnKt$ or $\nKt$ satisfies \emph{symmetry of information} with logarithmic error, then $\mathsf{FewEXP} \subseteq \mathsf{BPP}^{\mathsf{NP}}$.
\end{theorem}

To prove \Cref{thm:SoI_uniform_collapse}, the key ingredient is the following lemma, which states that if $\mathsf{FewEXP} \not\subseteq \mathsf{BPP}^{O}$ for some $O \in \PSPACE$, then there exists a pseudorandom generator with certain efficiency guarantees that fools randomized algorithms with oracle access to $O$.

	\begin{lemma}[Hardness-to-PRG]\label{lem:prg-rel}
		Let $O\in\PSPACE$.
		If $\mathsf{FewEXP} \not\subseteq \BPP^{O}$, then for every constant $\eta>0$ and every constant $b\in\mathbb{N}$,
		there exists a sequence of functions $G=\{G_n\}_{n\in\mathbb{N}}$ with
		\[
		G_n:\{0,1\}^{n^{\eta/2}}\to\{0,1\}^n
		\]
		such that:
		\begin{enumerate}
			\item Each $G_n$ is computable in single-valued nondeterministic time $2^{O(n^{\eta/2})}$ with $O(n^{\eta/2})$ bits of advice.
			\item For every distribution family $\{C_n\}$ of $O$-oracle Boolean circuits samplable in time $n^b$, there are infinitely many $n$ such that, with probability at least $1 - n^{-b}$ over $C \sim C_n$, $G_n$ $n^{-b}$-fools $C$.
		\end{enumerate}
	\end{lemma}
	\begin{proof}
Fix $O \in \PSPACE$, and assume $\mathsf{FewEXP} \not\subseteq \BPP^{O}$. We consider two cases.

\paragraph{Case 1: $\mathsf{EXP} \not\subseteq \BPP^{O}$.}  
The desired PRG follows directly from \Cref{t:io-PRG-EXP}.

\paragraph{Case 2: $\mathsf{EXP} \subseteq \BPP^{O}$.}  
Under the assumption that $\mathsf{FewEXP} \not\subseteq \BPP^{O}$, it follows that $\mathsf{FewEXP} \neq \mathsf{EXP}$. We show the following claim.

\begin{claim}\label{FewEXP-easy-witness}
	If $\mathsf{FewEXP} \neq \mathsf{EXP}$, then there exist a verifier $V$ and constants $\alpha,k>0$ such that for infinitely many $\ell$, there exist $x_\ell \in \{0,1\}^\ell$ and a set $W_\ell \subseteq \{0,1\}^{2^{\ell^\alpha}}$ satisfying:
	\begin{itemize}
		\item $V(x_\ell,w) = 1$ for every $w \in W_\ell$, where $V$ runs in time $2^{O(\ell^{\alpha})}$;
		\item $1 \leq |W_\ell| \leq 2^{\ell^k}$;
		\item for every $w \in W_\ell$, when $w$ is interpreted as the truth table of a Boolean function, it cannot be computed by any $O$-oracle circuit of size $\poly(\ell)$.
	\end{itemize}
\end{claim}

\begin{proof}[Proof of \Cref{FewEXP-easy-witness}]
	\renewcommand\qedsymbol{$\diamond$}
	
	Let $L \in \mathsf{FewEXP} \setminus \mathsf{EXP}$, and let $V$ be a verifier for $L$. For each $x \in \{0,1\}^\ell$, let $W_x$ denote the set of witnesses for $x$ with respect to $V$ (possibly empty). Without loss of generality, assume each witness has length exactly $2^{\ell^{\alpha}}$ for some constant $\alpha>0$. Also, by the definition of $\mathsf{FewEXP}$, there exists a constant $k>0$ such that $|W_x| \le 2^{\ell^k}$.
    
    Suppose, for the sake of contradiction, that for all but finitely many $x \in L$, some witness $w \in W_x$, when viewed as the truth table of a Boolean function on $\ell^{\alpha}$ input bits, has polynomial-size $O$-oracle circuit complexity.
	
	We claim that this assumption implies $L \in \mathsf{EXP}$. On input $x \in \{0,1\}^\ell$, enumerate all $O$-oracle circuits of size at most $\mathrm{poly}(\ell^{\alpha})=\poly(\ell)$. For each such circuit, evaluate it on all inputs to obtain a truth table $y \in \{0,1\}^{2^{\ell^\alpha}}$, and check whether $V(x,y) = 1$. Accept if and only if such a $y$ exists.
	
	Since $O \in \mathsf{PSPACE}$, each oracle query can be simulated in time $2^{\mathrm{poly}(\ell)}$, and the total number of circuits of size $\mathrm{poly}(\ell)$ is at most $2^{\mathrm{poly}(\ell)}$. Therefore, evaluating each circuit on all inputs and verifying $V(x,y) = 1$ can be done in time $2^{\mathrm{poly}(\ell)}$, giving an exponential-time algorithm for $L$.
	
	This contradicts the assumption that $L \notin \mathsf{EXP}$. Hence, for infinitely many yes-instances $x$, every witness in $W_x$ (which is nonempty) must have super-polynomial $O$-oracle circuit complexity.
\end{proof}

Next, we show that \Cref{FewEXP-easy-witness} allows us to efficiently obtain, in nondeterministic time with small advice, hard truth tables that have super-polynomial-size $O$-oracle circuit complexity.

\begin{claim}\label{c:hard-tt}
	If $\mathsf{FewEXP} \neq \mathsf{EXP}$, there exist a constant $\beta>0$ and a sequence $\{y_m\in\{0,1\}^{2^m}\}_m$ such that the following holds for infinitely many $m$:
	\begin{itemize}
		\item $y_m$, viewed as a truth table, is not computable by any $O$-oracle circuit of size $\poly(m)$.
		\item $y_m$ can be computed in single-valued nondeterministic time $2^{m^\beta}$ with $m^\beta$ bits of advice.
	\end{itemize}
\end{claim}

\begin{proof}[Proof of \Cref{c:hard-tt}]
	\renewcommand\qedsymbol{$\diamond$}
	
	Let $V$, $\alpha$, and $k$ be as in \Cref{FewEXP-easy-witness}. By that claim, for infinitely many lengths $\ell$, there exist an input $x_\ell \in \{0,1\}^\ell$ and a (nonempty) set of witnesses $W_\ell \subseteq \{0,1\}^{2^{\ell^{\alpha}}}$ that satisfies the conditions specified in \Cref{FewEXP-easy-witness}. Fix such a good $\ell$. We describe how to efficiently obtain one of these ``hard'' witnesses in single-valued nondeterministic time with small advice.
	
	Take the advice string to be $(x_\ell, |W_\ell|)$. We can now perform the following nondeterministic procedure to obtain the lexicographically first witness in $W_\ell$:
	\begin{enumerate}
		\item Nondeterministically guess $|W_\ell|$ distinct strings $y_1,\dots,y_{|W_\ell|} \in \{0,1\}^{2^{\ell^{\alpha}}}$.
		\item Verify for each $y_i$ that $V(x_\ell, y_i) = 1$.
		\item Output the lexicographically first $y_i$ that passes verification.
	\end{enumerate}
	
	This procedure is single-valued because, on any nondeterministic branch, the output is uniquely determined as the lex-first valid witness. The running time is at most $O(|W_\ell| \cdot 2^{\ell^{\alpha}}) = 2^{\mathrm{poly}(\ell)}$.
	
	Let $m = \ell^{\alpha}$. Then:
	\begin{itemize}
		\item The output of the above procedure is a witness in $W_\ell$ and has length $2^m$.
		\item The advice length is $O(\ell + \ell^k)$, which is at most $m^\beta$ for some constant $\beta>0$.
		\item The nondeterministic running time is $2^{\mathrm{poly}(\ell)} = 2^{m^\beta}$ for some $\beta>0$.
	\end{itemize}
	
	Finally, by \Cref{FewEXP-easy-witness}, this lex-first witness is hard against any $\poly(m)$ $O$-oracle circuits.
\end{proof}

Now, we can plug the hard truth tables obtained from \Cref{c:hard-tt} into the hardness-to-randomness framework to obtain a PRG.

Fix constants \(0 < \eta < 1\) and \(b \in \mathbb{N}\). Let \(\varepsilon \coloneqq \eta/2\), and let \(\delta \le \varepsilon\) and constant \(c \in \mathbb{N}\) be as guaranteed by \Cref{t:BFNW93}.

Let \(\{y_m \in \{0,1\}^{2^m}\}_m\) be the sequence from \Cref{c:hard-tt}. For infinitely many \(m\), the function \(y_m\), viewed as a Boolean function on \(m\) variables, does not admit \(O\)-oracle circuits of size \(\poly(m)\), and is computable in single-valued nondeterministic time \(2^{m^\beta}\) with \(m^\beta\) bits of advice, where $\beta>0$ is a constant.

Set
\[
n \coloneqq m^{\beta/\delta}.
\]
Note that we have \(m \le n^\delta\). By padding, we can extend \(y_m\) to a function
\[
z_{n^\delta} : \{0,1\}^{n^\delta} \to \{0,1\}
\]
by ignoring the first \(n^\delta - m\) input bits. Note that \(z_{n^\delta}\) inherits hardness: if it had an \(O\)-oracle circuit of size \(n^{c\delta}\), then by fixing the extra inputs we would obtain an \(O\)-oracle circuit of size \(\poly(m)\) computing \(y_m\), contradicting the hardness of \(y_m\).

Now define the generator
\[
G_n(s) \coloneqq F(z_{n^\delta}, s),
\]
where \(s \in \{0,1\}^{n^\varepsilon}\). By the guarantee of \Cref{t:BFNW93} and the hardness of \(z_{n^\delta}\), it follows that for infinitely many \(n\), the generator \(G_n\) \(n^{-b}\)-fools all \(O\)-oracle circuits of size at most \(n^b\).

It remains to verify efficiency. By \Cref{c:hard-tt}, the truth table \(y_m\) (and hence \(z_{n^\delta}\)) can be computed in single-valued nondeterministic time \(2^{m^\beta} = 2^{O(n^\varepsilon)}\) and with \(m^\beta = O(n^\varepsilon)\) bits of advice. Since \(F\) is computable in polynomial time, it follows that each \(G_n\) can be computed in single-valued nondeterministic time \(2^{O(n^\varepsilon)}\) with \(O(n^\varepsilon)\) bits of advice.
	\end{proof}

	Next, we use our PRG in \Cref{lem:prg-rel} to show that if $\mathsf{FewEXP}\not\subseteq \BPP^{O}$, where $O\in\PSPACE$, then the following meta-complexity problem must be hard. 

\begin{definition}
Let $s_1,s_2\colon\mathbb{N}\to\mathbb{N}$. Define the promise problem $\mathsf{Gap}\text{-}(\nKt \textnormal{ v.s. } \K)[s_1,s_2]$ as the promise problem $(\mathcal{YES},\mathcal{NO})$, where
\begin{align*}
    \mathcal{YES}
    &\eqdef
    \left\{ x \;\middle|\; \nKt(x)\le s_1(|x|) \right\},\\
    \mathcal{NO}
    &\eqdef
    \left\{ x \;\middle|\; \K(x) > s_2(|x|) \right\}.
\end{align*}
\end{definition}

\begin{lemma}\label{l:FewExp-hardness}
Let $O \in \PSPACE$. If $\mathsf{FewEXP} \not\subseteq \mathsf{BPP}^{O}$, then for every $\zeta>0$, $\mathsf{Gap}\text{-}(\nKt \textnormal{ v.s. } \K)[n^{\zeta},n-1]\not\in\BPP^{O}$.
\end{lemma}
\begin{proof}
Fix $O \in \PSPACE$, and assume $\mathsf{FewEXP} \not\subseteq \mathsf{BPP}^{O}$.

        Toward a contradiction, suppose $\mathsf{Gap}\text{-}(\nKt \textnormal{ v.s. } \K)[n^{\zeta},n-1]\in\BPP^{O}$ for some $\zeta>0$.
		Let $D$ be a probabilistic polynomial-time $O$-oracle algorithm such that, except with exponentially small probability over its internal randomness, $D$ rejects every string with $\nKt$-complexity $\leq n^{\zeta}$, and accepts every string with $\nKt$-complexity $>n$. Such a distinguisher can be obtained via standard error reduction followed by a union bound. Equivalently, we may view $D$ as a distribution over $O$-oracle circuits obtained by fixing its internal randomness.
		
Let $G$ be the PRG from \Cref{lem:prg-rel}. Note that since $G$ is computable in single-valued nondeterministic time $2^{O(n^{\eta/2})}$ with $O(n^{\eta/2})$ bits of advice, we have for all sufficiently large $n$ and every $r\in\bool^{n^{\zeta/2}}$,
\[
    \nKt(G_n(r))\leq O(n^{\zeta/2})\leq n^{\zeta}.
\]

By the correctness of the algorithm $D$, except with exponentially small probability over the choice of $D$, it holds that for all but finitely many $n$,
\[
D(G(r)) = 0
\]
for every $r \in \bool^{n^{\eta/2}}$.

On the other hand, with probability at least $1/2$ over $x\sim\bool^n$, we have $\K(x)>n-1$. It follows that
\[
\Pr_{x \sim \bool^n}[D(x)=1] \ge \tfrac12 .
\]
Hence, $D$ distinguishes the output distribution of the PRG from the uniform distribution, contradicting \Cref{lem:prg-rel}.
	\end{proof}
	
We can now complete the proof of \Cref{thm:SoI_uniform_collapse}
\begin{proof}[Proof of \Cref{thm:SoI_uniform_collapse}]
We prove the contrapositive. Suppose that $\mathsf{FewEXP} \not\subseteq \mathsf{BPP}^{\mathsf{NP}}$. Then, by \Cref{l:FewExp-hardness}, for every $\zeta>0$,
\[
\mathsf{Gap}\text{-}(\nKt \textnormal{ v.s. } \K)[n^{\zeta},n-1] \notin \BPP^{\NP}.
\]

Observe that if either $\gapmrnktp[n-O(\log n), n-1]$ or $\gapmnktp[n-O(\log n), n-1]$ were in $\pr\BPP^{||\NP}$, then $\mathsf{Gap}\text{-}(\nKt \textnormal{ v.s. } \K)[n^{\zeta},n-1]$ would also be in $\BPP^{\NP}$ for some $\zeta>0$. Indeed, every YES (resp. NO) instance of the latter problem is also a YES (NO) instance of the former.\footnote{Recall in particular that, for all strings $x$, $\K(x) \leq \rnKt(x) \leq \nKt(x)$.} This contradicts the previous paragraph. Hence, both $\gapmrnktp[n-O(\log n), n-1]$ and $\gapmnktp[n-O(\log n), n-1]$ are not in $\pr\BPP^{||\NP}$. It follows from \Cref{l:SoI-rnKt} and \Cref{l:SoI-nKt} that SoI does not hold for either $\rnKt$ or $\nKt$.
\end{proof}

\subsection{Nonuniform Inclusions}
In this subsection, we prove the second item of \Cref{thm:collapses_from_SoI}. Again, we present the proof for the case where the SoI error is $O(\log n)$, and it can be easily adapted to the case of $\poly(\log n)$ error.

\newcommand{\NoBad}{\mathsf{NoBad}}
\newcommand{\NoHang}{\mathsf{NoHang}}
\newcommand{\DTIME}{\mathsf{DTIME}}

\subsubsection{The Case of \texorpdfstring{$\nKt$}{}}\label{sec:nonuniform-collapse-nkt}
\begin{theorem}[SoI for $\nKt$ $\Rightarrow$ $\EXP^{\parallel\NP}\subseteq \NP/\poly\cap\coNP/\poly$]
	\label{thm:soi-to-expnp-collapse}
	Assume that $\nKt$ satisfies symmetry of information with logarithmic error, i.e., there exists a
	constant $c>0$ such that for all strings $x,y$,
	\begin{equation}\label{eq:soi-nkt-nonunif}
		\nKt(x,y)\ \ge\ \nKt(x)\ +\ \nKt(y\mid x)\ -\ c\log(|x|+|y|).
	\end{equation}
	Then
	\[
	\EXP^{\parallel\NP}\ \subseteq\ \NP/\poly\ \cap\ \coNP/\poly.
	\]
\end{theorem}

\begin{proof}
	Assume~\eqref{eq:soi-nkt-nonunif}. We combine the consequences already established in \Cref{sec:uniform_collapse} with two ingredients from the meta-complexity/pseudorandomness literature. This is implemented in four steps, as explained next.
	
	\medskip\noindent
	\textbf{Step 1: A large-gap $\MnKtP$ promise problem is in $\pr\P^{\parallel\NP}$.}
	Since we assume SoI holds for $\nKt$ with logarithmic error, by \Cref{l:SoI-nKt}, we get that $\gapmnktp[n-O(\log n), n-1]\in \pr\P^{\parallel\NP}$. This implies that for some constant $\delta\in(0,1)$,
	\begin{equation}\label{eq:gapmnktp-in-pnp}
		\gapmnktp[n^\delta,n-1]\ \in\ \pr\P^{\parallel\NP}.
	\end{equation}
	
	\medskip\noindent
	\textbf{Step 2: $\E^{\parallel\NP}$ reduces to a large-gap $\MnKtP$.}
	We now invoke a PRG/reconstruction completeness phenomenon for $\MnKtP$.
	
	\begin{lemma}[Nonuniform truth-table reduction for $\nKt$]\label{lem:enp-tt}
		\label{lem:e-to-gapmnktp-nonunif}
		For any constant $\delta\in(0,1)$ and language
		$L\in \E^{\parallel\NP}$, $L$ has a nonuniform polynomial-time truth-table reduction to
		$\gapmnktp[n^{\delta},n-1]$. Equivalently,
		\[
		\E^{\parallel\NP}\ \subseteq\ \P^{\gapmnktp[n^{\delta},n-1]}/\poly,
		\]
		where all oracle queries to $\MnKtP[n^{\delta},n-1]$ are nonadaptive.
	\end{lemma}

	\begin{proof}[Proof Sketch.]
		Fix a language $L\in\E^{\parallel\NP}$. We sketch a standard Nisan--Wigderson (NW) generator and reconstruction argument showing that,
		for a constant $d\ge 1$ as specified below, the language $L$ has a nonuniform polynomial-time
	    truth-table reduction to the promise problem $\MnKtP[m^{\delta},m-1]$ on inputs of length $m=n^d$.
		
		\medskip\noindent
		\emph{Define an NW/BFNW-style generator from the truth table.}
        Fix a constant $\varepsilon\in(0,1)$ and let $b\in\nat$ be large enough that $b\delta>\varepsilon$, where $\delta$ is as in the statement of this lemma. Let $\delta'<\varepsilon$, $a$, and $F$ be as guaranteed by \Cref{t:BFNW93} applied to $\varepsilon$ and $b$. Set $N \eqdef \lceil n^{1/\delta'}\rceil$ so that $N^{\delta'}\ge n$, and let $f\colon\{0,1\}^{N^{\delta'}}\to\{0,1\}$ be the function $L_n$, padded to ignore all but its last $n$ input bits. Set $m \eqdef N^b$ (so that $m = n^d$ with $d \eqdef b/\delta'$) and define the generator
    	\[
    	G_f(z)\ \eqdef \ F(\tt(f),z)\in\{0,1\}^{m},
    	\qquad z\in\{0,1\}^{N^{\varepsilon}}.
    	\]

		\medskip\noindent
		\emph{Generator outputs are $\nKt$-easy.}
        Fix any seed $z$. We describe a single-valued nondeterministic program that outputs $G_f(z)$.
    	The program hardwires $z$ together with the number $p$ of positive answers among all of the $\NP$
    	queries made by the $\E^{\parallel\NP}$ machine for $L$ across all $2^n$ inputs of length $n$.
    	It first simulates that machine on every input of length $n$, collecting the (non-adaptive) list of queries;
    	this takes time $2^{O(n)}$, and there are at most $2^{O(n)}$ queries. It then nondeterministically guesses
    	witnesses for exactly $p$ of these queries, verifies them, and outputs $\bot$ unless the count matches. Observe that this yields a single-valued nondeterministic construction of $\tt(f)$. The program then computes
	    $F(\tt(f),z)$, taking time $\poly(2^{N^{\delta'}})=\poly(2^{n})$ by \Cref{t:BFNW93}.
    
        The total running time is $t\le 2^{O(n)}$, and the description length is $|z|+|p|+O(1)\le N^{\varepsilon}+O(n)$.
    	Hence,
    	\[
    	\nKt(G_f(z))\ \le\ N^{\varepsilon} + O(n) + O(\log t)\ =\ N^{\varepsilon}+O(N^{\delta'})\ =\ O(N^{\varepsilon}).
    	\]
    	By the choice of $b$, we get $\nKt(G_f(z))\le m^{\delta}$ for all sufficiently
    	large $n$, so $G_f(z)$ is a yes-instance of $\MnKtP[m^{\delta},m-1]$ for every $z$.

		\medskip\noindent
		\emph{Uniform strings are $\nKt$-hard with constant probability.}
    	By a counting argument, a uniformly random $y\sim\mathcal{U}_m$ satisfies $\nKt(y)\ge m-1$ with probability at least $1/2$.
    	Therefore any language $D$ consistent with the promise problem $\MnKtP[m^{\delta},m-1]$ accepts every
    	$G_f(z)$ and rejects all but a $1/2$ fraction of $\mathcal{U}_m$, and in particular
    	\[
    	\left|\Pr_{z\sim\mathcal{U}_{N^{\varepsilon}}}[D(G_f(z))=1]-\Pr_{w\sim\mathcal{U}_{m}}[D(w)=1]\right|\ \ge\ \frac{1}{2} \geq \frac{1}{N^b}.
    	\]
		
		\medskip\noindent
		\emph{Reconstruction yields a nonuniform reduction.}
		By \Cref{t:BFNW93}, there is an oracle circuit $R^{(\text{-})}$ of
		size $N^{a}=\poly(n)$, making nonadaptive queries to $D$, such that $R^{D}(x)=f(x)$ for
		every $x\in\{0,1\}^{N^{\delta'}}$. Hardwiring $R$ as nonuniform advice yields a nonuniform
		polynomial-time truth-table reduction from $L$ to $\MnKtP[m^{\delta},m-1]$, with queries of
		length $m=\poly(n)$ depending only on the input and the advice. Note that this holds for any
		language $D$ consistent with the promise.

	\end{proof}
	
	Combining Lemma~\ref{lem:e-to-gapmnktp-nonunif} with~\eqref{eq:gapmnktp-in-pnp} and substituting the
	$\pr\P^{\parallel\NP}$-procedure for the parallel queries to the oracle yields
	\begin{equation}\label{eq:e-in-pnp-poly}
		\E^{\parallel\NP}\ \subseteq\ \P^{\parallel\NP}/\poly.
	\end{equation}

	\medskip\noindent
	\textbf{Step 3: Downward collapse to $\NP/\poly\cap\coNP/\poly$.}

	Applying Corollary~\ref{cor:su-collapse-E} to~\eqref{eq:e-in-pnp-poly} gives
	\begin{equation}\label{eq:e-in-np-poly}
		\E^{\parallel\NP}\subseteq \NP/\poly\cap\coNP/\poly.
	\end{equation}
	
	\medskip\noindent
	\textbf{Step 4: Standard padding from $\E^{\parallel\NP}$ to $\EXP^{\parallel\NP}$.}
	Let $L\in \EXP^{\parallel\NP}$. Then there exists $k\ge 1$ such that
	$L\in \DTIME(2^{n^k})^{\parallel\NP}$.
	Define the padded language $L'$ by
	\[
	L'(x0^{|x|^k-|x|})\ \eqdef \ L(x).
	\]
	Then $L'\in \E^{\parallel\NP}$, and hence by~\eqref{eq:e-in-np-poly} we have
	$L'\in \NP/\poly\cap\coNP/\poly$.
	Using the advice for length $|x|^k$ and running the corresponding verifier on the padded input
	$y=x0^{|x|^k-|x|}$ yields $\NP/\poly$ and $\coNP/\poly$ verifiers for $L$ itself (with advice size
	still polynomial in $|x|$). Therefore,
	\[
	\EXP^{\parallel\NP}\subseteq \NP/\poly\cap\coNP/\poly,
	\]
	as claimed.
\end{proof}

\subsubsection{The Case of \texorpdfstring{$\rnKt$}{}}\label{sec:nonuniform-collapse-rnkt}
\begin{theorem}[SoI for $\rnKt$ $\Rightarrow$ $\BPEXP^{\parallel\NP}\subseteq \NP/\poly\cap\coNP/\poly$]
	\label{thm:soi-to-bpexpnp-collapse}
	Assume that $\rnKt$ satisfies symmetry of information with logarithmic error, i.e., there exists a
	constant $c>0$ such that for all strings $x,y$,
	\begin{equation}\label{eq:soi-rnkt-nonunif}
		\rnKt(x,y)\ \ge\ \rnKt(x)\ +\ \rnKt(y\mid x)\ -\ c\log(|x|+|y|).
	\end{equation}
	Then
	\[
	\BPEXP^{\parallel\NP}\ \subseteq\ \NP/\poly\ \cap\ \coNP/\poly.
	\]
\end{theorem}

The proof of \Cref{thm:soi-to-bpexpnp-collapse} is similar to that of \Cref{thm:soi-to-expnp-collapse}.

We will make use of the following inclusion, which follows from standard techniques. We include a proof for completeness.

\begin{lemma}\label{lem:bpe-np-ame-n}
    $\BPE^{\parallel \NP} \subseteq \AM\E/O(n)$.
\end{lemma}
\begin{proof}

    Consider $L \in \BPE^{\|\NP}$, witnessed by a randomized machine $M^{\|\SAT}$ running in time
    $t = 2^{O(n)}$ on inputs of length $n$. Without loss of generality, assume $M^{\|\SAT}$ has failure probability $o(1)$. On input $x \in \bool^n$ with randomness $r \sim \mathcal{U}_t$, $M$ generates queries $q_1(x,r),\ldots,q_k(x,r)$ with $k \leq t$,
    receives oracle answers, and accepts or rejects after some $t$-time computation. For each
    $(x, r)$, the correct oracle answers are fixed values $b_1(x,r),\ldots,b_k(x,r) \in \bool$.

    For each $r \in \bool^t$, define
    \[S_{r} := \left\{(x, j) \in \bool^n \times [k] \mid b_j(x,r) = 1\right\}.\] 
    Let $k' \vcentcolon= k\cdot 2^n$ and
    \[
        \mu_n \;\vcentcolon= \; \underset{r \sim \mathcal{U}_t}{\mathbb{E}}\bigl[|S_r|\bigr] \in [0, k'],
    \]
    the expected number of ``yes'' queries on the whole truth-table over $r \in \bool^t$. Setting
    $\ell \vcentcolon= (k')^2 \cdot n$, we take the advice string to be the integer
    \[
        \alpha_n \;\vcentcolon= \; \lfloor \ell \cdot \mu_n \rfloor \ \in\ [0,\ \ell \cdot k'],
    \]
    so that $\ell\mu_n - 1 < \alpha_n \le \ell\mu_n$. Since $\ell \cdot k' \le 2^{O(n)}$, the advice
    $\alpha_n$ can be represented with $O(n)$ bits. 
    
    We now describe an Arthur-Merlin protocol with advice $\alpha_n$ to decide $L$ on inputs of length $n$. 
    \begin{quote}
        On input $x^*$, Arthur samples $\ell = 2^{O(n)}$ independent strings $r_1, \ldots, r_\ell \sim \mathcal{U}_t$ and sends them to Merlin.
        \\
        \\For each $i \in [\ell]$, Merlin sends:
        \begin{itemize}
            \item a set $\hat{S}_{i} \subseteq \bool^n \times [k]$, claimed to be the set of ``yes'' queries of $M^{\parallel \SAT}$ on randomness $r_i$ 
            (i.e., $S_{r_i}$ as defined above); 
            \item for each $(x, j) \in \hat{S}_{i}$, a witness $w_{x,i,j} \in \bool^{\leq t}$ for the $\SAT$ instance $q_j(x, r_i)$.
        \end{itemize}
        Arthur verifies Merlin's message in the following steps:
        \begin{itemize}
            \item for each $i \in [\ell]$, and $(x, j) \in \hat{S}_{i}$, verify that $w_{x, i,j}$ is a valid witness for $q_j(x, r_i)$, and reject if any is not;
            \item compute the total number of ``yes'' queries over all $r_i$,
            \[
                X := \sum_{i=1}^{\ell}|\hat{S}_{i}|,
            \]
            and reject unless
            \[
                \alpha_n - X \leq \ell/6;
            \]
            
            \item accept iff
            \[
                \frac{1}{\ell}\cdot\bigl|\bigl\{i \in [\ell] \;\big|\; M^{\hat{S}_{x^*,i}}(x^*, r_i)
                \text{ accepts } \bigr\}\bigr| \;\geq\; \frac{1}{2},
            \]
            where $\hat{S}_{x^*, i} := \{(x, j) \in \hat{S}_{i} \mid x = x^*\}$ and $M^{\hat{S}_{x^*,i}}$ denotes the computation of $M^{(-)}$ with oracle queries answered according to $\hat{S}_{x^*,i}$.
        \end{itemize}   
    \end{quote}
    Note that the size of Merlin's message is at most $2^{O(n)} \cdot \ell \cdot k \cdot t = 2^{O(n)}$. Also note that Arthur runs in time at most $2^{O(n)}$.

    Now we argue for the correctness of the protocol.

    \paragraph{Completeness.} Suppose $x^* \in L$ and Merlin is honest, i.e., $\hat{S}_i = S_{r_i}$ and all witnesses sent by Merlin are correct. Clearly, the first step of Arthur's verification passes. Further, since $\alpha_n \le \ell\mu_n$, we have $\alpha_n - X \le \ell\mu_n - X$, and hence by
    Hoeffding's Inequality,
    \begin{align*}
        \Pr_{r_1, ..., r_{\ell} \sim \mathcal{U}_t}\left[ \; \alpha_n - X > \ell/6 \; \right]
        &\leq \Pr_{r_1, ..., r_{\ell} \sim \mathcal{U}_t}\left[ \; \ell\mu_n - X > \ell/6 \; \right]
        \\&\leq \exp\left(-\frac{2\cdot(\ell/6)^2}{\ell\cdot(k')^2}\right)
        \\&\leq \exp(-n/18).
    \end{align*}
    So, with probability $1 - o(1)$, the second step of Arthur's verification passes.

    Since $x^* \in L$, the honest computation accepts $x^*$ on all but a $o(1)$-fraction of random $r$. Then, by a Chernoff bound, letting $A_i$ be an indicator random variable for the event that $M^{\parallel \SAT}(x^*,r_i)$ accepts, we have
    \begin{align}\label{eq:honest-accept}
        \Pr\!\left[\frac{1}{\ell}\sum_{i=1}^{\ell}
       A_i \leq \frac{9}{10}\right] \;\leq\; \exp(-\Omega(n)).
    \end{align}
    In the case that $\frac{1}{\ell}\sum_{i=1}^{\ell}
       A_i > \frac{9}{10}$, with all queries answered honestly, the third verification step passes. Overall, the protocol accepts with probability $1 - o(1)$.

    \paragraph{Soundness.}
    Suppose $x^* \notin L$. We show that for any Merlin strategy, Arthur rejects with
    probability $1 - o(1)$.
    
    Since Merlin must provide a valid witness for every $(x,j) \in \hat{S}_i$, the
    first verification step enforces $\hat{S}_i \subseteq S_{r_i}$ for all $i$, so
    Merlin can only answer such that 
    \[X \leq \sum_{i=1}^{\ell} |S_{r_i}| =: Y.\]
    By Hoeffding's inequality, with probability $1-\exp(-\Omega(n))$ we have
    $Y \leq \ell\mu_n + \ell/6 < \alpha_n + 1 + \ell/6$.
    In this event, if Merlin omits more than $\ell/3 + 1$ yes-queries (i.e.\ $Y - X > \ell/3 + 1$), then
    $\alpha_n - X > \ell/6$, and Arthur rejects in the second step.

    Therefore, supposing Merlin passes the second step, he must omit at most $\ell/3 + 1$ total ``yes'' queries, meaning in particular that $\hat{S}_{x^*,i} \neq S_{x^*,r_i}$ for at most $\ell/3 + 1$ indices $i$. On the remaining indices, $M^{\hat{S}_{x^*,i}}(x^*,r_i)$ agrees with the honest computation of $M^{\parallel \SAT}(x^*,r_i)$. With $x^*\notin L$, arguing as in \Cref{eq:honest-accept}, we have
    \[
        \Pr\!\left[\frac{1}{\ell}\sum_{i=1}^{\ell}
       A_i
        \;\geq\; \frac{1}{10}\right] \;\leq\; \exp(-\Omega(n)).
    \]
    Hence, with probability $1-\exp(-\Omega(n))$, the number of accepting indices in the third
    step is at most $\ell/3 + 1 + \ell/10 < \ell/2$, in which case Arthur rejects. 

    This concludes the proof.
        
\end{proof}

    We will need the following randomized analogue of \Cref{lem:e-to-gapmnktp-nonunif}.

\begin{lemma}[Nonuniform truth-table reduction for $\rnKt$]
    \label{lem:bpe-to-gapmnrktp-nonunif}
    For every constant $\delta\in(0,1)$ and language
    $L\in \BPE^{\parallel\NP}$, $L$ has a nonuniform polynomial-time truth-table reduction to
    $\gapmrnktp[n^\delta,n-1]$. Equivalently,
    \[
    \BPE^{\parallel\NP} \subseteq \P^{\gapmrnktp[n^\delta,n-1]}/\poly,
    \]
    where all oracle queries to $\gapmrnktp[n^\delta,n-1]$ are nonadaptive. If in addition we have
    $$
    \gapmrnktp[n^\delta,n-1] \in \pr \BPTIME[T]^{\parallel \NP},
    $$
    then 
    $$
    \BPE^{\parallel\NP} \subseteq \DTIME[\poly(T(\poly))]^{\parallel \NP}/\poly(T(\poly)).
    $$
\end{lemma}
\begin{proof}[Proof of \Cref{lem:bpe-to-gapmnrktp-nonunif}]

    We first argue that $\AM\E/O(n) \subseteq \P^{\gapmrnktp[n^{\delta},n-1]}/\poly$, following along the same lines as the proof of \Cref{lem:enp-tt}. Fix a language $L \in \AM\E/O(n)$.
    
    \medskip\noindent
    \emph{Define a generator from the truth table.}
    Fix a constant $\varepsilon\in(0,1)$ and let $b\in\nat$ be large enough that $b\delta>\varepsilon$, where $\delta$ is the constant of the statement. Let $\delta'<\varepsilon$, $a$, and $F$ be as guaranteed by \Cref{t:BFNW93} applied to $\varepsilon$ and $b$. Set $N \eqdef \lceil n^{1/\delta'}\rceil$, so that $N^{\delta'}\ge n$, and let $f\colon\{0,1\}^{N^{\delta'}}\to\{0,1\}$ be the function $L_n$, padded to ignore all but its last $n$ input bits. Set $m \eqdef N^b$, so that $m = n^{d}$ with $d \eqdef b/\delta'$, and define the generator
    \[
    G_f(z)\ \eqdef\ F(\tt(f),z)\in\{0,1\}^{m},
    \qquad z\in\{0,1\}^{N^{\varepsilon}}.
    \]
    
    \medskip\noindent
	\emph{Generator outputs are $\rnKt$-easy.}
    Fix any seed $z$. Without loss of generality, we may assume that there is a non-uniform verifier $V_L$ for $L$ with failure probability at most $2^{-2n}$. By a union bound, this implies that with probability greater than $2/3$ over $r \sim \mathcal{U}_{2^{O(n)}}$, for all $x \in \bool^n$,
    \[x \in L \iff \exists w \in \bool^{2^{O(n)}},\ V_L(x, r, w) = 1.\]
    We describe a randomized nondeterministic program that outputs $G_f(z)$. The program hardwires $z$ (of length $N^{\varepsilon}$), the non-uniformity of $V_L$, and the number $k \in [2^n]$ of inputs $x \in \bool^n$ with $L(x)=1$ (both of length $O(n)$). Given randomness $r$ as above, it nondeterministically guesses a set $Y \subseteq \bool^n$ of size $k$ together with witnesses $w_x \in \bool^{2^{O(n)}}$ for each $x \in Y$, and verifies that $V_L(x,r,w_x)=1$ for every $x \in Y$; if any verification fails, it outputs $\bot$. Observe that, since $|Y| = k = |L_n|$, every nondeterministic path not resulting in $\bot$ correctly classifies all yes- and no-instances of $L_n$; this yields a single-valued nondeterministic construction of $\tt(f)$. The program then computes $F(\tt(f),z)$, taking time $\poly(2^{N^{\delta'}})=\poly(2^{n})$ by \Cref{t:BFNW93}.

    The total running time is $t \le 2^{O(n)}$, and the description length is $|z| + O(n) + O(1) \le N^{\varepsilon} + O(n)$. Hence,
    \[
    \rnKt(G_f(z))\ \le\ N^{\varepsilon} + O(n) + O(\log t)\ =\ N^{\varepsilon} + O(N^{\delta'})\ =\ O(N^{\varepsilon}).
    \]
    By the choice of $b$, we get $\rnKt(G_f(z)) \le m^{\delta}$ for all sufficiently large $n$, so $G_f(z)$ is a yes-instance of $\gapmrnktp[m^{\delta},m-1]$ for every $z$.

    \medskip\noindent
    \emph{Uniform strings are $\rnKt$-hard with constant probability.}
    By a counting argument as in \Cref{l:incomp-strings}, a uniformly random $y \sim \mathcal{U}_m$ satisfies $\rnKt(y) \ge m-1$ with probability at least $1/2$. Therefore, any language $D$ consistent with the promise problem $\gapmrnktp[m^{\delta},m-1]$ breaks $G_f$ with advantage at least $1/2 \geq N^{-b}$, meeting the hypothesis of \Cref{t:BFNW93}.

    \medskip\noindent
    \emph{Reconstruction yields a non-uniform reduction.}
    By \Cref{t:BFNW93}, there is an oracle circuit $R^{(\text{-})}$ of size $N^{a}=\poly(n)$, making nonadaptive queries to $D$, such that $R^{D}(x)=f(x)$ for every $x \in \{0,1\}^{N^{\delta'}}$. Hardwiring $R$ as nonuniform advice and restricting $f$ to $L_n$ yields a nonuniform polynomial-time truth-table reduction from $L$ to $\gapmrnktp[m^{\delta},m-1]$, with queries of length $m=\poly(n)$ depending only on the input and the advice. In particular, this holds for any language $D$ consistent with $\gapmrnktp[m^{\delta},m-1]$; applying \Cref{lem:bpe-np-ame-n}, we obtain
    \[\BPE^{\parallel \NP} \subseteq \AM\E/O(n) \subseteq \P^{\parallel D}/\poly\]
    for any such $D$.

    For the ``in addition'' part of the lemma, suppose $\gapmrnktp[m^\delta,m-1] \in \pr \BPTIME[T]^{\parallel \NP}$, witnessed by a randomised oracle algorithm $A^{\parallel \NP}$ running in time $T$. Without loss of generality, we may assume that $A^{\parallel \NP}$ has failure probability at most $2^{-2m}$ on inputs of length $m$. This implies that with probability $1 - 2^{-\Omega(n)}$ over the internal randomness $r$ of $A^{\parallel \NP}$, $A^{\parallel \NP}(-; r)$ is consistent with $\gapmrnktp[m^{\delta}, m-1]$ on every input of length $m = \poly(n)$. Fix some good choice of randomness, $r^* \in \bool^{\poly(T(m))}$, and apply the reduction above to $D(-) \eqdef A^{\parallel \NP}(-; r^*)$. 
    
    We may then simulate any ${\sf TIME}[\poly(n)]^{\parallel \gapmrnktp[m^{\delta}, m-1]}/\poly(n)$ procedure by answering all queries via $A^{\parallel \NP}(-; r^*)$. Since all oracle queries are non-adaptive, we obtain
    $$
    \BPE^{\parallel\NP} \subseteq \DTIME[\poly(T(\poly))]^{|| \NP}/\poly(T(\poly))
    $$
    as desired.

\end{proof}

 We are now ready to describe the proof of \Cref{thm:soi-to-bpexpnp-collapse}.

\begin{proof}[Proof Sketch of \Cref{thm:soi-to-bpexpnp-collapse}]
The proof closely follows that of \Cref{thm:soi-to-expnp-collapse}.

Assuming SoI holds for $\rnKt$ with logarithmic error, by \Cref{l:SoI-rnKt} we obtain that $\gapmrnktp[n-O(\log n), n-1]\in \pr\BPP^{\parallel\NP}$, which further implies that for some constant $\delta>0$,
\[
\gapmrnktp[n^\delta,n-1]\ \in\ \pr\BPP^{\parallel\NP}.
\]
Next, by \Cref{lem:bpe-to-gapmnrktp-nonunif}, every language in $\BPE^{\parallel\NP}$ admits a nonuniform polynomial-time truth-table reduction to $\gapmrnktp[n^\delta,n-1]$. Therefore, we obtain
\[
\BPE^{\parallel\NP}\subseteq\P^{\parallel\NP}/\poly .
\]
Applying \Cref{cor:su-collapse-BPE}, we conclude that
\[
\BPE^{\parallel\NP}\ \subseteq\ \NP/\poly \cap \coNP/\poly.
\]
Finally, by a standard padding argument, the above yields
\[
\BPEXP^{\parallel\NP}\ \subseteq\ \NP/\poly \cap \coNP/\poly,
\]
as desired.
\end{proof}

\section{Refined Meta-Complexity Landscape Under  SoI: Proof of \texorpdfstring{\Cref{thm:consequences_SoI}}{}}\label{sec:meta_SoI}

\subsection{Nonuniform Upper Bound}

The proof of \Cref{thm:consequences_SoI} Item (\emph{i}) follows from a combination of results, as explained next.

\begin{proof}[Proof of \Cref{thm:consequences_SoI} Item \emph{(}i\emph{)}]
    Suppose that SoI with error $e(n) = O(\log n)$ holds for $\rnKt$. Then, by \Cref{thm:equivalence_SoI}, it follows that
    $$
\gapmrnktp[n- O(\log n), n-1] \in \pr \BPP^{||\NP}.
    $$
    Let $M(x)$ be a $\pr \BPP^{||\NP}$-machine that agrees with $\gapmrnktp[n- O(\log n), n-1]$. Since $M$ runs in polynomial time, every oracle query made by $M$ is of length $\poly(n)$, and the random string employed by $M$ is also of length $\poly(n)$. Consequently, we can deterministically compute the acceptance probability of $M$ on any given input in exponential time. In particular, this implies the existence of a language $L \in \EXP$ that agrees with the promise problem $\gapmrnktp[n- O(\log n), n-1]$.

    Under the initial SoI assumption, we can apply \Cref{thm:soi-to-bpexpnp-collapse} to $L$, which gives the upper bound $L \in \NP/\poly \cap \coNP/\poly$. In other words,
    $$
\gapmrnktp[n- O(\log n), n-1] \in \NP/\poly \cap \coNP/\poly.
    $$
    This completes the proof.
\end{proof}

\subsection{Uniform Lower Bound}

We will use the following result about $\MKtP$ by \citep{DBLP:journals/siamcomp/AllenderBKMR06}:
\begin{lemma}[\citep{DBLP:journals/siamcomp/AllenderBKMR06}]\label{lem:EXP-in-NP-MKtP}
    $\EXP\subseteq\NP^{\MKtP}$.
\end{lemma}

Additionally, we will use the following simple separation.

\begin{lemma}\label{lem:diag_ee}
    $\EE \nsubseteq \DTIME[n^{\poly(\log n)}]^{\parallel\NP}/n^{\poly(\log n)}$.
\end{lemma}

\begin{proof}
    For $n\in\nat$, let $B\vcentcolon= 2^{\lceil n^{2/3}\rceil}$. For every machine $M$ of description
    length at most $\log n$ and every advice string $\alpha\in\bool^{\le B}$, let
    $\tau_{M,\alpha}\in\bool^{2^n}$ be the table whose entry at $x\in\bool^n$ is the output of
    $M^{\SAT}(x,\alpha)$ after $B$ steps (taken to be $0$ if $M$ has not halted with output $1$).
    Let $\Lambda_n$ be the list of all such tables. Note that 
    \[
        |\Lambda_n|\ \le\ 2n\cdot 2^{B+1}\ <\ 2^{2^n}.
    \]
    Let $\tau_n$ be the lexicographically first string of $\bool^{2^n}$ not in $\Lambda_n$, and set
    $L_{\mathrm{diag}}(x)\vcentcolon=\tau_{|x|}(x)$.

    To compute $\tau_n$, first tabulate the answers to all $\SAT$ queries of length at most $B$ by
    exhaustive search, in time $2^{O(B)}$. Using this table, simulate each of the at most
    $2n\cdot2^{B+1}$ pairs $(M,\alpha)$ on all $2^n$ inputs for $B$ steps, and then check the first
    $|\Lambda_n|+1$ strings of $\bool^{2^n}$ against $\Lambda_n$. The total time is $2^{O(B)}$, so
    $L_{\mathrm{diag}}\in\DTIME[2^{O(2^n)}]\subseteq\EE$.

    Now let $L\in\DTIME[n^{(\log n)^c}]^{\parallel O}/n^{(\log n)^c}$ for some constant $c$ and some
    $O\in\NP$, witnessed by a machine $M$ and advice $(\alpha_n)_{n \in \nat}$. Let $M'$ simulate $M$, replacing each
    query to $O$ by its image under a fixed polynomial-time reduction to $\SAT$. On inputs of length
    $n$, $M'$ runs in time $\poly(n^{(\log n)^c}) \leq B$. Moreover, $|M'|\le\log n$ and $|\alpha_n|\le B$. For such $n$, $\tau_{M',\alpha_n}$ is the truth
    table of $L$ at length $n$, so it lies in $\Lambda_n$ and differs from $\tau_n$. Hence $L$ and
    $L_{\mathrm{diag}}$ differ at every sufficiently large input length.
\end{proof}

We are now ready to prove \Cref{thm:consequences_SoI} Item (\emph{ii}).

\begin{proof}[Proof of \Cref{thm:consequences_SoI} Item \emph{(}ii\emph{)}]
   We consider first the case of $\rnKt$. Assume that SoI holds for $\rnKt$ with error $e(n)=\poly(\log n)$. Then by \Cref{thm:equivalence_SoI}, we have
    \[
        \gapmrnktp[n-\poly(\log n),n-1]\in\pr\mathsf{BPTIME}[n^{\poly(\log n)}]^{||\NP}.
    \]
    Then by \Cref{lem:bpe-to-gapmnrktp-nonunif} and a standard padding argument,
    \begin{equation}
        \BPEXP^{||\NP}\subseteq\mathsf{DTIME}[n^{\poly(\log n)}]^{||\NP}/n^{\poly(\log n)}.\label{eq:nonuniform-BPEXP-collapse}
    \end{equation}
    Now assume that $\MKtP\in\NTIME[n^{\poly(\log n)}]\cap\coNTIME[n^{\poly(\log n)}]$. Then by \Cref{lem:EXP-in-NP-MKtP}, we get
    \begin{align*}
        \EXP\subseteq\NP^{\MKtP}\subseteq\NP^{\NTIME[n^{\poly(\log n)}]\cap\coNTIME[n^{\poly(\log n)}]}\subseteq\NTIME[n^{\poly(\log n)}],
    \end{align*}
    where the final inclusion follows from guessing the witnesses for both the $\NP$ and $\coNP$ machine for the oracle queries. Now by a padding argument, we get
    \[
        \EE=\DTIME[2^{2^{O(n)}}]\subseteq\NTIME[2^{\poly(n)}]=\NEXP.
    \]
    By \Cref{eq:nonuniform-BPEXP-collapse}, this gives
    \[
        \EE\subseteq\NEXP\subseteq\BPEXP^{||\NP}\subseteq\mathsf{DTIME}[n^{\poly(\log n)}]^{||\NP}/n^{\poly(\log n)}.
    \]
    However, this is in contradiction to \Cref{lem:diag_ee}.

    Similarly, assume that SoI holds for $\nKt$ with error $e(n)=\poly(\log n)$. Then analogously to the previous proof, but applying \Cref{lem:e-to-gapmnktp-nonunif} instead of \Cref{lem:bpe-to-gapmnrktp-nonunif}, it follows  that 
    $$
\EXP^{|| \NP} \subseteq \DTIME[n^{\poly(\log n)}]^{|| \NP}/\poly(n).
    $$
    Since $\NEXP \subseteq \EXP^{|| \NP}$, it is enough to note that this is sufficient to obtain a contradiction, by continuing as in the proof for $\rnKt$ described above.
\end{proof}

\section{Asymmetry, Hardness, and Explicit Constructions: Proof of  \texorpdfstring{\Cref{thm:equivalence_SoI}}{}}

\subsection{SoI to Meta-Complexity}
In this subsection, we show that if SoI holds for $\rnKt$ (resp. $\nKt$), then approximating $\rnKt$-complexity (resp. $\nKt$-complexity) becomes easy with parallel access to an $\NP$ oracle.

\subsubsection{The Case of \texorpdfstring{$\rnKt$}{}}\label{sec:soi-hardness-rnkt}
We first establish the following using ideas from \citep{HMO26}.

\begin{lemma}[Efficient Simulation Under SoI]\label{lem:rnk-t-rnkt}

    Assume SoI holds for $\rnKt$ with error $e(n) \geq \log n$. Then, there exists a time bound $t \leq 2^{O(e(n))}$ such that for every string $x \in \bool^n$, it holds that $\rnK^t(x) \leq \rnKt(x) + O(e(n))$.
    
\end{lemma}
\begin{proof}We adapt the proof of a similar result from \citep{HMO26} stated in the context of $\rKt$ complexity.\footnote{The proof given here admits slightly better parameters, and relies on an alternative argument suggested to us by Hanlin Ren.}

Let $n=|x|$. Let the $n$ bits of $x$ be $x_1x_2\cdots x_n$. For any integers $1\le i\le j\le n$, we define $x_{[i,j]}=x_i\cdots x_j$. For each $i\le j$, we define $s_{[i,j]}$ and $t_{[i,j]}$ to be the program size and running time of an optimal $\rnKt$ program for $x_{[i,j]}$ conditioning on $x_{[1,i-1]}$. If there are multiple programs achieving the optimal bound $\rnKt(x_{[i,j]})$, we take one with the smallest running time.
    
    We let $1 \leq i \leq n$ be the smallest integer such that $t_{[1,i]}\ge2^{C\cdot \log n + 1}\cdot 2^{e(n)}$, where $C \geq 1$ is a suitable universal constant defined below. If no such $i$ exists, then the optimal $\rnKt$ program for $x$ already runs in time $2^{O(e(n))}$, and the proof is finished. 
    
    Note that we can assume without loss of generality that $1 < i < n$. Indeed, if $C$ is large enough, using that $e(n) \geq \log n$ it is easy to see that $i > 1$. The case $i = n$ can be handled by an easy modification of the argument above. 
    
    Now, by symmetry of information, we have
    \begin{align}
        \rnKt(x)&\ge\rnKt(x_{[1,i]})+\rnKt(x_{[i+1,n]}\mid x_{[1,i]})-e(n)\notag\\
        &=s_{[1,i]}+\log(t_{[1,i]})+s_{[i+1,n]}+\log(t_{i+1,n})-e(n).\label{eq:SoI-1}
    \end{align}
    But on the other hand, since we can print $x$ by simply running the programs in sequence, we get\footnote{Here we use the aforementioned properties of the encoding of programs in the definition of Kolmogorov complexity.}
    \begin{align}
        \rnKt_{(2/3)^2}(x) & \le  s_{[1,i]}+s_{[i+1,n]}+O(1) + \log(O(t_{[1,i]}+t_{[i+1,n]}+n)) \notag  \\
        & \leq  s_{[1,i]}+s_{[i+1,n]}+\log(t_{[1,i]}+t_{[i+1,n]}+n) + C,  \label{eq:SoI-2}
    \end{align}
    for some universal constant $C \geq 1$ that can depend on the computational model. Furthermore, by applying success amplification (\Cref{l:success-amp-rnK}) to \Cref{eq:SoI-2}, we get
    \begin{equation}\label{eq:SoI-3}
        \rnKt(x) \leq s_{[1,i]}+s_{[i+1,n]}+\log(t_{[1,i]}+t_{[i+1,n]}+n) + C\cdot \log n.
    \end{equation}
    
    Combining \Cref{eq:SoI-1,eq:SoI-3}, we get
    \[
        \log(t_{[1,i]}+t_{[i+1,n]}+n)\ge\log(t_{[1,i]})+\log(t_{[i+1,n]})-e(n) - C\cdot \log n.
    \]
    By taking the exponential of both sides, and using that $e(n) \geq \log n$ and $t_{[1,i]} \geq 2^{C\log n+1} \cdot 2^{e(n)}$,  we get
    \[
        t_{[i+1,n]}\le\frac{1+\frac{n}{t_{[1,i]}}}{\frac{1}{2^{e(n) + C\log n}}-\frac{1}{t_{[1,i]}}}\le\frac{2}{\frac{1}{2^{e(n) + C\log n}}-\frac{1}{t_{[1,i]}}}\le 2^{C\log n + 2}\cdot 2^{e(n)}.
    \]
    Hence by combining the programs for $x_{[1,i-1]}$ and $x_{[i+1,n]}$, and by storing the bit $x_i$ and the position $i$, we can print $x$ in (nondeterministic) time $2^{O(e(n))}$. Therefore, it follows that
    \begin{align*}
        \rnK^{2^{O(e(n))}}(x)&\le s_{[1,i-1]}+s_{[i+1,n]}+ O(\log n)\\
        &\le\rnKt(x_{[1,i-1]})+\rnKt(x_{[i+1,n]}\mid x_{[1,i]})+ O(\log n)\\
        &\le\rnKt(x_{[1,i]})+\rnKt(x_{[i+1,n]}\mid x_{[1,i]})+ O(\log n)\\
        &\le\rnKt(x)+O(e(n)),\tag{By SoI for $\rnKt$}
    \end{align*}
    as desired. 
\end{proof}

As a consequence of the above, we obtain the following lemma, which completes the proof of Item (i) $\implies$ Item (ii) in \Cref{thm:equivalence_SoI}.

\begin{lemma}\label{l:SoI-rnKt}

    Assume SoI holds for $\rnKt$ with error $e(n) \geq \log n$. Then, $\gapmrnktp[s_1, s_2]$ is contained in $\pr \mathsf{BPTIME}[2^{O(e(n))}]^{||\NP}$, where $s_1(n) = n - O(e(n))$ and $s_2(n) = n-1$.
    
\end{lemma}
\begin{proof}

    For $t \in \nat$, we say that a string $y$ is an \emph{$\rnK^t$-description} of a string $x$ if 
    \[\Prob_{r\sim\bool^t}\mleft[
	\begin{aligned}
		&\forall w\in\bool^t, \Pi(y;w;r) \textrm{ outputs $x$ or $\bot$ in $t$ steps}; \\
        &\exists w \in \{0, 1\}^t, \Pi(y;w;r) \textrm{ outputs $x$ in $t$ steps}
	\end{aligned} \mright] \geq 2/3.\]
    (Note that we do not require $y$ to be of minimum length among $\rnK^t$-descriptions here.)

    By \Cref{lem:rnk-t-rnkt}, for some constant $b \in \nat$ and $t' = 2^{b\cdot e(n)}$, for every string $x \in \bool^n$, it holds that $\rnK^{t'}(x) \leq \rnKt(x) + b\cdot e(n)$. In particular, for every $x \in \bool^n$, there exists a string $m$ of length at most $\rnKt(x) + b \cdot e(n)$ such that $m$ is an $\rnK^{t'}$-description for $x$ and $\rnK^{t'}(x) = |m|$. 
    
    By our assumption of SoI for $\rnKt$,
    \begin{align}\label{eq:compress-m}
        \rnKt(m \mid x) &\leq \rnKt(m,x) - \rnKt(x) + e(n) \notag
        \\&\leq |m| + \log t' + O(\log n) - \rnKt(x) + e(n) \notag
        \\&\leq \rnK^{t'}(x) - \rnKt(x) + O(e(n)) \notag
        \\&\leq c\cdot e(n),
    \end{align}
    for some constant $c \in \nat$. This implies that, for $t = 2^{c \cdot e(n)}$,
    \[\rnK^{t}(m \mid x) \leq c \cdot e(n).\]

    As a first step, we will define a randomised algorithm $B$ equipped with non-adaptive access to a $\SAT$ oracle. Ultimately, this algorithm will be used as a sub-procedure in our algorithm to estimate $\rnKt$. Roughly, $B^{\parallel \SAT}$ will decide, for a given pair $(v,x)$, whether $v$ is an $\rnK^t$ description of a string $m \in \bool^s$ that is an $\rnK^{t'}$ description of $x$; if so, it will output $s$.
    
    For $n, t, t', s, i \in \nat$ and strings $y, r, r'$, define the following $\SAT$ queries.
    \begin{itemize}
        \item $Q_1(y, r, t, s) \in \SAT \iff \exists w \in \bool^{\leq t} \text{ s.t. } U(y; w; r)$ outputs some string $z \in \bool^s\backslash\{\bot\}$ in time at most $t$;
        \item $Q_2(y, r, t) \in \SAT \iff \exists w, w' \in \bool^{\leq t}$ such that either $U(y; w; r)$ fails to halt within $t$ steps, or $U(y; w; r) = z$ and $U(y; w'; r) = z'$ in time at most $t$, where $z \neq z'$ and $z, z' \neq \bot$.
        \item $Q_3(y, r, t, i) \in \SAT \iff \exists w \in \bool^{\leq t} \text{ s.t. } U(y; w; r) \text{ outputs some string } z \neq \bot$ in time at most $t$, and $z_i = 1$.
        \item $Q_4(y, r, r', t, t', n) \in \SAT \iff \exists w \in \bool^{\leq t} \text{ s.t. } U(y; w; r)$ outputs some string $z \neq \bot$ in time at most $t$, and $Q_1(z, r', t', n) \in \SAT$.
        \item $Q_5(y, r, r', t, t') \in \SAT \iff \exists w \in \bool^{\leq t} \text{ s.t. } U(y; w; r)$ outputs some string $z \neq \bot$ in time at most $t$, and $Q_2(z, r', t') \in \SAT$.
        \item $Q_6(y, r, r', t, t', i) \in \SAT \iff \exists w \in \bool^{\leq t} \text{ s.t. } U(y; w; r) \text{ outputs some string } z \neq \bot$ in time at most $t$, and $Q_3(z, r', t', i) \in \SAT$.
    \end{itemize}

    The algorithm $B^{\parallel \SAT}$ is as follows.
    \begin{enumerate}
        \item On input $(v, x)$, sample randomness $r_1, ..., r_{n^2} \sim \mathcal{U}_{t}$ and $r'_1, ..., r'_{n^2} \sim \mathcal{U}_{t'}$, where $n = |x|$, and $t$ and $t'$ are as defined above. 
        
        \item For each $j \in [n^2]$, $s \in [2n]$, and $i \in [s]$, query $Q_1((v,x), r_j, t, s)$, $Q_2((v,x), r_j, t)$, and $Q_3((v,x), r_j, t, i)$ to the $\SAT$ oracle, and let $z^{(j)}$ be the concatenation of the values of $\SAT(Q_3((v,x), r_j, t, i))$ for $i = 1, ..., s$. 
        
        \item For each $i, j, k \in [n^2]$, query $Q_4((v, x), r_j, r'_k, t, t', n)$, $Q_5((v, x), r_j, r'_k, t, t')$, and $Q_6((v, x), r_j, r'_k, t, t', i)$ to the $\SAT$ oracle, and let $z^{(j,k)}$ be the concatenation of the values of $\SAT(Q_6((v, x), r_j, r'_k, t, t', i))$ for $i = 1, ..., n$. 
        
        \item Let $S \subseteq [n^2]$ be the set of rounds $j$ in which $Q_1((v, x), r_j, t, s)$ was accepted for some choice of $s$ and $Q_2((v, x), r_j, t)$ was rejected. Let $m'$ be the majority value of $(z^{(j)})_{j \in S}$, and let $s' = |m'|$. Let $S' = \{j \in S \mid z^{(j)} = m'\}$. If $|S'|/n^2 \leq (1/2) + (1/10)$, halt and output $\bot$.
        
        \item Fix the lexicographically first $j \in S'$. Let $T \subseteq [n^2]$ be the set of rounds $k$ in which $Q_4((v, x), r_j, r'_k, t, t', n)$ was accepted and  $Q_5((v, x), r_j, r'_k, t, t')$ was rejected. Let $x'$ be the majority value of $z^{(j,k)}$ over all $k \in T$. Let $T' = \{k \in T \mid z^{(j,k)} = x'\}$. If $|T'|/n^2 > (1/2) + (1/10)$ and $x' = x$, halt and output $s'$. Otherwise, output $\bot$.
    \end{enumerate}
    Observe that $B^{\parallel \SAT}$ queries its oracle non-adaptively and that it runs in time $2^{O(e(n))}$.
    
    \begin{claim}\label{cl:B-SAT-1}
        If $(v,x)$ is an $\rnK^t$-description for some string $m \in \bool^s$, where $m$ is an $\rnK^{t'}$-description for $x$, then $B^{\parallel \SAT}(v, x)$ outputs $s$ with probability $1 - 2^{-\Omega(n^2)}$.
    \end{claim}
    \begin{proof}
        Since $(v,x)$ is an $\rnK^t$-description for $m$, with probability at least $2/3$ over a random $r$, it holds that $Q_1((v,x), r, t, s) \in \SAT$, $Q_2((v,x), r, t) \notin \SAT$, and the concatenation of $\SAT(Q_3((v,x), r_j, t, i))$ for $i = 1, ..., s$ equals $m$. By a concentration bound, with probability $1 - 2^{-\Omega(n^2)}$, the fraction of $r_j$ on which the above items hold is greater than $(1/2) + (1/10)$, in which case $m$ is the majority value of $(z^{(j)})_{j \in S}$, and $|S'|/n^2 \geq (1/2) + (1/10)$.

        When the above holds, for any $j \in S'$ (including the lexicographically first), $m$ is the unique string not equal to $\bot$ such that there exists $w \in \bool^{\leq t}$ with $U((x,v), w, r_j) = m$ in time at most $t$. Thus, the first condition of Queries 4-6 with inputs $(v,x)$, $r_j$, and $t$ is always true, with $z = m$.
        
        Since $m$ is an $\rnK^{t'}$-description for $x$, with probability at least $2/3$ over a random choice of $r'$, it holds that $Q_4((v, x), r_j, r', t, t', n) \in \SAT$, $Q_5((v, x), r_j, r'_k, t, t') \notin \SAT$, and moreover the concatenation of $\SAT(Q_6((v, x), r_j, r'_k, t, t', i))$ for $i = 1, ..., n$ equals $x$. By a concentration bound, with probability $1 - 2^{-\Omega(n^2)}$, the fraction of $r'_k$ on which the above items hold is greater than $(1/2) + (1/10)$, in which case $x$ is the majority value of $(z^{(j, k)})_{k \in T}$, and $|T'|/n^2 \geq (1/2) + (1/10)$. In this case, $B^{\parallel \SAT}$ halts and outputs $s$.
    \end{proof}

    \begin{claim}\label{cl:B-SAT-2}
        If $\rnK^{t'}_{1/2}(x) > s$, then for any $v \in \bool^*$, with probability at least $1 - 2^{-\Omega(n^2)}$, $B^{\parallel \SAT}(v,x)$ does not output $s$.
    \end{claim}
    \begin{proof}
        Suppose $\rnK^{t'}_{1/2}(x) > s$. Suppose that $B^{\parallel \SAT}(v,x)$ does not output $\bot$ in Step (4); otherwise, we are done. Let $m' \in \bool^{s}$ be the majority value of $(z^{(j)})_{j \in S}$, and fix $j$ as in Step (5). In particular, $m'$ is the unique string not equal to $\bot$ such that there exists $w \in \bool^{\leq t}$ with $U((x, v), w, r_j) = m'$ in time at most $t$. Note that $|m'|$ must be equal to $s$; otherwise, $B^{\parallel \SAT}$ would not output $s$, and we are done. Further, note that the first condition of Queries 4-6 with inputs $(v,x)$, $r_j$, and $t$ is always true, with $z = m'$.

        By definition of $\rnK^{t'}_{1/2}$, with probability at least $1/2$ over a random choice of $r'$, either $Q_4((v,x),r_j,r'_k,t,t',n)\notin\SAT$, $Q_5((v, x), r_j, r'_k, t, t') \in \SAT$, or the concatenation of $\SAT(Q_6((v, x), r_j, r'_k, t, t', i))$ for $i = 1, ..., n$ is not equal to $x$. By a concentration bound, with probability at least $1 - 2^{-\Omega(n^2)}$ over the choice of $r'_1, ..., r'_{n^2}$, the fraction of $r'_k$ with this property is at least $(1/2) - (1/100)$. 

        Recall that, if $x$ is not the majority value $x'$ of $(z^{j, k})_{k \in T}$, then $B^{\parallel \SAT}$ outputs $\bot$. On the other hand, if $x$ is $x'$, then by the above, the fraction of $k \in [n^2]$ such that that $Q_5((v, x), r_j, r'_k, t, t') \notin \SAT$ \emph{and} the concatenation of $\SAT(Q_6((v, x), r_j, r'_k, t, t', i))$ equals $x'$ is at most $(1/2) + (1/100)$. It follows that $|T'|/n^2 < (1/2) + (1/10)$, and $B^{\parallel \SAT}$ outputs $\bot$. 
    \end{proof}

    Define an algorithm $A^{\parallel \SAT}$ as follows.
    \begin{quote}

        On input $x \in \bool^n$, enumerate all $v \in \bool^{c \cdot e(n)}$, and let $s_v = B^{\parallel \SAT}(v, x)$. Return the minimum $s_v$. If $s_v = \bot$ for all $v$, return $2n$.
        
    \end{quote}
    Observe that, by definition of $B^{\parallel \SAT}$, $A^{\parallel \SAT}$ queries its oracle non-adaptively and runs in time $2^{O(e(n))}$. 

    \begin{claim}\label{cl:A-SAT-1}

        If $\rnK^{t'}(x) \leq s$, then $A^{\parallel \SAT}(x) \leq s$ with probability $1 - 2^{-\Omega(n)}$.
        
    \end{claim}
    \begin{proof}

        If $\rnK^{t'}(x) \leq s$, then by \Cref{eq:compress-m}, $\rnK^{t}(m \mid x) \leq c\cdot e(n)$, where $m \in \bool^{\leq s}$ is an $\rnK^{t'}$-description for $x$ such that $\rnK^{t'}(x) = |m|$. In particular, for some $v \in \bool^{c \cdot e(n)}$, $(v, x)$ is an $\rnK^t$-description for $m$. By \Cref{cl:B-SAT-1}, with probability at least $1 - 2^{-\Omega(n^2)}\ge1-2^{\Omega(n)}$, $B^{\parallel \SAT}(v, x)$ outputs $s_v = |m|$. In this case, $A^{\parallel \SAT}(x) = \min_v\{s_v\} \leq s$. 
    \end{proof}

    \begin{claim}\label{cl:A-SAT-2}

        If $\rnK^{t'}_{1/2}(x) > s$, then $A^{\parallel \SAT}(x) > s$ with probability $1 - 2^{-\Omega(n)}$.
        
    \end{claim}
    \begin{proof}

        By \Cref{cl:B-SAT-2}, for any $v \in \bool^{c \cdot e(n)}$, for any $s' \leq s$, the probability that $B^{\parallel \SAT}(v, x)$ outputs $s'$ is at most $2^{-\Omega(n^2)}$. The claim follows by a union bound.
    \end{proof}

    Finally, we explain how to use $A^{\parallel \SAT}$ to decide $\gapmrnktp[s_1, s_2]$. Specifically, we will set $s_1(n) = n - 2d \cdot e(n)$ and $s_2(n) = n - 1$, for a constant $d > b$ to be specified below. Let $\theta = n - d \cdot e(n)$. To decide $\gapmrnktp[s_1, s_2]$ on input $x$, simulate $A^{\parallel \SAT}(x)$, and accept iff the output is at most $\theta$.

    We now argue for the correctness of this approach. First suppose $\rnKt(x) \leq s_1(n) = n - 2d \cdot e(n)$. Then, by \Cref{lem:rnk-t-rnkt}, 
    \[\rnK^{t'}(x) \leq n - 2d \cdot e(n) + b \cdot e(n) < \theta.\]
    By \Cref{cl:A-SAT-1}, $A^{\parallel \SAT}(x) \leq \theta$ with high probability, in which case $x$ is accepted.

    Now suppose $\rnKt(x) \geq s_2(n) = n - 1$. Applying success amplification for $\rnKt$ (\Cref{l:success-amp-rnK}), there exists a polynomial $p$ such that 
    \begin{align}\label{eq:t'-12-SA}
        \rnK^{p(t')}(x) \leq \rnK^{t'}_{1/2}(x) + \log p(t').
    \end{align}
    Also observe that $\rnKt(x) \leq \rnK^{p(t')}(x) + \log p(t')$. Thus, with high probability,
    \begin{align*}
        A^{\parallel \SAT}(x) &\geq \rnK^{t'}_{1/2}(x) \tag{\Cref{cl:A-SAT-2}}
        \\&\geq \rnK^{p(t')}(x) - \log p(t') \tag{\Cref{eq:t'-12-SA}}
        \\&\geq \rnKt(x) - 2\log p(t')
        \\&\geq n - 2\log p(t') - 1
        \\&\geq n - 3k \cdot e(n)
        \\&> n - d\cdot e(n) = \theta,
    \end{align*}
    where we let $k \in \nat$ be such that $p(m)^b \leq m^k$ for all $m \in \nat$, and we set $d := 4k$. In this case, $x$ is rejected.

    This completes the proof of the lemma.
\end{proof}

\subsubsection{The Case of \texorpdfstring{$\nKt$}{}}\label{sec:soi-hardness-nkt}
Here, we extend the above results to $\nKt$, which will be needed in the proof of \Cref{thm:collapses_from_SoI}.

\begin{lemma}\label{lem:SoI-imply-short-time}
    Assume SoI holds for $\nKt$ with error $e(n) \geq \log n$. Then, there exists a time bound $t \leq 2^{O(e(n))}$ such that for every string $x \in \bool^n$, it holds that $\nK^t(x) \leq \nKt(x) + O(e(n))$.
\end{lemma}
\begin{proof}[Proof Sketch]
 The proof follows essentially the same argument as in \Cref{lem:rnk-t-rnkt}; see also the proof of \Cref{lem:cond-SoI-imply-short-time}.
\end{proof}

\begin{lemma}\label{l:SoI-nKt}
    Assume SoI holds for $\nKt$ with error $e(n) \geq \log n$. Then, $\gapmnktp[s_1, s_2]$ is contained in $\pr \mathsf{DTIME}[2^{O(e(n))}]^{||\NP}$, where $s_1(n) = n - O(e(n))$ and $s_2(n) = n-1$.
\end{lemma}
\begin{proof}[Proof Sketch]
The proof is a straightforward adaptation of that of \Cref{l:SoI-rnKt}, using \Cref{lem:SoI-imply-short-time} in place of \Cref{lem:rnk-t-rnkt}. We omit the details.
\end{proof}

\subsection{Meta-Complexity to SoI}

The following lemma completes the proof of Item (ii) $\implies$ Item (i) in \Cref{thm:equivalence_SoI}.

\begin{lemma}\label{lem:equivalence-ii-i}

    If $\gapmrnktp[s_1, s_2]\in \pr \mathsf{BPTIME}[n^{\poly(\log n)}]^{||\NP}$, where $s_1(n) = n - \poly(\log n)$ and $s_2(n) = n-1$, then SoI holds for $\rnKt$ with error $e(n) = \poly(\log n)$.    
    
\end{lemma}
\begin{proof}

    Let $D^{\parallel\NP}$ be a randomized quasi-polynomial time algorithm solving $\gapmrnktp[s_1, s_2]$, making non-adaptive queries to some oracle in $\NP$. Without loss of generality, we may assume that $D^{\parallel\NP}$ has failure probability $o(1)$.

    The proof of SoI will proceed by a hybrid argument. Fix $n \in \nat$. For $m \leq 2n$, let $G_m: \bool^{n} \times \bool^d \to \bool^m$ be the generator of \Cref{l:black-box-PRG}, where $d = O(\log^3n)$. Consider any strings $x, y \in \bool^n$, and let $k, k' \in \nat$  be parameters to be specified later satisfying $k, k' = \omega(1)$ and $k, k' \leq 2n$. By \Cref{l:incomp-strings}, for uniformly random strings $w \sim \mathcal{U}_k$ and $w\ \sim \mathcal{U}_{k'}$, with probability at least $1/2$,
    \[\rnKt(w, w') \geq k + k' - 1.\]
    Therefore,
    \[\Pr_{w \sim \mathcal{U}_k; \; w' \sim \mathcal{U}_{k'}}\left[D^{\parallel \NP}(w, w') = 0\right] \geq 1/2 - o(1),\]

    Toward a contradiction, suppose that 
    \[\Pr_{z \sim \mathcal{U}_d; \; w' \sim \mathcal{U}_{k'}; \; D}\left[D^{\parallel \NP}(G_{k}(x; z), w')  = 1\right] > \frac{2}{3}.\]
    Then, $D^{\parallel \NP}(-, \mathcal{U}_{k'})$ is a distinguisher for $G_k(x; -)$, in which case, by \Cref{cor:rnkt-recon}, 
    \[\rK^{\poly(n), \,\parallel D^{\parallel \NP}}(x) \leq k + O(\log^3 n).\]
    We can include a description of $D$ in the bound on $\rnK$, allowing non-adaptive queries to an $\NP$ oracle as needed. Since $D^{\parallel \NP}$ runs in time $n^{\poly(\log n)}$ on inputs of length $\poly(n)$, we have
    \[\rK^{n^{\poly(\log n)}, \,\parallel \NP}(x) \leq k + O(\log^3 n).\]
    By \Cref{l:rKt--parallel-NP}, using the definition of $\rnK$ to remove the non-adaptive $\NP$-oracle queries,
    \[\rnK_{n^{-\poly(\log n)}}^{n^{\poly(\log n)}}(x) \leq k + \poly(\log n).\]
    Applying success amplification (\Cref{l:success-amp-rnK}),
    \[\rnK^{n^{(\log n)^c}}(x) < k + \log^c n,\]
    for some constant $c \in \nat$. 
    
    Then, setting $k := \rnK^{n^{(\log n)^c}}(x) - \log^c n$ yields a contradiction.\footnote{We may assume without loss of generality that $k = \rnK^{n^{(\log n)^c}}(x) - \log^c n \geq \log n$; otherwise, $\rnKt(x) \leq \rnK^{n^{\poly(\log n)}}(x) + \poly(\log n) \leq \poly(\log n)$, and the lemma follows from the fact that $\rnKt(x, y) \geq \rnKt(y \mid x)$.}
    It follows that 
    \[\Pr_{z \sim \mathcal{U}_d; \; w' \sim \mathcal{U}_{k'}; \; D}\left[D^{\parallel \NP}(G_{k}(x; z), w') = 1\right] \leq \frac{2}{3}.\]

    By a very similar argument, assuming that the constant $c \in \nat$ is sufficiently large, setting $k' := \rnK^{n^{(\log n)^c}}(y \mid x) - \log^c n$ implies that 
    \[\Pr_{z, z' \sim \mathcal{U}_d; \; D}\left[D^{\parallel \NP}(G_{k}(x; z), G_{k'}(y; z')) = 1\right] \leq \frac{3}{4}.\]
    It follows that 
    \[\Pr_{z, z' \sim \mathcal{U}_d}\left[\; \rnKt(G_{k}(x; z), G_{k'}(y; z')) \leq k + k' - \poly(\log n)\; \right] \leq \frac{3}{4} + o(1).\]
    So, there exist $z, z' \in \bool^d$ such that 
    \begin{align*}
        \rnKt(G_{k}(x; z), G_{k'}(y; z') &> k + k' - \poly(\log n)
        \\&\geq \rnK^{n^{(\log n)^c}}(x) + \rnK^{n^{(\log n)^c}}(y \mid x) - \poly(\log n)
        \\&\geq \rnKt(x) + \rnKt(y \mid x) - \poly(\log n),
    \end{align*}
    where the last inequality follows from the fact that $\rnKt(u \mid v) \leq \rnK^t(u \mid v) + O(\log t)$ for any time-bound $t \in \nat$ and $u, v \in \bool^*$.

    Lastly, observe that 
    \[\rnKt(G_{k}(x; z), G_{k'}(y; z')) \leq \rnKt(x, y) + O(d) + O(\log n).\]
    Combined with the above, this yields
    \[\rnKt(x, y) \geq \rnKt(x) + \rnKt(y \mid x) - \poly(\log n),\]
    as desired.  
\end{proof}

\subsection{Natural Properties and Explicit Constructions}
In this section we show that \Cref{i:easiness-Gap-MrnKtP,i:natural-property-rnKt,i:hardness-explicit-construction} of \Cref{thm:equivalence_SoI} are equivalent. In fact, the equivalence between \Cref{i:natural-property-rnKt,i:hardness-explicit-construction} is straightforward: one can verify that they are in fact exactly the same statement. \Cref{i:easiness-Gap-MrnKtP} $\implies$ \Cref{i:natural-property-rnKt} is also trivial, by defining $\mathcal{YES}=\{x\mid\rnKt(x)\ge|x|-1\}$, and $\mathcal{NO}=\{x\mid \rnKt(x)\le |x|-\poly(\log |x|)\}$. So the only non-trivial implication we need to prove is \Cref{i:natural-property-rnKt} $\implies$ \Cref{i:easiness-Gap-MrnKtP}, which we restate as a lemma below:

\begin{lemma}[Worst-case to Average-case Reduction for $\gapmrnktp$]\label{lem:worst-avg-rnKt}
    Assume that there exist a constant $C$ and a promise problem $P=(\mathcal{YES},\mathcal{NO})\in\pr\BPTIME[n^{\log^C n}]^{||\NP}$ such that
    \begin{itemize}
        \item $|\mathcal{YES}\cap\{0,1\}^n|\ge2^{n-1}$, and
        \item For every $x\in\{0,1\}^n$ where $\rnKt(x)<n-\log^C n$, $x\in\mathcal{NO}$.
    \end{itemize}
    Then there exists a constant C', such that
    \[
        \gapmrnktp[n-\log^{C'}n,n-1]\in\pr\BPTIME[n^{\log^{C'}n}]^{||\NP}.
    \]
\end{lemma}

Using ideas similar to the worst-case to average-case reduction in \cite{DBLP:conf/focs/Hirahara18}, we are able to prove this lemma.

\begin{proof}
    Let $C'=C+10$. Let $A$ be the algorithm in $\pr\BPTIME[n^{\log^C n}]^{||\NP}$ which accepts any string in $\mathcal{YES}$ with probability at least $9/10$, and rejects any string in $\mathcal{NO}$ with probability at least $9/10$. Let $G$ be the pseudo-random generator of \Cref{l:black-box-PRG}, with $\varepsilon=1/n$ and $m=n-\frac{\log^{C'}n}{2}$. This would give $d=O(\log^3n)$. Our algorithm $B$ for deciding $\gapmrnktp[n-\log^{C'}n,n-1]$ works as follows:
    \begin{quote}
        On input $x$, $B$ uniformly samples a string $z\sim U_d$, and accepts iff $A(G(x,z))$ rejects.
    \end{quote}
    It is easy to see that $B$ makes non-adaptive queries to $\NP$ and runs in time $\poly(n)\cdot n^{\log^C n}$ , which is smaller than $n^{\log^{C'}(n)}$ for all large enough $n$. It remains to argue for the correctness of $B$.
    \begin{claim}\label{clm:worst-avg-rnKt-lower}
        For any $x$ such that $\rnKt(x)\le n-\log^{C'}n$, $B$ accepts $x$ with probability at least $9/10$.
    \end{claim}
    \begin{proof}
        For any $z\in\{0,1\}^d$, by storing $z$ in the description of the program and noticing that $G$ is poly-time computable, for large enough $n$ we have that
        \[
            \rnKt(G(x,z))\le\rnKt(x)+|z|+O(\log n)\le n-\log^{C'}(n)+O(\log ^3 n)\le n-\log^C(n).
        \]
        Hence $G(x,z)$ is in $\mathcal{NO}$, and $A$ rejects with probability at least $9/10$. In other words, $B(x)$ accepts with probability at least $9/10$.
    \end{proof}
    \begin{claim}\label{clm:worst-avg-rnKt-upper}
        For any $x$ such that $\rnKt(x)\ge n-1$, $B$ accepts $x$ with probability at most $7/10$.
    \end{claim}
    \begin{proof}
        For the sake of contradiction, assume there exists $x$ such that $\rnKt(x)\ge n-1$, but $B$ accepts $x$ with probability larger than $7/10$. By definition of $B$, we get
        \[
            \Prob_{z\sim U_d,A}\left[A^{||\NP}(G(x,z))=1\right]\le \frac3{10}.
        \]
        Notice that for $s$ uniformly sampled from $\{0,1\}^m$, with probability at least $1/2$ $s$ is in $\mathcal{YES}$, so we get
        \[
            \Prob_{s\sim U_m,A}\left[A^{||\NP}(s)=1\right]\ge\frac{9}{20}.
        \]
        Hence we have that
        \[
            \Prob_{s\sim U_m,A}\left[A^{||\NP}(s)=1\right]-\Prob_{z\sim U_d,A}\left[A^{||\NP}(G(x,z))=1\right]\ge\frac{3}{20}\ge\frac2n.
        \]
        In other words, $A^{||\NP}(\cdot)$ is a $2/n$-distinguisher for $G(x,\cdot)$. Then by \Cref{cor:rnkt-recon},
        we would get
        \[
            \rK^{\poly(n),||A^{||\NP}}(x)\le m+O(\log^3 n).
        \]
        Since $A$ runs in time $n^{\log^Cn}$, we can store $A$ in the description, and let the program simulate $A$ when constructing $x$. Then the program runs in time $\poly(n)\cdot n^{\log^Cn}\le n^{\log^{C+1}n}$ for all large enough $n$. This would give
        \[
            \rK^{n^{\log^{C+1}n},||\NP}(x)\le m+O(\log^3n).
        \]
        By \Cref{l:rKt--parallel-NP,l:success-amp-rnK}, for all large enough $n$ we have
        \[
            \rnK^{n^{\log ^{C+2}n}}(x)\le\rK^{n^{\log^{C+1}n},||\NP}(x)+O(\log^{C+2}n)\le m+O(\log^{C+2}n).
        \]
        Hence for all large enough $n$ we get
        \[
            \rnKt(x)\le\rnK^{n^{\log^{C+2}n}}(x)+\log^{C+3}n\le m+O(\log^{C+3}n)<n-1.
        \]
        But this contradicts the assumption that $\rnKt(x)\ge n-1$.
    \end{proof}
    Combining \Cref{clm:worst-avg-rnKt-lower} and \Cref{clm:worst-avg-rnKt-upper}, there is a constant gap between the accepting probability of two cases for $B$. By a standard success amplification argument, we obtain an algorithm $B'$ that decides $\gapmrnktp[n-\log^{C'}n,n-1]$ with error at most $1/3$, while running in time $n^{\log^{C'}n}$ and making non-adaptive queries to $\NP$.
\end{proof}

\section{A Quick Lower Bound Result via the Correspondence Theorem: Proof of \Cref{thm:application_correspondence}}\label{sec:quick_lower_bound}

In this section, for completeness, we provide a proof of \Cref{thm:application_correspondence}, which follows from \Cref{thm:consequences_SoI} and \Cref{thm:equivalence_SoI}.

\begin{proof}[Proof of \Cref{thm:application_correspondence}]

    First consider the case that, for some constant $c \in \nat$, SoI holds for $\rnKt$ with error $e(n) = (\log n)^c$. Then, by Item (ii) of \Cref{thm:consequences_SoI},
    \[\MKtP \notin \NTIME[n^{\poly(\log n)}] \cap \coNTIME[n^{\poly(\log n)}].\]
    Otherwise, for every constant $c \in \nat$, SoI does not hold for $\rnKt$ with error $e(n) = (\log n)^c$. In this case, by the implication from Item (ii) to Item (i) in \Cref{thm:equivalence_SoI} (as shown in \Cref{lem:equivalence-ii-i}), 
    \[\gapmrnktp[n-\poly(\log n),n-1]\notin\pr\BPTIME[n^{\poly(\log n)}]^{||\NP}.\]
    Since one of the conditions above must hold, at least one of the consequences must hold. This completes the proof.
\end{proof}

\section{Unconditional Hardness of \texorpdfstring{$\mathsf{MnKtP}$}{MnKtP}: Proof of \Cref{thm:intro_nKt_lb}}

In this section we prove the moreover part of \Cref{thm:intro_nKt_lb}, from which the main part follows directly. 
We first define our technical tool: $\rK$, but with restricted length of random bits.
\begin{definition}\label{d:rK-randomness}
    For any oracle $\mathcal{O}$, any time bound $t\in\mathbb{N}$, any length $r\in\mathbb{N}$, any real number $\delta\in[0,1]$, and any strings $x,y\in\{0,1\}^*$, we define $\rK_\delta^{(t;r),\mathcal{O}}(x\mid y)$ as
    \[
        \rK_\delta^{(t;r),\mathcal{O}}(x\mid y)=\min_{M}\left\{|M|\middle\vert\Prob_{w\sim\{0,1\}^r}\left[M^{\mathcal{O}}(y,w)=x\text{ in $t$ steps}\right]\ge\delta\right\}.
    \]
    We use $\rK_\delta^{(t;r),\mathcal{O}_{\le l}}(x\mid y)$ when $M$ only make queries to $\mathcal{O}$ on strings of length at most $l$. We may omit the conditional string $y$ when it is empty, simplify $(t; r)$ to $t$ when $r \ge t$, and omit $\delta$ when $\delta=2/3$.
\end{definition}

\subsection{Randomness Efficient Success Amplification for \texorpdfstring{$\rK$}{rK}}\label{ss:rK-rand-amp}
In this section, we show how to amplify the success probability of $\rK$ with only 2 times randomness used.
\begin{lemma}\label{l:rK-randomness-amplify}
     For any oracle $\mathcal{O}$, any time bound $t\in\mathbb{N}$, any length $n,m,r,l\in\mathbb{N}$, any real number $\delta\in(0,1]$, and any strings $x\in\{0,1\}^n,y\in\{0,1\}^m$ satisfying $n < t$, we have
     \[
        \rK_{2/3}^{(\poly(t/\delta);2r),\mathcal{O}_{\le l}}(x\mid y)\le\rK_\delta^{(t;r),\mathcal{O}_{\le l}}(x\mid y)+O(\log n/\delta).
     \]
\end{lemma}
Our proof is similar to \cite{DBLP:conf/icalp/LuO21}, but to save randomness, we use the pairwise independent sampling trick. We then use the following string isolation lemma (\cite[Lemma 15]{DBLP:conf/icalp/LuO21}), and save the correct hash value of $x$.
\begin{lemma}[String Isolation Lemma \cite{DBLP:conf/icalp/LuO21}]\label{l:string-isolation}
    There is a deterministic algorithm $A$ for which the following holds. For any set $W\subseteq\{0,1\}^n$ of size $l$, there exists a string $u\in\{0,1\}^{O(\log(n\cdot l))}$ such that $A(1^n,u)$ runs in $\poly(n)$ time and outputs a Boolean circuit that computes a function $H:\{0,1\}^n\to\{0,1\}^{O(\log l)}$ with the following property:
    \[
        H(w)\ne H(w')\text{ for every distinct pair }w,w'\in W.
    \]
    Moreover, the same guarantee is achieved by a random string $u$ of the same length with probability at least $0.99$.
\end{lemma}

\begin{proof}[Proof of \Cref{l:rK-randomness-amplify}]
    Let $s=\rK_\delta^{(t;r),\mathcal{O}_{\le l}}(x\mid y)$, and let $M$ be the corresponding program. Let $A$ be the algorithm of \Cref{l:string-isolation}. For any set $T\subseteq\{0,1\}^n$, we define $p_T$ as
    \[
        p_T=\Prob_{w\sim\mathcal{U}_r}\left[M^{\mathcal{O}_{\le l}}(y,w)\text{ outputs some $z\in T$ in $t$ steps}\right].
    \]
    And for any $\gamma\in(0,1)$, we define $S_\gamma$ as
    \[
        S_\gamma=\left\{z\in\{0,1\}^n\ \middle\vert\ p_{\{z\}}\ge \gamma\right\}.
    \]
    One can easily see that $|S_p|\le1/p$. Now applying \Cref{l:string-isolation} to $S_{\delta/8}$, there exists a string $u\in\{0,1\}^{O(\log n/\delta)}$ such that $A(1^n,u)$ runs in time $\poly(n)$ and outputs a Boolean circuit computing a function $H:\{0,1\}^n\to\{0,1\}^{O(\log 1/\delta)}$, such that
    \[
        H(z)\ne H(z')\text{ for every distinct pair }z,z'\in S_{\delta/8}.
    \]
    Our description is $(M,A,u,H(x))$. For now, assume $2^{r} \geq 256/\delta^2$. Our algorithm $B$ that recovers $x$ from the description runs as follows:
    \begin{enumerate}
        \item Set $k=256/\delta^2$. Sample $a,b\sim\mathcal{U}_r$, and compute $w_i=a\cdot i+b$ for $i\in[k]$, where the addition and multiplication are done in $\mathsf{GF}(2^r)$.
        \item For each $i\in[k]$, we run $M^{\mathcal{O}_{\le l}}(y,w_i)$ for $t$ steps, and let $x_i$ be its output. (If it does not halt within $t$ steps or the output is not of length $n$, we set $x_i=\bot$.)
        \item Let $S'$ be the set of strings appearing at least $\frac\delta2\cdot k$ times (ignoring $\bot$) in $\{x_i\}_{i\in[k]}$.
        \item Run $A(1^n,u)$ to obtain the circuit $C$. Compute $C(x')$ for every string $x'\in S'$. Output the first string $x'$ such that $C(x')$ is equal to the stored $H(x)$. (If no such string exists, output $\bot$.)
    \end{enumerate}
    Clearly the above algorithm runs in time $\poly(t/\delta)$, and only uses $2r$ random bits when sampling $a,b$. It remains to argue about the probability of outputting $x$. We will use the following claim.
    \begin{claim}\label{c:PI-sampling}
        With probability at least $2/3$ over $a,b\sim\mathcal{U}_r$, we have
        \begin{enumerate}
            \item $x\in S'$, and\label{i:PI-sampling-1}
            \item $S'\subseteq S_{\delta/8}$.\label{i:PI-sampling-2}
        \end{enumerate}
    \end{claim}
    \begin{proof}[Proof of \Cref{c:PI-sampling}]\renewcommand{\qedsymbol}{$\diamond$}
        For every set of strings $T\subseteq\{0,1\}^n$ and every $i\in[k]$, we define a random variable $V_{i,T}\in\{0,1\}$ such that $V_{i,T}=1$ iff $x_i\in T$. By properties of finite fields, $\{w_i\}_{i\in[k]}$ (and hence $\{x_i\}_{i\in[k]}$) are identically distributed and pairwise independent. Therefore, for any fixed $T$, the random variables $\{V_{i,T}\}_{i\in[k]}$ are also identically distributed and pairwise independent. From our assumptions on $x$, we have $p_{\{x\}}\ge\delta$, so the expected value of $V_{i, \{x\}}$ is at least $\delta$ for every $i\in[k]$. By Chebyshev's inequality, we get
        \begin{equation}
            \Prob\left[\sum_{i=1}^kV_{i,\{x\}}\le\frac\delta2\cdot k\right]\le\frac{\text{Var}[\sum_{i=1}^k V_{i,\{x\}}]}{(kp_{\{x\}}-\delta k/2)^2}\le\frac{kp_{\{x\}}}{k^2p_{\{x\}}^2/4}\le\frac4{k\delta}\le\frac\delta{64}.\label{eq:PI-sampling-1}
        \end{equation}
        That is, the probability of $x\notin S'$ is at most $\delta/64$. Next, we need to argue that with constant probability, $S'$ is contained in $S_{\delta/8}$, i.e. $S'$ does not contain any string of $\overline{S_{\delta/8}}$. First, we divide $\overline{S_{\delta/8}}$ into $d$ disjoint buckets $B_1,\ldots,B_d$, such that $p_{B_i}\in[\delta/8,\delta/4]$ holds for every $i\in[d-1]$. This can be done by putting the strings of $\overline{S_{\delta/8}}$ into buckets one by one, and whenever the total probability of the current bucket exceeds $\delta/8$, we add a new empty bucket. Since $p_{\{x'\}}\le\delta/8$ for any $x'\in\overline{S_{\delta/8}}$, the total probability of each bucket is in $[\delta/8,\delta/4]$ except for the last one, which might be smaller than $\delta/8$. Now for each bucket $B_j$, by Chebyshev's inequality, we have
        \begin{equation}
            \Prob\left[\sum_{i=1}^kV_{i,B_j}\ge\frac\delta2\cdot k\right]\le\frac{\text{Var}[\sum_{i=1}^k V_{i,B_j}]}{(kp_{B_j}-\delta k/2)^2}\le\frac{kp_{B_j}}{k^2\delta^2/16}\le\frac4{k\delta}\le\frac\delta{64}.\label{eq:PI-sampling-2}
        \end{equation}
        Since for each $i\in[d-1]$, $p_{B_i}\ge\delta/8$ holds, we have $d\le 8/\delta+1\le 16/\delta$. Applying union bound to \Cref{eq:PI-sampling-2}, we get
        \[
            \Prob\left[\exists j\in[d]\text{ s.t. }\sum_{i=1}^kV_{i,B_j}\ge\frac\delta2\cdot k\right]\le\frac{\delta}{64}\cdot d\le \frac14.
        \]
        Since each string $x'$ in $\overline{S_{\delta/8}}$ belongs to some bucket $B_j$, if $x'$ appears more than $\delta k/2$ times, then $B_j$ appears more than $\delta k/2$ times. Hence we can estimate the probability of $S'\nsubseteq S_{\delta/8}$:
        \begin{align}\label{eq:PI-sampling-3}
            \Prob\left[S'\not\subseteq S_{\delta/8}\right] &=\Prob\left[\exists x'\in\overline{S_{\delta/8}}\text{ s.t. }\sum_{i=1}^kV_{i,\{x'\}}\ge\frac\delta2\cdot k\right]\notag
            \\&\le\Prob\left[\exists j\in[d]\text{ s.t. }\sum_{i=1}^kV_{i,B_j}\ge\frac\delta2\cdot k\right]\notag
            \\&\le\frac14.
        \end{align}
        Applying union bound to \Cref{eq:PI-sampling-1,eq:PI-sampling-3}, with probability at least $3/4-\delta/64\ge2/3$ over $a,b\sim\mathcal{U}_r$, both \Cref{i:PI-sampling-1} and \Cref{i:PI-sampling-2} are satisfied.
    \end{proof}
    Now with probability at least $2/3$ over $a,b\sim\mathcal{U}_r$, we have both $x\in S'$ and $S'\subseteq S_{\delta/8}$. Conditioning on this, by \Cref{l:string-isolation}, $H(x)\ne H(x')$ for any $x'\in S'\setminus\{x\}$. Therefore, with probability $2/3$, the algorithm $B$ can always find the correct $x$.

    It remains to handle the case $2^{r} < 256/\delta^2$. Here, the algorithm is deterministic: it enumerates all $2^{r} < k = O(1/\delta^2)$ strings $w \in \bool^{r}$, runs
    $M^{\mathcal{O}_{\le l}}(y,w)$ for $t$ steps on each, and lets $S'$ be the set of outputs
    occurring for at least a $\delta$ fraction of the strings $w$. By the definition of $\delta$, we have
    $x \in S'$. By construction, $S' \subseteq S_{\delta}\subseteq S_{\delta/8}$, so the stored hash value $H(x)$ identifies $x$ uniquely within $S'$ by
    \Cref{l:string-isolation}. The algorithm therefore outputs $x$ with probability $1$, runs in time
    $2^{r}\cdot t\le\poly(t/\delta)$, and has the same description length as above.
    
\end{proof}

\subsection{Reconstruction from Oracle Access to Dense Language}
\begin{lemma}\label{l:MnKtP-reconstruction}
    There exists some constant $C\in\mathbb{N}$ such that the following holds. For any $b\in\mathbb{N}$ and any language $L$ satisfying $|L\cap\{0,1\}^n|\ge 2^n/n^b$ and any strings $x,y\in\{0,1\}^*$, let $n_1=|x|$ and $n_2=|y|$, and $n=n_1+n_2$. Then for any $m_1,m_2\in\mathbb{N}$ satisfying 
    \begin{align*}
        m_1&\le\rK^{(n^{C\cdot b};C\cdot n),L_{\le 3n}}(x\mid y)-C\cdot\log^3(n^b),\\
        m_2&\le\rK^{(n^{C\cdot b};C\cdot n),L_{\le 3n}}(y)-C\cdot\log^3(n^b),
    \end{align*}
    we define $G_1,G_2$ to be the PRG of \Cref{l:black-box-PRG}, with parameter $\varepsilon=1/2n^b$\emph{:}
    \begin{align*}
        G_1&:\{0,1\}^{n_1}\times\{0,1\}^{d_1}\to\{0,1\}^{m_1},\\
        G_2&:\{0,1\}^{n_2}\times\{0,1\}^{d_2}\to\{0,1\}^{m_2}.
    \end{align*}
    Then there exists $z_1\in\{0,1\}^{d_1},z_2\in\{0,1\}^{d_2}$ such that $G_1(x,z_1)\circ G_2(y,z_2)\in L$.
\end{lemma}
To prove this lemma, we need the following corollary regarding $\rK$ upper bound, which comes from the reconstruction procedure of \Cref{l:black-box-PRG}.
\begin{corollary}\label{c:rK-reconstruction}
    For sufficiently large $n,m\in\mathbb{N}$ such that $m\le 2n$, and any $\varepsilon\in(0,1)$, let $G$ be the generator of \Cref{l:black-box-PRG}. Consider any string $x\in\{0,1\}^n$, oracle $\mathcal{O}$, parameters $l,n',t\in\mathbb{N}$, string $y\in\{0,1\}^{n'}$, and algorithm $B^{(\text{-})}:\{0,1\}^m\times\{0,1\}^{n'}\times\{0,1\}^{r'}\to\{0,1\}$ running in time $t$. If it holds that
    \[
        \Prob_{\substack{z\sim\mathcal{U}_d\\w'\sim\mathcal{U}_{r'}}}\left[B^{\mathcal{O}_{\le l}}(G(x;z),y,w')=1\right]-\Prob_{\substack{w\sim\mathcal{U}_m\\ w'\sim\mathcal{U}_{r'}}}\left[B^{\mathcal{O}_{\le l}}(w,y,w')=1\right]\ge2\varepsilon,
    \]
    then we have
    \[
        \rK^{(\poly(n\cdot t/\varepsilon);O(r'+m)),\mathcal{O}_{\le l}}(x\mid y)\le m+|B|+O(\log ^3(n/\varepsilon)),
    \]
    where $|B|$ is the description length of $B$.
\end{corollary}
\begin{proof}
    By \Cref{l:rK-randomness-amplify}, it suffices to show that $\rK_{1/\poly(n/\varepsilon)}^{(\poly(t/\varepsilon);O(r'+m)),\mathcal{O}_{\le l}}(x\mid y)\le m+|B|+O(\log^3(n/\varepsilon))$. By Markov's inequality, we have
    \[
        \Prob_{w'\sim\mathcal{U}_{r'}}\left[\Prob_{z\sim\mathcal{U}_d}\left[B^{\mathcal{O}_{\le l}}(G(x;z),y,w')=1\right]-\Prob_{w\sim\mathcal{U}_m}\left[B^{\mathcal{O}_{\le l}}(w,y,w')=1\right]\ge\varepsilon\right]\ge\varepsilon.
    \]
    That is, with probability at least $\varepsilon$ over $w'\sim\mathcal{U}_{r'}$, $D_{w'}(\text{-}):=B^{\mathcal{O}_{\le l}}(\text{-},y,w')$ is a $\varepsilon$-distinguisher for $G$. Therefore by \Cref{l:black-box-PRG}, we have
    \[
        \Prob_{\substack{w'\sim\mathcal{U}_{r'}\\w\sim\mathcal{U}_{r}\\z\sim\mathcal{U}_d}}\left[R^{D_{w'}}(A(x,z),z,w)=x\right]\ge\frac1{\poly(n/\varepsilon)}.
    \]
    By averaging, there exists some $z\in\{0,1\}^{d}$ such that
    \[
        \Prob_{\substack{w'\sim\mathcal{U}_{r'}\\w\sim\mathcal{U}_{r}}}\left[R^{D_{w'}}(A(x,z),z,w)=x\right]\ge\frac1{\poly(n/\varepsilon)}.
    \]
    Therefore, we can store $z$, $A(x,z)$, as well as the description of $R$ and $B$. Then given $y$ and oracle access to $\mathcal{O}_{\le l}$, by sampling $w'\sim\mathcal{U}_{r'}$ and $w\sim\mathcal{U}_r$, we can simulate $R^{D_{w'}}(A(x,z),z,w)$ in time $\poly(t/\varepsilon)$ and recover $x$ with probability $1/\poly(n\cdot t/\varepsilon)$. Since $r=O(m)$, this gives
    \[
        \rK_{1/\poly(n/\varepsilon)}^{(\poly(t/\varepsilon);O(r'+m)),\mathcal{O}_{\le l}}(x\mid y)\le m+|B|+O(\log^3(n/\varepsilon)).
    \]
    Applying \Cref{l:rK-randomness-amplify} completes the proof.
\end{proof}

We may now prove \Cref{l:MnKtP-reconstruction}.

\begin{proof}[Proof of \Cref{l:MnKtP-reconstruction}]
    For the sake of contradiction, assume the contrary. Then we have
    \[
        \Prob_{\substack{z_1\sim \mathcal{U}_{d_1}\\z_2\sim\mathcal{U}_{d_2}}}[G_1(x;z_1)\circ G_2(y;z_2)\in L]=0.
    \]
    Notice that when $C$ is a large enough constant, $m_1\le n_1$ and $m_2\le n_2$ holds. Hence by our assumption on $L$, we have
    \[
        \Prob_{\substack{w_1\sim\mathcal{U}_{m_1}\\w_2\sim\mathcal{U}_{m_2}}}\left[w_1\circ w_2\in L\right]\ge\frac1{(m_1+m_2)^b}\ge\frac1{(n_1+n_2)^b}=\frac{1}{n^b}.
    \]
    Then by a hybrid argument, we either have
    \[
        \Prob_{\substack{w_1\sim\mathcal{U}_{m_1}\\z_2\sim\mathcal{U}_{d_2}}}[w_1\circ G_2(y;z_2)\in L]-\Prob_{\substack{z_1\sim \mathcal{U}_{d_1}\\z_y\sim\mathcal{U}_{d_2}}}[G_1(x;z_1)\circ G_2(y;z_2)\in L]\ge\frac1{2n^b},
    \]
    or
    \[
        \Prob_{\substack{w_1\sim\mathcal{U}_{m_1}\\w_2\sim\mathcal{U}_{m_2}}}\left[w_1\circ w_2\in L\right]-\Prob_{\substack{w_1\sim\mathcal{U}_{m_1}\\z_2\sim\mathcal{U}_{d_2}}}[w_1\circ G_2(y;z_2)\in L]\ge\frac1{2n^b}.
    \]
    For the former case, given $y$ and oracle access to $L$, there exists a uniform algorithm that $1/2n^b$-distinguishes $G_1(x;\mathcal{U}_{d_1})$ from $\mathcal{U}_{m_1}$ in linear time using linear randomness: given a string $s\in\{0,1\}^{m_1}$, it samples $z_2\sim\mathcal{U}_{d_2}$, and accepts iff $s\circ G_2(y;z_2)\in L$, which can be checked with oracle access to $L$ on length $m_1+m_2\le 3n$. Hence by \Cref{c:rK-reconstruction}, we have
    \[
        \rK^{(\poly(n^b);O(n)),L_{\le 3n}}(x\mid y)\le m_1+O(\log^3(n^b)).
    \]
    But by definition, $m_1=\rK^{(n^{C\cdot b};C\cdot n),L_{\le 3n}}(x\mid y)-C\cdot\log ^3(n^b)$, hence when $C$ is a large enough constant, we have a contradiction.
    
    Similarly, for the latter case, given oracle access to $L$, there exists a uniform algorithm that $1/2n^b$-distinguishes $G_2(y;\mathcal{U}_{d_2})$ from $\mathcal{U}_{m_2}$ in linear time using linear randomness: given a string $s\in\{0,1\}^{m_2}$, it samples $w_1\sim\mathcal{U}_{m_1}$, and accepts iff $w_1\circ s\in L$, which can be checked with oracle access to $L$ on length $m_1+m_2\le 3n$. Hence by \Cref{c:rK-reconstruction}, we have
    \[
        \rK^{(\poly(n^b);O(n)),L_{\le 3n}}(y)\le m_2+O(\log^3(n^b)).
    \]
    But by definition, $m_2=\rK^{(n^{C\cdot b};C\cdot n),L_{\le 3n}}(y)-C\cdot\log ^3(n^b)$, hence when $C$ is a large enough constant, we have a contradiction.
\end{proof}

\subsection{Derandomizing \texorpdfstring{$\rK^{L}$}{rK\^L}}\label{ss:derand-rK}
\begin{lemma}\label{l:derand-rK}
    There exists some constant $C$ such that the following holds. If there exists $\varepsilon,\gamma\in(0,1)$ such that a language $L$ satisfies the following three properties\emph{:}
    \begin{enumerate}
        \item $L\in\NTIME[2^{\varepsilon n}]\cap\coNTIME[2^{\varepsilon n}]$,
        \item $\forall x\in L$, $\nKt(x)\ge\gamma\cdot |x|$,
        \item $\forall n\in\mathbb{N}$, $L\cap\{0,1\}^n\ne\varnothing$,
    \end{enumerate}
    then for all large enough $n,m\in\mathbb{N}$, any $l,t,r\in\mathbb{N}$ and any strings $x\in\{0,1\}^n,y\in\{0,1\}^m$, we have
    \[
        \nKt(x\mid y)\le\rK^{(t;r),L_{\le l}}(x\mid y)+C\cdot(\varepsilon\cdot (n+m+l+r+1/\gamma)/\gamma+\log^3(nmlr/\gamma)+\log t).
    \]
\end{lemma}

\begin{proof}[Proof of \Cref{l:derand-rK}]
    Let $n'=C_1(n+m+l+r+\lceil \log t\rceil+\lceil1/\gamma\rceil)$, where $C_1$ is some constant whose value will be determined later, and let $N'=\lceil n'/\gamma\rceil$. Let $G:\{0,1\}^{N'}\times\{0,1\}^d\to\{0,1\}^r$ be the pseudorandom generator of \Cref{l:black-box-PRG}, with $\varepsilon=1/6$. Let $M^{(\text{-})}$ be the program of $\rK^{(t;r),L_{\le l}}(x\mid y)$, such that
    \[
        \Prob_{w\sim\mathcal{U}_r}\left[M^{L_{\le l}}(y,w)=x\text{ in $t$ steps}\right]\ge\frac23.
    \]
    Then we claim that given a string with high $\nKt$ complexity, we can derandomize $\rK^{L}$:
    \begin{claim}\label{c:derand-rK}
        If $C_1$ is some large enough constant, then for any $s\in\{0,1\}^{N'}$ satisfying $\nKt(s)\ge n'$, it holds that
        \[
            \Prob_{z\sim\mathcal{U}_d}\left[M^{L_{\le l}}(y,G(s,z))=x\text{ in $t$ steps}\right]>\frac12.
        \]
    \end{claim}
    \begin{proof}[Proof of \Cref{c:derand-rK}]\renewcommand{\qedsymbol}{$\diamond$}
        Suppose the contrary. Then it is not hard to see that the following algorithm $B^{(-)}$ with $L_{\le l}$ oracle $1/6$-distinguishes $G(s,\mathcal{U}_d)$ from $\mathcal{U}_r$, and runs in time $\poly(t)$:
        \begin{itemize}
            \item $B$ is given as input $u\in\{0,1\}^r$ along with $x$ and $y$ as above.
            \item $B$ runs $M^{L_{\le l}}(y,u)$ for $t$ steps, and accepts iff it halts and outputs $x$.
        \end{itemize}
        Since $B$ only need to store the description of $M^{(\text{-})}$, by \Cref{c:rK-reconstruction}, we have
        \begin{align*}
            \rK^{(\poly(t\cdot N');O(r)),L_{\le l}}(s\mid x,y)&\le r+|M^{(\text{-})}|+O(\log^3 N')+O(\log t)\\
            &=r+\rK^{(t;r),L_{\le l}}(x\mid y)+O(\log ^3 N')+O(\log t)\\
            &\le O(r+n+\log ^3 N'+\log t).
        \end{align*}
        Because we can store $x\in \bool^n,y\in\{0,1\}^m$, and a ``good'' choice of $O(r)$-length randomness for $\rK$, we get that
        \begin{equation}
            \K^{\poly(t\cdot N'),L_{\le l}}(s)\le O(r+n+m+\log ^3 N'+\log t).\label{eq:derand-rK-1}
        \end{equation}
        Since $L\in\NTIME[2^{\varepsilon n}]\cap\coNTIME[2^{\varepsilon n}]$, we get
        \begin{align*}
            \nKt(s)&\le \nK^{\poly(t\cdot N'\cdot 2^{\varepsilon l})}(s)+O(\log t+\log N'+\varepsilon l)\\
            &\le\nK^{\poly(t\cdot N'),L_{\le l}}(s)+O(\log t+\log N'+\varepsilon l)\tag{by \Cref{l:K-NP-coNP-superscript}}\\
            &\le\K^{\poly(t\cdot N'),L_{\le l}}(s)+O(\log t+\log N'+\varepsilon l)\\
            &\le O(n+m+l+r+\log t+\log^3 N').\tag{by \Cref{eq:derand-rK-1}}
        \end{align*}
        Notice that 
        \begin{align*}
            \log^3 N'&\le O(\log^3 C_1+\log^3(n+m+l+r+\log t+1/\gamma)+\log ^3(1/\gamma))\\
            &\le O(C_1+n+m+l+r+\log t+1/\gamma).
        \end{align*}
        Hence we have $\nKt(s)\le O(C_1+n+m+l+r+\log t+1/\gamma)$. But by our assumption, $\nKt(s)\ge n'=C_1\cdot(n+m+l+r+\lceil\log t\rceil+\lceil1/\gamma\rceil)$. Hence when $C_1$ is some large enough constant, we get a contradiction.
    \end{proof}
    Now let $C_1$ be a large enough constant such that \Cref{c:derand-rK} holds. We define a non-deterministic algorithm $A^{(\text{-})}$ such that, when given the description of $M^{(\text{-})}$ and $y$, as well as oracle access to $L_{\le N'}$, recovers $x$. $A^{L_{\le N'}}$ works as follows:
    \begin{enumerate}
        \item $A^{L_{\le N'}}$ non-deterministically guesses a string $s\in\{0,1\}^{N'}$.
        \item $A^{L_{\le N'}}$ checks if $s\in L$, using its oracle access to $L_{\le N'}$. If $s\notin L$, then it outputs $\bot$ and halts.
        \item $A^{L_{\le N'}}$ enumerates over every $z\in\{0,1\}^d$, and simulates $M^{L_{\le l}}(y,G(s,z))$ for $t$ steps. (Because $l\le N'$, $A^{L_{\le N'}}$ can answer the oracle queries by $M^{L_{\le l}}$.) Let $\{x_z\}_{z\in\{0,1\}^d}$ be the output of the simulations. If there exists some string $\hat x$ appearing more than $2^{d-1}$ times in $\{x_z\}_{z\in\{0,1\}^d}$, then $A^{L_{\le N'}}$ outputs $\hat x$; otherwise it outputs $\bot$.
    \end{enumerate}
    By \Cref{c:derand-rK}, whenever $\nKt(s)\ge n'$, the majority of $\{x_z\}_{z\in\{0,1\}^d}$ is $x$. Since for any $s\in L\cap\{0,1\}^{N'}$, $\nKt(s)\ge\gamma\cdot N'\ge n'$, for any non-deterministic guess of $s$, $A$ always outputs either $x$ or $\bot$. Because $L\cap\{0,1\}^{N'}\ne\varnothing$ for any $N'$, there always exists some non-deterministic guess of $s$ such that $A$ outputs $x$. Hence by the definition of $\nK$, we have
    \begin{equation}
        \nK^{\poly(2^d\cdot t\cdot N'),L_{\le N'}}(x\mid y)\le \rK^{L_{\le l},(t;r)}(x\mid y)+O(\log (t\cdot N')).\label{eq:derank-rK-2}
    \end{equation}
    Therefore we have
    \begin{align*}
        \nKt(x\mid y)&\le \nK^{\poly(2^{\varepsilon N'}\cdot 2^d\cdot t)}(x\mid y)+O(\varepsilon N'+d+\log t)\\
        &\le\nK^{\poly(2^d\cdot t\cdot N'),L_{\le N'}}(x\mid y)+O(\varepsilon N'+d+\log t)\tag{by \Cref{l:K-NP-coNP-superscript}}\\
        &\le\rK^{(t;r),L_{\le l}}(x\mid y)+O(\varepsilon N'+d+\log t)\tag{by \Cref{eq:derank-rK-2}}\\
        &\le\rK^{(t;r),L_{\le l}}(x\mid y)+O(\varepsilon N'+\log ^3 N'+\log t)\tag{Since $d=O(\log^3 N')$}
    \end{align*}
    Replacing $N'$ by $\lceil C_1(n+m+l+r+\lceil \log t\rceil+\lceil1/\gamma\rceil)/\gamma\rceil$ finishes the proof.
\end{proof}

\subsection{Iterative Construction of High-Complexity Strings}\label{ss:nKt-iter-cons}
Combining \Cref{l:MnKtP-reconstruction} and \Cref{l:derand-rK}, we have the following lemma, which we will use to prove \Cref{thm:intro_nKt_lb}.

\begin{lemma}\label{l:nKt-SoI}
    There exists some constant $C\in\mathbb{N}$ such that the following holds. Suppose there exists $\varepsilon,\gamma\in(0,1)$ and $b\in\mathbb{N}$ such that a language $L$ satisfies the following three properties\emph{:}
    \begin{enumerate}
        \item $L\in\NTIME[2^{\varepsilon n}]\cap\coNTIME[2^{\varepsilon n}]$,
        \item $\forall x\in L$, $\nKt(x)\ge\gamma\cdot|x|$,
        \item $\forall n\in\mathbb{N}$, $|L\cap\{0,1\}^n|\ge2^n/n^b$.
    \end{enumerate}
    For any strings $x,y\in\{0,1\}^*$, let $n_1=|x|$, $n_2=|y|$, and $n=n_1+n_2$. Consider any $m_1,m_2\in\mathbb{N}$ satisfying
    \begin{align*}
        m_1&\le \nKt(x\mid y)-C\cdot(\varepsilon \cdot(n+1/\gamma)/\gamma+\log^3(n^b/\gamma))\\
        m_2&\le \nKt(y)-C\cdot(\varepsilon \cdot(n+1/\gamma)/\gamma+\log^3(n^b/\gamma)).
    \end{align*}
    We define $G_1,G_2$ to be the PRG of \Cref{l:black-box-PRG}, with distinguishing parameter $1/2n^b$\emph{:}
    \begin{align*}
        G_1&:\{0,1\}^{n_1}\times\{0,1\}^{d_1}\to\{0,1\}^{m_1},\\
        G_2&:\{0,1\}^{n_2}\times\{0,1\}^{d_2}\to\{0,1\}^{m_2}.
    \end{align*}
    Then there exists $z_1\in\{0,1\}^{d_1},z_2\in\{0,1\}^{d_2}$ such that $G_1(x,z_1)\circ G_2(y,z_2)\in L$.
\end{lemma}
\begin{proof}
    Follows from \Cref{l:MnKtP-reconstruction} and \Cref{l:derand-rK}.
\end{proof}
\begin{proof}[Proof of \Cref{thm:intro_nKt_lb}]
    We prove the ``moreover'' part of this theorem, of which the main part is an immediate corollary. Fix some $\gamma \in (1/2, 1)$ and $b \in \nat$, and let $\varepsilon$ be a parameter to be fixed below. For the sake of contradiction, assume that there exists some language $L$ satisfying the following three properties:
    \begin{enumerate}
        \item $L\in\NTIME[2^{\varepsilon n}]\cap\coNTIME[2^{\varepsilon n}]$,
        \item $\forall x\in L$, $\nKt(x)\ge\gamma\cdot |x|$,
        \item for all sufficiently large $n\in\mathbb{N}$, $|L\cap\{0,1\}^n|\ge2^n/n^b$.
    \end{enumerate}
    Then we use the following algorithm to explicitly construct a long string $s\in L$ in short deterministic time with oracle access to $L$.
    \begin{algorithm}[H]
        \caption{Iterative construction of long string in $L$}
		\begin{algorithmic}[1]
			\Procedure{$A^L$}{$\gamma, 1^n$} 
                \State $s_0 \vcentcolon= \bot, N_1 \vcentcolon= n, \theta_1 \vcentcolon= \gamma n$.
                \For {$s' \in \{0, 1\}^n$ by lexicographical order}
                    \If {$s' \in L$}
                        \State $s_1 \leftarrow s'$.
                        \State \textbf{break}
                    \EndIf
                \EndFor
                \State $\gamma' \vcentcolon= (\gamma+1/2)/2$.
                \For {$i = 2, 3, \cdots$}
                    \State $m_{i, 1} \vcentcolon= \lceil\gamma'n/\gamma\rceil, m_{i, 2} \vcentcolon= \lceil\gamma'\theta_{i - 1}/\gamma\rceil$.
                    \State $N_i \vcentcolon= m_{i, 1} + m_{i, 2}, \varepsilon_i \vcentcolon= (n+N_{i-1})^{-b}/2$.
                    \State $\theta_i \vcentcolon= \gamma N_i$.
                    \State Instantiate $G_1: \{0, 1\}^n \times \{0, 1\}^{d_1} \to \{0, 1\}^{m_{i, 1}}$ in \Cref{l:black-box-PRG} with security parameter $\varepsilon_i$.
                    \State Instantiate $G_2: \{0, 1\}^{N_{i - 1}} \times \{0, 1\}^{d_2} \to \{0, 1\}^{m_{i, 2}}$ in \Cref{l:black-box-PRG} with security parameter $\varepsilon_i$.
                    \State $s_i \vcentcolon= \bot$.
                    \For {$x \in \{0, 1\}^n, z_1 \in \{0, 1\}^{d_1}, z_2 \in \{0, 1\}^{d_2}$ by lexicographical order}
                        \If {$G_1(x; z_1) \circ G_2(s_{i - 1}; z_2) \in L$}
                            \State $s_i \vcentcolon= G_1(x; z_1) \circ G_2(s_{i - 1}; z_2)$.
                            \State \textbf{break}
                        \EndIf
                    \EndFor
                    \If {$\theta_i > n/(2 - 2\gamma')$}
                        \State Output $s_i$.
                        \State\textbf{return}
                    \EndIf
                \EndFor
                \State Output $\bot$.
            \EndProcedure
		\end{algorithmic}
	\end{algorithm}
    Then we have the following claim for the above procedure $A^L$:
    \begin{claim}\label{c:nKt-iter-cons}
        For a fixed $\gamma$, if $\varepsilon$ is small enough, then for all large enough $n$, there exists some integer $i_0\le\log_{\gamma'}(1-\frac{1}{2\gamma'})+2$ such that
        \begin{enumerate}
            \item For any $i\le i_0$, $s_i\in L$, and
            \item $\theta_{i_0}>n/(2-2\gamma')$.
        \end{enumerate}
    \end{claim}
    \begin{proof}[Proof of \Cref{c:nKt-iter-cons}]\renewcommand{\qedsymbol}{$\diamond$}
        We need to prove two things: first, for every $i$ during the iteration, we can find some $s_i$ s.t. $s_i\ne\bot$, i.e. $s_i\in L$; second, our algorithm terminates at some $i_0$, i.e. $\theta_{i_0}>n/(2-2\gamma')$.
        
        To prove the first argument, we use induction on $i$. For $i=1$, because $L\cap\{0,1\}^n\ne\varnothing$, we can always find $s'\in L$, so $s_1\ne\bot$. Then suppose $s_{i-1}\ne\bot$ for some $i\ge 2$. Since $\gamma>1/2$, by \Cref{l:nKt-SoI}, there exists some universal constant $C$ such that whenever $m_{i,1}$ and $m_{i,2}$ satisfy
        \begin{align*}
            m_{i,1}&\le\nKt(x\mid s_{i-1})-C\cdot(\varepsilon (n+N_{i-1})+\log^3((n+N_{i-1})^b))\\
            m_{i,2}&\le\nKt(s_{i-1})-C\cdot(\varepsilon (n+N_{i-1})+\log^3((n+N_{i-1})^b)),
        \end{align*}
        there exists some $z_1\in\{0,1\}^{d_1},z_2\in\{0,1\}^{d_2}$ such that
        \[
            G_1(x,z_1)\circ G_2(s_{i-1},z_2)\in L.
        \]
        By inductive hypothesis, $s_{i-1}\in L$, hence $\nKt(s_{i-1})\ge\gamma|s_{i-1}|=\theta_{i-1}$. By a counting argument (\Cref{l:incomp-strings}), there also exists some $x\in\{0,1\}^n$ such that $\nKt(x\mid s_{i-1})\ge n$. Therefore, it suffices to show that
        \begin{align*}
            m_{i,1}&=\lceil\gamma'n/\gamma\rceil\le n-C\cdot(\varepsilon (n+N_{i-1})+\log^3((n+N_{i-1})^b))\\
            m_{i,2}&=\lceil\gamma'\theta_{i-1}/\gamma\rceil\le\theta_{i-1}-C\cdot(\varepsilon (n+N_{i-1})+\log^3((n+N_{i-1})^b).
        \end{align*}
        Since we did not terminate at $i-1$, we are guaranteed that $\theta_{i-1}\le n/(2-2\gamma')$, therefore $N_{i-1}\le\frac{n}{\gamma\cdot(2-2\gamma')}\le4n$ (By definition $\gamma\ge1/2$ and $\gamma'\le3/4$). Since $\gamma'/\gamma<1$, if $\varepsilon$ is small enough (i.e. $\varepsilon<\frac{\gamma-\gamma'}{20C}$), then for all large enough $n$, the above two inequalities holds, and we can find $z_1\in\{0,1\}^{d_1},z_2\in\{0,1\}^{d_2}$ such that $G_1(x,z_1)\circ G_2(s_{i-1},z_2)\in L$, in other words, $s_i$ exists.

        To prove the second argument, notice that $|s_1|=N_1=n$, and $N_i\ge\gamma'\cdot(N_{i-1}+n/\gamma)$. Expanding the recursion, we get
        \[
            N_i\ge\left((\gamma')^1+\cdots+(\gamma')^{i-1}\right)\cdot\frac{n}{\gamma}+(\gamma')^{i-1}\cdot n=\left(\frac{\gamma'}{\gamma}\cdot\frac{1-(\gamma')^{i-1}}{1-\gamma'}+(\gamma')^{i-1}\right)\cdot n.
        \]
        Hence we have
        \[
            \theta_i\ge\frac{\gamma'\cdot(1-(\gamma')^{i-1})}{1-\gamma'}\cdot n
        \]
        When $i-1\ge\log_{\gamma'}(1-\frac{1}{2\gamma'})$, we have $\theta_i\ge\frac{n}{2-2\gamma'}$. Hence the algorithm terminates when $i\le\log_{\gamma'}(1-\frac{1}{2\gamma'})+2$, i.e. $i_0\le\log_{\gamma'}(1-\frac{1}{2\gamma'})+2$.
    \end{proof}
    The running time of algorithm $A^L(\gamma,1^n)$ is bounded by $t_0 := 2^n\cdot\poly(n)\cdot2^{O(\log ^3 n)}\cdot i_0$, and we only need to store $n$ and $\gamma$. Without loss of generality, we may assume that $\gamma=1/2+1/2^k$ for some $k\in\mathbb{N}$, and we can store $k$, which takes $\lceil\log\log(\frac{1}{\gamma-1/2})\rceil$ bits. Also since $N_{i_0-1}\le 4n$, we have $N_{i_0}\le 5n$, so $A$ only need to query $L$ on inputs of length at most $5n$. Therefore
    \begin{equation}
        \K^{t_0,L_{\le5n}}(s_{i_0})\le O\left(\log n+\log\log\left(\frac{1}{\gamma-1/2}\right)\right).\label{eq:nKt-iter-cons-3}
    \end{equation}
    
    Let $c$ be the constant of \Cref{l:K-NP-coNP-superscript}. Applying that
    lemma with $T(m) = 2^{\varepsilon m}$ and query length $\ell \vcentcolon= 5n$, and using that
    $\K^{t_0, L_{\le 5n}}(s_{i_0}) \ge \nK^{t_0, L_{\le 5n}}(s_{i_0})$, we obtain
    \[
        \nK^{t_1}(s_{i_0}) \le \K^{t_0, L_{\le 5n}}(s_{i_0}) + O(\log\log t_0),
        \qquad
        t_1 \vcentcolon= t_0\cdot 2^{5\varepsilon n}\cdot \log^{c}\!\left(2^{5\varepsilon n}\cdot t_0\right).
    \]
    Since $\log t_0 = n + O(\log^3 n + \log i_0)$, we have
    $\log t_1 \le n + 5\varepsilon n + O(\log^3 n + \log i_0)$, and therefore
    \begin{align}
        \nKt(s_{i_0}) &\le \nK^{t_1}(s_{i_0}) + \log t_1 \notag\\
        &\le \K^{t_0, L_{\le 5n}}(s_{i_0}) + n + 5\varepsilon n + O(\log^3 n + \log i_0) \notag\\
        &\le n + 5\varepsilon n + O\left(\log^3 n + \log i_0 + \log\log\left(\frac{1}{\gamma-1/2}\right)\right)
        \tag{by \Cref{eq:nKt-iter-cons-3}}\\
        &\le n + 5\varepsilon n + O\left(\log^3 n + \log\log\left(\frac{1}{\gamma-1/2}\right)\right),
        \label{eq:nKt-iter-cons-1}
    \end{align}
    where the last inequality uses $i_0 \le O(\log\frac{1}{\gamma-1/2})$, which follows from
    $i_0\le\log_{\gamma'}(1-\frac{1}{2\gamma'})+2$ (\Cref{c:nKt-iter-cons}), $\gamma'=(\gamma+1/2)/2$
    and $1/2<\gamma<1$.

    On the other hand, since $s_{i_0}\in L$,
    \begin{align}
        \nKt(s_{i_0})\ \ge\ \gamma\cdot|s_{i_0}| = \theta_{i_0}\ \ge\ \frac{n}{2-2\gamma'}.
        \label{eq:nKt-iter-cons-2}
    \end{align}
    Writing $\Delta_\gamma \vcentcolon= \frac{1}{2-2\gamma'} - 1 > 0$, which is positive because
    $\gamma' > 1/2$, \Cref{eq:nKt-iter-cons-2} implies that $\nKt(s_{i_0}) \geq n + \Delta_{\gamma}\cdot n$; thus, the two bounds are contradictory for all sufficiently large $n$ provided
    $5\varepsilon < \Delta_\gamma/2$. Recalling the constraint from \Cref{c:nKt-iter-cons}, it
    therefore suffices to choose
    \[
        \varepsilon \ <\ \min\left\{\frac{\gamma-\gamma'}{20C},\ \frac{\Delta_\gamma}{10}\right\},
    \]
    where $C$ is the constant of \Cref{l:nKt-SoI}. This choice depends only on $\gamma$ and on the
    constants of \Cref{l:nKt-SoI} and \Cref{l:K-NP-coNP-superscript}, which completes the proof.

\end{proof}

\subsection{Limitations of the Iterative Construction for \texorpdfstring{$\gamma \leq 1/2$}{}}\label{ss:limitations_gamma}
We believe the statement is also true for $0<\gamma\le1/2$. However, our technique only works for $\gamma>1/2$. In order to see this, recall that we iteratively construct a string $s$ in time $2^{(1 +o(1))\cdot n}$ such that $\nKt(s)\ge(1+\Omega(1))\cdot n$, resulting in a contradiction. But by the iterative construction, we are only guaranteed that $\nKt(s)\ge(\gamma+\gamma^2+\gamma^3+\cdots)\cdot n-o(n)$, which is upper bounded by $\frac{\gamma}{1-\gamma}\cdot n-o(n)$. Hence when $\gamma\le 1/2$, $\frac{\gamma}{1-\gamma}\le 1$, and we don't have $\nKt(s)\ge(1+\Omega(1))\cdot n$, so the proof does not work here. 

A more intuitive explanation to this barrier arises from the SoI intuition behind this technical iterative construction (see \Cref{s:techniques}). In the setting of a dense language $L \subseteq \mathcal{R}_{\nKt[\gamma \cdot n]}$, we are able to prove that
\begin{equation}\nonumber
    \nKt(x, y) \ge \gamma(\nKt(x \mid y) + \nKt(y)) - C\varepsilon n
\end{equation}
by reconstruction (\Cref{l:MnKtP-reconstruction}) and derandomization (\Cref{l:derand-rK}). This degrades to a trivial inequality when $\gamma \le 1/2$, since
\[
    \nKt(x, y) \ge \max\{\nKt(x), \nKt(y)\} \ge \max\{\nKt(x \mid y), \nKt(y)\} \ge \frac{1}{2}(\nKt(x \mid y) + \nKt(y))
\]
holds unconditionally. Thus the derived SoI inequality does not seem helpful. On the other hand, the statement is non-trivial when $\gamma > 1/2$, and it is strong enough for us to derive a contradiction from it.

\section{Unconditional Hardness of \texorpdfstring{$\gapmrnktp$}{Gap-MrnKtP}: Proof of \Cref{thm:intro_rnKt_lbs}}

In this section, we prove the ``moreover'' parts in \Cref{thm:intro_rnKt_lbs}. These parts are strengthened versions of the plain complexity lower bounds for  $\gapmrnktp$ in \Cref{thm:intro_rnKt_lbs}. To see this, we observe the following simple lemma.

\begin{lemma}\label{l:implication_rnkt}
	Let $\theta: \mathbb{N} \to \mathbb{N}$ be such that $\theta(n) < n - 1$ and $\mathcal{C}$ be a complexity class. Suppose for every dense language $L \subseteq \mathcal{R}_{\rnKt[\theta(n)]}$, it is the case that $L\not\in \mathcal{C}$. Then $\gapmrnktp[\theta(n), n - 1] \not\in \mathsf{co}\textnormal{-}\mathcal{C}$.
\end{lemma}
\begin{proof}
We prove the contrapositive. Suppose $\gapmrnktp[\theta(n), n - 1] \in \mathsf{co}\textnormal{-}\mathcal{C}$. Then there exists a language $L_0 \in \mathsf{co}\textnormal{-}\mathcal{C}$ that is consistent with $\gapmrnktp[\theta(n), n - 1]$, in the sense that for every $x \in \{0,1\}^n$: if $\rnKt(x) \leq \theta(n)$, then $x \in L_0$, and if $\rnKt(x) \geq n-1$, then $x \notin L_0$. 

Let $L$ be the complement of $L_0$. Note that $L \in \mathcal{C}$. Moreover, by construction, $L \subseteq \mathcal{R}_{\rnKt[\theta(n)]}$. Finally, note that $L$ is dense, since $L_n$ contains the set $\{x \mid \rnKt(x) \geq n-1\}$, whose size is at least $2^n / 2$ by \Cref{l:incomp-strings}. 
\end{proof}

\subsection{Exponential \texorpdfstring{$\AMTIME \cap \coAMTIME$}{} Lower Bound via Iterative Construction} \label{ss:rnKt_lb_2}

In this subsection, we give proof to \Cref{thm:intro_rnKt_lbs}, \Cref{i:rnkt_lb_2}. The plan is similar to the proof of \Cref{thm:intro_nKt_lb}, where we break reconstructive black-box PRG (See \Cref{l:black-box-PRG}) under easiness assumption of $L$, then iteratively construct a long string in $L$ to derive contradiction. It shares same technical barrier $\gamma > 1/2$ with \Cref{thm:intro_nKt_lb}, as discussed in
\Cref{ss:limitations_gamma}.

\begin{lemma}\label{l:MrnKtP-reconstruction}
    For any $c > 0, 1/2 < \gamma < 1$, there exist constants $C_1, C_2 > 0$ such that for any $L \subseteq \mathcal{R}_{\rnKt[\gamma \cdot n]}$ satisfying $|L_n| \ge 2^n/n^c$, the following holds. 
    
    For all sufficiently large $n_1, n_2 \in \mathbb{N}$ and any string $x \in \{0, 1\}^{n_1}, y \in \{0, 1\}^{n_2}$, we fix any $m_1, m_2 \in \mathbb{Z}$ such that
    \begin{align*}
        m_1 & \le \rK^{(n_1 + n_2)^{C_1}, L_{\le 2n_1 + 2n_2}}(x \mid y) - C_2\log^3 (n_1 + n_2), \\
        m_2 & \le \rK^{(n_1 + n_2)^{C_1}, L_{\le 2n_1 + 2n_2}}(y) - C_2\log^3 (n_1 + n_2).
    \end{align*}
    Let $m \vcentcolon= m_1 + m_2$. If $m_1, m_2 > 0$, for any pseudorandom generators $G_1, G_2$ instantiated from \Cref{l:black-box-PRG} with parameters $\varepsilon \vcentcolon= m^{-c - 1}/2$ and
    \begin{align*}
        G_1&:\{0,1\}^{n_1}\times\{0,1\}^{d_1}\to\{0,1\}^{m_1},\\
        G_2&:\{0,1\}^{n_2}\times\{0,1\}^{d_2}\to\{0,1\}^{m_2},
    \end{align*}
    there exist $z_1 \in \{0, 1\}^{d_1}, z_2 \in \{0, 1\}^{d_2}$ such that
    \[
        G_1(x; z_1) \circ G_2(y; z_2) \in L.
    \]
\end{lemma}

\begin{proof}
        For the sake of contradiction, assume the contrary, i.e.,
        \[
            \Prob_{\substack{z_1\sim\mathcal{U}_{d_1}\\z_2\sim\mathcal{U}_{d_2}}}\left[G_1(x; z_1) \circ G_2(y;z_2) \in L\right] = 0.
        \]
        By density assumption on $L$, we have
        \[
            \Prob_{\substack{w_1\sim\mathcal{U}_{m_1}\\w_2\sim\mathcal{U}_{m_2}}}\left[w_1\circ w_2 \in L\right]\ge \frac{1}{m^c}.
        \]
        We consider two cases: either
        \[
            \Prob_{\substack{w_1\sim\mathcal{U}_{m_1}\\z_2\sim\mathcal{U}_{d_2}}}\left[w_1\circ G_2(y;z_2) \in L\right]\le \frac{1}{2m^c},
        \]
        or
        \[
            \Prob_{\substack{w_1\sim\mathcal{U}_{m_1}\\z_2\sim\mathcal{U}_{d_2}}}\left[w_1\circ G_2(y;z_2) \in L\right] \ge \frac{1}{2m^c}.
        \]
        
        For the former case, we design $D^{L}$ as follows: Given a string $s\in\{0,1\}^{m_2}$, it samples $w_1\sim\mathcal{U}_{m_1}$ and returns whether $w_1 \circ s \in L$. Note that $m_1 \le 2n_1, m_2 \le 2n_2$, we have query length here bounded by $2n_1 + 2n_2$, as well as $m \le \poly(n_1 + n_2), \varepsilon^{-1} \le \poly(n_1 + n_2)$. 
        
        We can easily see that $D^L$ $\varepsilon$-distinguishes $G_2(y;\mathcal{U}_{d_2})$ from $\mathcal{U}_{m_2}$ in $O(m)$ time. Applying \Cref{c:rK-reconstruction} on $G_2$, we obtain
        \[
            \rK^{\poly(m / \varepsilon), L_{\le 2n_1 + 2n_2}}(y)\le m_2 + O(\log^3(n_2/\varepsilon)) = m_2 + O(\log^3(n_1 + n_2)).
        \]
        But by definition, $m_2 \leq \rK^{(n_1 + n_2)^{C_1}, L_{\le 2n_1 + 2n_2}}(y) - C_2\log^3 (n_1 + n_2)$, hence when $C_1,C_2$ are large enough constants, we have a contradiction.
        
        Similarly, for the latter case, we design $D^L$ which $\varepsilon$-distinguishes $G_1(x; \mathcal{U}_{d_1})$ from $\mathcal{U}_{m_1}$ in polynomial time: On $s\in\{0,1\}^{m_1}$, it samples $z_2\sim\mathcal{U}_{d_2}$, returns whether $s \circ G_2(y; z_2) \in L$. Hence by \Cref{c:rK-reconstruction}, we have
        \[
            \rK^{\poly(m / \varepsilon), L_{\le 2n_1 + 2n_2}}(x\mid y) \le m_1+O(\log^3(n_1 + n_2)).
        \]
        But by definition, $m_1 \leq \rK^{(n_1 + n_2)^{C_1}, L_{\le 2n_1 + 2n_2}}(x\mid y)-C_2\log^3 (n_1 + n_2)$, we then obtain a contradiction when $C_1,C_2$ are large enough.
\end{proof}

We aim to derive contradiction from \Cref{l:MrnKtP-reconstruction}, as in the proof of \Cref{thm:intro_nKt_lb}, which finishes the proof of \Cref{thm:intro_rnKt_lbs}, \Cref{i:rnkt_lb_2}.

\begin{proof}[Proof of \Cref{thm:intro_rnKt_lbs}, \Cref{i:rnkt_lb_2}]

    Fix $c > 0, 1/2 < \gamma < 1$ and $L \subseteq \mathcal{R}_{\rnKt[\gamma \cdot n]}$ such that $|L_n| \ge 2^n/n^c$. For sake of contradiction we assume that 
    \begin{equation*}
        L \in \AMTIME[2^{\varepsilon n}] \cap \coAMTIME[2^{\varepsilon n}], 
    \end{equation*}
    where $\varepsilon > 0$ is chosen below independently of $L$.
    
    Similar to \Cref{l:nKt-SoI}, we use the following algorithm to explicitly construct a long string $s \in L$ in short deterministic time with oracle access to $L$.  
    \begin{breakablealgorithm}
        \caption{Iterative construction of long string in $L$}
		\begin{algorithmic}[1]
			\Procedure{$A^L$}{$\gamma, 1^n$} 
                \State $s_0 \vcentcolon= \bot, N_1 \vcentcolon= n, \theta_1 \vcentcolon= \gamma n$.
                \For {$s' \in \{0, 1\}^n$ by lexicographical order}
                    \If {$s' \in L$}
                        \State $s_1 \leftarrow s'$.
                        \State \textbf{break}
                    \EndIf
                \EndFor
                \State $\gamma' \vcentcolon= (\gamma+1/2)/2$.
                \For {$i = 2, 3, \cdots$}
                    \State $m_{i, 1} \vcentcolon= \lceil\gamma'n/\gamma\rceil, m_{i, 2} \vcentcolon= \lceil\gamma'\theta_{i - 1}/\gamma\rceil$.
                    \State $N_i \vcentcolon= m_{i, 1} + m_{i, 2}, \varepsilon_i \vcentcolon= N_i^{-c - 1}/2$.
                    \State $\theta_i \vcentcolon= \gamma N_i$.
                    \State Instantiate $G_1: \{0, 1\}^n \times \{0, 1\}^{d_1} \to \{0, 1\}^{m_{i, 1}}$ in \Cref{l:black-box-PRG} with security parameter $\varepsilon_i$.
                    \State Instantiate $G_2: \{0, 1\}^{N_{i - 1}} \times \{0, 1\}^{d_2} \to \{0, 1\}^{m_{i, 2}}$ in \Cref{l:black-box-PRG} with security parameter $\varepsilon_i$.
                    \State $s_i \vcentcolon= \bot$.
                    \For {$x \in \{0, 1\}^n, z_1 \in \{0, 1\}^{d_1}, z_2 \in \{0, 1\}^{d_2}$ by lexicographical order}
                        \If {$G_1(x; z_1) \circ G_2(s_{i - 1}; z_2) \in L$}
                            \State $s_i \vcentcolon= G_1(x; z_1) \circ G_2(s_{i - 1}; z_2)$.
                            \State \textbf{break}
                        \EndIf
                    \EndFor
                    \If {$\theta_i > n/(2 - 2\gamma')$}
                        \State Output $s_i$.
                        \State \textbf{return}
                    \EndIf
                \EndFor
                \State Output $\bot$.
            \EndProcedure
		\end{algorithmic}
    \end{breakablealgorithm}
    
    We hereby analyse behaviour of $A^L$.
    \begin{claim} \label{l:rnkt-SoI-algo}
        For any $1/2 < \gamma < 1$, there exists $i^* \in \mathbb{N}$ such that for all sufficiently large $n \in \mathbb{N}$, $A^L(\gamma, 1^n)$ outputs in $O_\gamma(2^{n + O(\log^3 n)})$ time the string $s_{i^*} \neq \bot$ with $\rnKt(s_{i^*}) > n/(2 - 2\gamma')$. Here, the constant factor hidden in big-$O$ notation depends on $\gamma$, but can be fixed when $\gamma$ is fixed. 
    \end{claim}

    To prove \Cref{l:rnkt-SoI-algo}, we prove the following crucial statement.
    \begin{claim} \label{l:rnkt-SoI-algo-induction} 
        Fix any $1/2 < \gamma < 1$, for all sufficiently large $n \in \mathbb{N}$, the following holds for all $i \in \mathbb{N}$: On input $(\gamma, 1^n)$, either $s_i \neq \bot$ is produced successfully, or the algorithm terminates with output $s_j$ for some $j < i$.
    \end{claim}
    
    \begin{proof}[Proof of \Cref{l:rnkt-SoI-algo-induction}] \renewcommand{\qedsymbol}{\ensuremath{\diamond}}
        We prove by induction. For $i = 1$, since $L$ is dense, it is guaranteed that $L_n \neq \emptyset$, and $s_1$ is the lexicographically smallest string in $L_n$.

        Now assume for $i \ge 2$, we have $s_i \in L$ and $\theta_i \le n/(2 - 2\gamma')$. We aim to prove that algorithm successfully finds $s_{i + 1} \in L$.

        By induction assumption, $\rnKt(s_i) \ge \gamma N_i = \theta_i$. Let $C_1, C_2 > 0$ be the constants from \Cref{l:MrnKtP-reconstruction}. Let $c'$ be the constant of \Cref{l:rK-AM-coAM-superscript}, and write
        $t_i \vcentcolon= (n + N_i)^{C_1}$ and $\ell_i \vcentcolon= 2n + 2N_i$. Applying the second part of
        \Cref{l:rK-AM-coAM-superscript} with $T(m) = 2^{\varepsilon m}$, and noting that
        $|s_i| = N_i \le \ell_i$, we obtain
        \[
            \rnK^{t_i'}(s_i) \le \rnK^{t_i, L_{\le \ell_i}}(s_i) + O(\log(n + N_i)),
            \qquad
            t_i' \vcentcolon= t_i \cdot \bigl(2^{\varepsilon \ell_i} \cdot |s_i| \cdot \log t_i\bigr)^{c'},
        \]
        so that $\log t_i' \le 2c'\varepsilon(n + N_i) + O(\log (n+N_i))$. Since every deterministic
        oracle program is in particular a randomized nondeterministic one, we have
        $\rnK^{t_i, L_{\le \ell_i}}(s_i) \le \rK^{t_i, L_{\le \ell_i}}(s_i)$, and therefore
        \begin{align*}
            \rnKt(s_i) &\le \rnK^{t_i'}(s_i) + \log t_i' \\
            &\le \rK^{t_i, L_{\le \ell_i}}(s_i) + 2c'\varepsilon(n + N_i) + O(\log(n + N_i)).
        \end{align*}

        Rearranging and using $\rnKt(s_i) \ge \theta_i$, we obtain
        \[
            \rK^{t_i, L_{\le \ell_i}}(s_i) \ \ge\ \theta_i - 2c'\varepsilon(n + N_i) - O(\log(n + N_i)).
        \]
        Note that, since $\gamma > 1/2$, we have $\gamma'/\gamma = 1/2 + 1/(4\gamma) < 1$. Recall also
        that $N_i \le 4n$, so $n + N_i \le 5n \le 5\theta_i/\gamma$. Hence, for $\varepsilon$ chosen
        small enough as a function of $\gamma$ and $c'$ (and independently of $L$), we get
        \begin{align} \label{ineq:rnkt_reconstruction_2}
            m_{i + 1, 2} = \frac{\gamma'}{\gamma} \theta_i \le \rK^{t_i, L_{\le \ell_i}}(s_i) - C_2\log^3(n + N_i).
        \end{align}
        
        By standard counting argument of $\rK$, there exists $x \in \{0, 1\}^n$ such that
        \[
            \rK^{(n + N_i)^{C_1}, L_{\le 2n + 2N_i}}(x \mid s_i) \ge n,
        \]
        so there exists $x \in \{0, 1\}^n$ satisfying
        \begin{align} \label{ineq:rnkt_reconstruction_1}
            m_{i + 1, 1} = \frac{\gamma'}{\gamma} n \le n - C_2\log^3(n + N_i) \le \rK^{(n + N_i)^{C_1}, L_{\le 2n + 2N_i}}(x \mid s_i) - C_2\log^3(n + N_i).
        \end{align}

        Combining \Cref{ineq:rnkt_reconstruction_1} and \Cref{ineq:rnkt_reconstruction_2}, we apply \Cref{l:MrnKtP-reconstruction}, which indicates for any sufficiently large $n$, $A^L(\gamma, 1^n)$ successfully produces $s_{i + 1} \in L$.
    \end{proof}

    \
    \begin{proof}[Proof of \Cref{l:rnkt-SoI-algo}] \renewcommand{\qedsymbol}{\ensuremath{\diamond}}
        Fix any $1/2 < \gamma < 1$. For any $i \ge 2$, observe that if $A^L(\gamma, 1^n)$ produces $s_i \neq \bot$, we will have
        \[
            \theta_i = \gamma'n + \gamma'\theta_{i - 1},
        \]
        and $\theta_1 = \gamma n$ holds. Simple telescoping gives
        \[ 
            \theta_i = \left(\sum_{j = 1}^{i - 1} \gamma'^j + \gamma'^{i - 1}\gamma \right)n.  
        \]
        Applying $\gamma' < \gamma < 1$, we have
        \begin{align} \label{ineq:theta}
            \theta_i > \left(\sum_{j = 1}^i \gamma'^j \right)n = \frac{\gamma'(1 - \gamma'^i)n}{1 - \gamma'}, 
            \quad 
            \theta_i < \left(\sum_{j = 1}^i \gamma^j \right)n < \frac{\gamma n}{1 - \gamma},
        \end{align}
        holding for every $i$. Note that $\gamma' = (1 + 2\gamma)/4 > 1/2$, we always have $\gamma' / (1 - \gamma') > 1 / (2 - 2\gamma')$, indicating that there exists $I_\gamma$ such that whenever $i \ge I_\gamma$, $\theta_i > n/(2 - 2\gamma')$. Combining this with \Cref{l:rnkt-SoI-algo-induction} we can derive that $A^L(\gamma, 1^n)$ terminates with output $s_{i^*} \in L$ for all sufficiently large $n$ where $i^* \le I_\gamma$. Given $L \subseteq \mathcal{R}_{\rnKt[\gamma \cdot n]}$, we have
        \[
            \rnKt(s_{i^*}) \ge \gamma |s_{i^*}| = \theta_{i^*} > n/(2 - 2\gamma').
        \]
        
        The running time is bounded by
        \begin{align*}
            O(2^n) + \sum_{i = 2}^{i^*} O(2^{n + O(\log^3 n) + O(\log^3 N_{i - 1})}) 
            & = O(2^n) + \sum_{i = 2}^{i^*} O(2^{n + O(\log^3 n) + O(\log^3 (\theta_{i - 1} / \gamma))}) \\
            & \le O(2^n) + I_{\gamma} \cdot O(2^{n + O(\log^3 n) + O(\log^3 (n / (1 - \gamma)))})
        \end{align*}
        Note that $I_{\gamma}$ is independent from $n$, the running time is therefore bounded by $O_\gamma(2^{n + O(\log^3 n)})$.
    \end{proof}

    Fix any $1/2 < \gamma < 1$. For all sufficiently large $n$, by \Cref{l:rnkt-SoI-algo}, we compute $s_{i^*} \vcentcolon= A^L(\gamma, 1^n)$, and we have $\rnKt(s_{i^*}) > n/(2 - 2\gamma')$. On the other hand, the running time bound given in \Cref{l:rnkt-SoI-algo} indicates
    \begin{align*}
        \rnK^{t(n), L_{\le \max\{N_1, N_2, \cdots, N_{i^*}\}}}(s_{i^*}) = O(\log n),
    \end{align*}
    where $t(n) = O(2^{n + O(\log^3 n)})$. By \Cref{ineq:theta}, we have
    \[
        \max\{N_1, N_2, \cdots, N_{i^*}\} = \frac{1}{\gamma}\max\{\theta_1, \theta_2, \cdots, \theta_{i^*}\} < \frac{n}{1 - \gamma}.
    \]
    Therefore,
    \[
        \rnK^{t(n), L_{\le n/(1 - \gamma)}}(s_{i^*}) = O(\log n).
    \]

    Let $\ell \vcentcolon= n/(1-\gamma)$. Applying the second part of
    \Cref{l:rK-AM-coAM-superscript} with $T(m) = 2^{\varepsilon m}$, and letting
    $t' \vcentcolon= t(n)\cdot\bigl(2^{\varepsilon \ell}\cdot|s_{i^*}|\cdot\log t(n)\bigr)^{c'}$ be the resulting time-bound from that lemma, we obtain
    \[
        \rnK^{t'}(s_{i^*}) \le \rnK^{t(n), L_{\le \ell}}(s_{i^*}) + O(\log n) = O(\log n).
    \]
    Substituting $t(n) = O(2^{n+O(\log^3 n)})$ and $|s_{i^*}| \le \ell = O(n)$ gives
    $\log t' \le n + \frac{c'\varepsilon n}{1-\gamma} + O(\log^3 n)$, and therefore
    \[
        \rnKt(s_{i^*}) \le \rnK^{t'}(s_{i^*}) + \log t'
        \le n + \frac{c'\varepsilon n}{1-\gamma} + O(\log^3 n).
    \]
    On the other hand, $\rnKt(s_{i^*}) > n/(2-2\gamma')$, and $1/(2-2\gamma') > 1$ since $\gamma' > 1/2$.
    Writing $\Delta_\gamma \vcentcolon= \frac{1}{2-2\gamma'} - 1 > 0$ and choosing $\varepsilon$ such that 
    $\varepsilon < \Delta_\gamma(1-\gamma)/(2c')$ and \Cref{l:rnkt-SoI-algo-induction} holds, the two bounds contradict each other for all
    sufficiently large $n$.

\end{proof}

\subsection{\texorpdfstring{$\AMTIME$}{} Lower Bound via Explicit Constructions} \label{ss:rnKt_lb_1}

We formally state the technical result from \cite{DBLP:conf/stoc/0001LL25}, which is mentioned in \Cref{s:techniques}. We begin with a formal definition of \emph{single-valued Arthur-Merlin protocol}, which is a natural interpretation from \cite{DBLP:conf/stoc/0001LL25}. 

\begin{definition}\label{d:svAM}
    A \emph{non-binary public coin protocol} $(P, V)$ is defined as follows: On any input $x$, $V, P$ interact by constantly many rounds; In the $i$-th round, $V$ (the verifier) sends plain uniform randomness $r_i \in \{0, 1\}^{n_i}$ to $P$; $P$ (the prover) then returns $w_i$ according to $x, r_1, w_1, \cdots, r_{i - 1}, w_{i - 1}, r_i$. After the last round, $V$ outputs some string $y$ as answer, which deterministically depends on the whole interaction transcript (all $w_j, r_j$'s) and input $x$. 

    Non-binary public coin protocol $(P, V)$ is a \emph{single-valued Arthur-Merlin} protocol computing $\{x_n\}_{n \in \mathbb{N}}$, if there exists $l: \mathbb{N} \to \mathbb{N}$ such that the following holds for every $n \in \mathbb{N}$:
    \begin{itemize}
        \item \textbf{Conformity:} There exists $\alpha_n \in \{0, 1\}^{l(n)}$ and an honest prover $P$ such that $(P, V)(1^n ,\alpha_n) = x_n$ holds with probability $1$ over randomness of $V$;
        \item \textbf{Resiliency:} For every $\zeta_n \in \{0, 1\}^{l(n)}$, there exists $y_n \in \{0, 1\}^n$ such that, for every (possibly with unbounded computational power) prover $P'$, $(P', V)(1^n, \zeta_n) \in \{y_n, \bot\}$ with probability at least $2/3$ over randomness of $V$. 
    \end{itemize}
    We refer to $l(n)$ as the \emph{advice complexity} of $(P, V)$.
\end{definition}

\begin{lemma}[Round reduction for single-valued Arthur-Merlin protocol]\label{l:round-reduction}
    For any $k$-round single-valued Arthur-Merlin protocol $(P,V)$ computing $\{x_n\}_{n\in\mathbb{N}}$ with advice complexity $l(n)$ where the verifier $V$ runs in time $T(n)$, there exists a 2-round single-valued Arthur-Merlin protocol $(P',V')$ computing $\{x_n\}_{n\in\mathbb{N}}$ with advice complexity $l(n)$ where the verifier $V$ runs in time $T(n)^{O_k(1)}$.
\end{lemma}
\begin{proof}
    We translate the single-valued Arthur-Merlin protocol into a decision problem, using the round reduction results for decision problems with Arthur-Merlin protocol, and finally translate it back into the single-valued Arthur-Merlin protocol.
    
    We define a promise problem $\mathcal{P}=(\mathcal{YES},\mathcal{NO})$, where
    \begin{align*}
        (1^n,\alpha,i,b)\in\mathcal{YES}&\iff\exists\text{ prover $P$},\Pr_V\left[(P,V)(1^n,\alpha)\ne\bot\land i\text{'th bit of }(P,V)(1^n,\alpha)\text{ is }b\right]\ge\frac35,\\
        (1^n,\alpha,i,b)\in\mathcal{NO}&\iff\forall\text{ prover $P'$}, \Pr_V\left[(P',V)(1^n,\alpha)\ne\bot\land i\text{'th bit of }(P',V)(1^n,\alpha)\text{ is }b\right]\le\frac25.
    \end{align*}
    By definition, $\mathcal{P}$ has a $T(n)$-time binary Arthur-Merlin protocol with $k$ rounds. By the standard round reduction and error reduction for binary Arthur-Merlin protocol\cite{DBLP:journals/jcss/BabaiM88,DBLP:conf/stoc/GoldwasserS86}, $\mathcal{P}$ has a $T(n)^{O_k(1)}$-time 2-round binary Arthur-Merlin protocol with perfect completeness and small soundness error, in other words, $\mathcal{P}\in\mathsf{prAMTIME}[T(n)^{O_k(1)}]$. Let $\hat V$ be the verifier in the protocol, then we have,
    \begin{align*}
        (1^n,\alpha,i,b)\in\mathcal{YES}&\implies\Pr_{r}\left[\exists w,\hat V(1^n,\alpha,i,b,r,w)=1\right]=1,\\
        (1^n,\alpha,i,b)\in\mathcal{NO}&\implies\Pr_{r}\left[\exists w,\hat V(1^n,\alpha,i,b,r,w)=1\right]\le\frac1{3n}.
    \end{align*}

    Next, we define a 2-round single-valued Arthur-Merlin protocol that computes $\{x_n\}_{n\in\mathbb{N}}$. Given input $(1^n,\alpha)$, the verifier $V'$ first samples $2n$ random strings $r^0_1,r^0_2,\ldots,r^0_n,r^1_1,r^1_2,\ldots,r^1_n$ and send them to the prover. Then the prover sends back $2n$ strings $w^0_1,w^0_2,\ldots,w^0_n,w^1_1,w^1_2,\ldots,w^1_n$ and a proposed answer $\hat x\in\{0,1\}^n$. The verifier then runs $\hat V(1^n,\alpha,i,b,r^b_i,w^b_i)=:v^b_i$ for each $i\in[n]$ and $b\in\{0,1\}$. If for some $i$, $v^0_i=v^1_i=0$, then $V'$ prints $\bot$. Otherwise, $V'$ outputs a string $x\in\{0,1\}^n$ defined as follows. For each $i\in[n]$, if $v^0_i=v^1_i=1$, then set the $i$'th bit of $x$ as the $i$'th bit of $\hat x$. Otherwise set the $i$'th bit of $x$ as the $b$ such that $v_i^b=1$.

    It is not hard to see that $V'$ runs in time $T(n)^{O_k(1)}$, and has advice complexity of $|\alpha|=l(n)$. It remains to argue about the correctness, i.e. the conformity and resiliency of $V'$.
    \begin{itemize}
        \item \textbf{Conformity:} By the conformity of $V$, there exists $\alpha_n$ and honest prover $P$ such that $(P,V)(1^n,\alpha_n)=x_n$ with probability 1 over randomness of $V$. Let $x_{n,i}$ denote the $i$'th bit of $x_n$. Then by definition of $\mathcal{P}$, we know that $(1^n,\alpha_n,i,x_{n,i})\in\mathcal{YES}$ for each $i\in[n]$. The honest prover set $\hat x=x_n$, select $w_i^{x_{n,i}}$ such that $v_i^{x_{n,i}}=\hat V(1^n,\alpha_n,i,x_{n,i},r^{x_{n,i}}_i,w_i^{x_{n,i}})=1$, and set $w_i^{1-x_{n,i}}$ as the all 0 string. If $v_i^{1-x_{n,i}}=\hat V(1^n,\alpha_n,i,1-x_{n,i},r^{1-x_{n,i}}_i,w_i^{1-x_{n,i}})=1$, then $V'$ would set the $i$'th bit of $x$ as $\hat x$, which is equal to $x_n$ because the prover is honest. Otherwise, $v_i^{1-x_{n,i}}=0$, and $V'$ would set the $i$'th bit of $x$ as $x_{n,i}$, which is also correct. To summarize, an honest prover would make $V'$ print $x_n$ with probability 1.
        \item \textbf{Resiliency:} By the resiliency of $V$, for any $\zeta_n\in\{0,1\}^{l(n)}$, there exists $y_n\in\{0,1\}^n$ such that, for any prover $P'$, $(P',V)(1^n,\zeta_n)\in\{y_n,\bot\}$ with probability at least $2/3$ over randomness of $V$. Let $y_{n,i}$ denote the $i$'th bit of $y_n$. Then by the resiliency of $V$, $(1^n,\zeta_n,i,1-y_{n,i})\in\mathcal{NO}$ for every $i\in[n]$. Thus for each $i\in[n]$, the probability of existing $w_{i}^{1-y_{n,i}}$ such that $\hat V(1^n,\alpha,i,1-y_{n,i},r_i^{1-y_{n,i}},w_i^{1-y_{n,i}})=1$ is at most $\frac1{3n}$. By a union bound, with probability at least $2/3$, for every $i\in[n]$, for every $w_{i}^{1-y_{n,i}}$, we have $v_{i}^{1-y_{n,i}}=\hat V(1^n,\alpha,i,1-y_{n,i},r_i^{1-y_{n,i}},w_i^{1-y_{n,i}})=0$. Conditioning on this, if for some $i$, $v_i^{y_{n,i}}=0$, then $V'$ would output $\bot$. Otherwise, for every $i$, $v_i^{y_{n,i}}=1$, then $V'$ would output $y_n$. Hence with probability at least $2/3$ over the randomness of $V'$, regardless of prover's strategy, $V'$ either prints $y_n$ or $\bot$.

    \end{itemize}
\end{proof}

Then we present below the main result from \cite{DBLP:conf/stoc/0001LL25} regarding explicit algorithms that hits any dense language in $\coAM$.

\begin{lemma}[See \cite{DBLP:conf/stoc/0001LL25}, Theorem 1.3] \label{l:explicit-construction}
    For every dense language $L \in \coAM$ and constant $k \ge 1$, there exists $\{x_n\}_{n \in \mathbb{N}}$ and single-valued Arthur-Merlin protocol $(P, V)$ computing $\{x_n\}_{n \in \mathbb{N}}$ such that
    \begin{itemize}
        \item \textbf{Advice efficiency:} The advice complexity of $(P, V)$ is $l(n) = 2^{\log^{1/k} n}$;
        \item \textbf{Time efficiency:} On input $(1^n, \zeta_n)$ where $|\zeta_n| = l(n)$, $V$ runs in time $2^{\log^{O(k)} n}$;
        \item \textbf{Hitting:} For infinitely many $n \in \mathbb{N}$, $x_n \in L$.
    \end{itemize}
\end{lemma}

Note that the original result in \cite{DBLP:conf/stoc/0001LL25} works for only languages with half density, i.e., $|L \cap \bool^n| \ge 2^n/2$. However, it is easy to amplify the density by parallel repetition, sketched as follows: For language $L$ and some $p: \mathbb{N} \to \mathbb{N}$, define $L'$ as
\[
    L'_n = \{ x \mid x_{[1\ : \ l]} \in L, \textrm { or } \cdots, \textrm{ or } x_{[(p(n) - 1) \cdot l + 1\ : \ p(n) \cdot l]} \in L, \textrm{ where } l = \lfloor |x| / p(n) \rfloor \}
\]
we can show $L \in \coAM$ implies $L' \in \coAM$. Furthermore, if $L$ only has inverse-polynomial density (i.e., $|L \cap \bool^n| \ge 2^n/n^c$), when $p(n)$ is some appropriately picked polynomial function, we can have $|L' \cap \bool^n| \ge 2^{n - 1}$, so the original explicit construction result can be applied to obtain single-valued Arthur-Merlin protocol $(P', V')$ hitting $L'$. Then we apply $\log p$ bits of additional advice to locate the $L$-hitting subsection of an $L'$-hitting output from $(P', V')$. 

Then we show that this non-binary Arthur-Merlin algorithm hits $\coAM$ properties with strings of small $\rnKt$ complexity.

\begin{lemma}\label{l:construction-complexity}

Let $L \in \coAM$ be a dense language, then for any $k > 1$, there exists infinitely many $n \in \mathbb{N}$ and $x_n \in \{0, 1\}^n$, such that $x_n \in L \cap \{0, 1\}^n$ and $\rnKt(x_n) < O(2^{\log^{1/k}n})$.
    
\end{lemma}

\begin{proof}
    Fix any constant $k$. By \Cref{l:explicit-construction}, there exists a constant-round single-valued Arthur-Merlin protocol computing $\{x_n\}_{n\in\mathbb{N}}$, such that for infinitely many $n\in\mathbb{N}$, $x_n\in L$. Then by \Cref{l:round-reduction}, there exists a 2-round single-valued Arthur-Merlin protocol $(P,V)$ computing $\{x_n\}_{n\in\mathbb{N}}$ in time $2^{\log^{O(k)}n}$ with advice complexity $2^{\log^{1/k}n}$. Let $\alpha_n\in\{0,1\}^{2^{\log^{1/k}n}}$ be the advice for $x_n$. Then by the conformity of $V$, we know that
    \[
        \Pr_{r}[\exists w,V(1^n,\alpha_n,r,w)=x_n]=1.
    \]
    By the resiliency of $V$, we have
    \[
        \Pr_{r}[\forall w, V(1^n,\alpha_n,r,w)\in\{x_n,\bot\}]\ge\frac23.
    \]
    Thus for at least $2/3$ of $r$, there exists $w$ such that $V(1^n,\alpha_n,r,w)=x_n$, and for every $w$, $V(1^n,\alpha_n,r,w)\in\{x_n,\bot\}$. This exactly matches the definition of $\rnK$, thus we get
    \[
        \rnK^{2^{\log^{O(k)}n}}(x_n)\le|V|+|\alpha_n|+O(\log n)\le O(2^{\log^{1/k}n}).
    \]
    By definition of $\rnKt$, we get
    \[
        \rnKt(x_n)\le \rnK^{2^{\log^{O(k)}n}}(x_n)+\log^{O(k)}n\le O(2^{\log^{1/k}n}).
    \]
\end{proof}

We are then ready to prove the central result of this subsection. 

\begin{proof}[Proof of \Cref{thm:intro_rnKt_lbs}, \Cref{i:rnkt_lb_1}]
    Towards contradiction we assume the contrary, i.e., there exists $\gamma \in (0, 1), d > 1, L \subseteq \mathcal{R}_{\rnKt[2^{\log^\gamma n}]}$ such that some algorithm $A \in \coAM[2^{\log^d n}]$ decides $L$, and $L$ is dense.

    Applying standard padding argument, we define $l(n) = 2^{\log^d n}$ and 
    \begin{equation*}
        L' = \{ x \circ y \mid x \in L, \lfloor l^{-1}(|x| + |y|) \rfloor = |x|\}.
    \end{equation*}
    
    To decide $L'$ on $x \in \{0, 1\}^n$, we might first compute $m = \lfloor l^{-1}(n) \rfloor$, which can be done in $\poly(n)$ time assuming that $l$ is efficiently computable. Then $L'$ can be decided in $\coAM$, just by truncating the input into first $m$ bits and run the $\coAM$ protocol in time $2^{\log^d m} = l(m) \le n$.

    By assumption, $L$ is dense, which immediately gives $L'$ is dense. Applying \Cref{l:construction-complexity}, we can conclude for any constant $k > 1$ that there exist infinitely many $n \in \mathbb{N}$ and $z_n \in L'_n$ such that $\rnKt(z_n) < O(2^{\log^{1/k} n})$. For such $z_n$, we decompose it into $z_n = x_n \circ y_n$ such that $|x_n| = \lfloor l^{-1}(n) \rfloor$, then 
    \begin{itemize}
        \item By definition of $L'$, given that $x_n \circ y_n \in L'_n$, we have $x_n \in L$. Recall that $L \subseteq \mathcal{R}_{\rnKt[2^{\log^\gamma n}]}$, we have $\rnKt(x_n) \ge 2^{\log^{\gamma/d} n}$;
        \item By definition of $\rnKt$, trivially, we have $\rnKt(x_n) \le \rnKt(x_n \circ y_n) + O(1) < O(2^{\log^{1/k} n})$.
    \end{itemize}
    This gives
    \begin{equation*}
        2^{\log^{\gamma/d} n}< O(2^{\log^{1/k} n}).
    \end{equation*}
    Note the arbitrariness of $k > 1$, we take $k = 2d/\gamma$, then
    \begin{equation*}
        O(2^{\log^{1/k} n}) = O(2^{\log^{\gamma/2d} n}) < 2^{\log^{\gamma/d} n}.
    \end{equation*}
    which contradicts the previous inequality.
\end{proof}

\subsection{Non-Uniform Lower Bound via Near-Maximum Circuit Lower Bounds}\label{ss:rnKt_lb_3}

We prove \Cref{thm:intro_rnKt_lbs}, \Cref{i:rnkt_lb_3}, which begins with some technical tools.

\begin{lemma}\label{t:prg-circuit}
    For any language $\mathcal{O} \in \AM \cap \coAM$ and any constant $\varepsilon > 0$, there exists a family of multisets $T = \{T_s \subseteq \{0,1\}^s\}_{s \in \mathbb{N}}$ such that each $T_s$ satisfies $|T_s| \leq \poly(s)$ and the following properties hold\emph{:}
    
    \begin{itemize}[leftmargin=*]
        \item[] \textbf{Efficiency:} For all sufficiently large $s$, we have $\rnKt(T_s) \leq 2^{\log^{\varepsilon} s}$. Here, the multiset $T_s$ with $k$ elements is represented as a binary string of the form $e_1\,01\,e_2\,01\,\dots\,01\,e_k$, where each $e_i$ is a $2s$-bit string encoding an element of $T_s$, with each bit duplicated.

        \item[] \textbf{Pseudorandomness:} For infinitely many $s \in \mathbb{N}$, the following holds: for every $\mathcal{O}$-oracle circuit $C \colon \{0,1\}^s \to \{0,1\}$ of size at most $s$,
        \[
            \left| \Prob_{x \sim \{0,1\}^s}[C(x) = 1] - \Prob_{x \sim T_s}[C(x) = 1] \right| \le \frac{1}{s}.
        \]
    \end{itemize}
\end{lemma}

Next, we show \Cref{t:prg-circuit}. We need the following tools.

\begin{lemma}[{\cite[Theorem 1.2]{DBLP:conf/stoc/0001LL25}}]\label{l:maximum-circuit}
    For any language $\mathcal{O} \in \AM \cap \coAM$ and any $0 < \varepsilon < 1$, $(\AMEXP \cap \coAMEXP)/2^{n^\varepsilon} \nsubseteq \SIZE^\mathcal{O}[2^n/n]$.
\end{lemma}

The following lemma gives upper bound on $\rnKt$ complexity of 
truth table decidable in $\AMEXP \cap \coAMEXP$ with short advice.

\begin{lemma}\label{l:AM-truth-table}
    For sufficiently large $n \in \mathbb{N}$ and $0 < \varepsilon < 1$, if $f: \{0, 1\}^* \to \{0, 1\}$ is computable in $(\AMEXP \cap \coAMEXP)/2^{n^\varepsilon}$, then $\rnKt(\tt(f_n)) = O(2^{n^\varepsilon})$.
\end{lemma}
\begin{proof}
    We first state a extended version of \Cref{l:rK-AM-coAM-superscript}, which handles oracle computed by non-uniform Arthur-Merlin protocol.

\begin{claim}\label{l:non-uniform-AM-coAM-zero-error}
    For any time-constructible, non-decreasing $T: \mathbb{N} \to \mathbb{N}$, $T(n) \ge n$, and any
    $\alpha: \mathbb{N} \to \mathbb{N}$, there exists a constant $c \in \nat$ such that for every
    $\mathcal{O} \in (\AMTIME[T(n)] \cap \coAMTIME[T(n)])/\alpha(n)$, all $x, y \in \bool^*$, all
    integers $\ell \ge 1$, and all sufficiently large $t \in \mathbb{N}$ with $t > |x|$, restricting
    query length $\ell$,
    \[
        \rnK^{t \cdot (T(\ell) \cdot |x| \cdot \log t)^{c}}(x \mid y)
        \ \le\ \rnK^{t,\mathcal{O}_{= \ell}}(x \mid y) + O(\log(\ell \cdot |x| \cdot \log t)) + \alpha(\ell).
    \]
\end{claim}
\begin{proof}[Proof Sketch of \Cref{l:non-uniform-AM-coAM-zero-error}]\renewcommand\qedsymbol{$\diamond$}
    The proof is basically identical to that of \Cref{l:rK-AM-coAM-superscript}, except that the non-uniform
    advice needed to simulate the Arthur--Merlin protocols on queries of length at most $\ell$ is
    stored in the description, contributing $\alpha(\ell)$ extra bits. 
\end{proof}
    
    The proof then is a straightforward application of \Cref{l:non-uniform-AM-coAM-zero-error}. Assume there exists $k > 0$ and $0 < \varepsilon < 1$ such that $f \in (\AMTIME[2^{n^k}] \cap \coAMTIME[2^{n^k}])/2^{n^\varepsilon}$. Hence there exists a uniform algorithm $M^{(\text{-})}$ such that, with access to oracle $f$ on input length $n$, it outputs $\tt(f_n)$ in time $O(2^n)$, so
    \[
        \rnK^{2^n, f_{=n}}(\tt(f_n)) = O(1).
    \]
    By \Cref{l:non-uniform-AM-coAM-zero-error}, we have 
    \begin{equation*}
        \rnKt(\tt(f_n)) \le \rnK^{2^{O(n^k)}}(\tt(f_n)) + O(n^k) + 2^{n^\varepsilon} = O(2^{n^\varepsilon}),
    \end{equation*}
    as desired.
\end{proof}

As an application of \Cref{l:maximum-circuit,l:AM-truth-table}, we get the following.

\begin{lemma}\label{l:rnKt-hard-tt}
    For any language $\mathcal{O} \in \AM \cap \coAM$ and any constant $\gamma > 0$, there exists a sequence $\{x_N\}_{N \in \mathbb{N}}$ such that the following hold\emph{:}
    \begin{itemize}[leftmargin=*]
        \item For every large enough $N$, $\rnKt(x_N)\leq 2^{\log^{\gamma} N}$.
        \item For infinitely many $N$ that are powers of $2$, the string $x_N$, interpreted as the truth table of a function $f \colon \{0,1\}^{\log N} \to \{0,1\}$, has $\mathcal{O}$-oracle-circuit complexity at least $N /\log N$.
    \end{itemize}
\end{lemma}
\begin{proof}
\Cref{l:maximum-circuit} states that there exists a language in $(\AMEXP \cap \coAMEXP)/2^{n^\gamma}$ whose truth tables have nearly maximal $\mathcal{O}$-oracle circuit complexity. \Cref{l:AM-truth-table} then states that each such truth table has $\rnKt$-complexity at most $O\mleft(2^{\log^{\gamma} N}\mright)$, where $N$ denotes the length of the truth table.
 \end{proof}

We now present the proof of \Cref{t:prg-circuit}.

\begin{proof}[Proof of \Cref{t:prg-circuit}]
The proof combines \Cref{l:rnKt-hard-tt} and \Cref{t:IW97-PRG}.

Let $\mathcal{O}\in\AM\cap\coAM$.
We start by applying \Cref{t:IW97-PRG} with $\varepsilon \vcentcolon= 1/2$, and let $c, d$ be the corresponding constants from the lemma.

Consider \Cref{l:rnKt-hard-tt} with $\gamma\vcentcolon=\varepsilon/2$ and let  $\{x_N\}$ be the corresponding sequence of strings.

Let $\mathcal{N} \subseteq \nat$ be the infinite set of powers of $2$ at which the string $x_N$ from
\Cref{l:rnKt-hard-tt} has $\mathcal{O}$-oracle-circuit complexity at least $N/\log N$. For each
$N \in \mathcal{N}$, let $s(N) \vcentcolon= \lfloor N^{1/c} \rfloor$, so that
$s(N)^{c} \le N < (s(N)+1)^{c}$. We call an integer $s$ \emph{good} if $s = s(N)$ for some
$N \in \mathcal{N}$; note that there are infinitely many good $s$.

For a good $s$, with $N \in \mathcal{N}$ the corresponding length, the output length of
$F(x_N,-)$ is at least $N^{1/c} \ge s$, and we define
\[
    T_s \vcentcolon= \left\{ F(x_N, z)_{[1:s]} \mid z \in \{0,1\}^{d \log N} \right\}.
\]
For every other $s$, define $T_s = \{0^s\}$; the efficiency bound below applies to every $s$, while pseudorandomness is claimed only at good $s$.

We first argue that $T_s$ satisfies the pseudorandom condition for good $s$. This follows from the fact that if $s$ is good, then the corresponding $x_N$ has high $\mathcal{O}$-oracle-circuit complexity (as a truth table) and the security guarantee provided by the pseudorandom generator $F$ in \Cref{t:IW97-PRG}.

Next, we address the efficiency condition. Note that for a (sufficiently large) good $s$, given the integer $s$ and the string $x_N$, where $N$ is the integer witnessing the goodness of $s$, we can construct $T_s$ in time $\poly(N)\leq \poly(s)$. By the efficiency condition stated in \Cref{l:rnKt-hard-tt}, we have 
\[
     \rnKt(x_N)\leq 2^{\log^{\gamma} N}.
\]
It follows that
\begin{align*}
        \rnKt(T_s) &\leq 2^{\log^{\gamma} N} + \log \poly(s) \\
        &\leq 2^{O(\log s)^{\gamma}} + O(\log s)\\
        &\leq 2^{\log^{\varepsilon} s},
\end{align*}
as desired. Also note that efficiency holds trivially when $s$ is not good.

\end{proof}

Putting it all together, we are now ready to show \Cref{thm:intro_rnKt_lbs}, \Cref{i:rnkt_lb_3}.
\begin{proof}[Proof of \Cref{thm:intro_rnKt_lbs}, \Cref{i:rnkt_lb_3}.]

    Towards contradiction, we assume the contrary. That is, there exists $\gamma > 0, \mathcal{O} \in \AM \cap \coAM, c > 1, d \ge 1$, $L \subseteq \mathcal{R}_{\rnKt[2^{\log^\gamma n}]}$ satisfying
    \begin{itemize}
        \item $L$ can be decided by $\mathcal{O}$-oracle circuit family $\{C^\mathcal{O}_n\}_{n \in \mathbb{N}}$ of size $2^{\log^d n}$;
        \item There exists $N \in \mathbb{N}$ such that $|L \cap \{0, 1\}^n| \ge 2^n/n^c$ for all $n > N$.
    \end{itemize} 

    We instantiate \Cref{t:prg-circuit} with $\varepsilon \vcentcolon= \gamma / (2(d + 1))$, obtain the corresponding family of multi-sets $T = \{T_s\}_{s \in \mathbb{N}}$. Consider the following algorithm $D: \{0, 1\}^* \to \{0, 1\}$ trying to break the pseudo-randomness:
    \begin{algorithm}[H]
		\caption{A non-uniform algorithm for breaking pseudo-randomness}
		\begin{algorithmic}[1]
			\Procedure{$D$}{$x$} 
    			\State $s \vcentcolon= |x|$.
    			\State $m \vcentcolon= 2^{(2/(c - 1) \cdot \log s)^{1/(d + 1)}}$.
                \If {$m \le N$}
                   \State Output $0$.
                \EndIf
                \State $k := m^{(c - 1)/2}$.
                \If {$km > s$}
                   \State Output $0$.
                \EndIf
                \For {$i \in [k]$}
                    \State Compute $a_i \vcentcolon= C^\mathcal{O}_m(x_{[(i - 1)m + 1 \ \vcentcolon \ im]})$.
                    \If {$a_i = 1$}
                        \State Output $1$.
                    \EndIf
                \EndFor
                \State Output $0$.
            \EndProcedure
		\end{algorithmic}
	\end{algorithm}

    Let $s_0 \in \mathbb{N}$ be large enough that $\rnKt(T_s) \le 2^{(\log s)^{\varepsilon}}$ for every
    $s \ge s_0$. Fix any $s \ge s_0$, any
    $x \in T_s$, and any block $u \vcentcolon= x_{[(i-1)m+1\,:\,im]}$ with $i \in [k]$, where $m$ and $k$ are as
    computed by $D$ on inputs of length $s$. Since $u$ can be recovered from $T_s$ together with the
    indices of $x$ and $i$, we have
    \[
        \rnKt(u) \ \le\ \rnKt(T_s) + O(\log s) \ \le\ 2^{(\log s)^{\varepsilon}} + O(\log s).
    \]
    By the definition of $m$ we have $(\log s)^{1/(d+1)} = \frac{c-1}{2}\log m$, and therefore $(\log s)^{\varepsilon} = O\bigl(\log^{\gamma/2} m\bigr)$. Hence, $\rnKt(u) < 2^{\log^{\gamma} m}$ for all sufficiently large $m$, so $u \notin L$ and
    $C^{\mathcal{O}}_m(u) = 0$. As this holds for every block $i \in [k]$, we get $D(x) = 0$ for every
    $x \in T_s$, and consequently, for all $s \ge s_0$,
    \[
        \Prob_{x \sim T_s}[D(x) = 1] = 0.
    \]

    Pick $s_1 = 2^{c/2 \cdot \log^{d + 1}N}$, we can verify that for any $s \ge s_1$, 
    \[
        m = 2^{(2/(c - 1) \cdot \log s)^{1/(d + 1)}} \ge 2^{(c/(c - 1))^{1/(d + 1)} \cdot \log N} > N.
    \]
    In this case, $|L \cap \{0, 1\}^m| \ge 2^m/m^c$ by our assumption on $L$. Furthermore, we have $km = m^{(c + 1)/2} = 2^{\log^{1/(d + 1)}s \cdot F(c)}$ where $F(c) = (2/(c - 1))^{1/(d + 1)} \cdot (c + 1)/2$ is universally constant, so we pick $s_2 \vcentcolon= \max\{16, 2^{F(c)^4}\}$, which gives
    \begin{align*}
        km = 2^{\log^{1/(d + 1)}s \cdot F(c)} \le 2^{\log^{1/(d + 1)}s \cdot \log^{1/4}s} \le 2^{\log^{3/4}s} < \frac{s}{2}.
    \end{align*}
    for all $s \ge s_2$ (the second last inequality follows from $d \ge 1$, while the last follows from $s \ge 16$). Now for all $s \ge \max\{s_1, s_2\}$ and $x \in \{0, 1\}^s$,
    \[  
        D(x) = \bigvee_{i = 1}^k C^\mathcal{O}_m(x_{[(i - 1)m + 1 \ \vcentcolon \ im]}). 
    \]
    Therefore,
    \begin{align*}
        \Prob_{x \sim \{0, 1\}^s} [D(x) = 0] & = \Prob_{x \sim \{0, 1\}^s} [x_{[1 \ \vcentcolon \ m]} \notin L, \cdots, x_{[(k - 1)m + 1 \ \vcentcolon \ km]} \notin L] \\
        & = \left( \Prob_{y \sim \{0, 1\}^m} [y \notin L] \right)^k \le \left(1 - \frac{1}{m^c} \right)^k \\
        & < \exp\left(-\frac{k}{m^c}\right).
    \end{align*}
    Given that $k = m^{(c - 1)/2}$, we have
    \[
        k = \frac{m^{c - 1}}{k} > m^c\ln\left(1 + \frac{1}{km}\right),
    \]
    immediately giving
    \[
        \exp\left(-\frac{k}{m^c}\right) < \exp\left(-m^c\ln\left(1 + \frac{1}{km}\right) \cdot \frac{1}{m^c}\right) = 1 - \frac{1}{km + 1} \le 1 - \frac{2}{s + 2} < 1 - \frac{1}{s}.
    \]

    In conclusion, for all $s \ge \max\{s_0, s_1, s_2\}$, we have
    \[
        \Prob_{x \sim \{0, 1\}^s} [D(x) = 1] - \Prob_{x \sim T_s} [D(x) = 1] > \frac{1}{s} - 0 = \frac{1}{s}.
    \]
    Finally, $D$ is an $\mathcal{O}$-oracle circuit of size at most $s$: it consists of $k$ copies of $C^{\mathcal{O}}_m$ joined by an OR gate, where $\log k = \frac{c-1}{2}\log m$ and $\log^d m$ are both $o(\log s)$. This contradicts the pseudo-randomness of $T = \{T_s\}_{s \in \mathbb{N}}$ in \Cref{t:prg-circuit}.
    
\end{proof}

\section{Unconditional Failure of Conditional SoI: Proof of \Cref{thm:failure-cond-SoI}}\label{sec:conditional_SoI}

    We say that \emph{conditional SoI holds for $\nKt$ with error $e(n)$}, if for any strings $x,y,z\in\{0,1\}^*$, we have 
\[
    \nKt(x,y\mid w)\ge\nKt(x\mid y,w)+\nKt(y\mid w)-e(|x|+|y|).
\]
Note that in this definition the error bound does not depend on $|w|$. We do not assume random access to $w$.

We need to impose the following mild conditions on the computational model in the definition of Kolmogorov complexity:\footnote{These are standard assumptions in Kolmogorov complexity and meta-complexity, and often appear only implicitly in proofs.}
\begin{itemize}
    \item If a program $M_1(w)$ runs in time $t_1$, and a program $M_2(w)$ runs in time $t_2$, then we can run both programs on input $w$ in time $O(t_1 + t_2)$.
    \item The resulting program has length at most $|M_1| + |M_2| + O(1)$.
\end{itemize}

Under these mild encoding assumptions, we establish unconditionally that SoI fails in the presence of a conditional string (\Cref{thm:failure-cond-SoI}).

To prove \Cref{thm:failure-cond-SoI}, we will need the following lemma.

\begin{lemma}[Efficient Simulation Under Conditional SoI]\label{lem:cond-SoI-imply-short-time}
    If conditional SoI holds for $\nKt$ with error $e(n) \geq \log n$, then for any $x,w\in\{0,1\}^*$, we have
    \[
        \nK^{O(2^{e(|x|)})}(x\mid w)\le\nKt(x\mid w)+O(e(|x|)).
    \]
\end{lemma}
\begin{proof} 

The proof is essentially the same as that of \Cref{lem:SoI-imply-short-time}, but with the addition of an auxiliary conditional string. We present the details here for completeness.

Let $n=|x|$. Let the $n$ bits of $x$ be $x_1x_2\cdots x_n$. For any integers $1\le i\le j\le n$, we define $x_{[i,j]}=x_i\cdots x_j$. For each $i\le j$, we define $s_{[i,j]}$ and $t_{[i,j]}$ to be the program size and running time of an optimal $\nKt$ program for $x_{[i,j]}$ conditioning on $w$ and $x_{[1,i-1]}$. If there are multiple programs achieving the optimal bound $\nKt(x_{[i,j]} \mid w)$, we take one with the smallest running time.
    
    We let $1 \leq i \leq n$ be the smallest integer such that $t_{[1,i]}\ge2^{C + 1}\cdot 2^{e(n)}$, where $C \geq 1$ is a suitable universal constant defined below. If no such $i$ exists, then the optimal $\nKt$ program for $x$ conditioning on $w$ already runs in time $O(2^{e(n)})$, and the proof is finished. 
    
    Note that we can assume without loss of generality that $1 < i < n$. Indeed, if $C$ is large enough, using that $e(n) \geq \log n$ it is easy to see that $i > 1$. The case $i = n$ can be handled by an easy modification of the argument above. 
    
    Now, by conditional symmetry of information, we have
    \begin{align}
        \nKt(x\mid w)&\ge\nKt(x_{[1,i]}\mid w)+\nKt(x_{[i+1,n]}\mid x_{[1,i]},w)-e(n)\notag\\
        &=s_{[1,i]}+\log(t_{[1,i]})+s_{[i+1,n]}+\log(t_{i+1,n})-e(n).\label{eq:cond-SoI-1}
    \end{align}
    But on the other hand, since we can print $x$ by simply running the programs in sequence, we get\footnote{Here we use the aforementioned properties of the encoding of programs in the definition of Kolmogorov complexity.}
    \begin{align}
        \nKt(x\mid w) & \le  s_{[1,i]}+s_{[i+1,n]}+O(1) + \log(O(t_{[1,i]}+t_{[i+1,n]}+n)) \notag  \\
        & \leq  s_{[1,i]}+s_{[i+1,n]}+\log(t_{[1,i]}+t_{[i+1,n]}+n) + C,  \label{eq:cond-SoI-2}
    \end{align}
    for some universal constant $C \geq 1$ that only depends on the computational model.
    Combining \Cref{eq:cond-SoI-1,eq:cond-SoI-2}, we get
    \[
        \log(t_{[1,i]}+t_{[i+1,n]}+n)\ge\log(t_{[1,i]})+\log(t_{[i+1,n]})-e(n) - C.
    \]
    By taking the exponential of both sides, and using that $e(n) \geq \log n$ and $t_{[1,i]} \geq 2^{C+1} \cdot 2^{e(n)}$,  we get
    \[
        t_{[i+1,n]}\le\frac{1+\frac{n}{t_{[1,i]}}}{\frac{1}{2^{e(n) + C}}-\frac{1}{t_{[1,i]}}}\le\frac{2}{\frac{1}{2^{e(n) + C}}-\frac{1}{t_{[1,i]}}}\le 2^{C + 2}\cdot 2^{e(n)}.
    \]
    Hence by combining the programs for $x_{[1,i-1]}$ and $x_{[i+1,n]}$, and by storing the bit $x_i$ and the position $i$, we can print $x$ in (nondeterministic) time $O(2^{e(n)})$. Therefore, it follows that
    \begin{align*}
        \nK^{O(2^{e(n)})}(x\mid w)&\le s_{[1,i-1]}+s_{[i+1,n]}+ O(\log n)\\
        &\le\nKt(x_{[1,i-1]}\mid w)+\nKt(x_{[i+1,n]}\mid x_{[1,i]}, w)+ O(\log n)\\
        &\le\nKt(x_{[1,i]}\mid w)+\nKt(x_{[i+1,n]}\mid x_{[1,i]}, w)+ O(\log n)\\
        &\le\nKt(x\mid w)+O(e(n)),\tag{By conditional SoI for $\nKt$}
    \end{align*}
    as desired.
\end{proof}

We will use the upper bound on the running time provided by \Cref{lem:cond-SoI-imply-short-time} to significantly speed-up the computation of $\nKt$, leading to a contradiction to an existing complexity lower bound. We need the following result.

\begin{theorem}[Hardness of Approximating $\nKt$; Corollary of   \Cref{thm:intro_nKt_lb}]\label{thm:lb_nkt}
   There exists $\varepsilon > 0$ such that the following result holds:
    $$
\mathsf{Gap}\text{-}\MnKtP[0.6n,0.9n] \notin \DTIME[2^{\varepsilon \cdot n}].
    $$
\end{theorem}

\begin{proof}
    It is enough to note that if $\mathsf{Gap}\text{-}\MnKtP[0.6n,0.9n] \in \DTIME[2^{\varepsilon \cdot n}]$, then there would be a dense language $L$ containing only strings of $\nKt$ complexity at least $0.6n$ such that $L \in \NTIME[2^{\varepsilon \cdot n}] \cap \coNTIME[2^{\varepsilon \cdot n}]$. This directly contradicts the ``moreover'' part of \Cref{thm:intro_nKt_lb}.
\end{proof}

We are now ready to prove the main result of this section.

\begin{proof}[Proof of \Cref{thm:failure-cond-SoI}]
    Suppose that conditional SoI holds for $\nKt$ with error $e(n)\le\polylog(n)$. We can  assume without loss of generality that $e(n) \geq \log n$. We show how to estimate $\nKt(x)$ with additive error  $n/10$  in deterministic time $2^{2^{\poly(\log\log(n))}}$. This contradicts the complexity lower bound of \Cref{thm:lb_nkt}. 
    
    We set $k=\lceil\frac{n}{100e(n)}\rceil$, and split $x$ into $k$ pieces $x=x_1\circ x_2\circ\cdots\circ x_k$, where each $x_i$ is of length at most $100e(n)$.  By \Cref{lem:cond-SoI-imply-short-time}, we have
    \[
        \nK^{O(2^{e(100e(n))})}(x_i\mid x_1\circ\cdots\circ x_{i-1})\le\nKt(x_i\mid x_1\circ\cdots\circ x_{i-1})+O(e(100e(n))).
    \]
    Therefore, for each $i$, instead of computing $\nKt(x_i\mid x_1\circ\cdots\circ x_{i-1})$ using an exhaustive search over small nondeterministic programs and assuming access to an $\NP$ oracle, we can compute $\nK^{O(2^{e(100e(n))})}(x_i\mid x_1\circ\cdots\circ x_{i-1})$ deterministically by brute-forcing over all small nondeterministic programs and over all nondeterministic choices, which takes time $2^{100e(n)} \cdot 2^{O(2^{e(100e(n))})}$. The total running time is at most
    \[
        n\cdot 2^{100e(n)}\cdot 2^{O(2^{e(100e(n))})}\le 2^{2^{\poly(\log\log(n))}}.
    \]
    By conditional SoI with error $e(n)$ (applied $k$ times) and \Cref{lem:cond-SoI-imply-short-time}, the total estimation error is at most
    \[
        k\cdot(2e(n)+O(e(100e(n))))\le \frac{n}{10}.
    \]
    This contradicts \Cref{thm:lb_nkt}, as alluded to above.
\end{proof}

\bibliographystyle{alpha}	
\bibliography{references}

@article{DBLP:journals/jcss/BabaiM88,
  author       = {L{\'{a}}szl{\'{o}} Babai and
                  Shlomo Moran},
  title        = {Arthur-Merlin Games: {A} Randomized Proof System, and a Hierarchy
                  of Complexity Classes},
  journal      = {J. Comput. Syst. Sci.},
  volume       = {36},
  number       = {2},
  pages        = {254--276},
  year         = {1988}
}

@inproceedings{DBLP:conf/random/Rudich97,
  author       = {Steven Rudich},
  title        = {Super-bits, Demi-bits, and {NP}/qpoly-natural Proofs},
  booktitle    = {International Workshop on Randomization and Approximation Techniques \emph{(RANDOM)}},
  pages        = {85--93},
  year         = {1997}
}

@article{ECCC2025TR215,
  author     = {Halley Goldberg and Jinqiao Hu and Zhenjian Lu and Jingyi Lyu and Igor C. Oliveira},
  title      = {Synergies Between Complexity Theory and Nondeterministic {K}olmogorov Complexity},
  journal    = {Electron. Colloquium Comput. Complex.},
  volume     = {TR25-215},
  year       = {2025},
  url        = {https://eccc.weizmann.ac.il/report/2025/215/}
}

@inproceedings{DBLP:conf/approx/Hirahara020,
  author       = {Shuichi Hirahara and
                  Osamu Watanabe},
  title        = {On Nonadaptive Security Reductions of Hitting Set Generators},
  booktitle    = {International Workshop on Randomization and Approximation Techniques \emph{(RANDOM)}},
  pages        = {15:1--15:14},
  year         = {2020}
}

@inproceedings{HMO26,
  author      = {Hu, Jinqiao and Manor, Yahel and Oliveira, Igor C.},
  title       = {Failure of Symmetry of Information for Randomized Computations},
  booktitle = {Symposium on Theory of Computing \emph{(STOC)}},
  year        = {2026}
}

@inproceedings{DBLP:conf/crypto/LiP25,
  author       = {Yanyi Liu and
                  Rafael Pass},
  title        = {Hardness Along the Boundary Towards One-Way Functions from the Worst-Case Hardness of Time-Bounded {K}olmogorov Complexity},
  booktitle    = {Annual Cryptology Conference \emph{(CRYPTO)}},
  year         = {2025}
}

@inproceedings{DBLP:conf/crypto/LiuP23,
  author       = {Yanyi Liu and
                  Rafael Pass},
  title        = {One-Way Functions and the Hardness of (Probabilistic) Time-Bounded
                  {K}olmogorov Complexity w.r.t. Samplable Distributions},
  booktitle    = {Annual Cryptology Conference \emph{(CRYPTO)}},
  pages        = {645--673},
  year         = {2023}
}

@article{DBLP:journals/jc/Sudan97,
  author       = {Madhu Sudan},
  title        = {Decoding of Reed Solomon Codes beyond the Error-Correction Bound},
  journal      = {J. Complex.},
  volume       = {13},
  number       = {1},
  pages        = {180--193},
  year         = {1997},
  url          = {https://doi.org/10.1006/jcom.1997.0439},
  doi          = {10.1006/JCOM.1997.0439},
  bibsource    = {dblp computer science bibliography, https://dblp.org}
}

@article{DBLP:journals/eatcs/LuO22,
  author       = {Zhenjian Lu and
                  Igor C. Oliveira},
  title        = {Theory and Applications of Probabilistic {K}olmogorov Complexity},
  journal      = {Bull. {EATCS}},
  volume       = {137},
  year         = {2022},
  url          = {http://bulletin.eatcs.org/index.php/beatcs/article/view/700},
  bibsource    = {dblp computer science bibliography, https://dblp.org}
}

@inproceedings{DBLP:conf/icalp/LuO21,
  author    = {Zhenjian Lu and
               Igor C. Oliveira},
  title     = {An Efficient Coding Theorem via Probabilistic Representations and
               Its Applications},
  booktitle = {International Colloquium on Automata, Languages, and Programming \emph{(ICALP)}},
  pages     = {94:1--94:20},
  year      = {2021}
}

@inproceedings{DBLP:conf/stoc/Hirahara23,
  author       = {Shuichi Hirahara},
  title        = {Capturing One-Way Functions via {NP}-Hardness of Meta-Complexity},
  booktitle    = {Symposium on Theory of Computing \emph{(STOC)}},
  pages        = {1027--1038},
  year         = {2023}
}

@article{GK22,
  author    = {Halley Goldberg and Valentine Kabanets},
  title     = {A Simpler Proof of the Worst-Case to Average-Case Reduction for
Polynomial Hierarchy via Symmetry of Information},
  journal   = {Electron. Colloquium Comput. Complex.},
  volume    = {7},
  pages     = {1--14},
  year      = {2022},
  url       ={https://eccc.weizmann.ac.il/report/2022/007/}
}

@inproceedings{GKLO22,
title     = {Probabilistic {K}olmogorov Complexity with Applications to Average-Case
               Complexity},
               author    = {Halley Goldberg and
               Valentine Kabanets and
               Zhenjian Lu and
               Igor C. Oliveira},
  booktitle = {Computational Complexity Conference \emph{(CCC)}},
  pages     = {16:1--16:60},
  year      = {2022}
}

@inproceedings{DBLP:conf/focs/Hirahara18,
  author    = {Shuichi Hirahara},
  title     = {Non-Black-Box Worst-Case to Average-Case Reductions within {NP}},
  booktitle = {Symposium on Foundations of Computer Science \emph{(FOCS)}},
  pages     = {247--258},
  year      = {2018}
}

@inproceedings{DBLP:conf/icalp/Oliveira19,
  author    = {Igor C. Oliveira},
  title     = {Randomness and Intractability in {K}olmogorov Complexity},
  booktitle = {International Colloquium on Automata, Languages, and Programming \emph{(ICALP)}},
  pages     = {32:1--32:14},
  year      = {2019}
}

@book{DBLP:series/txcs/LiV19,
	author    = {Ming Li and
	Paul M. B. Vit{\'{a}}nyi},
	title     = {An Introduction to Kolmogorov Complexity and Its Applications, 4th
	Edition},
	series    = {Texts in Computer Science},
	publisher = {Springer},
	year      = {2019},
	url       = {https://doi.org/10.1007/978-3-030-11298-1},
	doi       = {10.1007/978-3-030-11298-1},
	isbn      = {978-3-030-11297-4},
	bibsource = {dblp computer science bibliography, https://dblp.org}
}

@article{DBLP:journals/cc/BuhrmanLM05,
  author    = {Harry Buhrman and
               Troy Lee and
               Dieter van Melkebeek},
  title     = {Language compression and pseudorandom generators},
  journal   = {Comput. Complex.},
  volume    = {14},
  number    = {3},
  pages     = {228--255},
  year      = {2005}
}

@article{zvonkin1970complexity,
  title={The complexity of finite objects and the algorithmic concepts of randomness and information},
  author={Alexander K. Zvonkin and Leonid A. Levin},
  journal={UMN (Russian Math. Surveys)},
  volume={25},
  number={6},
  pages={83--124},
  year={1970}
}

@inproceedings{DBLP:conf/coco/Hirahara22,
  author    = {Shuichi Hirahara},
  title     = {Symmetry of Information from Meta-Complexity},
  booktitle = {Computational Complexity Conference \emph{(CCC)}},
  pages     = {26:1--26:41},
  year      = {2022},
  url       = {https://doi.org/10.4230/LIPIcs.CCC.2022.26},
  doi       = {10.4230/LIPIcs.CCC.2022.26},
  bibsource = {dblp computer science bibliography, https://dblp.org}
}

@article{DBLP:journals/tcs/LeeR05,
  author    = {Troy Lee and
               Andrei E. Romashchenko},
  title     = {Resource bounded symmetry of information revisited},
  journal   = {Theor. Comput. Sci.},
  volume    = {345},
  number    = {2-3},
  pages     = {386--405},
  year      = {2005},
  url       = {https://doi.org/10.1016/j.tcs.2005.07.017},
  doi       = {10.1016/j.tcs.2005.07.017},
  bibsource = {dblp computer science bibliography, https://dblp.org}
}

@article{DBLP:journals/jcss/AllenderKRR11,
  author    = {Eric Allender and
               Michal Kouck{\'{y}} and
               Detlef Ronneburger and
               Sambuddha Roy},
  title     = {The pervasive reach of resource-bounded {K}olmogorov complexity in computational
               complexity theory},
  journal   = {J. Comput. Syst. Sci.},
  volume    = {77},
  number    = {1},
  pages     = {14--40},
  year      = {2011}
}

@inproceedings{DBLP:conf/approx/GoldbergK24,
  author       = {Halley Goldberg and
                  Valentine Kabanets},
  title        = {Consequences of Randomized Reductions from {SAT} to Time-Bounded {K}olmogorov
                  Complexity},
  booktitle    = {Approximation, Randomization, and Combinatorial Optimization: Algorithms
                  and Techniques \emph{(APPROX/RANDOM)}},
  pages        = {51:1--51:19},
  year         = {2024}
}

@inproceedings{DBLP:conf/coco/Santhanam23,
  author       = {Rahul Santhanam},
  title        = {An Algorithmic Approach to Uniform Lower Bounds},
  booktitle    = {Computational Complexity Conference \emph{(CCC)}},
  pages        = {35:1--35:26},
  year         = {2023}
}

@inproceedings{DBLP:conf/focs/LuORS24,
  author       = {Zhenjian Lu and
                  Igor C. Oliveira and
                  Hanlin Ren and
                  Rahul Santhanam},
  title        = {On the Complexity of Avoiding Heavy Elements},
  booktitle    = {Symposium on Foundations of Computer Science \emph{(FOCS)}},
  pages        = {2403--2412},
  year         = {2024}
}

@phdthesis{TroyLeeThesis,
author={Troy Lee},
title={Kolmogorov complexity and formula lower bounds},
school={University of Amsterdam},
year={2006}
}

@inproceedings{DBLP:conf/tcc/HiraharaLO24,
  author       = {Shuichi Hirahara and
                  Zhenjian Lu and
                  Igor C. Oliveira},
  title        = {One-Way Functions and {pKt} Complexity},
  booktitle    = {Theory of Cryptography \emph{(TCC)}},
  pages        = {253--286},
  year         = {2024}
}

@inproceedings{DBLP:conf/stoc/OliveiraS17,
  author    = {Igor C. Oliveira and
               Rahul Santhanam},
  title     = {Pseudodeterministic constructions in subexponential time},
  booktitle = {Symposium on Theory of Computing \emph{(STOC)}},
  pages     = {665--677},
  year      = {2017}
}

@article{DBLP:journals/iandc/Levin84,
  author    = {Leonid A. Levin},
  title     = {Randomness Conservation Inequalities; Information and Independence
               in Mathematical Theories},
  journal   = {Information and Control},
  volume    = {61},
  number    = {1},
  pages     = {15--37},
  year      = {1984}
}

@article{DBLP:journals/cc/BabaiFNW93,
  author    = {L{\'{a}}szl{\'{o}} Babai and
               Lance Fortnow and
               Noam Nisan and
               Avi Wigderson},
  title     = {{BPP} Has Subexponential Time Simulations Unless {EXPTIME} has Publishable
               Proofs},
  journal   = {Computational Complexity},
  volume    = {3},
  pages     = {307--318},
  year      = {1993}
}

@inproceedings{DBLP:conf/focs/LiuP20,
  author    = {Yanyi Liu and
               Rafael Pass},
  title     = {On One-way Functions and {K}olmogorov Complexity},
  booktitle = {Symposium on Foundations of Computer Science \emph{(FOCS)}},
  pages     = {1243--1254},
  year      = {2020}
}

@article{DBLP:journals/iandc/LongpreW95,
  author    = {Luc Longpr{\'{e}} and
               Osamu Watanabe},
  title     = {On Symmetry of Information and Polynomial Time Invertibility},
  journal   = {Inf. Comput.},
  volume    = {121},
  number    = {1},
  pages     = {14--22},
  year      = {1995},
  url       = {https://doi.org/10.1006/inco.1995.1120},
  doi       = {10.1006/inco.1995.1120},
  bibsource = {dblp computer science bibliography, https://dblp.org}
}

@inproceedings{DBLP:conf/coco/HiraharaKLO24,
  author       = {Shuichi Hirahara and
                  Valentine Kabanets and
                  Zhenjian Lu and
                  Igor C. Oliveira},
  title        = {Exact Search-To-Decision Reductions for Time-Bounded {K}olmogorov Complexity},
  booktitle    = {Computational Complexity Conference \emph{(CCC)}},
  pages        = {29:1--29:56},
  year         = {2024}
}

@article{DBLP:journals/siamcomp/KlivansM02,
	author    = {Adam R. Klivans and
	Dieter van Melkebeek},
	title     = {Graph Nonisomorphism Has Subexponential Size Proofs Unless the Polynomial-Time
	Hierarchy Collapses},
	journal   = {{SIAM} J. Comput.},
	volume    = {31},
	number    = {5},
	pages     = {1501--1526},
	year      = {2002}
}

@inproceedings{DBLP:conf/coco/CarmosinoIKK16,
	author    = {Marco L. Carmosino and
	Russell Impagliazzo and
	Valentine Kabanets and
	Antonina Kolokolova},
	title     = {Learning Algorithms from Natural Proofs},
	booktitle = {Conference on Computational Complexity \emph{(CCC)}},
	pages     = {10:1--10:24},
	year      = {2016}
}

@inproceedings{DBLP:conf/stoc/Hirahara21,
  author    = {Shuichi Hirahara},
  title     = {Average-case hardness of {NP} from exponential worst-case hardness
               assumptions},
  booktitle = {Symposium on Theory of Computing \emph{(STOC)}},
  pages     = {292--302},
  year      = {2021}
}

@article{SU06,
  author       = {Ronen Shaltiel and
                  Christopher Umans},
  title        = {Pseudorandomness for Approximate Counting and Sampling},
  journal      = {Comput. Complex.},
  volume       = {15},
  number       = {4},
  pages        = {298--341},
  year         = {2006},
  url          = {https://doi.org/10.1007/s00037-007-0218-9},
  doi          = {10.1007/S00037-007-0218-9},
  bibsource    = {dblp computer science bibliography, https://dblp.org}
}

@article{DBLP:journals/siamcomp/AllenderBKMR06,
	author    = {Eric Allender and
	Harry Buhrman and
	Michal Kouck{\'{y}} and
	Dieter van Melkebeek and
	Detlef Ronneburger},
	title     = {Power from Random Strings},
	journal   = {{SIAM} J. Comput.},
	volume    = {35},
	number    = {6},
	pages     = {1467--1493},
	year      = {2006}
}

@article{DBLP:journals/jcss/ImpagliazzoW01,
	author    = {Russell Impagliazzo and
	Avi Wigderson},
	title     = {Randomness vs Time: Derandomization under a Uniform Assumption},
	journal   = {J. Comput. Syst. Sci.},
	volume    = {63},
	number    = {4},
	pages     = {672--688},
	year      = {2001}
}

@article{DBLP:journals/cc/TrevisanV07,
	author    = {Luca Trevisan and
	Salil P. Vadhan},
	title     = {Pseudorandomness and Average-Case Complexity Via Uniform Reductions},
	journal   = {Computational Complexity},
	volume    = {16},
	number    = {4},
	pages     = {331--364},
	year      = {2007}
}

@article{NW94,
	author    = {Noam Nisan and
	Avi Wigderson},
	title     = {Hardness vs Randomness},
	journal   = {J. Comput. Syst. Sci.},
	volume    = {49},
	number    = {2},
	pages     = {149--167},
	year      = {1994}
}

@inproceedings{IW97,
	author    = {Russell Impagliazzo and
	Avi Wigderson},
	title     = {{P = BPP} if {E} Requires Exponential Circuits: Derandomizing
	the {XOR} Lemma},
	booktitle = {Symposium on Theory of Computing \emph{(STOC)}},
	pages     = {220--229},
	publisher = {{ACM}},
	year      = {1997}
}

@article{DBLP:journals/tcs/Yap83,
  author       = {Chee{-}Keng Yap},
  title        = {Some Consequences of Non-Uniform Conditions on Uniform Classes},
  journal      = {Theor. Comput. Sci.},
  volume       = {26},
  pages        = {287--300},
  year         = {1983}
}

@inproceedings{DBLP:journals/eccc/KK25,
  author       = {Valentine Kabanets and Antonina Kolokolova},
  title        = {Chain Rules for Time-Bounded {K}olmogorov Complexity},
  booktitle = {Symposium on Theory of Computing \emph{(STOC)}},
  year        = {2026}
}

@article{DBLP:journals/eccc/HN25,
  author       = {Shuichi Hirahara and
                  Mikito Nanashima},
  title        = {Complexity-Theoretic Inductive Inference},
  journal      = {Electron. Colloquium Comput. Complex.},
  volume       = {{TR25-92}},
  year         = {2025}
}

@inproceedings{DBLP:conf/coco/GoldbergK23,
  author       = {Halley Goldberg and
                  Valentine Kabanets},
  title        = {Improved Learning from {K}olmogorov Complexity},
  booktitle    = {Computational Complexity Conference \emph{(CCC)}},
  pages        = {12:1--12:29},
  year         = {2023}
}

@phdthesis{RonneburgerThesis,
	author={Detlef Ronneburger},
	title={Kolmogorov Complexity and Derandomization},
	school={Rutgers University},
	year={2004}
}

@article{DBLP:journals/jacm/ShaltielU05,
  author    = {Ronen Shaltiel and
               Christopher Umans},
  title     = {Simple extractors for all min-entropies and a new pseudorandom generator},
  journal   = {J. {ACM}},
  volume    = {52},
  number    = {2},
  pages     = {172--216},
  year      = {2005},
  url       = {https://doi.org/10.1145/1059513.1059516},
  doi       = {10.1145/1059513.1059516},
  bibsource = {dblp computer science bibliography, https://dblp.org}
}

@inproceedings{DBLP:conf/focs/Hirahara20,
  author    = {Shuichi Hirahara},
  title     = {Characterizing Average-Case Complexity of {PH} by Worst-Case Meta-Complexity},
  booktitle = {Symposium on Foundations of Computer Science \emph{(FOCS)}},
  pages     = {50--60},
  year      = {2020}
}

@inproceedings{DBLP:conf/cie/Allender10,
  author       = {Eric Allender},
  title        = {Avoiding Simplicity Is Complex},
  booktitle    = {Programs, Proofs, Processes -- Conference on Computability in Europe \emph{(CiE)}},
  pages        = {1--10},
  year         = {2010}
}

@inproceedings{DBLP:conf/coco/AllenderKR03,
  author       = {Eric Allender and
                  Michal Kouck{\'{y}} and
                  Detlef Ronneburger and
                  Sambuddha Roy},
  title        = {Derandomization and Distinguishing Complexity},
  booktitle    = {Conference on Computational Complexity  \emph{CCC}},
  pages        = {209--220},
  year         = {2003}
}

@inproceedings{DBLP:conf/stoc/GoldwasserS86,
  author       = {Shafi Goldwasser and
                  Michael Sipser},
  title        = {Private Coins versus Public Coins in Interactive Proof Systems},
  booktitle    = {Symposium on Theory of Computing \emph{(STOC)}},
  pages        = {59--68},
  year         = {1986}
}

@inproceedings{DBLP:conf/focs/Ilango23,
  author       = {Rahul Ilango},
  title        = {{SAT} Reduces to the Minimum Circuit Size Problem with a Random Oracle},
  booktitle    = {Symposium on Foundations of Computer Science \emph{(FOCS)}},
  pages        = {733--742},
  year         = {2023}
}

@misc{SimonsMetaComplexity2023,
  title        = {Meta-{C}omplexity},
  howpublished = {Research Program at the Simons Institute for the Theory of Computing, UC Berkeley},
  note         = {\url{https://simons.berkeley.edu/programs/Meta-Complexity2023}},
year = {Spring 2023},
  url          = {https://simons.berkeley.edu/programs/Meta-Complexity2023}
}

@article{10.1145/3744715.3744722,
author = {Mazor, Noam and Pass, Rafael},
title = {Guest Column: On Cryptography and Meta-Complexity},
year = {2025},
issue_date = {June 2025},
publisher = {Association for Computing Machinery},
volume = {56},
number = {2},
issn = {0163-5700},
url = {https://doi.org/10.1145/3744715.3744722},
doi = {10.1145/3744715.3744722},
journal = {SIGACT News},
month = jun,
pages = {63–101},
numpages = {39}
}

@article{DBLP:journals/eatcs/Hirahara22,
  author    = {Shuichi Hirahara},
  title     = {Meta-Computational Average-Case Complexity: {A} New Paradigm Toward
               Excluding {H}euristica},
  journal   = {Bull. {EATCS}},
  volume    = {136},
  year      = {2022},
  url       = {http://bulletin.eatcs.org/index.php/beatcs/article/view/688},
  bibsource = {dblp computer science bibliography, https://dblp.org}
}

@inproceedings{DBLP:conf/icalp/LuS24,
  author       = {Zhenjian Lu and
                  Rahul Santhanam},
  title        = {Impagliazzo's Worlds Through the Lens of Conditional {K}olmogorov Complexity},
  booktitle    = {International Colloquium on Automata, Languages, and Programming \emph{(ICALP)}},
  pages        = {110:1--110:17},
  year         = {2024}
}

@article{RazRV02_jcss_journals,
  author = {Raz, Ran and Reingold, Omer and Vadhan, Salil P.},
  title = {{Extracting all the Randomness and Reducing the Error in {Trevisan's} Extractors}},
  journal = {J. Comput. Syst. Sci.},
  volume = {65},
  number = {1},
  pages = {97--128},
  year = {2002},
  doi = {10.1006/jcss.2002.1824},
  bibsource = {dblp computer science bibliography, https://dblp.org}
}

@inproceedings{DBLP:conf/stoc/HiraharaILNO23,
  author       = {Shuichi Hirahara and
                  Rahul Ilango and
                  Zhenjian Lu and
                  Mikito Nanashima and
                  Igor C. Oliveira},
  title        = {A Duality between One-Way Functions and Average-Case Symmetry of Information},
  booktitle = {Symposium on Theory of Computing \emph{(STOC)}},
  pages        = {1039--1050},
  year         = {2023}
}

@inproceedings{DBLP:conf/focs/HiraharaN23,
  author       = {Shuichi Hirahara and
                  Mikito Nanashima},
  title        = {Learning in {P}essiland via Inductive Inference},
  booktitle    = {Symposium on Foundations of Computer Science, \emph{(FOCS)}},
  pages        = {447--457},
  year         = {2023}
}

@inproceedings{DBLP:conf/focs/HiraharaLN24,
  author       = {Shuichi Hirahara and
                  Zhenjian Lu and
                  Mikito Nanashima},
title        = {Optimal Coding for Randomized {K}olmogorov Complexity and Its Applications},
	booktitle    = {Symposium on Foundations of Computer Science \emph{(FOCS)}},
  pages        = {369--378},
  year         = {2024}
}

@inproceedings{DBLP:conf/coco/Kabanets00,
  author       = {Valentine Kabanets},
  title        = {Easiness Assumptions and Hardness Tests: Trading Time for Zero Error},
  booktitle    = {Conference on Computational
                  Complexity \emph{(CCC)}},
  pages        = {150--157},
  year         = {2000},
  url          = {https://doi.org/10.1109/CCC.2000.856746},
  doi          = {10.1109/CCC.2000.856746},
  bibsource    = {dblp computer science bibliography, https://dblp.org}
}

@inproceedings{DBLP:conf/icalp/LuOZ22,
	author    = {Zhenjian Lu and
	Igor C. Oliveira and
	Marius Zimand},
	title     = {Optimal Coding Theorems in Time-Bounded {K}olmogorov Complexity},
	booktitle = {International Colloquium on Automata, Languages, and Programming \emph{(ICALP)}},
	pages     = {92:1--92:14},
	year      = {2022}
}

@inproceedings{HLO25,
  author    = {Jinqiao Hu and Zhenjian Lu and Igor C. Oliveira},
journal = {Preprint},
booktitle        = {Computational Complexity Conference \emph{(CCC)}},
  title     = {Hardness of Computing 
  Nondeterministic {Kolmogorov} Complexity},
  year      = {2026}
}

@article{DBLP:journals/siamcomp/BakerGS75,
  author       = {Theodore P. Baker and
                  John Gill and
                  Robert Solovay},
  title        = {Relativizations of the {P} =? {NP} Question},
  journal      = {{SIAM} J. Comput.},
  volume       = {4},
  number       = {4},
  pages        = {431--442},
  year         = {1975},
  url          = {https://doi.org/10.1137/0204037},
  doi          = {10.1137/0204037},
  bibsource    = {dblp computer science bibliography, https://dblp.org}
}

@article{DBLP:journals/ipl/Cook88,
  author       = {Stephen A. Cook},
  title        = {Short Propositional Formulas Represent Nondeterministic Computations},
  journal      = {Inf. Process. Lett.},
  volume       = {26},
  number       = {5},
  pages        = {269--270},
  year         = {1988},
  url          = {https://doi.org/10.1016/0020-0190(88)90152-4},
  doi          = {10.1016/0020-0190(88)90152-4},
  bibsource    = {dblp computer science bibliography, https://dblp.org}
}

@article{DBLP:journals/jcss/Fortnow00,
  author       = {Lance Fortnow},
  title        = {Time-Space Tradeoffs for Satisfiability},
  journal      = {J. Comput. Syst. Sci.},
  volume       = {60},
  number       = {2},
  pages        = {337--353},
  year         = {2000},
  url          = {https://doi.org/10.1006/jcss.1999.1671},
  doi          = {10.1006/JCSS.1999.1671},
  bibsource    = {dblp computer science bibliography, https://dblp.org}
}

@techreport{BuhrmanFortnow1995,
  author       = {Harry Buhrman and Lance Fortnow},
  title        = {Distinguishing Complexity and Symmetry of Information},
  institution  = {Department of Computer Science, The University of Chicago},
  type         = {Technical Report},
  number       = {TR-95-11},
  year         = {1995}
}

@article{DBLP:journals/eccc/ECCC-TR97-020,
  author       = {Oded Goldreich},
  title        = {A Sample of Samplers - {A} Computational Perspective on Sampling (survey)},
  journal      = {Electron. Colloquium Comput. Complex.},
  volume       = {{TR97-020}},
  year         = {1997},
  url          = {https://eccc.weizmann.ac.il/eccc-reports/1997/TR97-020/index.html},
  eprinttype    = {ECCC},
  eprint       = {TR97-020},
  bibsource    = {dblp computer science bibliography, https://dblp.org}
}

@inproceedings{DBLP:conf/coco/HiraharaLR23,
	author       = {Shuichi Hirahara and
	Zhenjian Lu and
	Hanlin Ren},
	title        = {Bounded Relativization},
	booktitle = {Computational Complexity Conference \emph{CCC}},
	pages        = {6:1--6:45},
	year         = {2023}
}

@unpublished{Khrushchev-June-2026,
  author = {Viacheslav Khrushchev and Bruno Bauwens},
  note = {Private Communication},
  year = {2026},
  month = {June}
}

@inproceedings{DBLP:conf/stoc/0001LL25,
  author       = {Lijie Chen and
                  Jiatu Li and
                  Jingxun Liang},
  title        = {Maximum Circuit Lower Bounds for Exponential-Time {A}rthur {M}erlin},
  booktitle    = {Symposium on Theory of Computing \emph{(STOC)}},
  pages        = {1348--1358},
  year         = {2025}
}

\appendix

\section{The Role of Adaptivity I: Ronneburger's Technique}\label{s:nKt_SoI}

    \cite{RonneburgerThesis} showed that Symmetry of Information (SoI) does not hold for $\Kt$. Since the argument relativizes, it follows that SoI does not hold for $\Kt^{\SAT}$ either (where the queries are allowed to be adaptive). Recall that $\nKt$ is related, within a constant, to $\Kt$ complexity with \emph{non-adaptive} queries to a $\SAT$ oracle (see \Cref{lem:ktA-leq-nkt} and \Cref{lem:nkt-leq-ktsat}). Nevertheless, the proof technique of \cite{RonneburgerThesis} does not seem to extend to show failure of SoI for $\nKt$.

To clarify why this technique falls short and to highlight the role of adaptivity, we sketch a proof that SoI fails in a hybrid setting: the left-hand side uses $\Kt$ with access to a $\SAT$ oracle with up to two levels of adaptivity, while the right-hand side terms employ $\nKt$. Whether symmetry of information fails for $\nKt$ remains an open problem.

    \begin{theorem}\label{thm:neg-soi-kt2sat-nkt}

        There is a constant $\varepsilon>0$ such that for every large enough $n\in\mathbb{N}$, there exist strings $x, y, \in \bool^n$ such that
        \[\Kt^{\parallel_2 \SAT}(x, y) < \nKt(x) + \nKt(y \mid x) - \varepsilon n.\]
        
    \end{theorem}

    Here, for an oracle $\mathcal{O}$ and a constant $d \in \N$, $\Kt^{\parallel_d\mathcal{O}}$ refers to the complexity measure $\Kt$ wherein the universal machine may make oracle calls to $\mathcal{O}$ with up to $d$ levels of adaptivity. 
    
    We will need the following lemma.

    \begin{lemma}\label{lem:neg-ktsat-soi}

       For every $n\in\mathbb{N}$, constant $d \in \N$, and language $A \in \NP$, there exist strings $x_1, ..., x_d\in\bool^{n}$ satisfying
        \[\sum_{i \in [d]}\Kt^{\parallel A}(x_i \mid x_1, ..., x_{i - 1}) > d\cdot\frac{n}{2}\]
        and        
        \[\Kt^{\parallel_d A}(x_1, ..., x_d) \leq \frac{n}{2} + O(\log n).\]
        
    \end{lemma}
    \begin{proof}

        We employ an idea from  \cite{RonneburgerThesis}. Fix $n, d \in \N$, and for $i \in [d]$, let $x_i$ be the lexicographically first string of length $n$ such that 
        \[\Kt^{\parallel A}(x_i \mid x_1, ..., x_{i - i}) > \frac{n}{2}.\] 
        Note that
        \[\sum_{i \in [d]}\Kt^{\parallel A}(x_i \mid x_1, ..., x_{i - 1}) > d\cdot\frac{n}{2}.\]

        Now, consider the following brute-force algorithm to find $x_i$ given $(x_1, ..., x_{i-1})$.
        \begin{quote}

            Let $S$ be an empty set at first. For every $j \in [n/2]$, for every $w \in \bool^{j}$, letting $t = 2^{n/2 - j}$, let $z_w$ be the output of a universal non-adaptive $A$-oracle TM run on input $(w, x_1, ..., x_{i-1})$ for $t$ steps. (If the input does not encode a program making all of its oracle queries in parallel, do not output anything.) Add the output $z_w$ to $S$. 
            
            Once the above is complete, sort $S$ lexicographically, and let $x_i$ be the lexicographically smallest string not in $S$.

            We note that when simulating the procedure described above, we can collect the set of queries to the SAT oracle requested by each machine and submit all queries simultaneously to the oracle. For this reason, non-adaptive access to the SAT oracle is sufficient.
            
        \end{quote}
        Observe that the algorithm described above runs in time at most
        \[\frac{n}{2}\cdot 2^j \cdot 2^{n/2 - j} + O(|S|\cdot\log|S|) < n^2 \cdot 2^{n/2}\]
        and requires program size at most $O(\log n)$ provided non-adaptive access to an $A$ oracle. 
        
        Moreover, repeating the above $d$ times to produce all strings $x_1, ..., x_d$ takes time at most $d \cdot n^2 \cdot 2^{n/2}$, making queries to a  $A$ oracle with at most $d$ levels of adaptivity. We obtain
        \begin{align*}
            \Kt^{\parallel_d A}(x_1, ..., x_d) &\leq O(\log n) + \log(d \cdot n^2\cdot 2^{n/2})
            \\&\leq \frac{n}{2} + O(\log n).
        \end{align*}
        This completes the proof of the lemma. 

    \end{proof}

    We are now ready to prove Theorem~\ref{thm:neg-soi-kt2sat-nkt}.

    \begin{proof}[Proof of  Theorem~\ref{thm:neg-soi-kt2sat-nkt}]

        Applying Lemma~\ref{lem:neg-ktsat-soi} with $d = 2$ and $A = \SAT$, for every $n \in \N$, we obtain strings $x, y \in \bool^n$ such that 
        \[\Kt^{\parallel_2 \SAT}(x, y) \leq \Kt^{\parallel \SAT}(x) + \Kt^{\parallel \SAT}(y \mid x) - \frac{n}{2} + O(\log n).\]
        Then, by an optimized variant of \Cref{lem:ktA-leq-nkt} with $c=1$,\footnote{This can be obtained by reducing a nondeterministic time-$t$ computation to a SAT instance of size $t \cdot \poly(\log t)$, where the reduction runs in time $t \cdot \poly(\log t)$ (see \citep{DBLP:journals/ipl/Cook88, DBLP:journals/jcss/Fortnow00}).}
        \[\Kt^{\parallel_2 \SAT}(x, y) < \nKt(x) + \nKt(y \mid x) - \frac{n}{4},\]
        as desired.
    \end{proof}

\section{The Role of Adaptivity II: Nonadaptive Reconstruction, Explicit Constructions, and Failure of SoI}\label{sec:nonadaptive}

In order to highlight the relevance of (non)adaptivity and the challenges that it poses, we show in this section that if the pseudorandom generator construction of Trevisan and Vadhan~\cite{DBLP:journals/cc/TrevisanV07} can be modified to admit a \emph{non-adaptive} reconstruction routine, then $\mathsf{SoI}$ fails for $\rnKt$. 

To this end, we first observe that the algorithm of Oliveira and Santhanam~\cite{DBLP:conf/stoc/OliveiraS17}, which builds on~\cite{DBLP:journals/cc/TrevisanV07} to obtain pseudodeterministic explicit constructions for dense $\BPP$-properties\footnote{\cite{DBLP:conf/stoc/OliveiraS17} state their result for dense $\P$-properties; however, the argument extends straightforwardly to dense $\BPP$-properties.}, (a) is $\PSPACE$-relativizing (in the sense of~\cite{DBLP:conf/coco/HiraharaLR23}), and (b) can be extended to work for dense $\pr\BPP$-properties with the aid of small advice.

\begin{theorem}\label{t:OS17-pspace}
	For any $O \in \PSPACE$, any dense $\pr\BPP^{O}$ property $P$, and any $\gamma > 0$, there exists an infinitely-often pseudodeterministic explicit construction for $P$ that runs in time $2^{n^{\gamma}}$, given oracle access to $O$ and $n^{\gamma}$ bits of advice.\footnote{In other words, if $P = (\mathcal{YES}_n,\mathcal{NO}_n)$, then there are infinitely many values of $n$ such that we can pseudodeterministically construct in subexponential time with a small amount of advice a string $x_n \notin \mathcal{NO}_n$.}
\end{theorem}

\begin{proof}[Proof Sketch]
	The proof follows closely the argument of~\cite{DBLP:conf/stoc/OliveiraS17}. We begin by constructing a pseudodeterministic PRG.
	
	\begin{claim}\label{c:pdt-PRG}
		Let $O \in \PSPACE$. For every $\varepsilon > 0$ and $b \in \mathbb{N}$, there exists a sequence $\{G_n\}_{n \in \mathbb{N}}$, where
		\[
		G_n : \bool^{n^{\varepsilon}} \to \bool^{n^b},
		\]
		computable in pseudodeterministic time $2^{O(n^{\varepsilon})}$ with oracle access to $O$, such that the following holds. For every $n^b$-time samplable distribution family $\{\mathcal{C}_n\}_{n \in \mathbb{N}}$ of Boolean $O$-oracle circuits, there exist infinitely many $n$ such that, with probability at least $1 - n^{-b}$ over $C \sim \mathcal{C}_n$, the generator $G_n$ $(n^{-b})$-fools $C$.
	\end{claim}
	
	\begin{proof}[Proof of Claim~\ref{c:pdt-PRG}]
		\renewcommand\qedsymbol{$\diamond$}
		We consider two cases.
		
		\paragraph{Case 1: $\PSPACE \subseteq \BPP^{O}$.}
		In this case, by a standard padding argument,
		\[
		\mathsf{DSPACE}\mleft[2^{O(n)}\mright] \subseteq \BPE^{O}.
		\]
		Also, by a diagonalization argument, there exists a language
		\[
		L \in \mathsf{DSPACE}\mleft[2^{O(n)}\mright]
		\]
		such that for all but finitely many input lengths $n$, the language $L$ does not admit $O$-oracle circuits of size $2^{0.9n}$. Indeed, using $2^{O(n)}$ space, one can compute the lexicographically first truth table of length $2^n$ that is not computed by any $O$-oracle circuit of size $2^{0.9n}$. This is possible since $O \in \PSPACE$, allowing such circuits to be simulated within $2^{O(n)}$ space.
		
		By the containment above, we obtain $L \in \BPE^{O}$. Such a hard language yields a pseudorandom generator with logarithmic seed length, computable in pseudodeterministic polynomial time with oracle access to $O$, and secure against non-uniform polynomial-size $O$-oracle circuits~\cite{IW97, DBLP:journals/siamcomp/KlivansM02, DBLP:conf/stoc/OliveiraS17}.
		
		\paragraph{Case 2: $\PSPACE \not\subseteq \BPP^{O}$.}
		In this case, the desired PRG follows directly from the constructions of~\cite{DBLP:journals/jcss/ImpagliazzoW01, DBLP:journals/cc/TrevisanV07} (see Theorem~\ref{t:io-PRG}).
	\end{proof}
	
	Now suppose, toward a contradiction, that the statement of the theorem is false. Then there exist $O \in \PSPACE$, a dense $\pr\BPP^{O}$ property $P = (\mathcal{YES}, \mathcal{NO})$, and some $\gamma > 0$ such that for every algorithm $A$, for all sufficiently large $n$, and for every advice string $a_n \in \bool^{n^{\gamma}}$, if $A(a_n)$ runs in pseudodeterministic time $2^{n^{\varepsilon}}$, then its canonical output lies in $\mathcal{NO}$.
	
	Let $D$ be a probabilistic polynomial-time $O$-oracle algorithm such that, except with exponentially small probability over its internal randomness, $D$ accepts every string in $\mathcal{YES}$ and rejects every string in $\mathcal{NO}$. Such a distinguisher can be obtained via standard error reduction and a union bound. We may equivalently view $D$ as a distribution over $O$-oracle circuits obtained by fixing its internal randomness. Let $A(a_n)$ be the procedure $G(a_n)$ computing the PRG from Claim~\ref{c:pdt-PRG} with $\varepsilon = \gamma/2$ and seed string fixed to $a_n$.
	
	By the assumptions above, with all but exponentially small probability over the choice of $D$, it holds that for all but finitely many $n$,
	\[
	\Pr_{A}\bigl[D(A(a_n)) = 0\bigr] \geq \tfrac{2}{3}
	\]
	for every $a_n \in \bool^{n^{\varepsilon}}$.
	
	On the other hand, since $P$ is dense, we have
	\[
	\Pr_{x \sim \bool^n}[D(x) = 1] \geq \tfrac{1}{2}.
	\]
	Thus, $D$ distinguishes the output of the PRG from the uniform distribution, contradicting Claim~\ref{c:pdt-PRG}.
\end{proof}

As a corollary of Theorem~\ref{t:OS17-pspace}, we obtain the following result.

\begin{corollary}\label{cor:explicit_adaptive_NP}
	For any $O \in \NP$, any dense $\pr\BPP^{O}$ property $P$, and any $\gamma > 0$, there exists an infinitely-often pseudodeterministic explicit construction for $P$ that runs in time $2^{n^{\gamma}}$, given oracle access to $O$ and $n^{\gamma}$ bits of advice.
\end{corollary}

It is not difficult to extend the proof of this result to the setting of quasi-polynomial time, and it is natural to ask whether one can further modify the above explicit construction so that the oracle access in both the property and the construction algorithm is \emph{non-adaptive}. By \Cref{thm:equivalence_SoI}, this would allow us to refute SoI for $\rnKt$, as explained below.

\begin{proposition}\label{prop:approach_failure}
    Suppose that for any $O \in \NP$, any dense $\pr\mathsf{BPTIME}[n^{\poly(\log n)}]^{||O}$ property $P$, and any $\gamma > 0$, there exists an infinitely-often pseudodeterministic explicit construction for $P$ that runs in time $2^{n^{\gamma}}$, given nonadaptive oracle access to $O$ and $n^{\gamma}$ bits of advice. Then symmetry of information does not hold for $\rnKt$ with error $e(n)=\poly(\log n)$.
\end{proposition}

\begin{proof}[Proof Sketch.] By \Cref{thm:equivalence_SoI}, it is enough to show that the assumption of \Cref{prop:approach_failure} implies the following statement:
\begin{quote}
            For any promise problem $P=(\mathcal{YES},\mathcal{NO})\in\pr\BPTIME[n^{\poly(\log n)}]^{||\NP}$ such that $|\mathcal{YES}\cap\{0,1\}^n|\ge2^{n-1}$, there exists $t(\cdot),\ell(\cdot):\N\to\N$ where $\ell(n)+\log t(n)<n-\poly(\log(n))$ for each $n$, and an algorithm $A(1^n, \cdot)$ such that, for infinitely many $n$, there exists advice $\alpha_n\in\{0,1\}^{\ell(n)}$ and a string $x_n\in\{0,1\}^n\setminus\mathcal{NO}$, such that
            \[
            \Prob_{r\sim\bool^{t(n)}}\mleft[
	\begin{aligned}
		&\forall w\in\bool^{t(n)}, A(1^n,\alpha_n,w,r) \textrm{ outputs $x$ or $\bot$ in $t(n)$ steps}; \\
        &\exists w \in \{0, 1\}^{t(n)}, A(1^n,\alpha_n,w,r) \textrm{ outputs $x$ in $t(n)$ steps}
	\end{aligned} \mright] \geq\frac23.
            \]
        \end{quote}
   To see this, let $P = (\mathcal{YES},\mathcal{NO})$ be any such property, and consider the  pseudodeterministic algorithm $B^{|| \NP}(1^n,\beta_n,r)$ for $P$ given by the hypothesis of \Cref{prop:approach_failure}, where $\beta_n$ denotes the correct advice string on a good input length $n$, and $r$ denotes its randomness. To establish the desired statement, it is sufficient to obtain from $B^{|| \NP}(1^n,\beta_n, r)$ an algorithm $A(1^n,\alpha_n,w,r)$ as above, where $A$ also gets a nondeterministic string $w$ and a bounded-length string $\beta_n$ extending $\alpha_n$, but $A$ does not have oracle access to $||\NP$.

   From the perspective of time-bounded Kolmogorov complexity, this is essentially equivalent to a proof that $\rK^{|| \NP}$ can be simulated by $\rnK$ (with some overhead), which is the randomized analogue of the simulation of $\Kt^{||\NP}$ by $\nKt$ obtained by storing as advice the number of positive answers given by the oracle. It is possible to implement the desired simulation in the randomized setting using the techniques described in \Cref{sec:nonuniform-collapse-rnkt}.
\end{proof}

A natural attempt is to modify the proof of \Cref{c:pdt-PRG} by considering two different cases: $\PSPACE \subseteq \BPP^{\| O}$ and $\PSPACE \not\subseteq \BPP^{\| O}$. For the first case, by inspecting the proof and the PRG constructions of~\cite{IW97, DBLP:journals/siamcomp/KlivansM02, DBLP:conf/stoc/OliveiraS17}, one can still obtain a PRG that is secure against polynomial-size $O$-oracle circuits and computable in pseudodeterministic polynomial time using only \emph{non-adaptive} oracle access to $O$. Such a PRG would yield a pseudodeterministic explicit construction for dense $\pr\BPP^{\| O}$ properties, with running time $2^{n^{\gamma}}$ given \emph{non-adaptive} oracle access to $O$ and $n^{\gamma}$ bits of advice.

Therefore, it remains to handle the second case. More specifically:

\begin{problem}
	Show a non-adaptive version of Theorem~\ref{t:io-PRG}. That is, prove the following. Let $O\in\NP$. Assume $\PSPACE \not\subseteq \BPP^{\| O}$. Then for every $\varepsilon > 0$ and $b \in \mathbb{N}$, there exists a sequence $\{G_n\}_{n \in \mathbb{N}}$, where
	\[
	G_n : \bool^{n^{\varepsilon}} \to \bool^{n^b}
	\]
	is computable in time $2^{O(n^{\varepsilon})}$, such that the following holds. For every $n^b$-time samplable distribution family $\{\mathcal{C}_n\}_{n \in \mathbb{N}}$ of Boolean $O$-oracle circuits with non-adaptive oracle queries, there exist infinitely many $n$ such that, with probability at least $1 - n^{-b}$ over $C \sim \mathcal{C}_n$, the generator $G_n$ $(n^{-b})$-fools $C$.
\end{problem}

Recall that the PRG in Theorem~\ref{t:io-PRG} is obtained via the hardness--vs.--randomness paradigm. At a high level, one takes the truth table of a (presumably) hard $\PSPACE$-complete language $L$ with suitable structural properties~\cite{DBLP:journals/cc/TrevisanV07}, and plugs it into a standard PRG construction (e.g.,~\cite{NW94}) that admits a reconstruction procedure. Roughly speaking, the reconstruction property asserts that if there exists an efficient distinguisher (e.g., a $\pr\BPP^{\| O}$ algorithm) that breaks the security of the PRG, then one can use this distinguisher to obtain a circuit for computing $L$, thereby contradicting the assumed hardness of $L$. 

However, to make this reconstruction \emph{uniform}, the reconstruction algorithm must also be able to answer oracle queries to the language $L$ itself. To address this issue, one exploits the \emph{downward self-reducibility} of the special language $L$. Informally, this property allows one to compute $L$ on inputs of a given length by making oracle queries to $L$ on strictly smaller input lengths. Using this, one can iteratively construct a sequence of bounded-size circuits, starting from small input lengths and building up to the desired length.

Now suppose that $L$ is hard against $\pr\BPP^{\| O}$. In order to contradict this assumed hardness, one must ensure that the above reconstruction procedure uses only non-adaptive queries to $O$. However, the oracle queries arising from this iterative construction are inherently sequential: at the $i$-th stage, we build a bounded-size circuit $C_i$ that computes $L$ on inputs of length $i$, where the construction of $C_i$, while making non-adaptive queries to $O$, requires the algorithm to query the previously constructed circuit $C_{i-1}$. Consequently, the overall pattern of oracle access is necessarily adaptive. This presents a fundamental obstacle to obtaining a non-adaptive version of the construction.

Finally, we mention that existing results strongly suggest that if a non-adaptive reconstruction procedure exists,  it must employ non-black-box techniques (unless $\PSPACE \subseteq \AM \cap \coAM$). We refer to 
\citep{DBLP:conf/approx/Hirahara020} and reference therein for more details and related discussions.

\end{document}